\newif\iflongversion
\longversiontrue

\documentclass[10pt,twoside]{article}
\usepackage{fullpage}
\usepackage{times}
\usepackage{doi}

\usepackage{graphicx}
\usepackage{subcaption}

\usepackage{fancyhdr}
\usepackage{booktabs}
\usepackage{graphicx}
\usepackage{multirow}
\usepackage{hyperref}

\usepackage{xcolor}
\usepackage{paralist}
\usepackage{wrapfig}

\usepackage{etoolbox}
\BeforeBeginEnvironment{wrapfigure}{\setlength{\intextsep}{0pt}}

\usepackage{bm}
\usepackage[bbsets,Dfprime,setrelation]{math}

\usepackage{wasysym}
\usepackage[prefixflatinterpret,bracketmodalinterpret,fixformat,silentconst,sidenotecalculus,longseqcontext,seqinsist]{logic}
\usepackage[prefixflatinterpret,bracketmodalinterpret,fixformat,silentconst,differentialdL,simplenames]{dL}
\def\leftrule{L}%
\def\rightrule{R}%
\newcommand{\bebecomes}{\mathrel{::=}}
\newcommand{\alternative}{~|~}

\usepackage[proofatend,proofonly]{endproof2}

\newcommand{\fvarA}{\phi}
\newcommand{\fvarB}{\psi}
\newcommand{\rfvar}{P}
\newcommand{\relaxineq}[1]{#1_{\scriptscriptstyle{>}}^{\scriptscriptstyle{\geq}}}
\newcommand{\strictineq}[1]{#1_{\scriptscriptstyle{\geq}}^{\scriptscriptstyle{>}}}
\newcommand{\invvar}{I}
\newcommand{\rcfvar}{C}
\newcommand{\rgvar}{G}
\newcommand{\rrfvar}{R}
\newcommand{\rsfvar}{S}
\newcommand{\rsfvarhat}{\widetilde{S}}
\newcommand{\rtfvar}{T}

\newsavebox{\Lightningval}%
\sbox{\Lightningval}{\mbox{\lightning}}

\newsavebox{\Rval}%
\sbox{\Rval}{$\scriptstyle\mathbb{R}$}
\relax

  \newdimen\linferenceRulehskipamount%
  \newdimen\lcalculuscollectionvskipamount%
\renewcommand{\linferenceRuleNameSeparation}{\hspace{4pt}}

\definecolor{vblue}{rgb}{.1,.15,.62}
\definecolor{vgreen}{rgb}{.1,.5,0}
\definecolor{vgray}{rgb}{.35,.35,.35}
\definecolor{vred}{rgb}{.7,0,0}

\renewcommand*{\irrulename}[1]{\text{\textcolor{vgray}{\upshape#1}}}%

\renewcommand{\axkey}[1]{#1}

\renewcommand*{\lie}[3][]
{\mathcal{L}_{#2}^{\ifthenelse{\equal{#1}{}}{}{^{\left(#1\right)}}}(#3)}
\renewcommand*{\lied}[3][]{\overset{\bm .}{#3}\ifthenelse{\equal{#1}{}}{}{{}^{(#1)}}}

\renewcommand{\siglied}[3][]{\overset{\bm .}{#3}{}^{\Dostar{#1}}}

\renewcommand{\Dostar}[1]{\ifthenelse{\equal{#1}{}}{(*)}{-(*)}}
\newcommand{\sigliedsai}[3][]{\siglied[#1]{#2}{#3}}

\newcommand{\cmp}{\succcurlyeq}

\newcommand{\pmc}{\preccurlyeq}

\newcommand{\bdr}[1]{\partial #1}
\renewcommand{\interior}[1]{{\mathring{#1}}}

\newcommand{\timeset}{\mathbb{T}}

\usepackage{prettyref}
\newcommand{\rref}[2][]{\prettyref{#2}}
\newrefformat{sec}{Section\,\ref{#1}}
\newrefformat{subsec}{Section\,\ref{#1}}
\newrefformat{def}{Def.\,\ref{#1}}
\newrefformat{thm}{Theorem\,\ref{#1}}
\newrefformat{prop}{Proposition\,\ref{#1}}
\newrefformat{rem}{Remark\,\ref{#1}}
\newrefformat{lem}{Lemma\,\ref{#1}}
\newrefformat{cor}{Corollary\,\ref{#1}}
\newrefformat{ex}{Example\,\ref{#1}}
\newrefformat{cex}{Counterexample\,\ref{#1}}
\newrefformat{tab}{Table\,\ref{#1}}
\newrefformat{fig}{Fig.\,\ref{#1}}
\newrefformat{case}{case\,\ref{#1}}
\newrefformat{foot}{Footnote\,\ref{#1}}
\newrefformat{app}{Appendix\,\ref{#1}}

\newcommand{\prooflink}[2]{\hyperlink{#1}{Proof} in~\rref{#2}}

\newenvironment{proofsketcha}[2][ ]{\proof[Proof Sketch (\prooflink{#2}{#1})]}{\endproof}
\newenvironment{proofsketchb}[2][ ]{\proof[\prooflink{#2}{#1}]}{\endproof}

\newtheorem{theorem}{Theorem}
\newtheorem{lemma}[theorem]{Lemma}
\newtheorem{proposition}[theorem]{Proposition}
\newtheorem{corollary}[theorem]{Corollary}
\newtheorem{conjecture}[theorem]{Conjecture}

\theoremstyle{remark}
\newtheorem{remark}{Remark}
\newtheorem{example}{Example}
\newtheorem{counterexample}[example]{Counterexample}{\itshape}{}

\definecolor{highlightred}{rgb}{.7, 0.0, 0.0}
\iflongversion
\usepackage{ulem}
\newcommand{\highlight}[1]{\textcolor{vblue}{\uline{#1}}}
\newcommand{\bluec}[1]{\textcolor{vblue}{#1}}
\newcommand{\greenc}[1]{\textcolor{vgreen}{#1}}
\newcommand{\redc}[1]{\textcolor{vred}{#1}}
\else
\newcommand{\highlight}[1]{\textcolor{vblue}{\underline{#1}}}
\fi

\newcommand{\ptermA}{p}
\newcommand{\ptermB}{q}
\newcommand{\cofterm}{g}

\renewcommand{\allvars}{\mathbb{V}}
\newcommand{\States}{\mathbb{S}}

\newcommand{\I}{\dLint[const=I,state=\omega]}

\usepackage{ upgreek }
\newcommand{\solvar}{\boldsymbol\upvarphi}
\newcommand{\timevar}{t}

\newcommand{\boundedf}{B}
\newcommand{\constt}[1]{#1()}

\newcommand{\exlinear}{\ensuremath{\alpha_l}}
\newcommand{\exnonlinear}{\ensuremath{\alpha_n}}
\newcommand{\exblowup}{\ensuremath{\alpha_b}}
\newcommand{\lnorm}[1]{{{\norm{#1}}_{\infty}}}
\newcommand{\dotp}[2]{#1 \stimes #2}

\newcommand{\bigchi}{\ensuremath{\mathcal{X}}}

\newcommand{\footnotesizeoff}{}%
\makeatletter
\newcommand{\Vast}{\bBigg@{4.8}}
\newcommand{\oset}[3][0ex]{%
  \mathrel{\mathop{#3}\limits^{
    \vbox to#1{\kern-2\ex@
    \hbox{$\scriptstyle#2$}\vss}}}}
\makeatother

\newcommand{\osetf}[2]{\oset[1.6ex]{#1}{#2}}

\newrefformat{app}{Appendix\,\ref{#1}}

\newcommand{\openset}{\mathcal{O}}

\newcommand*{\lrefine}{\longrightarrow}
\newcommand*{\limprefinechain}[1]{\osetf{#1}{\lrefine}}

\usepackage{tikz}
\usetikzlibrary{arrows}

\newcommand*\circled[1]{\tikz[baseline=(char.base)]{
            \node[shape=circle,draw,inner sep=1pt] (char) {#1};}}

\newcommand{\bform}{\circled{B}\xspace}
\newcommand{\dform}{\circled{D}\xspace}

\begin{document}

\title{An Axiomatic Approach to Existence and Liveness for Differential Equations}

\author{Yong Kiam Tan \and
Andr\'e Platzer
\thanks{
  Computer Science Department, Carnegie Mellon University, Pittsburgh, USA
  {\{yongkiat$|$aplatzer\}@cs.cmu.edu}
  }
}
\date{}

\maketitle
\thispagestyle{empty}

\begin{abstract}
This article presents an axiomatic approach for deductive verification of existence and liveness for ordinary differential equations (ODEs) with differential dynamic logic (\dL).
The approach yields proofs that the solution of a given ODE exists long enough to reach a given target region without leaving a given evolution domain.
Numerous subtleties complicate the generalization of discrete liveness verification techniques, such as loop variants, to the continuous setting.
For example, ODE solutions may blow up in finite time or their progress towards the goal may converge to zero.
These subtleties are handled in \dL by successively refining ODE liveness properties using ODE invariance properties which have a complete axiomatization.
This approach is widely applicable: several liveness arguments from the literature are surveyed and derived as special instances of axiomatic refinement in \dL.
These derivations also correct several soundness errors in the surveyed literature, which further highlights the subtlety of ODE liveness reasoning and the utility of an axiomatic approach.
An important special case of this approach deduces (global) existence properties of ODEs, which are a fundamental part of every ODE liveness argument.
Thus, all generalizations of existence properties and their proofs immediately lead to corresponding generalizations of ODE liveness arguments.
Overall, the resulting library of common refinement steps enables both the sound development and justification of new ODE existence and of liveness proof rules from \dL axioms.
These insights are put into practice through an implementation of ODE liveness proofs in the \KeYmaeraX theorem prover for hybrid systems.
\end{abstract}


\section{Introduction}
\label{sec:introduction}

Hybrid systems are mathematical models describing discrete and continuous dynamics, and interactions thereof.
This flexibility makes them natural models of cyber-physical systems (CPSs) which feature interactions between discrete computational control and continuous real world physics \cite{10.2307/j.ctt17kkb0d,Platzer18}.
Formal verification of hybrid systems is of significant practical interest because the CPSs they model frequently operate in safety-critical settings.
Verifying properties of the differential equations describing the continuous dynamics present in hybrid system models is a key aspect of any such endeavor.

This article focuses on deductive verification of \emph{existence} and \emph{liveness}\footnote{%
The form of ODE liveness considered in this article is in the sense of Owicki and Lamport~\cite{DBLP:journals/toplas/OwickiL82} for concurrent programs within their (linear) temporal logic.
Liveness for ODEs has sometimes been called \emph{eventuality}~\cite{DBLP:journals/siamco/PrajnaR07,DBLP:conf/fm/SogokonJ15} and \emph{reachability}~\cite{DBLP:conf/emsoft/TalyT10}.
To minimize ambiguity, this article refers to the property as \emph{liveness}, with a precise formal definition in~\rref{sec:background}.
Other advanced notions of liveness for ODEs are discussed in~\rref{sec:relatedwork}.} properties
for ordinary differential equations (ODEs), i.e., the question whether an ODE solution exists for long enough to reach a given region without leaving its domain of evolution.
Such questions can be phrased naturally in differential dynamic logic (\dL)~\cite{DBLP:journals/logcom/Platzer10,DBLP:conf/lics/Platzer12a,DBLP:journals/jar/Platzer17,Platzer18}, a logic for \emph{deductive verification} of hybrid systems whose relatively complete axiomatization \cite{DBLP:conf/lics/Platzer12b,DBLP:journals/jar/Platzer17} lifts ODE verification results to hybrid systems, and whose theorem prover, \KeYmaeraX~\cite{DBLP:conf/cade/FultonMQVP15}, enables an implementation.

For discrete systems, methods for proving liveness are well-known: loop variants show that discrete loops eventually reach a desired goal \cite{DBLP:books/sp/Harel79}, while temporal logic is used to specify and study liveness properties in concurrent and infinitary settings~\cite{DBLP:books/daglib/0077033,DBLP:journals/toplas/OwickiL82}.
However, the deduction of (continuous) ODE liveness properties is hampered by several difficulties:
\begin{inparaenum}[\it i)]
\item solutions of ODEs may converge towards a goal without ever reaching it,
\item solutions of nonlinear ODEs may blow up in finite time leaving insufficient time for the desired goal to be reached, and
\item the goal may be reachable but only by (illegally) leaving the evolution domain constraint.
\end{inparaenum}
In contrast, invariance properties for ODEs are better understood~\cite{DBLP:journals/logcom/Platzer10,DBLP:conf/tacas/GhorbalP14,DBLP:conf/emsoft/LiuZZ11} and have a complete \dL axiomatization~\cite{DBLP:journals/jacm/PlatzerT20}.
Motivated by the aforementioned difficulties, this article presents \dL axioms enabling systematic, step-by-step refinement of ODE liveness properties, where each step is justified using an ODE invariance property.
This refinement approach is a powerful framework for understanding ODE liveness arguments because it brings the full deductive power of \dL's ODE invariance proof rules to bear on liveness proofs.
Using invariance (or safety) properties to deduce liveness is a well-known proof technique for (discrete) concurrent systems~\cite{DBLP:books/daglib/0077033,DBLP:journals/toplas/OwickiL82}.
This article shows that those ideas work just as well in the continuous setting---as long as the aforementioned difficulties are appropriately handled.

To demonstrate the applicability of the approach, this article surveys several arguments from the literature and derives them all as (corrected) \dL proof rules, see \rref{tab:survey}.
This logical presentation has two key benefits:
\begin{itemize}
\item The proof rules are \emph{syntactically derived} from sound axioms of \dL, which guarantees their correctness.
Many of the surveyed arguments contain subtle soundness errors (\rref{tab:survey}, middle and right).
These errors do not diminish the surveyed work.
Rather, they emphasize the need for an axiomatic, uniform way of presenting and analyzing ODE liveness arguments instead of relying on ad hoc approaches.
\item The approach identifies common refinement steps that form a basis for the surveyed liveness arguments drawn from various applications (\rref{tab:survey}, left).
This library of building blocks enables sound development and justification of new ODE liveness proof rules, e.g., by generalizing individual refinement steps or by exploring different combinations of those steps.
Corollaries~\ref{cor:higherdv},~\ref{cor:boundedandcompact}, and~\ref{cor:combination} are examples of new ODE liveness proof rules that can be derived and justified from the uniform approach that this article follows.
\end{itemize}

\begin{table}
\centering
\caption{Surveyed ODE liveness arguments with \highlight{highlighting in blue} for soundness-critical corrections identified in this article. The applications (and corrections, if any) for each surveyed argument is briefly described here. The referenced corollaries are corresponding derived proof rules with details of the corrections.}
\label{tab:survey}

\begin{tabular}{@{}lllll@{}}
\toprule
Application & \multicolumn{2}{l}{Without Domain Constraints} & \multicolumn{2}{l}{With Domain Constraints} \\ \midrule
Hybrid systems verification~\cite{DBLP:journals/logcom/Platzer10}  & OK                         & (Cor.~\ref{cor:atomicdvcmp})  & \highlight{if open/closed, initially false} & (Cor.~\ref{cor:atomicdvcmpQ})   \\
Auto. ODE verification~\cite{DBLP:conf/hybrid/PrajnaR05,DBLP:journals/siamco/PrajnaR07} & \multicolumn{2}{l}{\cite[Rem. 3.6]{DBLP:journals/siamco/PrajnaR07} \highlight{is incorrect}} & \highlight{if conditions checked globally} & (Cor.~\ref{cor:prq})     \\
Finding basin of attraction~\cite{DBLP:journals/siamco/RatschanS10} &  \highlight{if compact} & (Cor.~\ref{cor:rs}) & \highlight{if  compact} & (Cor.~\ref{cor:rsq}) \\
Staging set ODE liveness proofs~\cite{DBLP:conf/fm/SogokonJ15}          & OK                      & (Cor.~\ref{cor:SP}) &  OK                     & (Cor.~\ref{cor:SPQ}) \\
Switching logic synthesis~\cite{DBLP:conf/emsoft/TalyT10}         &  \highlight{if global solutions} & (Cor.~\ref{cor:tt}) & \highlight{if global solutions} & (Cor.~\ref{cor:ttq}) \\ \bottomrule
\end{tabular}
\end{table}

This article extends the authors' earlier conference version~\cite{DBLP:conf/fm/TanP19}.
The key new insight is that all of the aforementioned liveness arguments (\rref{tab:survey}) are based on reducing liveness properties of ODEs to assumptions about sufficient existence duration for their solutions.
In fact, many of those arguments become significantly simpler (and sound) when the ODEs of concern are assumed to have global solutions, i.e., they do not blow up in finite time.
It is reasonable and commonplace to make such an assumption for the continuous dynamics in models of CPSs~\cite[Section 6]{10.2307/j.ctt17kkb0d}.
After all, mechanical systems do not simply cease to exist after a short time!
An example from control theory is Lyapunov stability which guarantees global solutions near stable equilibria~\cite[Definition 4.1]{MR1201326}\cite[Theorem 3.1]{MR2381711}.
Control systems are designed to always operate near stable equilibria and so always have global solutions.
Logically though, making an \emph{a priori} assumption of global existence for ODEs means that the correctness of any subsequent verification results for the ODEs and hybrid system models are conditional on an unproved existence duration hypothesis.
While global existence is known to hold for linear systems, even the simplest nonlinear ODEs (see \rref{sec:globexist}) fail to meet the hypothesis without further assumptions.
This article therefore adopts the view that (global) existence should be \emph{proved} rather than \emph{assumed} for the continuous dynamics in hybrid systems.

The new contributions of this article beyond the authors' earlier conference version \cite{DBLP:conf/fm/TanP19} are:
\begin{itemize}
\item \rref{sec:globexist} presents deductive \dL proofs of global existence for ODE solutions.
Together with the liveness proofs of Sections~\ref{sec:nodomconstraint} and~\ref{sec:withdomconstraint}, this enables \emph{unconditional} proofs of ODE liveness properties entirely within the uniform \dL refinement framework without existence presuppositions.

\item \rref{sec:impl} shows how to apply the insights from Sections~\ref{sec:globexist}--\ref{sec:withdomconstraint} in practice.
This includes:
\begin{inparaenum}[\it i)]
\item the design of proof rules that are practically useful and well-suited for implementation (\rref{subsec:complex}) and
\item the design of proof support to aid users in existence and liveness proofs (\rref{subsec:support}).
\end{inparaenum}
\end{itemize}

The practical insights of~\rref{sec:impl} are drawn from an implementation of ODE liveness proof rules in the \KeYmaeraX theorem prover for hybrid systems~\cite{DBLP:conf/cade/FultonMQVP15}.
The unconditional liveness proofs enabled by Sections~\ref{sec:globexist}--\ref{sec:withdomconstraint} fit particularly well with an implementation in \KeYmaeraX because axiomatic refinement closely mirrors \KeYmaeraX's design principles.
\KeYmaeraX implements \dL's uniform substitution calculus~\cite{DBLP:journals/jar/Platzer17} and it is designed to minimize the soundness-critical code that has to be trusted in order to guarantee its verification results.
On top of this soundness-critical core, \KeYmaeraX's non-soundness-critical tactics framework~\cite{DBLP:conf/itp/FultonMBP17} adds support and automation for proofs.
Liveness proofs are similarly based on a series of small refinement steps which are, in turn, implemented as (untrusted) tactics based on a small basis of derived refinement axioms.
More complicated liveness arguments, such as those from~\rref{tab:survey} or from new user insights, are implemented by piecing those tactics together using tactic combinators~\cite{DBLP:conf/itp/FultonMBP17}.
The implementation required minor changes to ${\approx}155$ lines of soundness-critical code in \KeYmaeraX, while the remaining ${\approx}1500$ lines implement ODE existence and liveness proof rules as tactics.
These additions suffice to prove all of the examples in this article and in ODE models elsewhere~\cite{DBLP:conf/fm/SogokonJ15,DBLP:journals/ral/BohrerTMSP19} (\rref{subsec:autoexamples}).

Throughout this article, core \dL axioms underlying the refinement approach are presented as lemmas, which are summarized and proved in~\rref{app:coreproofs}.
Existence and liveness proof rules that are derived syntactically from those axioms, e.g.,~\rref{tab:survey}, are listed as corollaries and their derivations are given in~\rref{app:derivationproofs}.
Counterexamples explaining the soundness errors in~\rref{tab:survey} are given in~\rref{app:counterexamples}.

\section{Background: Differential Dynamic Logic}
\label{sec:background}

This section reviews the syntax and semantics of \dL, focusing on its continuous fragment which has a complete axiomatization for ODE invariants~\cite{DBLP:journals/jacm/PlatzerT20}.
Full presentations of \dL, including its discrete fragment, are available elsewhere~\cite{DBLP:conf/lics/Platzer12a,DBLP:journals/jar/Platzer17,Platzer18}.

\subsection{Syntax}
\label{subsec:syntax}

The grammar of \dL terms is as follows, where $x \in \allvars$ is a variable and $c \in \rationals$ is a rational constant.
These terms are polynomials over the set of variables $\allvars$:
\[
	\ptermA,\ptermB~\bebecomes~x \alternative c \alternative \ptermA + \ptermB \alternative \ptermA \cdot \ptermB
\]

Notably, \dL also supports term language extensions~\cite{DBLP:journals/jacm/PlatzerT20} that would enable, e.g., the use of exponentials and trigonometric functions.
These extensions are compatible with the results presented in this article, but are omitted for simplicity as they are not the main focus of this article.

The grammar of \dL formulas is as follows, where $\sim {\in}~\{=,\neq,\geq,>,\leq,<\}$ is a comparison operator and $\alpha$ is a hybrid program:
\begin{align*}
  \fvarA,\fvarB~\bebecomes&~\overbrace{\ptermA \sim \ptermB \alternative \fvarA \land \fvarB \alternative \fvarA \lor \fvarB \alternative \lnot{\fvarA} \alternative \lforall{v}{\fvarA} \alternative \lexists{v}{\fvarA}}^{\text{First-order formulas of real arithmetic }\rfvar,\ivr} \alternative \dbox{\alpha}{\fvarA} \alternative\ddiamond{\alpha}{\fvarA}
\end{align*}

The notation $\ptermA \cmp \ptermB$ (resp.~$\pmc$) is used when the comparison operator can be either $\geq$ or $>$ (resp.~$\leq$ or $<$).
Other standard logical connectives, e.g., $\limply,\lbisubjunct$, are definable as in classical logic.
Formulas not containing the modalities $\dibox{\cdot},\didia{\cdot}$ are formulas of first-order real arithmetic and are written as $\rfvar,\ivr$.
The box ($\dbox{\alpha}{\fvarA}$) and diamond ($\ddiamond{\alpha}{\fvarA}$) modality formulas express dynamic properties of the hybrid program $\alpha$.
This article focuses on \emph{continuous} programs, where $\alpha$ is given by a system of ODEs $\pevolvein{\D{x}=\genDE{x}}{\ivr}$.
Here, $\D{x}=\genDE{x}$ is an $n$-dimensional system of differential equations, $\D{x_1}=f_1(x), \dots, \D{x_n}=f_n(x)$, over variables $x = (x_1,\dots,x_n)$, where the LHS $\D{x_i}$ is the time derivative of $x_i$ and the RHS $f_i(x)$ is a polynomial over variables $x$.
The ODEs $\D{x}=\genDE{x}$ are \emph{autonomous} as they do not depend explicitly on time on the RHS~\cite{MR1201326}.
A useful transformation is to add a clock variable $t$ to the system with $\D{x}=\genDE{x,t}, \D{t}=1$ if time dependency on the RHS is needed.
The domain constraint $\ivr$ specifies the set of states in which the ODE is allowed to evolve continuously.
When there is no domain constraint, i.e., $\ivr$ is the formula $\ltrue$, the ODE is also written as $\D{x}=\genDE{x}$.
For $n$-dimensional vectors $x,y$, the dot product is $\dotp{x}{y} \mdefeq \sum_{i=1}^{n}{x_i y_i} $ and $\norm{x}^2\mdefeq \sum_{i=1}^{n}{x_i^2}$ denotes the squared Euclidean norm.
Other norms are explicitly defined in this article when used.

\begin{figure}
\centering
\includegraphics[width=0.22\textwidth]{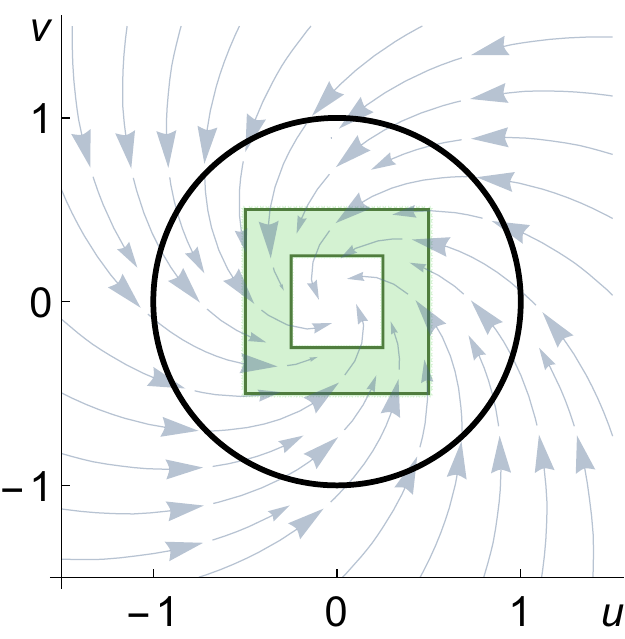}
\qquad\qquad
\includegraphics[width=0.22\textwidth]{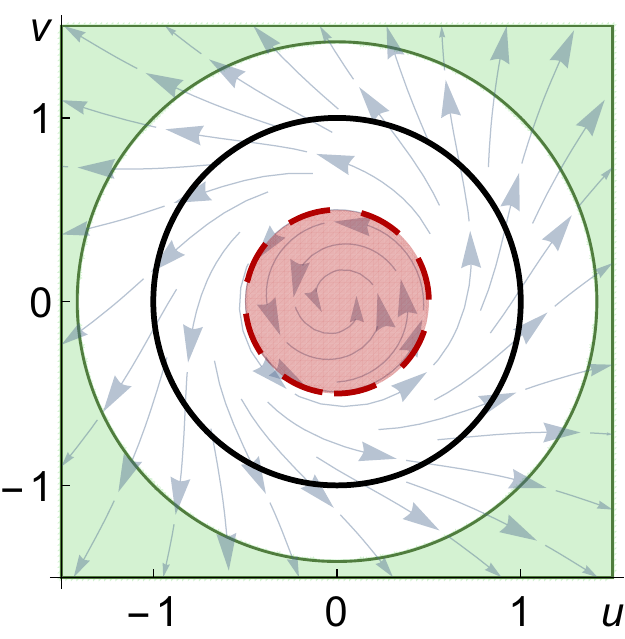}
\caption{Visualization of $\exlinear$ (left) and $\exnonlinear$ (right). Solutions of $\exlinear$ globally spiral towards the origin. In contrast, solutions of $\exnonlinear$ spiral inwards within the inner \redc{red} disk (dashed boundary), but spiral outwards otherwise. For both ODEs, solutions starting on the black unit circle eventually enter their respective shaded \greenc{green} goal regions.
The ODE $\exnonlinear$ also exhibits finite-time blow up of solutions outside the \redc{red} disk.}
\label{fig:odeexamples}
\end{figure}

\begin{figure}
\centering
\includegraphics[width=0.3\textwidth]{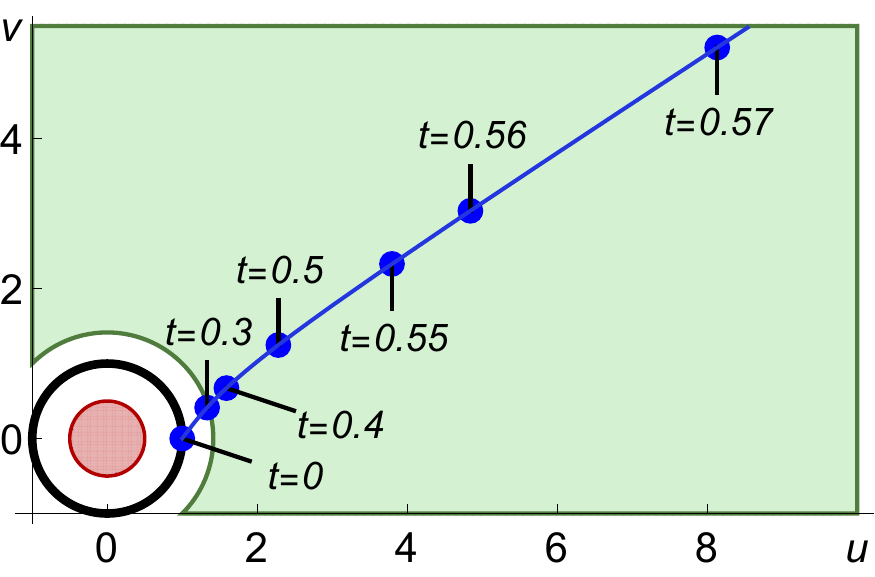}
\qquad\qquad
\includegraphics[width=0.45\textwidth]{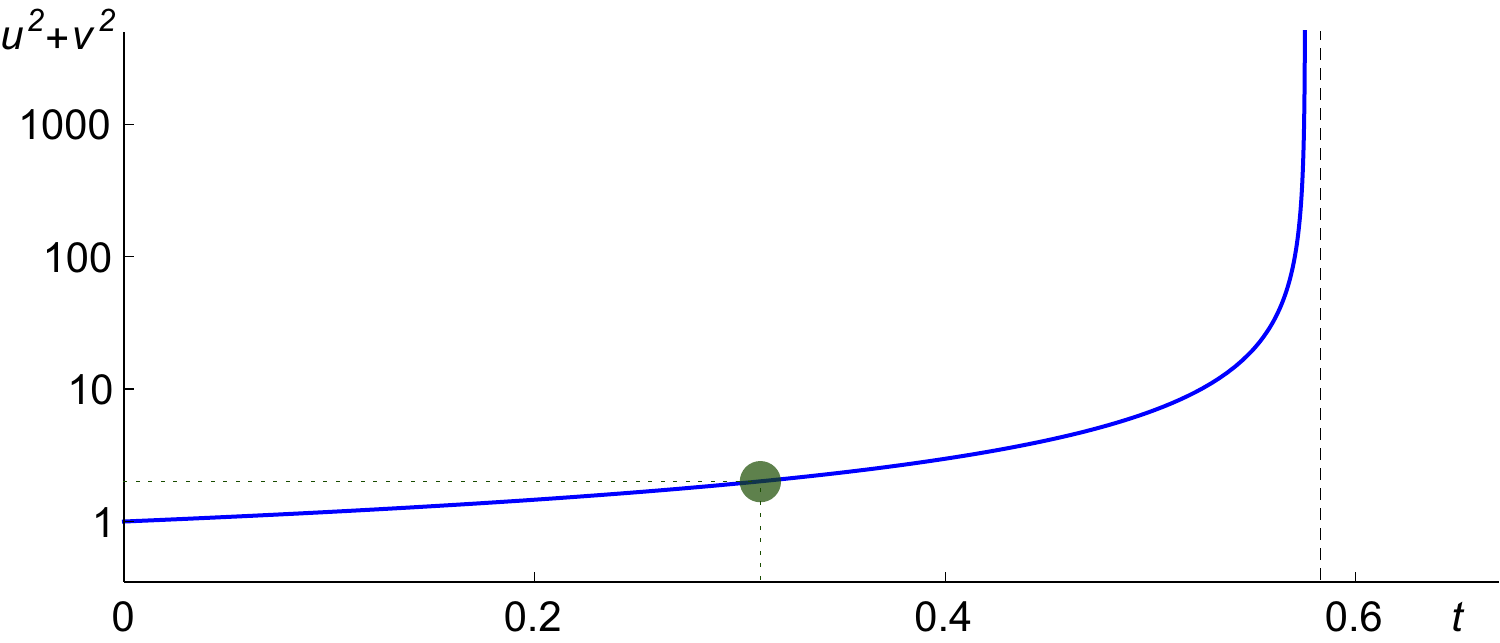}
\caption{Two views of the ODE $\exnonlinear$ evolving from initial state $u=1, v = 0$ over time $\timevar$.
The left plot shows its trajectory in the $u,v$ plane (cf.~\rref{fig:odeexamples}) while the right plot shows the squared Euclidean norm $u^2+v^2$ evolving over time $\timevar$ (with logarithmic scaling for the vertical axis).
The solution blows up in finite time with norm approaching $\infty$ as $t$ approaches $0.58$ (rounded up, black dashed asymptote).
Nevertheless, the solution reaches the \greenc{green} goal region $u^2 + v^2 \geq 2$ from~\rref{fig:odeexamples} at $t \approx 0.31$ (rounded up, \greenc{green} dot) before blowing up.}
\label{fig:odeblowup}
\end{figure}

The following two running example ODEs $\exlinear$ and $\exnonlinear$ are visualized in~\rref{fig:odeexamples} with directional arrows corresponding to their RHS evaluated at points on the plane:
\begin{align}
\exlinear    &\mnodefequiv \D{u}=-v-u,\D{v}=u-v \label{eq:exlinear} \\
\exnonlinear &\mnodefequiv \D{u}=-v-u(\frac{1}{4}-u^2-v^2),\D{v}=u-v(\frac{1}{4}-u^2-v^2) \label{eq:exnonlinear}
\end{align}

The ODE $\exlinear$ is \emph{linear} because its RHS depends linearly on $u$ and $v$ while $\exnonlinear$ is \emph{nonlinear} because of the cubic terms in its RHS.
The nonlinearity of $\exnonlinear$ results in more complex behavior for its solutions, e.g., the difference in spiraling behavior inside or outside the \redc{red} disk shown in~\rref{fig:odeexamples}.
In fact, solutions of $\exnonlinear$ blow up in finite time iff they start outside the disk characterized by $u^2+v^2 \leq \frac{1}{4}$, whereas finite-time blow up is impossible for linear ODEs like $\exlinear$~\cite{Chicone2006,Walter1998}.
An illustration of finite-time blow up for $\exnonlinear$ from an initial state outside the \redc{red} disk is shown in~\rref{fig:odeblowup}.
This phenomenon is precisely defined and investigated in~\rref{sec:globexist}, which enables formal proofs of the aforementioned (absence of) finite-time blow up.

When terms (or formulas) appear in contexts involving ODEs $\D{x}=\genDE{x}$, it is sometimes necessary to restrict the set of free variables they are allowed to mention.
In this article, these restrictions are always stated explicitly and are also indicated as arguments to terms (or formulas), e.g., $\ptermA()$ means the term $\ptermA$ does not mention any of $x_1,\dots,x_n$ as free variables, while $\rfvar(x)$ means the formula $\rfvar$ may mention all of them.
This understanding of variable dependencies is made precise using function and predicate symbols in \dL's uniform substitution calculus~\cite{DBLP:journals/jar/Platzer17}.

\subsection{Semantics}
\label{subsec:semantics}

States $\iget[state]{\I} : \allvars \to \reals$ assign real values to each variable in $\allvars$; the set of all states is written $\States$.
The semantics of polynomial term $\ptermA$ in state $\iget[state]{\I} \in \States$ is the real value $\ivaluation{\I}{\ptermA}$ of the corresponding polynomial function evaluated at $\iget[state]{\I}$.
The semantics of \dL formula $\fvarA$ is defined compositionally~\cite{DBLP:journals/jar/Platzer17,Platzer18} as the set of states $\imodel{\I}{\fvarA} \subseteq \States$ in which that formula is true.
The semantics of first-order logical connectives are defined as usual, e.g., $\imodel{\I}{\fvarA \land \fvarB} = \imodel{\I}{\fvarA} \cap \imodel{\I}{\fvarB}$.
For ODEs, the semantics of the modal operators is defined directly as follows.
Let $\iget[state]{\I} \in \States$ and $\solvar : [0, T) \to \States$ (for some $0<T\leq\infty$), be the unique solution maximally extended to the right~\cite{Chicone2006,Walter1998} for the ODE $\D{x}=\genDE{x}$ with initial value $\solvar(0)=\iget[state]{\I}$, then:
\begin{align*}
 \m{\imodels{\I}{\dbox{\pevolvein{\D{x}=\genDE{x}}{\ivr} }{\fvarA}}}~\text{iff}~&\text{for all}~0 \leq \tau < T~\text{where}~\solvar(\zeta)\,{\in}\,\imodel{\I}{\ivr}~\text{for all}~0 \leq \zeta \leq \tau\text{:}~ \solvar(\tau) \in \imodel{\I}{\fvarA} \\
 \m{\imodels{\I}{\ddiamond{\pevolvein{\D{x}=\genDE{x}}{\ivr}}{\fvarA}}}~\text{iff}~&\text{there exists}~0 \leq \tau < T~\text{such that}~
 \solvar(\zeta) \in \imodel{\I}{\ivr}~\text{for all}~0 \leq \zeta \leq \tau
 ~\text{and}~\solvar(\tau) \in \imodel{\I}{\fvarA}
\end{align*}

Informally, the \emph{safety} property $\dbox{\pevolvein{\D{x}=\genDE{x}}{\ivr}}{\fvarA}$ is true in initial state $\iget[state]{\I}$ if \emph{all} states reached by following the ODE from $\iget[state]{\I}$ while remaining in the domain constraint $\ivr$ satisfy postcondition $\fvarA$.
Dually, the \emph{liveness} property $\ddiamond{\pevolvein{\D{x}=\genDE{x}}{\ivr}}{\fvarA}$ is true in initial state $\iget[state]{\I}$ if \emph{some} state which satisfies the postcondition $\fvarA$ is eventually reached in \emph{finite} time by following the ODE from $\iget[state]{\I}$ while staying in domain constraint $\ivr$ at all times.
Figure~\ref{fig:odeexamples} suggests that formulas\footnote{$\lnorm{\cdot}$ denotes the supremum norm, with $\lnorm{x} \mnodefequiv \max_{i=1}^{n} \abs{x_i}$ for an $n$-dimensional vector $x$. The inequality $\lnorm{(u,v)} \leq \frac{1}{2}$ is expressible in first-order real arithmetic as $u^2 \leq \frac{1}{4} \land v^2 \leq \frac{1}{4}$. Similarly, $\frac{1}{4} \leq \lnorm{(u,v)}$ is expressible as $\frac{1}{16} \leq u^2 \lor \frac{1}{16}\leq v^2$.}
$\ddiamond{\exlinear}{\big(\frac{1}{4} \leq \lnorm{(u,v)} \leq \frac{1}{2}\big)}$ and
$\ddiamond{\exnonlinear}{u^2 + v^2 \geq 2}$ are true for initial states $\iget[state]{\I}$ on the unit circle.
These liveness properties are rigorously proved in Examples~\ref{ex:linproof} and~\ref{ex:nonlinproof} respectively.

Variables $y \in \allvars \setminus \{x\}$ not occurring on the LHS of ODE $\D{x}=\genDE{x}$ remain constant along solutions $\solvar : [0, T) \to \States$ of the ODE, with $\solvar(\tau)(y) = \solvar(0)(y)$ for all $\tau \in [0,T)$.
Since only the values of $x=(x_1,\dots,x_n)$ change along the solution $\solvar$, the solution may also be viewed geometrically as a trajectory in $\reals^n$, dependent on the initial values of the constant \emph{parameters} $y$.
Similarly, the values of terms and formulas depend only on the values of their free variables~\cite{DBLP:journals/jar/Platzer17}.
Thus, terms (or formulas) whose free variables are all parameters for $\D{x}=\genDE{x}$ also have provably constant (truth) values along solutions of the ODE.
For formulas $\fvarA$ that only mention free variables $x$, $\imodel{\I}{\fvarA}$ can also be viewed geometrically as a subset of $\reals^n$.
Such a formula is said to \emph{characterize} a (topologically) open (resp. closed, bounded, compact) set with respect to variables $x$ iff the set $\imodel{\I}{\fvarA} \subseteq \reals^n$ is topologically open (resp. closed, bounded, compact) with respect to the Euclidean topology.
These topological conditions are used as side conditions for some of the axioms and proof rules in this article.
In~\rref{app:topsidecalculus}, a more general definition of these side conditions is given for formulas $\fvarA$ that mention parameters $y$.
These side conditions are decidable~\cite{Bochnak1998} when $\fvarA$ is a formula of first-order real arithmetic and there are simple syntactic criteria for checking if they hold (\rref{app:topsidecalculus}).

Formula $\fvarA$ is valid iff $\imodel{\I}{\fvarA} = \States$, i.e., $\fvarA$ is true in all states.
If the formula \(\invvar \limply \dbox{\pevolvein{\D{x}=\genDE{x}}{\ivr}}{\invvar}\) is valid, the formula $\invvar$ is an \emph{invariant} of the ODE $\pevolvein{\D{x}=\genDE{x}}{\ivr}$.
Unfolding the semantics, this means that from any initial state $\iget[state]{\I}$ satisfying $\invvar$, all states reached by the solution of the ODE $\D{x}=\genDE{x}$ from $\iget[state]{\I}$ while staying in the domain constraint $\ivr$ satisfy $\invvar$.
Similarly, if the liveness formula $\rrfvar \limply \ddiamond{\pevolvein{\D{x}=\genDE{x}}{\ivr}}{\rfvar}$ is valid then, for all initial states $\iget[state]{\I}$ satisfying assumptions $\rrfvar$, the target region $\rfvar$ can be reached in finite time by following the ODE solution from $\iget[state]{\I}$ while remaining in the domain constraint $\ivr$.

\subsection{Proof Calculus}
\label{subsec:proofcalculus}
\irlabel{qear|\usebox{\Rval}}
\irlabel{notr|$\lnot$\rightrule}
\irlabel{notl|$\lnot$\leftrule}
\irlabel{orr|$\lor$\rightrule}
\irlabel{orl|$\lor$\leftrule}
\irlabel{andr|$\land$\rightrule}
\irlabel{andl|$\land$\leftrule}
\irlabel{implyr|$\limply$\rightrule}
\irlabel{implyl|$\limply$\leftrule}
\irlabel{equivr|$\lbisubjunct$\rightrule}
\irlabel{equivl|$\lbisubjunct$\leftrule}
\irlabel{id|id}
\irlabel{cut|cut}
\irlabel{weakenr|W\rightrule}
\irlabel{weakenl|W\leftrule}
\irlabel{existsr|$\exists$\rightrule}
\irlabel{existsrinst|$\exists$\rightrule}
\irlabel{alll|$\forall$\leftrule}
\irlabel{alllinst|$\forall$\leftrule}
\irlabel{allr|$\forall$\rightrule}
\irlabel{existsl|$\exists$\leftrule}
\irlabel{iallr|i$\forall$}
\irlabel{iexistsr|i$\exists$}

All derivations are presented in a classical sequent calculus with the usual rules for manipulating logical connectives and sequents.
The semantics of \emph{sequent} \(\lsequent{\Gamma}{\fvarA}\) is equivalent to the formula \((\landfold_{\fvarB \in\Gamma} \fvarB) \limply \fvarA\) and a sequent is \emph{valid} iff its corresponding formula is valid.
Completed branches in a sequent proof are marked with $\lclose$.
First-order real arithmetic is decidable~\cite{Bochnak1998} so proof steps are labeled with \irref{qear} whenever they follow from real arithmetic.
An axiom (schema) is \emph{sound} iff all its instances are valid.
Proof rules are \emph{sound} iff validity of all premises (above the rule bar) entails validity of the conclusion (below the rule bar).
Axioms and proof rules are \emph{derivable} if they can be deduced from sound \dL axioms and proof rules.
Soundness of the \dL axiomatization ensures that derived axioms and proof rules are sound~\cite{DBLP:journals/jar/Platzer17,Platzer18}.

The \dL proof calculus (briefly recalled below) is \emph{complete} for ODE invariants~\cite{DBLP:journals/jacm/PlatzerT20}, i.e., any true ODE invariant expressible in first-order real arithmetic can be proved in the calculus.
The proof rule~\irref{dIcmp} (below) uses the \emph{Lie derivative} of polynomial $\ptermA$ with respect to the ODE $\D{x}=\genDE{x}$, which is defined as the term $\lie[]{\genDE{x}}{\ptermA} \mdefeq \sum_{x_i\in x} \Dp[x_i]{\ptermA} f_i(x)$. Higher Lie derivatives $\lied[i]{\genDE{x}}{\ptermA}{}$ are defined inductively:
$ \lied[0]{\genDE{x}}{\ptermA}{} \mdefeq p, \lied[i+1]{\genDE{x}}{\ptermA}{} \mdefeq \lie{\genDE{x}}{\lied[i]{\genDE{x}}{\ptermA}{}}{}, \lied[]{\genDE{x}}{\ptermA} \mdefeq \lied[1]{\genDE{x}}{\ptermA}{}$.
Syntactically, Lie derivatives $\lied[i]{\genDE{x}}{\ptermA}$ are polynomials in the term language and they are provably definable in \dL using differentials~\cite{DBLP:journals/jar/Platzer17}.
Semantically, the value of Lie derivative $\lied[]{\genDE{x}}{\ptermA}$ is equal to the time derivative of the value of $\ptermA$ along solution $\solvar$ of the ODE $\D{x}=\genDE{x}$.

\begin{lemma}[Axioms and proof rules of \dL~\cite{DBLP:journals/jar/Platzer17,Platzer18,DBLP:journals/jacm/PlatzerT20}]
\label{lem:dlaxioms}
The following are sound axioms and proof rules of \dL.

\begin{calculuscollection}
\begin{calculus}
\cinferenceRule[diamond|$\didia{\cdot}$]{diamond axiom}
{\linferenceRule[equiv]
  {\lnot\dbox{\alpha}{\lnot{\rfvar}}}
  {\axkey{\ddiamond{\alpha}{\rfvar}}}
}
{}
\end{calculus}
\quad\quad\quad\quad
\begin{calculus}
\cinferenceRule[K|K]{K axiom / modal modus ponens} %
{\linferenceRule[impl]
  {\dbox{\alpha}{(\rrfvar \limply \rfvar)}}
  {(\dbox{\alpha}{\rrfvar}\limply\axkey{\dbox{\alpha}{\rfvar}})}
}{}
\end{calculus}\\
\begin{calculus}
\dinferenceRule[dIcmp|dI$_\cmp$]{}
{\linferenceRule
  {\lsequent{\ivr}{\lied[]{\genDE{x}}{\ptermA}\geq\lied[]{\genDE{x}}{\ptermB}}
  }
  {\lsequent{\Gamma,\ptermA \cmp \ptermB }{\dbox{\pevolvein{\D{x}=\genDE{x}}{\ivr}}{\ptermA \cmp \ptermB}} }
  \quad
}{where $\cmp$ is either $\geq$ or $>$}
\end{calculus}\\
\begin{calculus}
\dinferenceRule[dC|dC]{}
{\linferenceRule
  {\lsequent{\Gamma}{\dbox{\pevolvein{\D{x}=\genDE{x}}{\ivr}}{\rcfvar}}
  &\lsequent{\Gamma}{\dbox{\pevolvein{\D{x}=\genDE{x}}{\ivr \land \rcfvar}}{\rfvar}}
  }
  {\lsequent{\Gamma}{\dbox{\pevolvein{\D{x}=\genDE{x}}{\ivr}}{\rfvar}}}
}{}
\dinferenceRule[MbW|M${\dibox{'}}$]{}
{\linferenceRule
  {\lsequent{\ivr,\rrfvar}{\rfvar} \quad \lsequent{\Gamma}{\dbox{\pevolvein{\D{x}=\genDE{x}}{\ivr}}{\rrfvar}}}
  {\lsequent{\Gamma}{\dbox{\pevolvein{\D{x}=\genDE{x}}{\ivr}}{\rfvar}}}
}{}
\end{calculus}
\qquad
\begin{calculus}
\dinferenceRule[dW|dW]{}
{\linferenceRule
  {\lsequent{\ivr}{\rfvar}}
  {\lsequent{\Gamma}{\dbox{\pevolvein{\D{x}=\genDE{x}}{\ivr}}{\rfvar}}}
}{}
\dinferenceRule[MdW|M${\didia{'}}$]{}
{\linferenceRule
  {\lsequent{\ivr,\rrfvar}{\rfvar} \quad \lsequent{\Gamma}{\ddiamond{\pevolvein{\D{x}=\genDE{x}}{\ivr}}{\rrfvar}}}
  {\lsequent{\Gamma}{\ddiamond{\pevolvein{\D{x}=\genDE{x}}{\ivr}}{\rfvar}}}
}{}
\end{calculus}
\end{calculuscollection}
\end{lemma}
\begin{proofsketchb}[app:basecalculus]{proof:proof1}\end{proofsketchb}

Axiom~\irref{diamond} expresses the duality between the box and diamond modalities.
It is used to switch between the two in proofs and to dualize axioms between the box and diamond modalities.
Axiom~\irref{K} is the modus ponens principle for the box modality.
Differential invariants~\irref{dIcmp} say that if the Lie derivatives obey the inequality $\lied[]{\genDE{x}}{\ptermA} \geq \lied[]{\genDE{x}}{\ptermB}$, then $\ptermA \cmp \ptermB$ is an invariant of the ODE.
Differential cuts~\irref{dC} say that if one can separately prove that formula $\rcfvar$ is always satisfied along the solution, then $\rcfvar$ may be assumed in the domain constraint when proving the same for formula $\rfvar$.
In the box modality, solutions are restricted to stay in the domain constraint $\ivr$.
Thus, differential weakening~\irref{dW} says that postcondition $\rfvar$ is always satisfied along solutions if it is already implied by the domain constraint.
Using \irref{dW+K+diamond}, the final two monotonicity proof rules~\irref{MbW+MdW} for differential equations are derivable.
They strengthen the postcondition from $\rfvar$ to $\rrfvar$, assuming domain constraint $\ivr$, for the box and diamond modalities respectively.

Notice that the premises of several proof rules in~\rref{lem:dlaxioms}, e.g.,~\irref{dIcmp+dW}, discard all assumptions $\Gamma$ on initial states when moving from conclusion to premises.
This is necessary for soundness because the premises of these rules internalize reasoning that happens \emph{along} solutions of the ODE $\pevolvein{\D{x}=\genDE{x}}{\ivr}$ rather than in the initial state.
On the other hand, the truth value of constant assumptions $\constt{\rfvar}$ do not change along solutions, so they can be soundly kept across rule applications~\cite{Platzer18}.
These additional constant contexts are useful when working with assumptions on symbolic parameters, e.g., $\constt{v} > 0$ to model a (constant) positive velocity.

Besides rules~\irref{dIcmp+dC+dW}, the key to completeness for ODE invariance proofs in \dL is the \emph{differential ghosts}~\cite{DBLP:journals/jar/Platzer17,DBLP:journals/jacm/PlatzerT20} axiom shown below.
The $\exists$ quantifier in the axiom can be replaced with a $\forall$ quantifier.
\[
\cinferenceRule[DG|DG]{differential ghost}
{\linferenceRule[equiv]
  {\lexists{y}{\dbox{\pevolvein{\D{x}=\genDE{x},\D{y}=a(x)y+b(x)}{\ivr(x)}}{\rfvar(x)}}}
  {\axkey{\dbox{\pevolvein{\D{x}=\genDE{x}}{\ivr(x)}}{\rfvar(x)}}}
}{}
\]

Axiom~\irref{DG} says that, in order to prove safety postcondition $\rfvar(x)$ for the ODE $\D{x}=\genDE{x}$, it suffices to prove $\rfvar(x)$ for a larger system with an added ODE $\D{y}=a(x)y+b(x)$ that is linear in the ghost variable $y$ (because $a(x),b(x)$ do not mention $y$).
Intuitively, this addition is sound because the ODE $\D{x}=\genDE{x}$ does not mention the added variables $y$, and so the evolution of $\D{x}=\genDE{x}$ should be unaffected by the addition of an ODE for $y$.
However, this intuition only works if the additional ODEs do not unsoundly restrict the duration of the original solution by blowing up too early~\cite{DBLP:journals/jar/Platzer17}.
The linearity restriction prevents such a blow up.
Using axiom~\irref{DG} in a proof appears counterintuitive because the axiom tries to prove a property for a seemingly easier (lower-dimensional) ODE by instead studying a more difficult (higher-dimensional) one!
Yet, the~\irref{DG} axiom, is crucially used for completeness because it enables mathematical (or geometric) transformations to be carried out syntactically in the \dL proof calculus~\cite{DBLP:journals/jacm/PlatzerT20}.
This completeness result only requires a scalar version of~\irref{DG} that adds one ghost variable at time.
More general vectorial versions of the axiom (where $a(x)$ is a matrix and $b(x)$ is a vector) have also been used elsewhere~\cite{DBLP:journals/jacm/PlatzerT20}.
This article uses a new vectorial generalization that allows differential ghosts with provably bounded ODEs to be added.

\begin{lemma}[Bounded differential ghosts]
\label{lem:boundeddg}
The following bounded differential ghosts axiom~\irref{BDG} is sound, where $y=(y_1,\dots,y_m)$ is a $m$-dimensional vector of fresh variables (not appearing in $x$), $g(x,y)$ is a corresponding $m$-dimensional vector of terms, and $\norm{y}^2$ is the squared Euclidean norm of $y$. Term $\ptermA(x)$ and formulas $\rfvar(x),\ivr(x)$ are dependent only on free variables $x$ (and not $y$).

\begin{calculus}
\cinferenceRule[BDG|BDG]{}
{\linferenceRule[impll]
  {\dbox{\pevolvein{\D{x}=\genDE{x},\D{y}=g(x,y)}{\ivr(x)}}{\,\norm{y}^2 \leq \ptermA(x)}}
  {\big( \dbox{\pevolvein{\D{x}=\genDE{x}}{\ivr(x)}}{\rfvar(x)} \lbisubjunct \axkey{\dbox{\pevolvein{\D{x}=\genDE{x}, \D{y}=g(x,y)}{\ivr(x)}}{\rfvar(x)}}\big)}
}{}
\end{calculus}
\end{lemma}
\begin{proofsketchb}[app:basecalculus]{proof:proof2}\end{proofsketchb}

Like~\irref{DG}, axiom~\irref{BDG} allows an arbitrary vector of ghost ODEs $\D{y}=g(x,y)$ to be added syntactically to the ODEs.
However, it places no syntactic restriction on the RHS of the ODE (such as linearity in axiom~\irref{DG}).
For soundness,~\irref{BDG} instead adds a new precondition with a bound \(\norm{y}^2 \leq \ptermA(x)\) in terms of $x$ on the squared norm of $y$ along solutions of the augmented ODE.
This syntactic precondition ensures that $y$ cannot blow up before $x$, so that solutions of $\D{x}=\genDE{x},\D{y}=g(x,y)$ have existence intervals as long as those of the solutions of $\D{x}=\genDE{x}$.
\rref{sec:globexist} shows how to prove these preconditions so that axiom~\irref{BDG} enables ODE existence proofs through the refinement approach of~\rref{sec:livenessaxioms}.

\section{ODE Liveness via Box Refinements}
\label{sec:livenessaxioms}

This section explains step-by-step refinement for proving ODE liveness properties in \dL.
Suppose that an initial liveness property $\ddiamond{\pevolvein{\D{x}=\genDE{x}}{\ivr_0}}{\rfvar_0}$ is known for the ODE $\D{x}=\genDE{x}$.
How could this be used to prove a desired liveness property $\ddiamond{\pevolvein{\D{x}=\genDE{x}}{\ivr}}{\rfvar}$ for that ODE?
Logically, this amounts to proving:
\begin{equation}
\ddiamond{\pevolvein{\D{x}=\genDE{x}}{\ivr_0}}{\rfvar_0} \limply \ddiamond{\pevolvein{\D{x}=\genDE{x}}{\ivr}}{\rfvar}
\label{eq:refinementimpl}
\end{equation}

\begin{wrapfigure}[10]{r}{0.24\textwidth}
\centering
\includegraphics[width=0.22\textwidth]{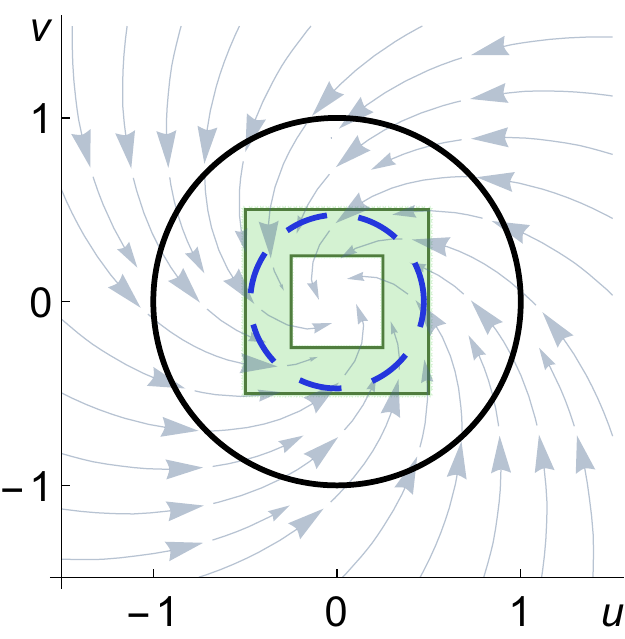}
\end{wrapfigure}

Proving implication~\rref{eq:refinementimpl} \emph{refines} knowledge of the initial liveness property to the desired liveness property.
As an example of such a refinement, consider the desired liveness property $\ddiamond{\exlinear}{\big(\frac{1}{4} \leq \lnorm{(u,v)} \leq \frac{1}{2}\big)}$ for ODE $\exlinear$~\rref{eq:exlinear} starting from the initial circle $u^2+v^2=1$ (cf.~\rref{fig:odeexamples}).
Suppose the initial liveness property $\ddiamond{\exlinear}{u^2+v^2 = \frac{1}{4}}$ is already proved, e.g., using the techniques of~\rref{sec:nodomconstraint}.
As visualized on the right, ODE solutions starting from the black circle $u^2+v^2 = 1$ eventually reach the dashed \bluec{blue} circle $u^2+v^2 = \frac{1}{4}$.
Since the \bluec{blue} circle is entirely contained in the \greenc{green} goal region, solutions that reach it must (trivially) also reach the goal region.
Formally, the following instance of implication~\rref{eq:refinementimpl} is provable because formula $\rfvar_0 \limply \rfvar$ is provable.
\begin{equation}
\ddiamond{\exlinear}{\underbrace{\big(u^2+v^2 = \frac{1}{4}\big)}_{\rfvar_0}} \limply \ddiamond{\exlinear}{\underbrace{\big(\frac{1}{4} \leq \lnorm{(u,v)} \leq \frac{1}{2}\big)}_{\rfvar}}
\label{eq:refinementimpl1}
\end{equation}

\begin{wrapfigure}[9]{r}{0.24\textwidth}
\centering
\includegraphics[width=0.22\textwidth]{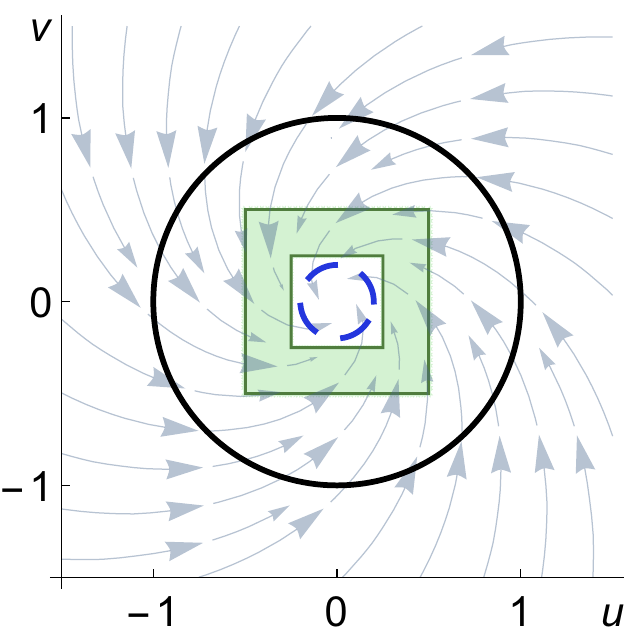}
\end{wrapfigure}

Similarly, if the implication between domain constraints $\ivr_0 \limply \ivr$ is provable, then implication~\rref{eq:refinementimpl} is proved by monotonicity, because any solution staying in the smaller domain $\ivr_0$ must also stay in the larger domain $\ivr$.
However, neither of these monotonicity-based arguments are sufficiently powerful for liveness proofs because they do not account for the specific ODE $\D{x}=\genDE{x}$ under consideration at all. Returning to the ODE $\exlinear$, suppose instead that the initial (known) liveness property is $\ddiamond{\exlinear}{u^2+v^2 = \frac{1}{25}}$. This is visualized on the right with a smaller dashed \bluec{blue} circle.
The following instance of implication~\rref{eq:refinementimpl} is also valid for solutions starting from the black circle $u^2+v^2=1$, but it does \emph{not} follow from a straightforward monotonicity argument because the smaller dashed \bluec{blue} circle $u^2+v^2 = \frac{1}{25}$ is not contained in the \greenc{green} goal region (formula $\rfvar_0 \limply \rfvar$ is not valid).
\begin{equation}
\ddiamond{\exlinear}{\underbrace{\big(u^2+v^2 = \frac{1}{25}\big)}_{\rfvar_0}} \limply \ddiamond{\exlinear}{\underbrace{\big(\frac{1}{4} \leq \lnorm{(u,v)} \leq \frac{1}{2}\big)}_{\rfvar}}
\label{eq:refinementimpl2}
\end{equation}%

A proof of implication~\rref{eq:refinementimpl2} requires additional information about solutions of the ODE $\exlinear$, namely, that they are continuous and the system $\exlinear$ is planar.
Informally, observe that it is impossible to draw a line (without lifting your pen off the page) that connects the black circle to the (smaller) dashed \bluec{blue} circle without crossing the \greenc{green} goal region.
The continuous solutions of $\exlinear$ are analogous to such lines and therefore must enter the \greenc{green} goal region before reaching the \bluec{blue} circle.
To formalize such reasoning, this article's approach is built on refinement axioms that conclude instances of implication~\rref{eq:refinementimpl}, like~\rref{eq:refinementimpl1} and~\rref{eq:refinementimpl2}, from box modality formulas involving the ODE $\D{x}=\genDE{x}$.
The following are four ODE refinement axioms of \dL that are used for the approach.
Crucially, these axioms are \emph{derived} from the box modality axioms presented in~\rref{subsec:proofcalculus} by exploiting the logical duality between the box and diamond modalities of \dL.
This makes it possible to build liveness arguments from a sound and parsimonious logical foundation.

\begin{lemma}[Diamond ODE refinement axioms]
\label{lem:diarefaxioms}
The following $\didia{\cdot}$ ODE refinement axioms are derivable in \dL.
In axioms~\irref{dBDG+dDDG}, $y=(y_1,\dots,y_m)$ is an $m$-dimensional vector of fresh variables (not appearing in $x$) and $g(x,y)$ is a corresponding $m$-dimensional vector of terms.
Terms $p(x),L(x),M(x)$ and formulas $\rfvar(x),\ivr(x)$ are dependent only on free variables $x$ (and not $y$).

\begin{calculuscollection}
\begin{calculus}
\dinferenceRule[Prog|K${\didia{\&}}$]{}
{
\linferenceRule[impl]
  {\dbox{\pevolvein{\D{x}=\genDE{x}}{\ivr \land \lnot{\rfvar}}}{\lnot{\rgvar}}}
  {\big(\ddiamond{\pevolvein{\D{x}=\genDE{x}}{\ivr}}{\rgvar} \limply \axkey{\ddiamond{\pevolvein{\D{x}=\genDE{x}}{\ivr}}{\rfvar}}\big)}
}{}

\dinferenceRule[dDR|DR${\didia{\cdot}}$]{}
{\linferenceRule[impl]
  {\dbox{\pevolvein{\D{x}=\genDE{x}}{\rrfvar}}{\ivr}}
  {\big( \ddiamond{\pevolvein{\D{x}=\genDE{x}}{\rrfvar}}{\rfvar} \limply \axkey{\ddiamond{\pevolvein{\D{x}=\genDE{x}}{\ivr}}{\rfvar}}\big)}
}{}

\dinferenceRule[dBDG|BDG${\didia{\cdot}}$]{}
{\linferenceRule[impll]
  {\dbox{\pevolvein{\D{x}=\genDE{x},\D{y}=g(x,y)}{\ivr(x)}}{\,\norm{y}^2 \leq \ptermA(x)}}
  {\big( \ddiamond{\pevolvein{\D{x}=\genDE{x}}{\ivr(x)}}{\rfvar(x)} \limply \axkey{\ddiamond{\pevolvein{\D{x}=\genDE{x}, \D{y}=g(x,y)}{\ivr(x)}}{\rfvar(x)}}\big)}
}{}

\dinferenceRule[dDDG|DDG${\didia{\cdot}}$]{}
{\linferenceRule[impll]
  {\dbox{\pevolvein{\D{x}=\genDE{x},\D{y}=g(x,y)}{\ivr(x)}}{\,2\dotp{y}{g(x,y)} \leq L(x) \norm{y}^2 + M(x)}}
  {\big( \ddiamond{\pevolvein{\D{x}=\genDE{x}}{\ivr(x)}}{\rfvar(x)} \limply \axkey{\ddiamond{\pevolvein{\D{x}=\genDE{x}, \D{y}=g(x,y)}{\ivr(x)}}{\rfvar(x)}}\big)}
}{}
\end{calculus}
\end{calculuscollection}
\end{lemma}
\begin{proofsketchb}[app:refinementcalculus]{proof:proof3}\end{proofsketchb}

Axiom~\irref{Prog} is best understood in the contrapositive.
Formula $\dbox{\pevolvein{\D{x}=\genDE{x}}{\ivr \land \lnot{\rfvar}}}{\lnot{\rgvar}}$ says $\rgvar$ never happens along the solution while $\lnot{\rfvar}$ holds.
Thus, the solution cannot get to $\rgvar$ unless it gets to $\rfvar$ first.
Axiom~\irref{Prog} formalizes the informal reasoning used for implication~\rref{eq:refinementimpl2} above in the contrapositive, with $\rgvar \mnodefequiv u^2+v^2 = \frac{1}{25}$ and $\rfvar \mnodefequiv \big(\frac{1}{4} \leq \lnorm{(u,v)} \leq \frac{1}{2}\big)$.
In the (partial) derivation shown below, the left premise requires a proof that the dashed \bluec{blue} circle $\rgvar$ cannot be reached while staying outside the \greenc{green} goal region $\rfvar$ while the right premise requires a proof of the initial liveness property $\ddiamond{\exlinear}{\big(u^2+v^2 = \frac{1}{25}\big)}$ for $\exlinear$.
In a sequent calculus proof, refinement steps are naturally read from top-to-bottom (downwards), while deduction steps, i.e., axiom or rule applications, are read bottom-to-top (upwards).

~\\
\noindent%
\begin{minipage}[c]{0.01\textwidth}%
\noindent$\overset{\substack{\scalebox{1}{\textbf{Deduction}\hidewidth}\mathstrut}}{\Vast\uparrow}$
\end{minipage}
\begin{minipage}[c]{0.95\textwidth}
{\footnotesizeoff\renewcommand{\arraystretch}{1.3}%
\begin{sequentdeduction}[array]
  \linfer[Prog]{
    \linfer[]{
      \vdots
    }
    {\lsequent{u^2+v^2=1}{\dbox{\pevolvein{\exlinear}{\lnot{\big(\frac{1}{4} \leq \lnorm{(u,v)} \leq \frac{1}{2}\big)}}}{\lnot{\big(u^2+v^2 = \frac{1}{25}\big)}}}}
    !
    \linfer[]{
      \vdots
    }
    {\lsequent{u^2+v^2=1}{\ddiamond{\exlinear}{\big(u^2+v^2 = \frac{1}{25}\big)}}}
  }
  {\lsequent{u^2+v^2=1}{\ddiamond{\exlinear}{\big(\frac{1}{4} \leq \lnorm{(u,v)} \leq \frac{1}{2}\big)}}}
\end{sequentdeduction}
}%
\end{minipage}
\begin{minipage}[c]{0.01\textwidth}
$\underset{\hidewidth\substack{\mathstrut\scalebox{1}{\textbf{Refinement}}}}{\Vast\downarrow}$
\end{minipage}
~\\~\\

In refinement axiom~\irref{dDR}, formula $\dbox{\pevolvein{\D{x}=\genDE{x}}{\rrfvar}}{\ivr}$ says that the ODE solution never leaves $\ivr$ while staying in $\rrfvar$, so if the solution gets to $\rfvar$ within $\rrfvar$, then it also gets to $\rfvar$ within $\ivr$.
The latter two refinement axioms~\irref{dBDG+dDDG} are both derived from~\irref{BDG}.
The (nested) refinement in both axioms says that, if the ODE $\D{x}=\genDE{x}$ can reach $\rfvar(x)$, then the ODE $\D{x}=\genDE{x},\D{y}=g(x,y)$, with the added variables $y$, can also reach $\rfvar(x)$.
Axiom~\irref{dBDG} is the derived diamond version of~\irref{BDG}, obtained by directly dualizing the inner equivalence of~\irref{BDG} with~\irref{diamond} and propositional simplification.
The intuition behind~\irref{dBDG} is identical to~\irref{BDG}: if the added ghost ODEs $y$ never blow up in norm, then they do not affect whether the solution of the original ODEs $\D{x}=\genDE{x}$ can reach $\rfvar(x)$.

Axiom~\irref{dDDG} is a derived, differential version of~\irref{dBDG}.
Instead of bounding the squared norm $\norm{y}^2$ explicitly,~\irref{dDDG} instead limits the rate of growth of the ghost ODEs by bounding the Lie derivative\footnote{In \dL's uniform substitution calculus~\cite{DBLP:journals/jar/Platzer17}, this Lie derivative is written directly as the differential term $\D{(\norm{y}^2)}$ which can be soundly and syntactically rewritten using \dL's differential axioms~\cite{DBLP:journals/jar/Platzer17}.} $\lie{\D{x}=\genDE{x},\D{y}=g(x,y)}{\norm{y}^2} = 2\dotp{y}{g(x,y)}$ of the squared norm.
This derivative bound in turn implicitly bounds the squared norm of the ghost ODEs by the solution of the linear differential equation $\D{z} = L(x)z+M(x)$, with dependency on the value of $x$ along solutions of the ODE $\D{x}=\genDE{x}$.
This ensures that premature blow-up of $y$ before $x$ itself blows up is impossible.
Any refinement step using axiom~\irref{dDDG} can also use axiom~\irref{dBDG} since the former is derived from the latter.
The advantage of~\irref{dDDG} is it builds in canonical differential reasoning steps once-and-for-all (see proof of~\rref{lem:diarefaxioms} and~\rref{sec:globexist}) which simplifies the refinement proof.

Axioms~\irref{Prog+dDR+dBDG+dDDG} all prove implication~\rref{eq:refinementimpl} in just one refinement step.
Logical implication is transitive though, so a sequence of such steps can be chained together to prove implication~\rref{eq:refinementimpl}.
This is shown in~\rref{eq:refinementchain}, with neighboring implications informally chained together for illustration:
{\footnotesizeoff
\begin{align}
\ddiamond{\pevolvein{\D{x}=\genDE{x}}{\ivr_0}}{\rfvar_0} &\overbrace{\lrefine}^{\hidewidth \irref{dDR}~\text{with}~\dbox{\pevolvein{\D{x}=\genDE{x}}{\ivr_0}}{\ivr_1} \quad\quad\;\; \hidewidth} \ddiamond{\pevolvein{\D{x}=\genDE{x}}{\ivr_1}}{\rfvar_0} \overbrace{\lrefine}^{\hidewidth \;\;\quad\quad \irref{Prog}~\text{with}~\dbox{\pevolvein{\D{x}=\genDE{x}}{\ivr_1 \land \lnot{P_1}}}{\lnot{P_0}} \hidewidth} \ddiamond{\pevolvein{\D{x}=\genDE{x}}{\ivr_1}}{\rfvar_1} \lrefine \cdots \lrefine \ddiamond{\pevolvein{\D{x}=\genDE{x}}{\ivr}}{\rfvar}
\label{eq:refinementchain}
\end{align}
}%

With its side conditions, i.e., the box modality formulas, proven, the chain of refinements~\rref{eq:refinementchain} proves the desired implication~\rref{eq:refinementimpl}.
However, a proof of the liveness property \(\ddiamond{\pevolvein{\D{x}=\genDE{x}}{\ivr}}{\rfvar}\) on the right still needs a proof of the hypothesis \(\ddiamond{\pevolvein{\D{x}=\genDE{x}}{\ivr_0}}{\rfvar_0}\) at the beginning of the chain.
Typically, this hypothesis is a (simple) existence assumption for the differential equation.
Formalizing and proving such existence properties is the focus of~\rref{sec:globexist}.
Those proofs are also based on refinements and make use of axioms~\irref{dBDG+dDDG}.

Refinement with axiom~\irref{dDR} requires proving the formula $\dbox{\pevolvein{\D{x}=\genDE{x}}{\rrfvar}}{\ivr}$.
Na\"ively, one might expect that adding $\lnot{\rfvar}$ to the domain constraint should also work, i.e., the solution only needs to be in $\ivr$ while it has not yet gotten to $\rfvar$:
\[
\cinferenceRule[badaxiom|DR$\didia{\cdot}$\usebox{\Lightningval}]{}
{
\linferenceRule[impl]
  {\dbox{\pevolvein{\D{x}=\genDE{x}}{\rrfvar \land \lnot{\rfvar}}}{\ivr} }
  {\big(\ddiamond{\pevolvein{\D{x}=\genDE{x}}{\rrfvar}}{\rfvar} \limply \axkey{\ddiamond{\pevolvein{\D{x}=\genDE{x}}{\ivr}}{\rfvar}}\big)}
}{}
\]

This conjectured axiom is unsound (indicated by $\mbox{\lightning}$) as the solution could sneak out of $\ivr$ exactly when it crosses from $\lnot{\rfvar}$ into $\rfvar$.
In continuous settings, the language of topology makes precise what this means.
The following topological refinement axioms soundly restrict what happens at the crossover point:

\begin{lemma}[Topological ODE refinement axioms]
\label{lem:diatopaxioms}
The following topological $\didia{\cdot}$ ODE refinement axioms are sound.
In axiom~\irref{CORef}, $\rfvar,\ivr$ either both characterize topologically open or both characterize topologically closed sets over variables $x$.

\begin{calculuscollection}
\begin{calculus}
\cinferenceRule[CORef|COR]{}
{
\linferenceRule[impl]
  {\lnot{\rfvar} \land \dbox{\pevolvein{\D{x}=\genDE{x}}{\rrfvar \land \lnot{\rfvar}}}{\ivr} }
  {\big(\ddiamond{\pevolvein{\D{x}=\genDE{x}}{\rrfvar}}{\rfvar} \limply \axkey{\ddiamond{\pevolvein{\D{x}=\genDE{x}}{\ivr}}{\rfvar}}\big)}
}{}

\cinferenceRule[SARef|SAR]{}
{
\linferenceRule[impl]
  {\dbox{\pevolvein{\D{x}=\genDE{x}}{\rrfvar \land \lnot{(\rfvar \land \ivr)}}}{\ivr}}
  {\big( \ddiamond{\pevolvein{\D{x}=\genDE{x}}{\rrfvar}}{\rfvar} \limply \axkey{\ddiamond{\pevolvein{\D{x}=\genDE{x}}{\ivr}}{\rfvar}}\big)}
}{}
\end{calculus}
\end{calculuscollection}
\end{lemma}
\begin{proofsketchb}[app:refinementcalculus]{proof:proof4}\end{proofsketchb}

Axiom~\irref{CORef} is the more informative topological refinement axiom.
Like the (unsound) axiom candidate~\irref{badaxiom}, it allows formula $\lnot{\rfvar}$ to be assumed in the domain constraint when proving the box refinement.
For soundness though, axiom~\irref{CORef} has crucial topological side conditions on formulas $\rfvar,\ivr$ so it can only be used when these conditions are met.
Several variations of~\irref{CORef} are possible (with similar soundness proofs), but they require alternative topological restrictions and additional topological notions.
One useful variation involving the topological interior is given in~\rref{lem:closeddomref}.
When these topological restrictions are enforced syntactically, axiom~\irref{CORef} is derived from \dL's real induction axiom~\cite{DBLP:journals/jacm/PlatzerT20}.
For the sake of generality, this article gives semantic topological side conditions with associated semantic soundness proofs in~\rref{app:refinementcalculus}.

Axiom~\irref{SARef} applies more generally than~\irref{CORef} but only assumes the less informative formula $\lnot{(\rfvar \land \ivr)}$ in the domain constraint for the box modality.
Its proof crucially relies on $\ivr$ being a formula of real arithmetic so that the set it characterizes has tame topological behavior~\cite{Bochnak1998}, see the proof in~\rref{app:refinementcalculus} for more details.
By topological considerations, axiom~\irref{SARef} is also sound if formula $\rfvar$ (or resp. $\ivr$) characterizes a topologically closed (resp. open) set over the ODE variables $x$.
These additional cases are also proved in~\rref{app:refinementcalculus} without relying on the fact that $\ivr$ is a formula of real arithmetic.

\section{Finite-Time Blow Up and Global Existence}
\label{sec:globexist}

This section explains how global existence properties can be proved for a given ODE $\D{x}=\genDE{x}$, subject to assumptions $\Gamma$ about the initial states for the ODE.
The existence and uniqueness theorems for ODEs~\cite{Chicone2006,Walter1998} guarantee that polynomial ODEs like $\D{x}=\genDE{x}$ always have a unique, right-maximal solution from any initial state, $\solvar : [0, T) \to \States$ for some $0<T\leq\infty$.
However, these theorems give no guarantees about the precise duration $T$.
In particular, ODEs can exhibit a technical phenomenon known as \emph{finite-time blow up of solutions}~\cite{Chicone2006}, where $\solvar$ is only defined on a bounded time interval $[0,T)$ with $T < \infty$.
Additionally, it is possible that such finite-time blow up phenomena only happens for \emph{some} initial conditions (and corresponding solutions) of the ODE.
Moreover, these initial conditions (with finite-time blow up) may not be relevant to the model of concern, especially when the dynamics of real world systems are controlled to stay away from the blow up.
For example, $\exnonlinear$~\rref{eq:exnonlinear} exhibits finite blow up of solutions only outside the \redc{red} disk as shown in~\rref{fig:odeexamples} and the blow up occurs well after its solutions have reached the target region, see~\rref{fig:odeblowup}.

As an additional example for this section, consider the following nonlinear ODE:
\begin{align}
\exblowup \mnodefequiv \D{v}=-v^2 \label{eq:exblowup}
\end{align}

The solution to this ODE is $v(\timevar) = \frac{v_0}{v_0+\timevar}$, where $v_0 \neq 0$ is the initial value of $v$ at time $\timevar=0$ (if $v_0=0$, then $v(\timevar)=0$ for all $\timevar$).
If $v_0 < 0$ initially, then this solution is only defined to the right for the finite time interval $[0,-v_0)$, because the denominator $v_0+\timevar$ is $0$ at $t=-v_0$.
On the other hand, for $v_0 \geq 0$, the existence interval to the right is $[0,\infty)$.
Thus, $\exblowup$ exhibits finite-time blow up of solutions, but only for $v_0 < 0$.

\subsection{Global Existence Proofs}
\label{sec:proveglobexist}

The discussion above uses the mathematical solution $v(\timevar)$ of the ODE \exblowup~\rref{eq:exblowup} as a function of time.
For deductive proofs, the (global) existence of solutions can be expressed in \dL as a special form of an ODE liveness property.
The first step is to add a fresh variable $\timevar$ with $\D{\timevar}=1$ that tracks the progress of time\footnote{For consistency, the ODE $\D{x}=\genDE{x}$ is assumed to not mention $\timevar$ even if this is not always strictly necessary.}, see~\rref{subsec:syntax}.
Then, using a fresh variable $\tau$ not in $x,\timevar$, the following formula syntactically expresses that the ODE has a global solution because its solutions can reach time $\tau$, for any arbitrary $\tau$:
\begin{align}
\lforall{\tau}{\ddiamond{\pevolve{\D{x}=\genDE{x},\D{\timevar}=1}}{\,\timevar > \tau}}
\label{eq:existence}
\end{align}

Proving formula~\rref{eq:existence} shows global existence of solutions for the ODE $\D{x}=\genDE{x}$.
The simplest instance of~\rref{eq:existence} is for the ODE $\D{\timevar}=1$ by itself without any ODE $\D{x}=\genDE{x}$.
In this case, the formula~\rref{eq:existence} is valid because $\D{\timevar}=1$ is an ODE with constant RHS and its solution exists for all time.
The axiom~\irref{TEx} below expresses this fact and it is derived directly from the solution axiom of \dL~\cite{DBLP:journals/jar/Platzer17}:

\begin{lemma}[Time existence]
\label{lem:timeexist}
The following axiom is derivable in \dL.

\begin{calculuscollection}
\begin{calculus}
\dinferenceRule[TEx|TEx]{}
{
  \axkey{\lforall{\tau}{\ddiamond{\pevolve{\D{\timevar}=1}}{\timevar > \tau}}}
}{}
\end{calculus}
\end{calculuscollection}
\end{lemma}
\begin{proofsketchb}[app:refinementcalculus]{proof:proof5}\end{proofsketchb}

Other instances of \rref{eq:existence} can be proved using axioms~\irref{dBDG+dDDG} with appropriate assumptions about the initial conditions for the additional ODEs $\D{x}=\genDE{x}$.
This is exemplified for the ODE $\exblowup$ next.

\begin{example}[Velocity of particle with air resistance]
\label{ex:velocity}
The ODE $\exblowup$ can be viewed as a model of the velocity of a particle that is slowing down due to air resistance.
Of course, it does not make physical sense for the velocity of such a particle to ``blow up''.
However, the solution of $\exblowup$ only exists globally if the particle starts with positive initial velocity $v>0$, otherwise, it only has short-lived solutions.
The reason is that $\exblowup$ only makes physical sense for positive velocities $v>0$, so that the air resistance term $-v^2$ slows the particle down instead of speeding it up.
Indeed, global existence~\rref{eq:existence} can be proved for $\exblowup$ if its initial velocity is positive, i.e., the \dL formula \(v > 0 \limply \lforall{\tau}{\ddiamond{\pevolve{\D{v}=-v^2,\D{\timevar}=1}}{\timevar > \tau}}\) is valid.

The derivation below starts with basic propositional steps (\irref{implyr+allr}), after which axiom~\irref{dDDG} is used with $\D{v}=-v^2$ as the differential ghost equation with the trivial choice of bounds $L \mnodefeq 0, M \mnodefeq 0$.
This yields two premises, the right of which is proved by~\irref{TEx}.
The resulting left premise requires proving the formula $2\dotp{v}{(-v^2)} \leq 0$ along the ODE.
Mathematically, this says that the derivative of the squared norm $v^2$ is non-negative along $\exblowup$, so that $v^2$ is non-increasing and cannot blow up.\footnote{The fact that $v^2$ is non-increasing can also be used in an alternative derivation with axiom~\irref{dBDG} and the bound $\ptermA \mnodefeq v_0^2$, where $v_0$ syntactically stores the initial value of $v$.}
An~\irref{MbW} step strengthens the postcondition to $v>0$ since $v>0$ implies $2\dotp{v}{(-v^2)} \leq 0$ in real arithmetic.
The resulting premise is an invariance property for $v > 0$ which is provable in \dL (proof omitted~\cite{DBLP:journals/jacm/PlatzerT20}).
The initial assumption $v>0$ is crucially used in this step, as expected.

{\footnotesizeoff\renewcommand{\arraystretch}{1.2}
\begin{sequentdeduction}[array]
  \linfer[implyr+allr]{
  \linfer[dDDG]{
    \linfer[MbW]{
    \linfer[]{
      \lclose
    }
    {\lsequent{v > 0}{\dbox{\pevolve{\D{v}=-v^2,\D{\timevar}=1}}{\,v > 0}}}
    }
    {\lsequent{v > 0}{\dbox{\pevolve{\D{v}=-v^2,\D{\timevar}=1}}{\,2\dotp{v}{(-v^2)} \leq 0}}} !
    \linfer[TEx]{
      \lclose
    }
    {\lsequent{}{\ddiamond{\D{\timevar}=1}}{\timevar > \tau}}
  }
    {\lsequent{v > 0}{\ddiamond{\pevolve{\D{v}=-v^2,\D{\timevar}=1}}{\timevar > \tau}}}
  }
  {\lsequent{}{v>0 \limply \lforall{\tau}{\ddiamond{\pevolve{\D{v}=-v^2,\D{\timevar}=1}}{\timevar > \tau}}}}
\end{sequentdeduction}
}%

\end{example}

\rref{sec:livenessaxioms} offers another view of the derivation above as a single refinement step in the chain~\rref{eq:refinementchain}, recall that refinement steps are read from top-to-bottom.
Here, an initial existence property for the ODE $\D{\timevar}=1$ is refined to the desired existence property for the ODE $\D{v}=-v^2,\D{\timevar}=1$.
The refinement step is justified using~\irref{dDDG} with the box modality formula $\dbox{\pevolve{\D{v}=-v^2,\D{\timevar}=1}}{\,2\dotp{v}{(-v^2)} \leq 0}$.
{\footnotesizeoff%
\begin{align*}
  \ddiamond{\D{\timevar}=1}{\timevar > \tau}
  {\limprefinechain{\irref{dDDG}}}
  \ddiamond{\D{v}=-v^2,\D{\timevar}=1}{\timevar > \tau}
\end{align*}
}%

This chain can be extended to prove global existence for more complicated ODEs $\D{x}=\genDE{x}$ in a stepwise fashion, and (possibly) alternating between uses of~\irref{dDDG} or~\irref{dBDG} for the refinement step.
To do this, note that any ODE $\D{x}=\genDE{x}$ can be written in \emph{dependency order}, where each group $y_i$ is a vector of variables and each $g_i$ corresponds to the respective vectorial RHS of the ODE for $y_i$ for $i=1,\dots,k$.
The RHS of each $\D{y_i}$ is only allowed to depend on the preceding vectors of variables (inclusive) $y_1,\dots,y_i$.
\begin{align}
\underbrace{\D{y_1} = g_1(y_1), \D{y_2} = g_2(y_1,y_2), \D{y_3} = g_3(y_1,y_2,y_3), \dots, \D{y_k} = g_k(y_1,y_2,y_3,\dots,y_k)}_{\D{x}=\genDE{x}~\text{written in dependency order}}
\label{eq:deporder}
\end{align}

\begin{corollary}[Dependency order existence]
\label{cor:deporderexist}
Let the ODE $\D{x}=\genDE{x}$ be in dependency order~\rref{eq:deporder}, and $\tau$ be a fresh variable not in $x,\timevar$.
The following rule with $k$ stacked premises is derived from~\irref{dBDG+dDDG} and~\irref{TEx}, where the postcondition of each premise $\rfvar_i$ for $1 \leq i \leq k$ can be chosen to be either of the form:
\begin{enumerate}
\item[\bform] $\rfvar_i \mnodefequiv \norm{y_i}^2 \leq \ptermA_i(t,y_1,\dots,y_{i-1})$ for some term $\ptermA_i$ with the indicated dependencies, or,
\item[\dform] $\rfvar_i \mnodefequiv 2\dotp{y_i}{g_i(y_1,\dots,y_i)} \leq L_i(t,y_1,\dots,y_{i-1}) \norm{y_i}^2+M_i(t,y_1,\dots,y_{i-1})$ for some terms $L_i,M_i$ with the indicated dependencies.
\end{enumerate}
\begin{calculuscollection}
\begin{calculus}
\dinferenceRule[DEx|DEx]{}
{\linferenceRule
  {
      \begin{array}{l}
    \lsequent{\Gamma}{\dbox{\pevolve{\D{y_1}=g_1(y_1),\D{\timevar}=1}}{\rfvar_1}} \\
    \lsequent{\Gamma}{\dbox{\pevolve{\D{y_1}=g_1(y_1),\D{y_2}=g_2(y_1,y_2),\D{\timevar}=1}}{\rfvar_2}} \\
    \quad\vdots \\
    \lsequent{\Gamma}{\dbox{\pevolve{\D{y_1}=g_1(y_1),\dots,\D{y_k}=g_k(y_1,\dots,y_k),\D{\timevar}=1}}{\rfvar_k}}
    \end{array}
  }
  {\lsequent{\Gamma}{\lforall{\tau}{\ddiamond{\pevolve{\D{x}=\genDE{x},\D{\timevar}=1}}{\timevar > \tau}}}}
}{}
\end{calculus}
\end{calculuscollection}
\end{corollary}
\begin{proofsketcha}[app:globexistproofs]{proof:proof6}
The derivation proceeds (backwards) by successive refinements using either~\irref{dBDG} for premises corresponding to the form \bform or~\irref{dDDG} for those corresponding to \dform, with the ghost equations for $g_i$ and the respective bounds $\ptermA_i$ or $L_i,M_i$ at each step for $i=k,\dots,1$.
\end{proofsketcha}

Rule~\irref{DEx} corresponds to a refinement chain~\rref{eq:refinementchain} of length $k$, with successive~\irref{dBDG+dDDG} steps, e.g.:
{\footnotesizeoff%
\begin{align*}
  \ddiamond{\D{\timevar} = 1}{\timevar > \tau}
  {\limprefinechain{\irref{dBDG}}}
  \ddiamond{\D{y_1}{=}g_1(y_1),\D{\timevar} = 1}{\timevar > \tau}
  {\limprefinechain{\irref{dDDG}}}
  \cdots
  \limprefinechain{}
  \ddiamond{\D{y_1}{=}g_1(y_1),\dots,\D{y_k}{=}g_k(y_1,\dots,y_k),\D{\timevar} = 1}{\timevar > \tau}
\end{align*}
}%

In rule~\irref{DEx} any choice of the shape of premises (\bform and \dform) is sound as these correspond to an underlying choice of axiom~\irref{dBDG+dDDG} to apply at each refinement step, respectively.
Another source of flexibility arises when choosing the dependency ordering~\rref{eq:deporder} for the ODE $\D{x}=\genDE{x}$, as long as the requisite dependency requirements are met.
For example, one can always choose the coarsest dependency order $y_1 \mnodefequiv x, g_1 \mnodefequiv \genDE{x}$ to directly prove global existence in one step using appropriate choice of bounds $L_1,M_1$.
The advantage of using finer dependency orders in~\irref{DEx} is it allows the user to choose the bounds $L_i,M_i$ in a step-by-step manner for $i=1,\dots,k$.
On the other hand, the flexibility of rule~\irref{DEx} can also be a drawback because it relies on manual effort from users to choose the partition and to prove the resulting premises.
\rref{sec:derivedexist} explains useful recipes for using the flexibility behind rule~\irref{DEx}, e.g.,  Corollaries~\ref{cor:globalexistaffine} and~\ref{cor:boundedexistgen}, while~\rref{subsec:support} further explains how proof support can help users in those proofs.

The discussion thus far proves global existence for ODEs with an explicit time variable $\timevar$.
This is not a restriction for the liveness proofs in later sections of this article because such a fresh time variable can always be added using the rule~\irref{dGt} below, which is derived from~\irref{DG}.
The rule also adds the assumption $\timevar=0$ initially without loss of generality for ease of proof.
\[
\dinferenceRule[dGt|dGt]{diff ghost clock}
{\linferenceRule
  {\lsequent{\Gamma,\timevar=0} {\ddiamond{\pevolvein{\D{x}=\genDE{x},\D{\timevar}=1}{\ivr}}{\rfvar}}}
  {\lsequent{\Gamma}{\ddiamond{\pevolvein{\D{x}=\genDE{x}}{\ivr}}{\rfvar}}}
}{}
\]

\subsection{Derived Existence Axioms}
\label{sec:derivedexist}
For certain classes of ODEs and initial conditions, there are well-known mathematical techniques to prove global existence of solutions.
These techniques have purely syntactic renderings in \dL as special cases of~\irref{dBDG+dDDG}, and \irref{DEx}.
In particular, this section shows how axioms~\irref{GEx+BEx} (shown below), which were proved semantically in the earlier conference version~\cite{DBLP:conf/fm/TanP19}, can be derived syntactically.
The refinement approach also yields natural generalizations of these axioms.

\subsubsection{Globally Lipschitz ODEs}
A function $f : \reals^m \to \reals^n$ is \emph{globally Lipschitz continuous} if there is a (positive) Lipschitz constant $C \in \reals$ such that the inequality $\norm{\genDE{x}-\genDE{y}} \leq C \norm{x-y}$ holds for all $x,y \in \reals^m$, where $\norm{\cdot}$ are appropriate norms.
Since norms are equivalent on finite dimensional vector spaces~\cite[\S5.V]{Walter1998}, without loss of generality, the Euclidean norm is used for the following discussion.
An ODE $\D{x}=\genDE{x}$ is \emph{globally Lipschitz} if its RHS $\genDE{x}$ is globally Lipschitz continuous and solutions of such ODEs always exist globally for all time~\cite[\S10.VII]{Walter1998}.
Global Lipschitz continuity is satisfied, e.g., by $\exlinear$~\rref{eq:exlinear}, and more generally by linear (or even affine) ODEs of the form $\D{x}=Ax$, where $A$ is a matrix of (constant) parameters~\cite{Walter1998} because of the following (mathematical) inequality with Lipschitz constant $\norm{A}$, i.e., the (matrix-Euclidean) Frobenius norm of $A$:
\[ \norm{Ax - Ay} = \norm{A (x-y)} \leq \norm{A}\norm{x-y} \]

This calculation uses the Euclidean norm $\norm{\cdot}$, which is not a term in \dL (\rref{subsec:syntax}) because it is not a polynomial.
This syntactic exclusion is not an oversight: it is crucial to the soundness of \dL that such non-differentiable terms are excluded from its syntax.
For example, $\norm{x}$ is not differentiable at $x=0$.
Thus, a subtle technical challenge in proofs is to appropriately rephrase mathematical inequalities, typically involving norms, into ones that can be reasoned about soundly also in the presence of differentiation.
In this respect, the Euclidean norm is useful, because expanding the inequality $0 \leq (1-\norm{x})^2$ and rearranging yields:
\begin{align}
2\norm{x} \leq 1+\norm{x}^2
\label{eq:norminequality}
\end{align}

Notice that, unlike the Euclidean norm $\norm{x}$, the RHS of the square inequality~\rref{eq:norminequality} can be represented syntactically.
Indeed, the squared Euclidean norm is already used in axioms~\irref{BDG+dBDG+dDDG}.
To support intuition, the proof sketches below continue to use mathematical inequalities involving Euclidean norms, while the proofs in the appendix use rephrasings with~\rref{eq:norminequality} instead.
The following corollary shows how global existence for globally Lipschitz ODEs is derived using a norm inequality as a special case of rule~\irref{DEx}.

\begin{corollary}[Global existence]
\label{cor:globalexistbase}
The following global existence axiom is derived from~\irref{dDDG} in \dL, where $\tau$ is a fresh variable not in $x,\timevar$, and $\D{x}=\genDE{x}$ is globally Lipschitz.

\begin{calculuscollection}
\begin{calculus}
\cinferenceRule[GEx|GEx]{}
{
  \axkey{\lforall{\tau}{\ddiamond{\pevolve{\D{x}=\genDE{x},\D{\timevar}=1}}{\timevar > \tau}}}
  \qquad
}{}
\end{calculus}
\end{calculuscollection}
\end{corollary}
\begin{proofsketcha}[app:globexistproofs]{proof:proof7}
Let $C$ be the Lipschitz constant for $f$.
The proof uses~\irref{dDDG} and two (mathematical) inequalities.
The first inequality~\rref{eq:globlipschitzboundone} bounds $\norm{f(x)}$ linearly in $\norm{x}$.
The constant $0$ is chosen here to simplify the resulting arithmetic.
\begin{align}
\norm{f(x)} &= \norm{f(x) - f(0) + f(0)} \leq \norm{f(x)-f(0)} + \norm{f(0)} \leq C\norm{x-0} + \norm{f(0)} = C\norm{x} + \norm{f(0)}
\label{eq:globlipschitzboundone}
\end{align}

The second inequality uses bound~\rref{eq:globlipschitzboundone} on $\norm{f(x)}$ to further bound $2\dotp{x}{f(x)}$ linearly in $\norm{x}^2$ along the ODE with appropriate choices of $L,M$ that only depend on the (positive) Lipschitz constant $C$ and $\norm{f(0)}$.
\begin{align}
2\dotp{x}{f(x)} &\leq 2\norm{x}\norm{f(x)} \overset{\rref{eq:globlipschitzboundone}}{\leq} 2\norm{x} \big( C \norm{x} + \norm{f(0)} \big) = 2C \norm{x}^2  + 2\norm{x}\norm{f(0)} \label{eq:globlipschitzboundtwo}\\
& \overset{\rref{eq:norminequality}}{\leq} 2C \norm{x}^2  + (1+\norm{x}^2) \norm{f(0)} = \underbrace{\big(2C + \norm{f(0)}\big)}_{L} \norm{x}^2  + \underbrace{\norm{f(0)}}_{M} \nonumber \qedhere
\end{align}
\end{proofsketcha}

The derivation of axiom~\irref{GEx} uses~\irref{dDDG}, but global existence extends to more complicated ODEs with the aid of~\irref{DEx} as long as appropriate choices of $L,M$ can be made.
A useful example of such an extension is global existence for ODEs that have an \emph{affine dependency order}~\rref{eq:deporder}, i.e., each $\D{y_i}=g_i(y_1,\dots,y_i)$ is affine in $y_i$ with $\D{y_i} = A_i(y_1,\dots,y_{i-1}) y_i + b_i(y_1,\dots,y_{i-1})$ where $A_i,b_i$ are respectively matrix and vector terms with appropriate dimensions and the indicated variable dependencies.

\begin{corollary}[Affine dependency order global existence]
\label{cor:globalexistaffine}
Axiom~\irref{GEx} is derivable from~\irref{dDDG} in \dL for ODEs $\D{x}=\genDE{x}$ with affine dependency order.
\end{corollary}
\begin{proofsketcha}[app:globexistproofs]{proof:proof8}
The proof is similar to~\rref{cor:globalexistbase} but uses~\irref{DEx} to prove global existence step-by-step for the dependency order.
It uses the following (mathematical) inequality and corresponding choices of $L_i,M_i$ (shown below) for $i=1,\dots,k$ at each step:
\begin{align}
2\dotp{y_i}{(A_iy_i + b_i)} &= 2(\dotp{y_i}{(A_iy_i)} + \dotp{y_i}{b_i}) \leq 2\norm{A_i}\norm{y_i}^2 + 2\norm{y_i}\norm{b_i} \nonumber \\
&\leq 2\norm{A_i}\norm{y_i}^2 + (1+\norm{y_i}^2)\norm{b_i} = \underbrace{(2\norm{A_i}+\norm{b_i})}_{L_i}\norm{y_i}^2 + \underbrace{\norm{b_i}}_{M_i}
\label{eq:globexistaffinebound}
\end{align}

This inequality is very similar to the one used for~\rref{cor:globalexistbase}, where $\norm{A_i}$ corresponds to $C$, and $\norm{b_i}$ corresponds to $\norm{f(0)}$.
The difference is that terms $L_i,M_i$ are allowed to depend on the preceding variables $y_1,\dots,y_{i-1}$.
Importantly for soundness, both terms meet the appropriate variable dependency requirements of~\irref{dDDG} because the terms $A_i,b_i$ are not allowed to depend on $y_i$ in the affine dependency order.
\end{proofsketcha}

With the extended refinement chain underlying~\irref{DEx},~\rref{cor:globalexistaffine} enables more general proofs of global existence for certain multi-affine ODEs that are not necessarily globally Lipschitz.

\begin{example}[Multi-affine ODE]
\label{ex:multiaffine}
Consider the multi-affine ODE $\D{u}=u, \D{v}=uv$.
The RHS of this ODE is given by the function $\begin{pmatrix} u \\ v \end{pmatrix} \mapsto \begin{pmatrix} u \\ uv \end{pmatrix}$ which is not globally Lipschitz.\footnote{
For the function to be globally Lipschitz, there must exist a constant $C \in \reals$ such that for all $\begin{pmatrix} u_1 \\ v_1 \end{pmatrix}, \begin{pmatrix} u_2 \\ v_2 \end{pmatrix} \in \reals^2$, the norm inequality $\norm{\begin{pmatrix} u_1-u_2 \\ u_1v_1 -u_2v_2\end{pmatrix}} \leq C \norm{\begin{pmatrix} u_1-u_2 \\ v_1 -v_2\end{pmatrix}}$ is satisfied.
No such $C$ exists because the $u_1v_1 -u_2v_2$ component on the LHS grows quadratically while the corresponding component $v_1-v_2$ on the RHS grows linearly (consider setting $u_i=v_i$ for $i=1,2$).}
Nevertheless, the ODE meets the dependency requirements of~\rref{cor:globalexistaffine} and has provable global solutions.

The following derivation illustrates the proof of~\rref{cor:globalexistaffine}.
In the first step, rule~\irref{DEx} is used with dependency order $y_1 \mnodefequiv u, y_2 \mnodefequiv v$ and Lipschitz constants $L_1(t) \mnodefeq 2, L_2(u,t) \mnodefeq 2u, M_1(t) \mnodefeq 0, M_2(u,t) \mnodefeq 0$.
The dependency requirements of the Lipschitz constants, notably for $L_2$, are satisfied by these choices and the resulting premises are proved by~\irref{dW+qear} because the postconditions are valid real arithmetic formulas.
{\footnotesizeoff%
\begin{sequentdeduction}[array]
\linfer[DEx]{
  \linfer[dW]{
  \linfer[qear]{
    \lclose
  }
    {\lsequent{}{2(u)(u) \leq (2) u^2}}
  }
  {\lsequent{}{\dbox{\pevolve{\D{u}=u,\D{\timevar}=1}}{2(u)(u) \leq (2) u^2}}} !
  \linfer[dW]{
  \linfer[qear]{
    \lclose
  }
    {\lsequent{}{2(v)(uv) \leq (2u) v^2}}
  }
  {\lsequent{}{\dbox{\pevolve{\D{u}=u,\D{v}=uv,\D{\timevar}=1}}{2(v)(uv) \leq (2u) v^2}}}
}
  {\lsequent{}{\lforall{\tau}{\ddiamond{\pevolve{\D{u}=u, \D{v}=uv,\D{\timevar}=1}}{\timevar > \tau}}}}
\end{sequentdeduction}
}%

Observe that the premises of~\irref{DEx} remove the ODEs for $u,v$ in a step-by-step fashion.
This is the key for generalizing global existence for globally Lipschitz ODEs~\cite[\S10.VII]{Walter1998} to more general classes of ODEs.
\end{example}

\subsubsection{Bounded Existence}
Returning to the example ODEs $\exnonlinear$~\rref{eq:exnonlinear} and $\exblowup$~\rref{eq:exblowup}, observe that axiom~\irref{GEx} applies to neither of those ODEs because they do not have affine dependency order.
As observed earlier in~\rref{fig:odeexamples} and~\rref{ex:velocity} respectively, neither $\exnonlinear$ nor $\exblowup$ have global solutions from all initial states.
Although~\rref{ex:velocity} shows how global existence for $\exblowup$ can be proved from assumptions motivated by physics, it is also useful to have general axioms (similar to~\irref{GEx}) corresponding to well-known mathematical techniques for proving global existence of solutions for nonlinear ODEs under particular assumptions.
One such mathematical technique is briefly recalled next.

Suppose that the solution of ODE $\D{x}=\genDE{x}$ is trapped within a bounded set (whose compact closure is contained in the domain of the ODE), then, the ODE solution exists globally~\cite[Corollary 2.5]{MR2381711}\cite[Theorem 3.3]{MR1201326}.
In control theory, this principle is used to show the global existence of solutions near stable equilibria~\cite{MR2381711,MR1201326}.
It also applies in case the model of interest has state variables that are \emph{a priori} known to range within a bounded set~\cite[Section 6]{10.2307/j.ctt17kkb0d}.

This discussion suggests that the following formula is valid for any ODE $\D{x}=\genDE{x}$, where $\boundedf(x)$ characterizes a bounded set over the variables $x$ so the assumption $\dbox{\D{x}=\genDE{x}}{\boundedf(x)}$ says that the ODE solution is always trapped within the bounded set characterized by $\boundedf(x)$.
\begin{align}
\dbox{\D{x}=\genDE{x}}{\boundedf(x)} \limply \lforall{\tau}{\ddiamond{\pevolve{\D{x}=\genDE{x},\D{\timevar}=1}}{\timevar > \tau}}
\label{eq:existence3}
\end{align}

Formula~\rref{eq:existence3} is (equivalently) rewritten succinctly in the following corollary by negating the box modality.

\begin{corollary}[Bounded existence]
\label{cor:boundedexistbase}
The following bounded existence axiom is derived from~\irref{dBDG} in \dL, where $\tau$ is a fresh variable not in $x,\timevar$, and formula $\boundedf(x)$ characterizes a bounded set over variables $x$.

\begin{calculuscollection}
\begin{calculus}
\dinferenceRule[BEx|BEx]{}
{
  \axkey{\lforall{\tau}{\ddiamond{\pevolve{\D{x}=\genDE{x},\D{\timevar}=1}}{(\timevar > \tau \lor \lnot{\boundedf(x)})}}}
}{}
\end{calculus}
\end{calculuscollection}
\end{corollary}
\begin{proofsketcha}[app:globexistproofs]{proof:proof9}
The squared norm $\norm{x}^2$ function is continuous in $x$ so it is bounded above by a constant $D$ on the compact closure of the set characterized by $\boundedf(x)$.
The proof uses axiom~\irref{dBDG} with $\ptermA(x) \mnodefeq D$ and rephrases formula~\rref{eq:existence3} with axiom~\irref{diamond}.\qedhere
\end{proofsketcha}

\begin{example}[Trapped solutions]
\label{ex:trappedsol}
Axiom~\irref{BEx} proves global existence for $\exnonlinear$~\rref{eq:exnonlinear} within the compact disk $u^2+v^2 \leq \frac{1}{4}$ by showing that solutions starting in the disk are trapped in it.
After the first~\irref{allr} step, a~\irref{Prog} step adds a disjunction to the postcondition.
On the resulting right premise, axiom~\irref{BEx} finishes the proof.
The left premise is an invariance property of the ODE (see~\rref{fig:odeexamples}), whose elided proof is easy~\cite{DBLP:journals/jacm/PlatzerT20}.%
{\footnotesizeoff\renewcommand{\arraystretch}{1.3}%
\begin{sequentdeduction}[array]
\linfer[allr]{
\linfer[Prog]{
  \linfer[]{
    \lclose
  }
  {\lsequent{u^2+v^2 \leq \frac{1}{4}}{\dbox{\pevolvein{\exnonlinear,\D{\timevar}=1}{\lnot{(\timevar > \tau)}}}{(u^2+v^2 \leq \frac{1}{4})}}} !
  \linfer[BEx]{
    \lclose
  }
  {\lsequent{}{\ddiamond{\pevolve{\exnonlinear,\D{\timevar}=1}}{(\timevar > \tau \lor \lnot{(u^2+v^2 \leq \frac{1}{4})})}}}
}
  {\lsequent{u^2+v^2 \leq \frac{1}{4}}{\ddiamond{\pevolve{\exnonlinear,\D{\timevar}=1}}{\timevar > \tau}}}
}
  {\lsequent{u^2+v^2 \leq \frac{1}{4}}{\lforall{\tau}{\ddiamond{\pevolve{\exnonlinear,\D{\timevar}=1}}{\timevar > \tau}}}}
\end{sequentdeduction}
}%
\end{example}

Axiom~\irref{BEx} removes the global Lipschitz (or affine dependency) requirement of~\irref{GEx} but weakens the postcondition to say that solutions must either exist for sufficient duration or blow up and leave the bounded set characterized by formula $\boundedf(x)$.
Like axiom~\irref{GEx}, axiom~\irref{BEx} is derived by refinement using axiom~\irref{dBDG}.
This commonality yields a more general version of~\irref{BEx}, which also incorporates ideas from~\irref{GEx}.

\begin{corollary}[Dependency order bounded existence]
\label{cor:boundedexistgen}
Consider the ODE $\D{x}=\genDE{x}$ in dependency order~\rref{eq:deporder}, and where $\tau$ is a fresh variable not in $x,\timevar$.
The following axiom is derived from~\irref{dBDG+dDDG} in \dL, where the indices $i=1\dots,k$ are partitioned into two disjoint index sets $L,N$ such that:
\begin{itemize}
\item For each $i \in L$, $\D{y_i}=g_i(y_1,\dots,y_i)$ is affine in $y_i$.
\item For each $i \in N$, $B_i(y_i)$ characterizes a bounded set over the variables $y_i$.
\end{itemize}

\begin{calculuscollection}
\begin{calculus}
\dinferenceRule[GBEx|GBEx]{}
{
  \axkey{\lforall{\tau}{\ddiamond{\pevolve{\D{x}=\genDE{x},\D{\timevar}=1}}{\big(\timevar > \tau \lor \lorfold_{i\in N}\lnot{\boundedf_i(y_i)}\big)}}}
}{}
\end{calculus}
\end{calculuscollection}
\end{corollary}
\begin{proofsketcha}[app:globexistproofs]{proof:proof10}
The derivation is similar to rule~\irref{DEx}, with an internal~\irref{dDDG} step (similar to~\irref{GEx}) for $i \in L$ and an internal~\irref{dBDG} step (similar to~\irref{BEx}) for $i \in N$.\qedhere
\end{proofsketcha}

The index set $L$ in~\rref{cor:boundedexistgen} indicates those variables of $\D{x}=\genDE{x}$ whose solutions are guaranteed to exist globally (with respect to the other variables).
On the other hand, the index set $N$ indicates the variables that may cause finite-time blow up of solutions.
The postcondition of axiom~\irref{GBEx} says that solutions either exist for sufficient duration or they blow up and leave one of the bounded sets indexed by $N$.
An immediate modeling application of~\rref{cor:boundedexistgen} is to identify which of the state variables in a model must be proved (or assumed) to take on bounded values~\cite[Section 6]{10.2307/j.ctt17kkb0d}.
This idea underlies the automated existence proof support discussed in~\rref{subsec:support}.

\subsection{Completeness for Global Existence}
The derivation of the existence axioms~\irref{GEx+BEx+GBEx} and rule~\irref{DEx} illustrate the use of liveness refinement for proving existence properties.
Moreover,~\irref{dBDG} is the sole ODE diamond refinement axiom underlying these derivations (recall~\irref{dDDG} is derived from~\irref{dBDG}).
This yields a natural question: are there ODEs whose solutions exist globally, but whose global existence \emph{cannot} be proved syntactically using~\irref{dBDG}?
The next completeness result gives a conditional completeness answer: \emph{all} global existence properties can be proved using~\irref{dBDG}, if the corresponding ODE solutions are \emph{syntactically representable}.

\begin{proposition}[Global existence completeness]
\label{prop:globalexistcomplete}
If the ODE $\D{x}=\genDE{x}$ has a global solution representable in the \dL term language, then the global existence formula~\rref{eq:existence} is derivable for $\D{x}=\genDE{x}$ from axiom~\irref{dBDG}.
\end{proposition}
\begin{proofsketcha}[app:globexistproofs]{proof:proof11}
Suppose that ODE $\D{x}=\genDE{x}$ has a global solution syntactically represented in \dL as term $X(t)$ dependent only on the free variable $t$, the (symbolic) initial values $x_0$ of variables $x$, and the (constant) parameters for the ODE.
The equality $x=X(t)$ is provable along the ODE $\D{x}=\genDE{x},\D{t}=1$ because solutions are equational invariants~\cite{DBLP:journals/jar/Platzer17,DBLP:journals/jacm/PlatzerT20}.
The proof uses~\irref{dBDG} with the bounding term $\ptermA \mnodefeq \norm{X(t)}^2$, so that the required hypothesis of~\irref{dBDG}, i.e., $\dbox{\pevolve{\D{x}=\genDE{x},\D{t}=1}}{\norm{x}^2 \leq \norm{X(t)}^2}$ proves trivially using the equality $x=X(t)$.
\end{proofsketcha}

The following remark illustrates the usage and limitations of~\rref{prop:globalexistcomplete}.

\begin{remark}[Syntactically representable solutions]
\label{rem:repsol}
Consider the example ODE $\D{u}=u, \D{v}=uv$ proved to have global solutions in~\rref{ex:multiaffine}.
Mathematically, its solution is given by the following functions (defined for all $\timevar \in \reals$), where $u_0, v_0$ are the initial values of $u, v$ at time $\timevar=0$ and $\exp$ is the real exponential function.
\begin{align}
u(\timevar) = u_0\exp(\timevar), v(\timevar) = \frac{v_0}{\exp(u_0)} \exp(u_0\exp(\timevar))
\label{eq:solution}
\end{align}

Since the solution~\rref{eq:solution} is defined globally, \rref{prop:globalexistcomplete} seemingly provides an alternative way to prove global existence for the ODE.
The caveat is that~\rref{prop:globalexistcomplete} only applies when the solution is \emph{syntactically representable} in the term language.
The term language of this article (\rref{subsec:syntax}) only accepts polynomial solutions.
However, \dL's term language extensions~\cite{DBLP:journals/jacm/PlatzerT20} considerably extends the class of syntactically representable solutions to include\footnote{Nevertheless, even such a syntactic extension is insufficient because Turing machines can be simulated by solutions of polynomial differential equations~\cite[Theorem 2]{GRACA2008330}.
It is possible to construct polynomial ODEs whose solutions do not blow up, but grow like the (Turing-computable) Ackermann function, i.e., faster than any tower of exponentials.}, e.g., towers of exponentials, like those found in~\rref{eq:solution}.
Thus, the usefulness of~\rref{prop:globalexistcomplete} is limited by which solutions are syntactically representable.

Notably though, the proof in~\rref{prop:globalexistcomplete} actually only requires a provable upper bound $X(t)$ with $\norm{x}^2 \leq \norm{X(t)}^2$, rather than an equality.
Such an upper bound, if syntactically representable in \dL, also suffices for proving global existence.
The complicated closed form solution~\rref{eq:solution} also highlights the advantage of axioms~\irref{dBDG+dDDG} and their use in the derived axioms of Corollaries~\ref{cor:globalexistbase}--\ref{cor:boundedexistgen} because they implicitly deduce global existence~\emph{without} needing an explicitly representable solution for the ODEs.
\end{remark}

\section{Liveness Without Domain Constraints}
\label{sec:nodomconstraint}

This section presents proof rules for liveness properties of ODEs $\D{x}=\genDE{x}$ without domain constraints, i.e., where $\ivr$ is the formula $\ltrue$.
Errors and omissions in the surveyed techniques are \highlight{highlighted in blue}.

\subsection{Differential Variants}
The fundamental technique for verifying liveness of discrete loops are loop variants, i.e., well-founded quantities that increase (or decrease) across each loop iteration.
\emph{Differential variants}~\cite{DBLP:journals/logcom/Platzer10} are their continuous analog, where the value of a given term $\ptermA$ is shown to increase along ODE solutions by showing that its rate of change is bounded below by a positive constant $\constt{\varepsilon} > 0$ along those solutions.
Recall this article's syntactic convention (\rref{subsec:syntax}) that term $\constt{\varepsilon}$ is not allowed to depend on any of the free variables $x_1, \dots, x_n$ appearing in the ODE and must therefore remain constant along the ODE solution.
\begin{corollary}[Atomic differential variants~\cite{DBLP:journals/logcom/Platzer10}]
\label{cor:atomicdvcmp}
The following proof rules (where $\cmp$ is either $\geq$ or $>$) are derivable in \dL.
Terms $\constt{\varepsilon},\constt{\ptermA_0}$ are constant for ODE $\D{x}=\genDE{x},\D{\timevar}=1$.
In rule~\irref{dVcmp}, the ODE $\D{x}=\genDE{x}$ has provable global solutions.\\
\begin{calculuscollection}
\begin{calculus}
\dinferenceRule[dVcmpA|dV$_\cmp^{\Gamma}$]{}
{\linferenceRule
  { \lsequent{\lnot{(\ptermA \cmp 0)}}{\lied[]{\genDE{x}}{\ptermA}\geq \constt{\varepsilon}}
  }
  {\lsequent{\Gamma,\ptermA = \constt{\ptermA_0},\timevar = 0,\ddiamond{\pevolve{\D{x}=\genDE{x},\D{\timevar}=1}}{\big(\constt{\ptermA_0} + \constt{\varepsilon} \timevar > 0\big)}}{\ddiamond{\pevolve{\D{x}=\genDE{x},\D{\timevar}=1}}{\ptermA \cmp 0}} }
}{}

\dinferenceRule[dVcmp|dV$_\cmp$]{}
{\linferenceRule
  { \lsequent{\lnot{(\ptermA \cmp 0)}}{\lied[]{\genDE{x}}{\ptermA}\geq \constt{\varepsilon}}
  }
  {\lsequent{\Gamma, \constt{\varepsilon} > 0}{\ddiamond{\pevolve{\D{x}=\genDE{x}}}{\ptermA \cmp 0}} }
}{}
\end{calculus}
\end{calculuscollection}
\end{corollary}
\begin{proofsketcha}[app:livenodomproofs]{proof:proof12}
Rule~\irref{dVcmpA} is derived from axiom \irref{Prog} with $\rgvar \mnodefequiv \big( \constt{\ptermA_0} + \constt{\varepsilon} \timevar > 0 \big)$:

{\footnotesizeoff%
\begin{sequentdeduction}[array]
\linfer[Prog]{
  \lsequent{\Gamma,\ptermA = \constt{\ptermA_0},\timevar = 0}{\dbox{\pevolvein{\D{x}=\genDE{x},\D{\timevar}=1}{\lnot{(\ptermA \cmp 0)}}}{\big( \constt{\ptermA_0} + \constt{\varepsilon} \timevar \leq 0 \big)}}
}
  {\lsequent{\Gamma,\ptermA = \constt{\ptermA_0},\timevar = 0,\ddiamond{\pevolve{\D{x}=\genDE{x},\D{\timevar}=1}}{\big(\constt{\ptermA_0} + \constt{\varepsilon} \timevar > 0\big)}}{\ddiamond{\pevolve{\D{x}=\genDE{x},\D{\timevar}=1}}{\ptermA \cmp 0}}}
\end{sequentdeduction}
}%

Monotonicity \irref{MbW} strengthens the postcondition to $\ptermA \geq \constt{\ptermA_0} + \constt{\varepsilon} \timevar$ with the domain constraint $\lnot{(\ptermA \cmp 0)}$.
A subsequent use of \irref{dIcmp} completes the derivation:
{\footnotesizeoff%
\begin{sequentdeduction}[array]
  \linfer[MbW]{
  \linfer[dIcmp]{
    \lsequent{\lnot{(\ptermA \cmp 0)}}{\lied[]{\genDE{x}}{\ptermA}\geq \constt{\varepsilon}}
  }
    {\lsequent{\Gamma, \ptermA = \constt{\ptermA_0}, \timevar=0}{\dbox{\pevolvein{\D{x}=\genDE{x},\D{\timevar}=1}{\lnot{(\ptermA \cmp 0)}}}{\big( \ptermA \geq \constt{\ptermA_0} + \constt{\varepsilon} \timevar \big)}}}
  }
  {\lsequent{\Gamma,\ptermA = \constt{\ptermA_0}, \timevar=0}{\dbox{\pevolvein{\D{x}=\genDE{x},\D{\timevar}=1}{\lnot{(\ptermA \cmp 0)}}}{\big( \constt{\ptermA_0} + \constt{\varepsilon} \timevar \leq 0 \big)}}}
\end{sequentdeduction}
}%

Rule~\irref{dVcmp} is derived in~\rref{app:livenodomproofs} as a corollary of rule~\irref{dVcmpA} because the ODE $\D{x}=\genDE{x}$ is assumed to have solutions which (provably) exist globally.
\end{proofsketcha}

\begin{figure}[t]%

\begin{subfigure}{0.45\textwidth}
\centering
\includegraphics[width=0.9\textwidth]{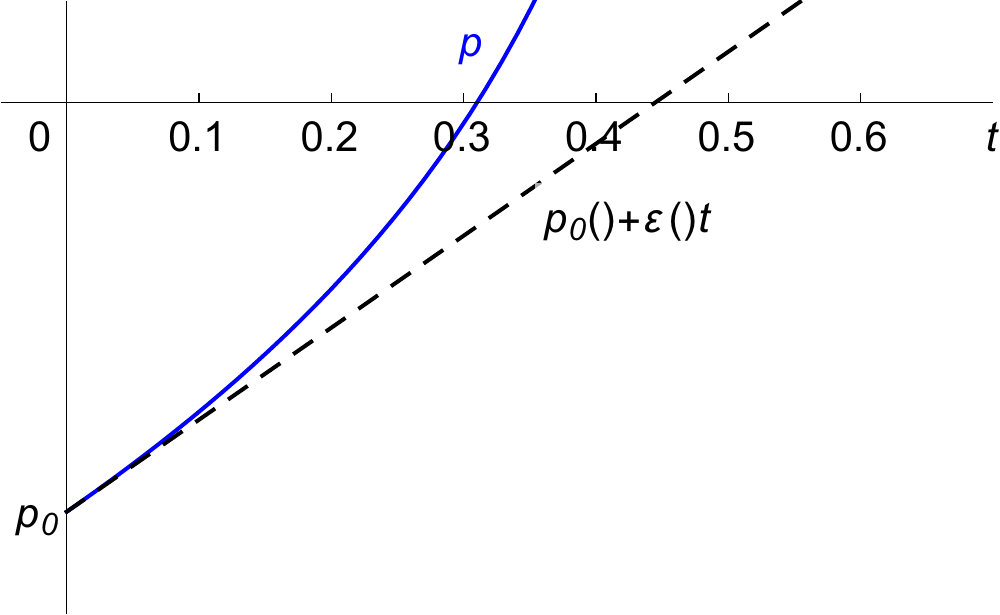}
\caption{$p \mnodefeq 2u+3v-4$, initial value $u=1, v=0$\\~}
\label{fig:dvexamples1}
\end{subfigure}%
\begin{subfigure}{0.45\textwidth}
\centering
\includegraphics[width=0.9\textwidth]{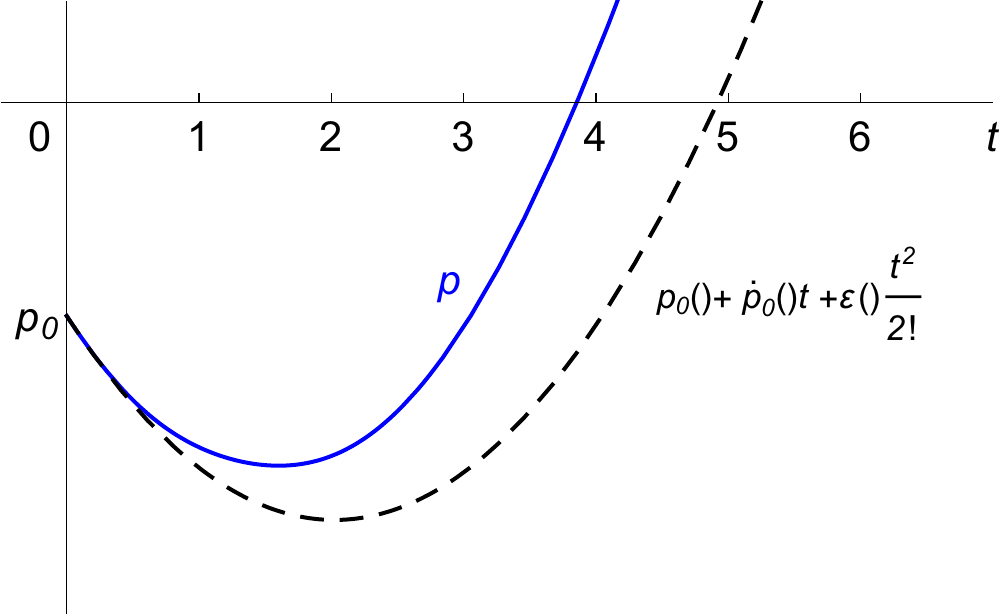}
\caption{$p \mnodefeq -2 + 2 u - u^2 - 5 u^3 + 5 v + u v - 2 u^3 v + v^2 - 5 u v^2 - u v^3$, initial value $u=-0.52, v=0$}
\label{fig:dvexamples2}
\end{subfigure}\\
\begin{subfigure}{0.45\textwidth}
\centering
\includegraphics[width=0.9\textwidth]{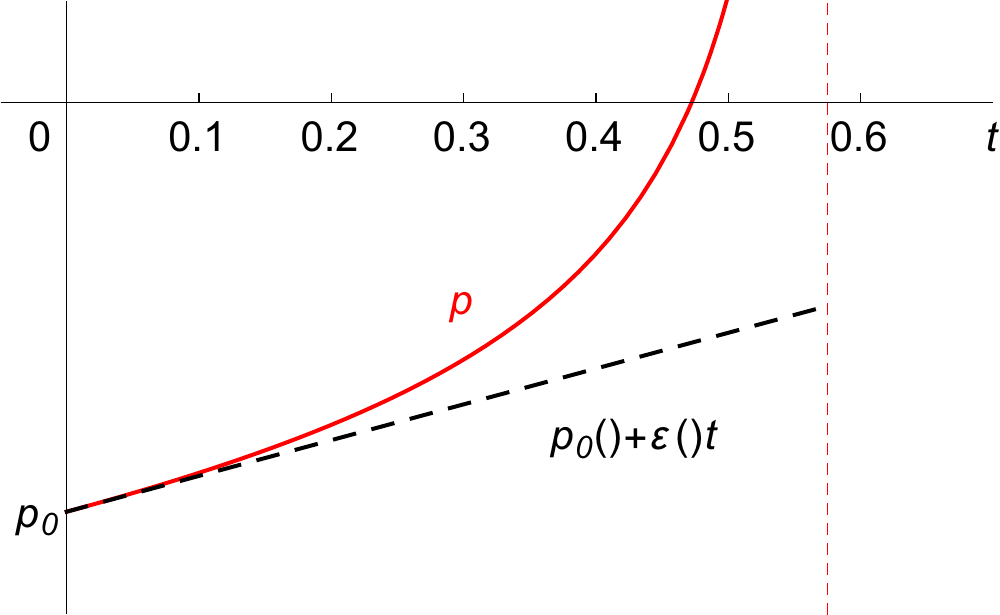}
\caption{$p \mnodefeq u+v-3$, initial value $u=1, v=0$}
\label{fig:dvexamples3}
\end{subfigure}%
\begin{subfigure}{0.45\textwidth}
\centering
\includegraphics[width=0.9\textwidth]{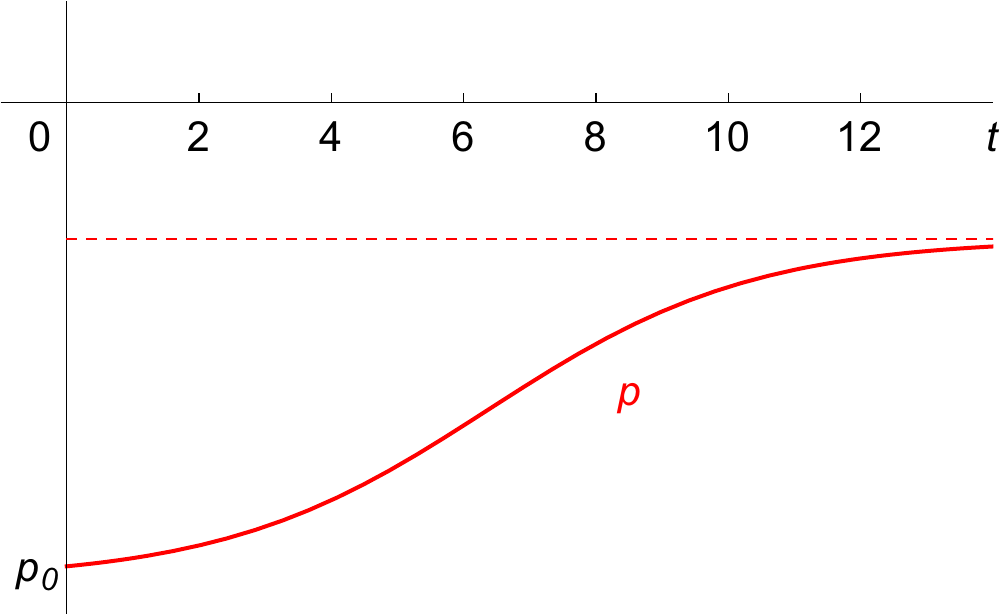}
\caption{$p \mnodefeq -10(u^2+v^2)-1$, initial value $u=0.49, v=0$}
\label{fig:dvexamples4}
\end{subfigure}%
\caption{The solid \bluec{blue} and \redc{red} curves show the value of various terms $\ptermA$ evaluated along solutions of the ODE $\exnonlinear$~\rref{eq:exnonlinear} from respective initial values $\ptermA_0$ over time $\timevar$.
The \bluec{blue} curves in Figs.~\ref{fig:dvexamples1} and~\ref{fig:dvexamples2} are respectively bounded below by a dashed black line (corresponding to~\rref{cor:atomicdvcmp}) and a dashed quadratic curve (corresponding to~\rref{cor:higherdv}) which imply that $\ptermA$ is eventually non-negative along their respective ODE solutions.
The \redc{red} curve in~\rref{fig:dvexamples3} is also bounded below by the dashed black line, but its solution only exists for $0.575$ time units (vertical \redc{red} dashed asymptote) so the direct bound from~\rref{cor:atomicdvcmp} fails, even though $\ptermA$ is eventually non-negative along the solution.
The \redc{red} curve in~\rref{fig:dvexamples4} has strictly positive time derivative but asymptotically increases towards a negative value (horizontal \redc{red} dashed asymptote).}
\label{fig:dvexamples}
\end{figure}

In both rules~\irref{dVcmpA+dVcmp}, the lower bound $\constt{\varepsilon} > 0$ on the Lie derivative $\lied[]{\genDE{x}}{\ptermA}$ ensures that the value of $\ptermA$ strictly increases along solutions to the ODE.
Geometrically, as illustrated in~\rref{fig:dvexamples1}, the value of $\ptermA$ is bounded below over time $\timevar$ by the line $\constt{\ptermA_0} + \constt{\varepsilon} \timevar$ with offset $\constt{\ptermA_0}$ and positive slope $\constt{\varepsilon}$.
Since $\constt{\ptermA_0} + \constt{\varepsilon} \timevar$ is non-negative for sufficiently large values of $\timevar$, the (lower bounded) value of $\ptermA$ is also eventually non-negative.

Two key subtleties underlying rules~\irref{dVcmpA+dVcmp} are illustrated in Figs.~\ref{fig:dvexamples3} and~\ref{fig:dvexamples4}.
The first subtlety, shown in~\rref{fig:dvexamples3}, is that ODE solutions must exist for sufficiently long for $\ptermA$ or, more precisely its lower bound, to become non-negative.
This is usually left as a soundness-critical side condition in liveness proof rules~\cite{DBLP:journals/logcom/Platzer10,DBLP:conf/fm/SogokonJ15}, but any such side condition is antithetical to approaches for minimizing the soundness-critical core in implementations~\cite{DBLP:journals/jar/Platzer17} because it requires checking the (semantic) condition that solutions exist for sufficient duration.
The conclusion of rule~\irref{dVcmpA} formalizes this side condition as an assumption.
In contrast, rule~\irref{dVcmp} requires provable global existence for the ODEs (provable as in~\rref{sec:globexist}).
The rest of this article similarly develops ODE liveness proof rules that rely on the global existence proofs from~\rref{sec:globexist}.
In all subsequent proof rules, the ODE $\D{x}=\genDE{x}$ is said to have \emph{provable global solutions} if the global existence formula~\rref{eq:existence} for $\D{x}=\genDE{x}$ is provable.
For example, if $\D{x}=\genDE{x}$ were globally Lipschitz (or, as a special case, linear), then its global existence can be proven using axiom~\irref{GEx} from Corollaries~\ref{cor:globalexistbase} and~\ref{cor:globalexistaffine}.
For uniformity, all proof steps utilizing this assumption are marked with~\irref{GEx}, although proofs of global existence could use various other techniques described in~\rref{sec:globexist}.
All subsequent proof rules can be soundly presented with sufficient duration assumptions like~\irref{dVcmpA}, but those are omitted for brevity.

The second subtlety, shown in~\rref{fig:dvexamples4}, is that rules~\irref{dVcmpA+dVcmp} crucially need a~\emph{constant} positive lower bound on the Lie derivative $\lied[]{\genDE{x}}{\ptermA}\geq \constt{\varepsilon}$ for soundness~\cite{DBLP:journals/logcom/Platzer10} instead of merely requiring $\lied[]{\genDE{x}}{\ptermA} > 0$.
In the latter case, even though the value of $\ptermA$ is strictly increasing along solutions, it is not guaranteed to become non-negative in finite time because the rate of increase can itself converge to zero.
In fact, as~\rref{fig:dvexamples4} shows, $\ptermA$ may stay negative by asymptotically increasing towards a negative value as $\timevar$ approaches $\infty$.

\begin{example}[Linear liveness]
\label{ex:linproof}
The liveness property that \rref{fig:odeexamples} suggested for the linear ODE $\exlinear$~\rref{eq:exlinear} is proved by rule~\irref{dVcmp}.
The proof is shown on the left below and visualized on the right.
The first monotonicity step~\irref{MdW} strengthens the postcondition to the inner \bluec{blue} circle $u^2+v^2 = \frac{1}{4}$ contained within the \greenc{green} goal region, see refinement~\rref{eq:refinementimpl1}.
Next, since solutions satisfy $u^2+v^2=1$ initially (black circle), the~\irref{Prog} step expresses an intermediate value property: to show that the \emph{continuous} solution eventually reaches $u^2+v^2 = \frac{1}{4}$, it suffices to show that it eventually reaches $u^2+v^2 \leq \frac{1}{4}$ (also see \rref{cor:tt} below).
The postcondition is rearranged before~\irref{dVcmp} is used with $\constt{\varepsilon} \mnodefeq \frac{1}{2}$.
Its premise is proved by~\irref{qear} because the Lie derivative of $\frac{1}{4} - (u^2+v^2)$ with respect to $\exlinear$ is $2(u^2+v^2)$, which is bounded below by $\frac{1}{2}$ under the assumption $\frac{1}{4} - (u^2+v^2) < 0$.

\noindent
\begin{minipage}[b]{0.5\textwidth}
{\footnotesizeoff\renewcommand{\arraystretch}{1.3}%
\centering
\begin{sequentdeduction}[array]
  \linfer[MdW]{
  \linfer[Prog]{
  \linfer[]{
  \linfer[dVcmp]{
  \linfer[]{
  \linfer[qear]{
    \lclose
  }
    {\lsequent{\frac{1}{4} < u^2+v^2}{2(u^2+v^2) \geq \frac{1}{2}}}
  }
    {\lsequent{\frac{1}{4} - (u^2+v^2) < 0}{2(u^2+v^2) \geq \frac{1}{2}}}
  }
    {\lsequent{u^2+v^2=1}{\ddiamond{\exlinear}{\frac{1}{4} - (u^2+v^2) \geq 0}}}
  }
    {\lsequent{u^2+v^2=1}{\ddiamond{\exlinear}{u^2+v^2 \leq \frac{1}{4}}}}
  }
    {\lsequent{u^2+v^2=1}{\ddiamond{\exlinear}{u^2+v^2 = \frac{1}{4}}}}
  }
  {\lsequent{u^2+v^2=1}{\ddiamond{\exlinear}{\big(\frac{1}{4} \leq \lnorm{(u,v)} \leq \frac{1}{2}\big)}}}
\end{sequentdeduction}
}%
~\\
\end{minipage}\qquad\qquad%
\begin{minipage}[b]{0.22\textwidth}
\includegraphics[width=\textwidth]{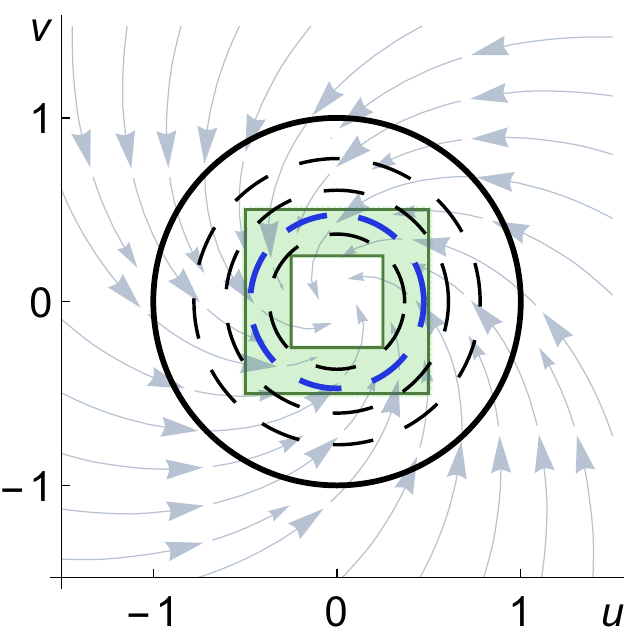}%
\vspace{3mm}
\end{minipage}%

The Lie derivative calculation shows that the value of $u^2+v^2$ decreases along solutions of $\exlinear$ with rate (at least) $\frac{1}{2}$ per unit time.
This is visualized by the shrinking (dashed) circles with radii eventually smaller than $\frac{1}{4}$.
Since the initial states satisfy $u^2+v^2 = 1$, a concrete upper bound on the time required for the solution to satisfy $u^2+v^2 \leq \frac{1}{4}$ is given by $(1 - \frac{1}{4}) \mathbin{/} \frac{1}{2} = \frac{3}{2}$ time units.
It is also instructive to examine the chain of refinements~\rref{eq:refinementchain} underlying the proof above.
Since $\exlinear$ is a linear ODE, the first~\irref{dVcmp} step refines the initial liveness property from~\irref{GEx}, i.e., that solutions exist globally (so for at least $\frac{3}{2}$ time units), to the property $u^2+v^2 \leq \frac{1}{4}$.
Subsequent refinement steps can be read off from the steps above from top-to-bottom:
{\footnotesizeoff%
\begin{align*}
  \ddiamond{\pevolve{\exlinear,\D{\timevar}=1}}{\timevar > \frac{3}{2}}
  ~{\limprefinechain{\irref{dVcmp}}}~
  \ddiamond{\exlinear}{u^2+v^2 \leq \frac{1}{4}}
  ~{\limprefinechain{\irref{Prog}}}~
  \ddiamond{\exlinear}{u^2+v^2 = \frac{1}{4}}
  ~{\limprefinechain{\irref{MdW}}}~
  \ddiamond{\exlinear}{\big(\frac{1}{4} \leq \lnorm{(u,v)} \leq \frac{1}{2}\big)}
\end{align*}
}%
\end{example}

The latter two steps illustrate the idea behind the next two surveyed proof rules.
In their original presentation~\cite{DBLP:conf/emsoft/TalyT10}, the ODE $\D{x}=\genDE{x}$ is only assumed to be \highlight{locally Lipschitz continuous}, which is insufficient for global existence of solutions, making the original rules unsound. See~\rref{app:counterexamples} for counterexamples.
Compared to~\rref{cor:atomicdvcmp},~\rref{cor:tt} below uses the fact that the value of differential variant $\ptermA$ evolves continuously along an ODE solution so it changes from $\ptermA \leq 0$ to $\ptermA > 0$ via $\ptermA = 0$.

\begin{corollary}[Equational differential variants~\cite{DBLP:conf/emsoft/TalyT10}]
\label{cor:tt}
The following proof rules are derivable in \dL.
Term $\constt{\varepsilon}$ is constant for ODE $\D{x}=\genDE{x}$, and the ODE has \highlight{provable global solutions} for both rules.\\
\begin{calculuscollection}
\begin{calculus}
\dinferenceRule[dVeq|dV$_=$]{}
{\linferenceRule
  { \lsequent{\ptermA < 0}{\lied[]{\genDE{x}}{\ptermA}\geq \constt{\varepsilon}}
  }
  {\lsequent{\Gamma,\constt{\varepsilon} > 0, \ptermA \leq 0}{\ddiamond{\pevolve{\D{x}=\genDE{x}}}{\ptermA = 0}} }
}{}
\end{calculus}
\qquad
\begin{calculus}
\dinferenceRule[TT|dV${_{=}^M}$]{}
{\linferenceRule
  { \lsequent{\ptermA = 0}{\rfvar}
   &\lsequent{\ptermA < 0}{\lied[]{\genDE{x}}{\ptermA}\geq \constt{\varepsilon}}
  }
  {\lsequent{\Gamma,\constt{\varepsilon} > 0, \ptermA \leq 0}{\ddiamond{\pevolve{\D{x}=\genDE{x}}}{\rfvar}} }
}{}
\end{calculus}
\end{calculuscollection}%
\end{corollary}
\begin{proofsketchb}[app:livenodomproofs]{proof:proof13}\end{proofsketchb}

The view of~\irref{dVcmp} as a refinement of~\irref{GEx} in~\rref{ex:linproof} also yields generalizations of~\irref{dVcmp} to higher Lie derivatives.
Indeed, it suffices that \emph{any} higher Lie derivative $\lied[k]{\genDE{x}}{\ptermA}$ is bounded below by a positive constant $\constt{\varepsilon}$ rather than just the first.
Geometrically, this guarantees that $\ptermA$ is bounded below by a degree $k$ polynomial in time variable $\timevar$ that is non-negative for large enough $\timevar$, see~\rref{fig:dvexamples2} for an illustration with $k=2$.

\begin{corollary}[Atomic higher differential variants]
\label{cor:higherdv}
The following proof rule (where $\cmp$ is either $\geq$ or $>$) is derivable in \dL.
Term $\constt{\varepsilon}$ is constant for ODE $\D{x}=\genDE{x}$, $k \geq 1$ is a freely chosen natural number, and the ODE has provable global solutions.\\
\begin{calculuscollection}
\begin{calculus}
\dinferenceRule[dVcmpK|dV$_\cmp^k$]{}
{\linferenceRule
  { \lsequent{\lnot{(\ptermA \cmp 0)}}{\lied[k]{\genDE{x}}{\ptermA}\geq \constt{\varepsilon}}
  }
  {\lsequent{\Gamma, \constt{\varepsilon} > 0}{\ddiamond{\pevolve{\D{x}=\genDE{x}}}{\ptermA \cmp 0}} }
}{}
\end{calculus}
\end{calculuscollection}
\end{corollary}
\begin{proofsketcha}[app:livenodomproofs]{proof:proof14}
Since $\lied[k]{\genDE{x}}{\ptermA}$ is strictly positive, all lower Lie derivatives $\lied[i]{\genDE{x}}{\ptermA}$ of $\ptermA$ for $i<k$, including $\ptermA \mnodefeq \lied[0]{\genDE{x}}{\ptermA}$, eventually become positive.
The derivation uses a sequence of~\irref{dC+dIcmp} steps.
\end{proofsketcha}

\subsection{Staging Sets}
The \emph{staging sets}~\cite{DBLP:conf/fm/SogokonJ15} proof rule adds flexibility to rules such as~\irref{TT} above by allowing users to choose a staging set formula $\rsfvar$ that \emph{the ODE can only leave by entering the goal region} $\rfvar$.
Staging sets are leaky invariants in the sense that they are almost invariant, except that they can be left by reaching the goal.
This staging property is expressed in the contrapositive by the box modality formula $\dbox{\pevolvein{\D{x}=\genDE{x}}{\lnot{\rfvar}}}{\rsfvar}$.%

\begin{corollary}[Staging sets~\cite{DBLP:conf/fm/SogokonJ15}]
\label{cor:SP}
The following proof rule is derivable in \dL.
Term $\constt{\varepsilon}$ is constant for ODE $\D{x}=\genDE{x}$, and the ODE has provable global solutions.\\
\begin{calculuscollection}
\begin{calculus}
\dinferenceRule[SP|SP]{}
{
\linferenceRule
  { \lsequent{\Gamma}{\dbox{\pevolvein{\D{x}=\genDE{x}}{\lnot{\rfvar}}}{\rsfvar}}
   &\lsequent{\rsfvar}{\ptermA \leq 0 \land \lied[]{\genDE{x}}{p}\geq \constt{\varepsilon}}
  }
  {\lsequent{\Gamma,\constt{\varepsilon}>0}{\ddiamond{\pevolve{\D{x}=\genDE{x}}}{\rfvar}} }
}{}
\end{calculus}
\end{calculuscollection}
\end{corollary}
\begin{proofsketcha}[app:livenodomproofs]{proof:proof15}
The derivation starts by using refinement axiom \irref{Prog} with $\rgvar \mnodefequiv \lnot{\rsfvar}$.
The rest of the derivation is similar to~\irref{dVcmpA+dVcmp}.
\end{proofsketcha}

The added choice of staging set formula $\rsfvar$ allows users to choose a staging set that, e.g., enables a liveness proof that uses a simpler differential variant $\ptermA$.
Furthermore, proof rules can be significantly simplified by choosing $\rsfvar$ with desirable topological properties.
For example, all of the liveness proof rules derived so far either have an explicit sufficient duration assumption (like~\irref{dVcmpA}) or assume that the ODEs have provable global solutions (like~\irref{dVcmp} using axiom~\irref{GEx}).
An alternative is to use axiom~\irref{BEx}, by choosing the staging set formula $\rsfvar(x)$ to characterize a bounded or compact set over the variables $x$.

\begin{corollary}[Bounded/compact staging sets]
\label{cor:boundedandcompact}
The following proof rules are derivable in \dL.
Term $\constt{\varepsilon}$ is constant for $\D{x}=\genDE{x}$.
In rule \irref{SPb}, formula $\rsfvar$ characterizes a bounded set over variables $x$.
In rule \irref{SPc}, it characterizes a compact, i.e., closed and bounded, set over those variables.\\
{\footnotesizeoff%
\begin{calculuscollection}
\begin{calculus}
\dinferenceRule[SPb|SP$_b$]{}
{
\linferenceRule
  { \lsequent{\Gamma}{\dbox{\pevolvein{\D{x}=\genDE{x}}{\lnot{\rfvar}}}{\rsfvar}}
   &\lsequent{\rsfvar}{\lied[]{\genDE{x}}{p} \geq \constt{\varepsilon}}
  }
  {\lsequent{\Gamma,\constt{\varepsilon} > 0}{\ddiamond{\pevolve{\D{x}=\genDE{x}}}{\rfvar}} }
}{}
\end{calculus}
\qquad
\begin{calculus}
\dinferenceRule[SPc|SP$_c$]{}
{
\linferenceRule
  { \lsequent{\Gamma}{\dbox{\pevolvein{\D{x}=\genDE{x}}{\lnot{\rfvar}}}{\rsfvar}}
   &\lsequent{\rsfvar}{\lied[]{\genDE{x}}{p} > 0}
  }
  {\lsequent{\Gamma}{\ddiamond{\pevolve{\D{x}=\genDE{x}}}{\rfvar}} }
}{}
\end{calculus}
\end{calculuscollection}
}%
\end{corollary}
\begin{proofsketcha}[app:livenodomproofs]{proof:proof16}
Rule~\irref{SPb} is derived using axiom~\irref{BEx} with differential variant $\ptermA$ to establish a time bound.
Rule~\irref{SPc} is an arithmetical corollary of~\irref{SPb}, using the fact that continuous functions on compact domains attain their extrema.
\end{proofsketcha}

\begin{example}[Nonlinear liveness]
\label{ex:nonlinproof}
The liveness property that \rref{fig:odeexamples} suggested for the nonlinear ODE $\exnonlinear$~\rref{eq:exnonlinear} is proved using rule~\irref{SPc} by choosing the staging set formula $\rsfvar \mnodefequiv 1 \leq u^2+v^2 \leq 2$ and the differential variant $\ptermA \mnodefeq u^2+v^2$.
The proof is shown on the left and visualized on the right below; the goal $u^2 + v^2 \geq 2$ is shown in \greenc{green} while $\rsfvar$ is shown as a \bluec{blue} annulus.
The Lie derivative $\lied[]{\genDE{x}}{\ptermA}$ with respect to $\exnonlinear$ is $2(u^2+v^2)(u^2+v^2-\frac{1}{4})$, which is bounded below by $\frac{3}{2}$ in $\rsfvar$.
Thus, the right premise of~\irref{SPc} closes trivially.
The left premise requires proving that $\rsfvar$ is an invariant within the domain constraint $\lnot{(u^2+v^2\geq 2)}$.
Intuitively, this is true because the \bluec{blue} annulus can only be left by entering the goal.
Its elided invariance proof is easy~\cite{DBLP:journals/jacm/PlatzerT20}.

\noindent
\begin{minipage}[b]{0.6\textwidth}
{\footnotesizeoff\renewcommand{\arraystretch}{1.3}%
\centering
\begin{sequentdeduction}[array]
  \linfer[SPc]{
 \linfer[cut+qear]{
  \linfer[]{
    \lclose
  }
  {\lsequent{\rsfvar}{\dbox{\pevolvein{\exnonlinear}{\lnot{(u^2+v^2\geq 2)}}}{\rsfvar}}}
  }
  {\lsequent{u^2+v^2=1}{\dbox{\pevolvein{\exnonlinear}{\lnot{(u^2+v^2\geq 2)}}}{\rsfvar}}} !
  \linfer[qear]{
    \lclose
  }
  {
  \lsequent{\rsfvar}{\lied[]{\genDE{x}}{\ptermA} > 0}}
  }
  {\lsequent{u^2+v^2=1}{\ddiamond{\exnonlinear}{u^2 + v^2 \geq 2}}}
\end{sequentdeduction}
}%
~\\~\\
\end{minipage}\qquad\qquad
\begin{minipage}[b]{0.22\textwidth}
\includegraphics[width=\textwidth]{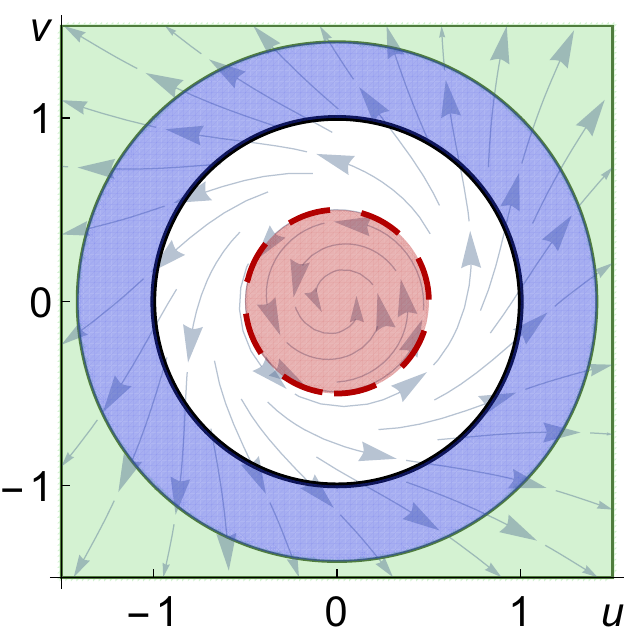}
\end{minipage}%

This proof exploits the flexibility provided by staging sets in two ways.
First, the formula $\rsfvar$ is chosen to characterize a compact set (as required by rule~\irref{SPc}).
As explained in~\rref{sec:globexist}, solutions of $\exnonlinear$ can blow up in finite time which necessitates the use of~\irref{BEx} for proving its liveness properties.
Second, $\rsfvar$ cleverly \emph{excludes} the \redc{red} disk (dashed boundary) characterized by $u^2+v^2 \leq \frac{1}{4}$.
Solutions of $\exnonlinear$ behave differently in this region, e.g., the Lie derivative $\lied[]{\genDE{x}}{\ptermA}$ is \emph{non-positive} in this disk.
The chain of refinements~\rref{eq:refinementchain} behind this proof can be seen from the derivation of rules~\irref{SPb+SPc} in~\rref{app:livenodomproofs}.
The chain starts from the initial liveness property~\irref{BEx} with concrete\footnote{The value of $u^2+v^2$ grows at rate $\frac{3}{2}$ per time unit along solutions and the initial states satisfy $u^2+v^2 = 1$.
Thus, a lower bound on time required to leave the staging set (when $u^2 + v^2 > 2$) is $(2 -1) \mathbin{/} \frac{3}{2} = \frac{2}{3}$ time units.} time bound $\frac{2}{3}$.
The first~\irref{Prog} step shows that the staging set is ultimately exited ($\ddiamond{\pevolve{\exnonlinear}}{\lnot{\rsfvar}}$), while the latter shows the desired liveness property:
\[
  \ddiamond{\pevolve{\exnonlinear,\D{\timevar}=1}}{(\timevar > \frac{2}{3} \lor \lnot{\rsfvar})}
  {\limprefinechain{\irref{Prog}}}
  \ddiamond{\pevolve{\exnonlinear}}{\lnot{\rsfvar}}
  {\limprefinechain{\irref{Prog}}}
  \ddiamond{\pevolve{\exnonlinear}}{u^2 + v^2 \geq 2}
\]
\end{example}

The need to use axiom \irref{BEx} (or otherwise, assume global existence) is subtle and is often overlooked in the surveyed liveness arguments.
An example of this is an incorrect claim~\cite[Remark 3.6]{DBLP:journals/siamco/PrajnaR07} that a liveness argument~\cite[Theorem 3.5]{DBLP:journals/siamco/PrajnaR07} works \highlight{without assuming that the relevant sets are bounded}.
This article's axiomatic approach can be used to find and fix errors involving these subtleties.
As another example, the following \emph{set Lyapunov function} proof rule adapts ideas from the literature~\cite[Theorem 2.4, Corollary 2.5]{DBLP:journals/siamco/RatschanS10} for proving liveness when the postcondition $\rfvar$ characterizes an open set.
The latter assumption on $\rfvar$ enables a convenient choice of staging set in rule~\irref{SPc} because $\lnot{\rfvar}$ characterizes a closed set.

\begin{corollary}[Set Lyapunov functions~\cite{DBLP:journals/siamco/RatschanS10}]
\label{cor:rs}
The following proof rule is derivable in \dL.
Formula $K$ characterizes a \highlight{compact set} over variables $x$, while formula $\rfvar$ characterizes an open set over those variables.\\
\begin{calculuscollection}
\begin{calculus}
\dinferenceRule[RS|SLyap]{}
{
\linferenceRule
  {
     \lsequent{\ptermA \geq 0}{K}
    &\lsequent{\lnot{\rfvar},K}{\lied[]{\genDE{x}}{\ptermA} > 0}
  }
  {\lsequent{\Gamma, \ptermA \cmp 0}{\ddiamond{\pevolve{\D{x}=\genDE{x}}}{\rfvar}} }
}{}
\end{calculus}
\end{calculuscollection}
\end{corollary}
\begin{proofsketcha}[app:livenodomproofs]{proof:proof17}
Rule~\irref{RS} is derived from rule~\irref{SPc} with $\rsfvar \mnodefequiv \lnot{\rfvar} \land K$, since $\lnot{\rfvar}$ characterizes a closed set, and the intersection of a closed set with a compact set is compact.
\end{proofsketcha}

Rule~\irref{RS} was claimed~\cite[Theorem 2.4, Corollary 2.5]{DBLP:journals/siamco/RatschanS10} to hold for any \highlight{closed} set $K$, when, in fact, $K$ crucially needs to be compact as seems to have been assumed implicitly in the proofs~\cite{DBLP:journals/siamco/RatschanS10}.

\section{Liveness With Domain Constraints}
\label{sec:withdomconstraint}

This section presents proof rules for liveness properties $\pevolvein{\D{x}=\genDE{x}}{\ivr}$ with domain constraint $\ivr$.
These properties are more subtle than liveness without domain constraints, because the limitation to a domain constraint $\ivr$ may make it impossible for an ODE solution to reach a desired goal region without leaving $\ivr$.

\begin{wrapfigure}[10]{r}{0.24\textwidth}
\centering
\includegraphics[width=0.22\textwidth]{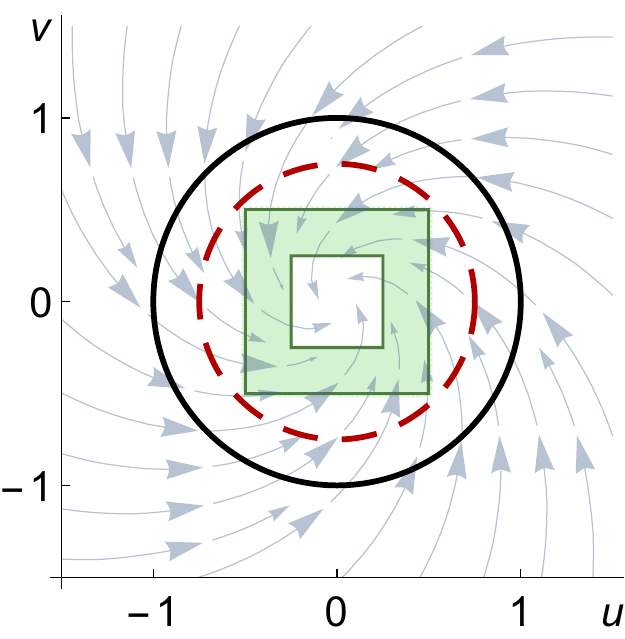}
\end{wrapfigure}

Consider the following liveness property for $\exlinear$~\rref{eq:exlinear} (visualized on the right), which adds domain constraint $\ivr \mnodefequiv u^2+v^2 \not= \frac{9}{16}$ restricting solutions from entering the \redc{red} dashed circle before reaching the \greenc{green} goal region.
\begin{equation}
\ddiamond{\pevolvein{\exlinear}{u^2+v^2 \not= \frac{9}{16}}}{\big(\frac{1}{4} \leq \lnorm{(u,v)} \leq \frac{1}{2}\big)}
\label{eq:baddomain}
\end{equation}

As proved in~\rref{ex:linproof}, solutions starting from the black circle $u^2+v^2=1$ reach the \greenc{green} goal region.
However, the continuous solutions must cross the \redc{red} dashed circle $u^2+v^2 = \frac{9}{16}$ to reach the goal, see discussion of implication~\rref{eq:refinementimpl2}.
This violates the domain constraint and falsifies~\rref{eq:baddomain} for initial states on the black circle.

Axiom~\irref{dDR} with $\rrfvar \mnodefequiv \ltrue$ provides one way of soundly and directly generalizing the proof rules from~\rref{sec:nodomconstraint}, as shown in the following example.

\begin{example}[Nonlinear liveness with domain]
\label{ex:nonlindomproof}
The liveness property $u^2+v^2=1 \limply \ddiamond{\exnonlinear}{u^2 + v^2 \geq 2}$ was proved in~\rref{ex:nonlinproof} for the nonlinear ODE $\exnonlinear$~\rref{eq:exnonlinear}.
The following derivation proves a stronger liveness property with the added domain constraint~$1 \leq u^2 + v^2$ by extending the proof from~\rref{ex:nonlinproof} with a~\irref{dDR} refinement step.
The resulting left premise is an invariance property of the ODE whose proof is elided~\cite{DBLP:journals/jacm/PlatzerT20}; intuitively, solutions starting from $u^2+v^2=1$ grow outwards, and so they remain in the domain $1 \leq u^2 + v^2$ (see~\rref{fig:odeexamples}).
The resulting right premise is proved in~\rref{ex:nonlinproof}.

{\footnotesizeoff\renewcommand{\arraystretch}{1.3}%
\centering
\begin{sequentdeduction}[array]
  \linfer[dDR]{
    \linfer[]{
      \lclose
    }
    {\lsequent{u^2+v^2=1}{\dbox{\exnonlinear}{1 \leq u^2 + v^2}}} !
    \linfer[]{
      \lclose
    }
    {\lsequent{u^2+v^2=1}{\ddiamond{\exnonlinear}{u^2 + v^2 \geq 2}}}
  }
  {\lsequent{u^2+v^2=1}{\ddiamond{\pevolvein{\exnonlinear}{1 \leq u^2 + v^2}}{u^2 + v^2 \geq 2}}}
\end{sequentdeduction}
}%

\end{example}

However, liveness arguments become much more intricate when attempting to generalize beyond domain constraint refinement with~\irref{dDR}, e.g., recall the unsound conjecture~\irref{badaxiom}.
Indeed, unlike the technical glitches of~\rref{sec:nodomconstraint}, this article uncovers several subtle soundness-critical errors in the literature.
With \dL's deductive approach, these intricacies are isolated to the topological axioms (\rref{lem:diatopaxioms}) which have been proved sound once-and-for-all.
Errors and omissions in the surveyed techniques are again \highlight{highlighted in blue}.

The following proof rule generalizes differential variants~\irref{dVcmp} to handle domain constraints.
Like rule~\irref{dVcmp}, the differential variant $\ptermA$ is guaranteed to eventually become non-negative along solutions with constant positive lower bound $\lied[]{\genDE{x}}{\ptermA}\geq \constt{\varepsilon}$ on its Lie derivative.
The additional twist is that the domain constraint $\ivr$ must be proved to hold as long as $\ptermA$ is still negative, i.e., while the goal has not been reached.
This is expressed in the contrapositive by the formula $\dbox{\pevolvein{\D{x}=\genDE{x}}{\lnot{(\ptermA \cmp 0)}}}{\ivr}$ in the left premise of the rule.

\begin{corollary}[Atomic differential variants with domains~\cite{DBLP:journals/logcom/Platzer10}]
\label{cor:atomicdvcmpQ}
The following proof rule (where $\cmp$ is either $\geq$ or $>$) is derivable in \dL.
Term $\constt{\varepsilon}$ is constant for the ODE $\D{x}=\genDE{x}$, and the ODE has provable global solutions.
\highlight{Formula $\ivr$ characterizes a closed (resp. open) set when $\cmp$ is $\geq$ (resp. $>$).}\\
\begin{calculuscollection}
\begin{calculus}
\dinferenceRule[dVcmpQ|dV$_\cmp\&$]{}
{\linferenceRule
  {
    \lsequent{\Gamma}{\dbox{\pevolvein{\D{x}=\genDE{x}}{\lnot{(\ptermA \cmp 0)}}}{\ivr}}
   &\lsequent{\lnot{(\ptermA \cmp 0)}, \ivr}{\lied[]{\genDE{x}}{\ptermA}\geq \constt{\varepsilon}}
  }
  {\lsequent{\Gamma,\constt{\varepsilon} > 0,\highlight{\lnot{(\ptermA \cmp 0)}}}{\ddiamond{\pevolvein{\D{x}=\genDE{x}}{\ivr}}{\ptermA \cmp 0}} }
}{}
\end{calculus}
\end{calculuscollection}
\end{corollary}
\begin{proofsketcha}[app:livewithdomproofs]{proof:proof18}
The derivation uses axiom~\irref{CORef} choosing $\rrfvar \mnodefequiv \ltrue$, noting that $\ptermA \geq 0$ (resp. $\ptermA > 0$) characterizes a topologically closed (resp. open) set so the appropriate topological requirements of~\irref{CORef} are satisfied.
The highlighted $\highlight{\lnot{(\ptermA \cmp 0)}}$ assumption is crucial for soundly using axiom~\irref{CORef}:
{\footnotesizeoff%
\begin{sequentdeduction}[array]
\linfer[CORef]{
  \lsequent{\Gamma}{\dbox{\pevolvein{\D{x}=\genDE{x}}{\lnot{(\ptermA \cmp 0)}}}{\ivr}} !
  \linfer[]{
  \linfer[]{
    \lsequent{\lnot{(\ptermA \cmp 0)}, \ivr}{\lied[]{\genDE{x}}{\ptermA}\geq \constt{\varepsilon}}
  }
    {\dots}
  }
  {\lsequent{\Gamma,\constt{\varepsilon} > 0}{\ddiamond{\pevolve{\D{x}=\genDE{x}}}{p\cmp 0}}}
}
{\lsequent{\Gamma,\constt{\varepsilon} > 0, \highlight{\lnot{(\ptermA \cmp 0)}}}{\ddiamond{\pevolvein{\D{x}=\genDE{x}}{\ivr}}{\ptermA \cmp 0}}}
\end{sequentdeduction}
}%

The derivation steps on the right premise are similar to the ones used in~\irref{dVcmp} although an intervening~\irref{dC} step is used to additionally assume $\ivr$ in the antecedents.
\end{proofsketcha}

Rule~\irref{dVcmpQ} uses the topological refinement axiom~\irref{CORef} to extend the refinement chain for~\irref{dVcmp} as follows:
\begin{equation}
  \cdots
  {\limprefinechain{\irref{dVcmp}}}
  \ddiamond{\pevolve{\D{x}=\genDE{x}}}{\ptermA \cmp 0}
  {\limprefinechain{\irref{CORef}}}
  \ddiamond{\pevolvein{\D{x}=\genDE{x}}{\ivr}}{\ptermA \cmp 0}
\label{eq:corefchain}
\end{equation}

A subtle advantage of placing the refinement~\irref{CORef} at the end of the refinement chain~\rref{eq:corefchain} is it decouples reasoning about domain constraint $\ivr$ from earlier steps refinement steps.
Notably, earlier refinement steps like~\irref{dVcmp} in the chain above can focus on handling other subtleties, such as sufficient duration existence of solutions (\rref{sec:globexist}), \emph{without} worrying about domain constraints.
The original presentation of rule~\irref{dVcmpQ}~\cite{DBLP:journals/logcom/Platzer10} omits the highlighted $\highlight{\lnot{(\ptermA \cmp 0)}}$ assumption, but the rule is unsound without it.
In addition, the original presentation uses a form of syntactic weak negation~\cite{DBLP:journals/logcom/Platzer10}, which is unsound for open postconditions, as pointed out earlier~\cite{DBLP:conf/fm/SogokonJ15}.
See~\rref{app:counterexamples} for counterexamples.

The proofs of the next two corollaries also make use of axiom~\irref{CORef} to derive the proof rule~\irref{TTQ}~\cite{DBLP:conf/emsoft/TalyT10} and the adapted rule~\irref{RSQ}~\cite{DBLP:journals/siamco/RatschanS10}.
These rules respectively generalize~\irref{TT} and~\irref{RS} from~\rref{sec:nodomconstraint} to handle domain constraints.
The technical glitches in their original presentations~\cite{DBLP:journals/siamco/RatschanS10,DBLP:conf/emsoft/TalyT10}, which were identified in \rref{sec:nodomconstraint}, remain highlighted here.

Like rule~\irref{dVcmpQ}, rules~\irref{dVeqQ+TTQ} below have an additional premise requiring that the domain constraint $\ivr$ provably holds while the goal has not yet been reached $\dbox{\pevolvein{\D{x}=\genDE{x}}{\ptermA < 0}}{\ivr}$.

\begin{corollary}[Equational differential variants with domains~\cite{DBLP:conf/emsoft/TalyT10}]
\label{cor:ttq}
The following proof rules are derivable in \dL.
Term $\constt{\varepsilon}$ is constant for the ODE $\D{x}=\genDE{x}$, and the ODE has \highlight{provable global solutions} for both rules. %
Formula $\ivr$ characterizes a closed set over variables $x$.\\
\begin{calculuscollection}
\begin{calculus}
\dinferenceRule[dVeqQ|dV$_=\&$]{}
{\linferenceRule
  { \lsequent{\Gamma}{\dbox{\pevolvein{\D{x}=\genDE{x}}{\ptermA < 0}}{\ivr}}
    &\lsequent{\ptermA < 0, \ivr}{\lied[]{\genDE{x}}{\ptermA}\geq \constt{\varepsilon}}
  }
  {\lsequent{\Gamma,\constt{\varepsilon} > 0, \ptermA \leq 0, \ivr}{\ddiamond{\pevolvein{\D{x}=\genDE{x}}{\ivr}}{\ptermA = 0}} }
}{}

\dinferenceRule[TTQ|dV${_{=}^M\&}$]{}
{\linferenceRule
  {
    \lsequent{\ivr, \ptermA = 0}{\rfvar}
   &\lsequent{\Gamma}{\dbox{\pevolvein{\D{x}=\genDE{x}}{\ptermA < 0}}{\ivr}} \quad
    \lsequent{\ptermA < 0, \ivr}{\lied[]{\genDE{x}}{\ptermA}\geq \constt{\varepsilon}}
  }
  {\lsequent{\Gamma,\constt{\varepsilon} > 0, \ptermA\leq 0, \ivr}{\ddiamond{\pevolvein{\D{x}=\genDE{x}}{\ivr}}{\rfvar}} }
}{}
\end{calculus}
\end{calculuscollection}
\end{corollary}
\begin{proofsketcha}[app:livewithdomproofs]{proof:proof19}
Rules~\irref{dVeqQ+TTQ} are both derived from rule~\irref{dVcmpQ} with $\geq$ for $\cmp$, since $\ivr$ characterizes a closed set.
Their derivations are respectively similar to the derivation of~\irref{dVeq+TT} from~\irref{dVcmp} and require the \highlight{provable global solutions} assumption for soundly applying rule~\irref{dVcmpQ}.
\qedhere
\end{proofsketcha}

Rule~\irref{RSQ} below has identical premises to the corresponding~\irref{RS} rule (without domain constraints).
The additional insight is that, assuming $\ptermA > 0$ is true initially, those same premises can be used to conclude the stronger liveness property $\ddiamond{\pevolvein{\D{x}=\genDE{x}}{\ptermA > 0}}{\rfvar}$ because $\ptermA$ can be additionally proved to stay positive along the solutions using the premises.
This stronger conclusion can be used with a monotonicity step to prove more general liveness properties with an arbitrary domain constraint $\ivr$ as exemplified by rule~\irref{RSQM}.

\begin{corollary}[Set Lyapunov functions with domains~\cite{DBLP:journals/siamco/RatschanS10}]
\label{cor:rsq}
The following proof rules are derivable in \dL.
Formula $K$ characterizes a \highlight{compact set} over variables $x$, while formula $\rfvar$ characterizes an open set over those variables.\\
\begin{calculuscollection}
\begin{calculus}
\dinferenceRule[RSQ|SLyap$\&$]{}
{
\linferenceRule
  {
     \lsequent{\ptermA \geq 0}{K}
     \quad
     \lsequent{\lnot{\rfvar},K}{\lied[]{\genDE{x}}{\ptermA} > 0}
  }
  {\lsequent{\Gamma, \ptermA > 0}{\ddiamond{\pevolvein{\D{x}=\genDE{x}}{\ptermA > 0}}{\rfvar}} }
}{}

\dinferenceRule[RSQM|SLyap$^M\&$]{}
{
\linferenceRule
  {
     \lsequent{\ptermA \geq 0}{K}
     \quad
     \lsequent{\lnot{\rfvar},K}{\lied[]{\genDE{x}}{\ptermA} > 0}
     \quad
     \lsequent{\ptermA > 0}{\ivr}
  }
  {\lsequent{\Gamma, \ptermA > 0}{\ddiamond{\pevolvein{\D{x}=\genDE{x}}{\ivr}}{\rfvar}} }
}{}

\end{calculus}
\end{calculuscollection}
\end{corollary}
\begin{proofsketcha}[app:livewithdomproofs]{proof:proof20}
Rule~\irref{RSQM} is derived from rule~\irref{RSQ} by monotonicity on the domain constraints with the additional premise $\lsequent{\ptermA > 0}{\ivr}$.
Rule~\irref{RSQ} is derived from~\irref{RS} after a refinement step with~\irref{CORef} since both formulas $\ptermA > 0$ and $\rfvar$ characterize open sets as sketched below.
{\footnotesizeoff%
\begin{sequentdeduction}[array]
  \linfer[CORef]{
    \linfer[]{
    \linfer[]{
      \lsequent{\ptermA \geq 0}{K} !
      \lsequent{\lnot{\rfvar},K}{\lied[]{\genDE{x}}{\ptermA} > 0}
    }
    {\dots}
    }
    {\lsequent{\Gamma, \ptermA > 0}{\dbox{\pevolvein{\D{x}=\genDE{x}}{\lnot{\rfvar}}}{\ptermA > 0}}}
    !
    \linfer[RS]{
     \lsequent{\ptermA \geq 0}{K} !
     \lsequent{\lnot{\rfvar},K}{\lied[]{\genDE{x}}{\ptermA} > 0}
    }
    {\lsequent{\Gamma, \ptermA > 0}{\ddiamond{\pevolve{\D{x}=\genDE{x}}}{\rfvar}}}
  }
  {\lsequent{\Gamma, \ptermA > 0}{\ddiamond{\pevolvein{\D{x}=\genDE{x}}{\ptermA > 0}}{\rfvar}} }
\end{sequentdeduction}
}%

The left premise proves the invariance of $\ptermA > 0$ for ODE $\D{x}=\genDE{x}$ with domain constraint $\rfvar$.
The elided derivation (see \hyperlink{proof:proof20}{proof}) reduces to two premises which are identical to those of rule~\irref{RS}.
The right premise uses rule~\irref{RS}, which necessitates the \highlight{compactness} assumption for formula $K$ for soundness.
\qedhere
\end{proofsketcha}

The following staging sets with domain constraints proof rule~\irref{SPQ}~\cite{DBLP:conf/fm/SogokonJ15} generalizes rule~\irref{SP} using axiom \irref{SARef}.
Notably, unlike the preceding rules, rule~\irref{SPQ} requires no topological assumptions\footnote{Aside from this article's standing assumption that $\rfvar, \ivr$ are formulas of first-order real arithmetic which is crucial for the soundness of axiom~\irref{SARef}.} about the domain constraint $\ivr$ nor of the goal region $\rfvar$ so it can be used in proofs of more general liveness properties.

\begin{corollary}[Staging sets with domains~\cite{DBLP:conf/fm/SogokonJ15}]
\label{cor:SPQ}
The following proof rule is derivable in \dL.
Term $\constt{\varepsilon}$ is constant for ODE $\D{x}=\genDE{x}$, and the ODE has provable global solutions.\\
\begin{calculuscollection}
\begin{calculus}
\dinferenceRule[SPQ|SP$\&$]{}
{
\linferenceRule
  { \lsequent{\Gamma}{\dbox{\pevolvein{\D{x}=\genDE{x}}{\lnot{(\rfvar \land \ivr)}}}{\rsfvar}}
   &\lsequent{\rsfvar}{\ivr \land \ptermA \leq 0 \land \lied[]{\genDE{x}}{p}\geq \constt{\varepsilon}}
  }
  {\lsequent{\Gamma,\constt{\varepsilon}>0}{\ddiamond{\pevolvein{\D{x}=\genDE{x}}{\ivr}}{\rfvar}} }
}{}
\end{calculus}
\end{calculuscollection}
\end{corollary}
\begin{proofsketcha}[app:livewithdomproofs]{proof:proof21}
The derivation starts with a~\irref{SARef} refinement step.
On the resulting left premise, an~\irref{MbW} monotonicity step yields the left premise and first (leftmost) conjunct of the right premise of rule~\irref{SPQ}.
On the resulting right premise, rule~\irref{SP} is used with a similar (see full~\hyperlink{proof:proof21}{proof}) monotonicity step, which yields the remaining conjuncts of the right premise of rule~\irref{SPQ}.
{\footnotesizeoff%
\begin{sequentdeduction}[array]
\linfer[SARef]{
  \linfer[MbW]{
    \lsequent{\Gamma}{\dbox{\pevolvein{\D{x}=\genDE{x}}{\lnot{(\rfvar \land \ivr)}}}{\rsfvar}} !
    \lsequent{\rsfvar}{\ivr}
  }
  {\lsequent{\Gamma}{\dbox{\pevolvein{\D{x}=\genDE{x}}{\lnot{(\rfvar \land \ivr)}}}{\ivr}}} !
  \linfer[SP]{
    \linfer[]{
      \lsequent{\rsfvar}{\ptermA \leq 0 \land \lied[]{\genDE{x}}{p}\geq \constt{\varepsilon}}
    }
    {\dots}
  }
  {\lsequent{\Gamma}{\ddiamond{\pevolve{\D{x}=\genDE{x}}}{\rfvar}}}
}
{\lsequent{\Gamma}{\ddiamond{\pevolvein{\D{x}=\genDE{x}}{\ivr}}{\rfvar}}}
\\[-\normalbaselineskip]\tag*{\qedhere}
\end{sequentdeduction}
}%
\end{proofsketcha}

The rules derived in Corollaries~\ref{cor:atomicdvcmpQ}--\ref{cor:SPQ} demonstrate the flexibility of \dL's refinement approach for deriving the surveyed liveness arguments as proof rules.
Indeed, their derivations are mostly straightforward adaptations of the corresponding rules presented in~\rref{sec:withdomconstraint}, with the appropriate addition of either a~\irref{CORef} or~\irref{SARef} axiomatic refinement step.
Moreover, the derived rules are sound, in contrast to the liveness arguments which were missing subtle assumptions in the literature (summarized in~\rref{tab:survey}).
The flexibility (and soundness) of this article's approach is not limited to the surveyed liveness arguments because refinement steps can also be freely mixed-and-matched for specific liveness questions.

\begin{example}[Strengthening]
\label{ex:strengthen}
The liveness property $u^2+v^2=1 \limply \ddiamond{\exnonlinear}{u^2 + v^2 \geq 2}$ for $\exnonlinear$~\rref{eq:exnonlinear} was proved in~\rref{ex:nonlinproof} using the staging set formula $\rsfvar \mnodefequiv 1 \leq u^2+v^2 \leq 2$, and provably strengthened in~\rref{ex:nonlindomproof} by adding the domain constraint $u^2+v^2 \geq 1$ with a~\irref{dDR} refinement.
Since $\rsfvar$ and $u^2 + v^2 \geq 2$ characterize closed sets, the refinement axiom~\irref{CORef} proves an even stronger liveness property with the strengthened domain $\rsfvar$, as shown in the derivation below.
The derivation starts with axiom~\irref{CORef} which yields three premises.
The leftmost premise is proved by~\irref{qear} since it is a real arithmetic fact, the middle premise proves because $\rsfvar$ is an invariant of the ODE $\exnonlinear$ (proof elided~\cite{DBLP:journals/jacm/PlatzerT20}), and the rightmost premise is proved in~\rref{ex:nonlinproof}.

{\footnotesizeoff\renewcommand{\arraystretch}{1.3}%
\renewcommand{\linferPremissSeparation}{\hspace{5pt}}%
\centering
\begin{sequentdeduction}[array]
  \linfer[CORef]{
    \linfer[qear]{
      \lclose
    }
    {\lsequent{u^2+v^2=1}{\lnot{(u^2 + v^2 \geq 2)}}}
    !
    \linfer[]{
      \lclose
    }
    {\lsequent{u^2+v^2=1}{\dbox{\pevolvein{\exnonlinear}{\lnot{(u^2 + v^2 \geq 2)}}}{\rsfvar}}} !
    \linfer[]{
      \lclose
    }
    {\lsequent{u^2+v^2=1}{\ddiamond{\exnonlinear}{u^2 + v^2 \geq 2}}}
  }
  {\lsequent{u^2+v^2=1}{\ddiamond{\pevolvein{\exnonlinear}{\rsfvar}}{u^2 + v^2 \geq 2}}}
\end{sequentdeduction}
}%

Axiom~\irref{CORef} extends the chain of refinements~\rref{eq:refinementchain} from~\rref{ex:nonlinproof} as follows:
\[
  \ddiamond{\pevolve{\exnonlinear,\D{\timevar}=1}}{(\timevar > \frac{2}{3} \lor \lnot{\rsfvar})}
  ~{\limprefinechain{\irref{Prog}}}~
  \ddiamond{\pevolve{\exnonlinear}}{\lnot{\rsfvar}}
  ~{\limprefinechain{\irref{Prog}}}~
  \ddiamond{\pevolve{\exnonlinear}}{u^2 + v^2 \geq 2}
  ~{\limprefinechain{\irref{CORef}}}~
  \ddiamond{\pevolvein{\exnonlinear}{\rsfvar}}{u^2 + v^2 \geq 2}
\]

The alternative staging set formula $\rsfvarhat \mnodefequiv 1 \leq u^2+v^2 < 2$ can also be used to prove \rref{ex:nonlinproof} with a similar refinement chain (using~\irref{SPb} instead of~\irref{SPc}), but $\rsfvarhat$ does \emph{not} characterize a closed set.
The topological restriction of axiom~\irref{CORef} crucially prevents its unsound use (indicated by $\usebox{\Lightningval}$):
\[
  \underbrace{\ddiamond{\pevolve{\exnonlinear,\D{\timevar}=1}}{(\timevar > \frac{2}{3} \lor \lnot{\rsfvarhat})}
  ~{\limprefinechain{\irref{Prog}}}~
  \ddiamond{\pevolve{\exnonlinear}}{\lnot{\rsfvarhat}}
  ~{\limprefinechain{\irref{Prog}}}~
  \ddiamond{\pevolve{\exnonlinear}}{u^2 + v^2 \geq 2}}_{\text{Similar to~\rref{ex:nonlinproof}}}
  ~\underbrace{\vphantom{\frac{2}{3}}\limprefinechain{\irref{CORef}\usebox{\Lightningval}}}_{\hidewidth\text{Unsound step!}\hidewidth}~
  \ddiamond{\pevolvein{\exnonlinear}{\rsfvar}}{u^2 + v^2 \geq 2}
\]

The liveness property $\ddiamond{\pevolvein{\exnonlinear}{\rsfvarhat}}{u^2 + v^2 \geq 2}$ is unsatisfiable because $\rsfvarhat$ does not overlap with $u^2 + v^2 \geq 2$.
The weakening of an inequality between domain constraints $\rsfvar$ and $\rsfvarhat$ leads to a wholly different conclusion!
\end{example}

The refinement approach also enables the discovery of new, general liveness proof rules by combining the underlying refinement steps in alternative ways.
As an example, the following chimeric proof rule combines ideas from Corollaries~\ref{cor:higherdv},~\ref{cor:boundedandcompact}, and~\ref{cor:SPQ}:

\begin{corollary}[Combination proof rule]
\label{cor:combination}
The following proof rule is derivable in \dL.
Formula $\rsfvar$ characterizes a compact set over variables $x$.\\
\begin{calculuscollection}
\begin{calculus}
\dinferenceRule[SPcQ|SP$_c^k\&$]{}
{
\linferenceRule
  { \lsequent{\Gamma}{\dbox{\pevolvein{\D{x}=\genDE{x}}{\lnot{(\rfvar \land \ivr)}}}{\rsfvar}}
   &\lsequent{\rsfvar}{\ivr \land \lied[k]{\genDE{x}}{p} > 0}
  }
  {\lsequent{\Gamma}{\ddiamond{\pevolvein{\D{x}=\genDE{x}}{\ivr}}{\rfvar}} }
}{}
\end{calculus}
\end{calculuscollection}
\end{corollary}
\begin{proofsketcha}[app:livewithdomproofs]{proof:proof22}
The derivation combines ideas from the derivations of~\irref{dVcmpK} (generalizing~\irref{dVcmp} to higher derivatives), \irref{SPc} (compact staging sets), and \irref{SPQ} (refining domain constraints).
\end{proofsketcha}

The logical approach of \dL derives complicated proof rules like~\irref{SPcQ} from a small set of sound logical axioms, which ensures their correctness.
The proof rule~\irref{PRQ} below is derived from rule~\irref{SPcQ} (for $k\mnodefeq1$) and is adapted from the literature~\cite[Theorem 3.5]{DBLP:journals/siamco/PrajnaR07}, where additional restrictions were imposed on the sets characterized by $\Gamma,\rfvar,\ivr$, and different conditions were given compared to the left premise of~\irref{PRQ} (\highlight{highlighted below}).
These original conditions were overly permissive as they are checked on sets that are smaller than necessary for soundness. See~\rref{app:counterexamples} for counterexamples to those original conditions.

\begin{corollary}[Compact eventuality~\cite{DBLP:journals/siamco/PrajnaR07}]
\label{cor:prq}
The following proof rule is derivable in \dL.
Formula $\ivr \land \lnot{\rfvar}$ characterizes a compact set over variables $x$.\\
\begin{calculuscollection}
\begin{calculus}
\dinferenceRule[PRQ|E$_c\&$]{}
{
\linferenceRule
  { \highlight{\lsequent{\Gamma}{\dbox{\pevolvein{\D{x}=\genDE{x}}{\lnot{(\rfvar \land \ivr)}}}{\ivr}}}
   &\lsequent{\ivr,\lnot{\rfvar}}{\lied[]{\genDE{x}}{p} > 0}
  }
  {\lsequent{\Gamma}{\ddiamond{\pevolvein{\D{x}=\genDE{x}}{\ivr}}{\rfvar}} }
}{}
\end{calculus}
\end{calculuscollection}
\end{corollary}
\begin{proofsketcha}[app:livewithdomproofs]{proof:proof23}
Rule~\irref{PRQ} is derived from~\irref{SPcQ} (for $k\mnodefeq1$), with $\rsfvar \mnodefequiv \ivr \land \lnot{\rfvar}$.
\end{proofsketcha}

\section{ODE Liveness Proofs in Practice}
\label{sec:impl}

The preceding sections show how axiomatic refinement can be used to fruitfully navigate and understand the zoo of ODE existence and liveness arguments from various applications (\rref{tab:survey}).
The generality of the approach enables the sound and foundational derivation of those arguments from a parsimonious basis of refinement steps.
This section provides a complementary study of how the refinement approach and its derived ODE existence and liveness proof rules are best implemented in practice.
There are two canonical approaches for such an implementation:
\begin{enumerate}
\item \label{itm:approach1} Implement the foundational refinement steps and let users build their own arguments using those steps, e.g., by following the derivations and proofs from Sections~\ref{sec:globexist}--\ref{sec:withdomconstraint}.
\item \label{itm:approach2} Implement the zoo of proof rules from Sections~\ref{sec:globexist}--\ref{sec:withdomconstraint} directly and let users pick from from those rules for their particular ODE liveness applications.
\end{enumerate}

The low-level flexibility of Approach~\rref{itm:approach1} is also its drawback in practice because users need to tediously reconstruct high-level ODE liveness arguments from basic refinements for each proof.
Approach~\rref{itm:approach2} provides users with those high-level arguments but limits users to proof rules that have been implemented, which squanders the generality of the refinement approach.
Moreover, users would still need to navigate the redundancies and tradeoffs among the zoo of proof rules to select one that is best-suited for their proof.
To account for these drawbacks, this section advocates for a middle ground between those two extremes: implementations should provide users with the basic refinement steps, bundled with a set of carefully curated, high-level proof rules (\rref{subsec:complex}) and associated proof support (\rref{subsec:support}) that help users navigate the common cases in their liveness proofs.

These ideas are put into practice through an implementation of ODE existence and liveness proof rules in \KeYmaeraX~\cite{DBLP:conf/cade/FultonMQVP15}.
Proof rules and proof support are implemented as \emph{tactics} in \KeYmaeraX~\cite{DBLP:conf/itp/FultonMBP17}, which are not soundness-critical.
Such an arrangement allows for the implementation of useful ODE liveness proof rules and their associated proof support with \KeYmaeraX's sound kernel as a safeguard against implementation errors or mistakes in their derivations and side conditions.
This core design decision underlying \KeYmaeraX is discussed elsewhere~\cite{DBLP:conf/cade/FultonMQVP15,DBLP:conf/itp/FultonMBP17,DBLP:journals/jar/Platzer17}.
All of the ODE liveness examples in this article have been formally proved in \KeYmaeraX (\rref{subsec:autoexamples}).
By leveraging existing infrastructure in \KeYmaeraX, the implementation can also be used as part of liveness proofs for hybrid systems.
It is used for the liveness proofs of a case study involving a robot model driving along circular arcs in the plane~\cite{DBLP:journals/ral/BohrerTMSP19}.

The basic refinements steps from~\rref{sec:livenessaxioms} and the proof rules in Sections~\ref{sec:globexist}--\ref{sec:withdomconstraint} are mostly straightforward to implement by following their respective proofs.
Thus, Sections~\ref{subsec:complex} and~\ref{subsec:support} focus on a select number of new proof rules and proof support that are beneficial in the implementation.
For the sake of completeness, syntactic derivations of all liveness proof rules presented in these sections are given in~\rref{app:implementationproofs}.

\subsection{Liveness Proof Rules}
\label{subsec:complex}

Atomic differential variants~\irref{dVcmp} is a useful primitive proof rule to implement in \KeYmaeraX because many ODE liveness proof rules, e.g.,~\irref{TT+SP}, derive from it.
From a practical perspective though, rule~\irref{dVcmp} as presented in~\rref{cor:atomicdvcmp} still requires users to provide a choice of the constant $\constt{\varepsilon}$, e.g., the proof in~\rref{ex:linproof} uses $\constt{\varepsilon} \mnodefeq \frac{1}{2}$.
The following slight rephrasing of~\irref{dVcmp} enables a more automated implementation.

\begin{corollary}[Existential atomic differential variants~\cite{DBLP:journals/logcom/Platzer10}]
\label{cor:atomicdvcmpexist}
The following proof rule (where $\cmp$ is either $\geq$ or $>$) is derivable in \dL, where $\varepsilon$ is a fresh variable and ODE $\D{x}=\genDE{x}$ has provable global solutions.

\begin{calculuscollection}
\begin{calculus}
\dinferenceRule[dVcmpE|dV$_\cmp^\exists$]{}
{\linferenceRule
  { \lsequent{\Gamma}{\lexists{\varepsilon>0}{\lforall{x}{\big(\lnot{(\ptermA \cmp 0)} \limply \lied[]{\genDE{x}}{\ptermA}\geq \varepsilon} \big)}}
  }
  {\lsequent{\Gamma}{\ddiamond{\pevolve{\D{x}=\genDE{x}}}{\ptermA \cmp 0}} }
}{}
\end{calculus}
\end{calculuscollection}
\end{corollary}
\begin{proofsketcha}[app:implementationproofs]{proof:proof24}
Rule~\irref{dVcmpE} is derived from~\irref{dVcmp} as a corollary.
\end{proofsketcha}

Just like rule~\irref{dVcmp}, rule~\irref{dVcmpE} requires a positive lower bound $\varepsilon > 0$ on the derivative of $\ptermA$ along solutions.
The difference is that the premise of rule~\irref{dVcmpE} is rephrased to ask a purely arithmetical question about the existence of a suitable choice for $\varepsilon$.
This can be decided automatically to save user effort in identifying $\varepsilon$, but such automation comes at added computational cost because the decision procedure must \emph{find} a suitable instance of $\varepsilon$ for the $\lexists{}$ quantifier (or decide that none exist) rather than simply \emph{check} a user-provided instance.
Thus, the implementation gives users control over the desired degree of automation in their proof by giving them the option of either invoking an arithmetic decision procedure~\irref{qear} on the premise of~\irref{dVcmpE} or manually instantiating the existential quantifier with a specific term for $\varepsilon$.

Another useful variation of rule~\irref{dVcmp} is its \emph{semialgebraic} generalization, i.e., where the goal region is described by a formula $\rfvar$ formed from conjunctions and disjunctions of (in)equalities.
Rules~\irref{TT+SP} provide examples of such a generalization, but they are indirect generalizations because users must still identify an underlying (atomic) differential variant $\ptermA$ as input when applying either rule.
In contrast, the new semialgebraic generalization of~\irref{dVcmp} below directly examines the syntactic structure of the goal region described by formula $\rfvar$.
Its implementation is enabled by \KeYmaeraX's ODE invariance proving capabilities which are, in turn, based on \dL's complete axiomatization for ODE invariants~\cite{DBLP:journals/jacm/PlatzerT20}.

\begin{corollary}[Semialgebraic differential variants]
\label{cor:semialgdv}
Let $b$ be a fresh variable, and term $\constt{\varepsilon}$ be constant for ODE $\D{x}=\genDE{x},\D{\timevar}=1$.
Let $\rfvar$ be a semialgebraic formula in the following normal form~\cite[Eq 5]{DBLP:journals/jacm/PlatzerT20}, and $\rgvar_\rfvar$ be its corresponding $\varepsilon$-progress formula (also in normal form):
\[\rfvar \mequiv \lorfold_{i=0}^{M} \Big(\landfold_{j=0}^{m(i)} \ptermA_{ij} \geq 0 \land \landfold_{j=0}^{n(i)} \ptermB_{ij} > 0\Big) \qquad \rgvar_\rfvar \mequiv \lorfold_{i=0}^{M} \Big(\landfold_{j=0}^{m(i)} \ptermA_{ij} - (b + \constt{\varepsilon}t) \geq 0 \land \landfold_{j=0}^{n(i)} \ptermB_{ij} - (b + \constt{\varepsilon}t) \geq 0 \Big) \]
The following proof rule is derivable in \dL, where the ODE $\D{x}=\genDE{x}$ has provable global solutions, and $\sigliedsai{\genDE{x}}{(\lnot{\rfvar})}, \sigliedsai{\genDE{x}}{(\rgvar_\rfvar)}$ are semialgebraic progress formulas~\cite[Def. 6.4]{DBLP:journals/jacm/PlatzerT20}\footnote{%
The arithmetic formula $\sigliedsai{\genDE{x}}{\rfvar}$ exactly characterizes that the ODE $\D{x}=\genDE{x}$ makes local progress in $\rfvar$ for some nonzero duration, see prior work~\cite[Def 6.4, Thm. 6.6]{DBLP:journals/jacm/PlatzerT20}.
The semialgebraic progress formula operator commutes with logical negation for semialgebraic formulas $\rfvar$, i.e., the equivalence $\sigliedsai{\genDE{x}}{(\lnot{\rfvar})} \lbisubjunct \lnot{(\sigliedsai{\genDE{x}}{\rfvar})}$ is provable~\cite[Cor 6.7]{DBLP:journals/jacm/PlatzerT20}.
Hence, local progress into $\lnot{\rfvar}$ is equivalent to the solution locally staying away from $\rfvar$ for nonzero duration.} with respect to $\D{x}=\genDE{x},\D{\timevar}=1$.\\
\begin{calculuscollection}
\begin{calculus}
\dinferenceRule[dV|dV]{}
{\linferenceRule
  { \lsequent{\lnot{\rfvar}, \sigliedsai{\genDE{x}}{(\lnot{\rfvar})}, \rgvar_\rfvar}{\sigliedsai{\genDE{x}}{(\rgvar_\rfvar)}}
  }
  {\lsequent{\Gamma, \constt{\varepsilon} > 0}{\ddiamond{\pevolve{\D{x}=\genDE{x}}}{\rfvar}} }
}{}

\dinferenceRule[dVE|dV$^\exists$]{}
{\linferenceRule
  { \lsequent{\Gamma}{
    \lexists{\varepsilon>0}{\lforall{b}{\lforall{t}{\lforall{x}{\big(\lnot{\rfvar} \land \sigliedsai{\genDE{x}}{(\lnot{\rfvar})} \land \rgvar_\rfvar \limply \sigliedsai{\genDE{x}}{(\rgvar_\rfvar)}\big)}}}}}
  }
  {\lsequent{\Gamma, \constt{\varepsilon} > 0}{\ddiamond{\pevolve{\D{x}=\genDE{x}}}{\rfvar}} }
}{}
\end{calculus}
\end{calculuscollection}
\end{corollary}
\begin{proofsketcha}[app:implementationproofs]{proof:proof25}
Rule~\irref{dVE} is derived from~\irref{dV} similar to the derivation of rule~\irref{dVcmpE} from~\irref{dVcmp}.
The derivation of~\irref{dV} is similar to rules~\irref{dVcmpA+dVcmp}, but replaces the use of rule~\irref{dIcmp} with the complete ODE invariance proof rule~\cite[Theorem 6.8]{DBLP:journals/jacm/PlatzerT20}.
The fresh variable $b$ is used as a lower bound of the value of all polynomials $\ptermA_{ij},\ptermB_{ij}$ appearing in the description of $\rfvar$ along solutions of the ODE.
\end{proofsketcha}

The intuition behind rule~\irref{dV} is similar to rule~\irref{dVcmp}: as long as the solution has not yet reached the goal $\rfvar$, it grows towards $\rfvar$ at ``rate'' $\constt{\varepsilon}$.
The technical challenge is how to formally phrase the ``rate'' of growth for a semialgebraic formula $\rfvar$, which does not have a well-defined notion of derivative.
Rule~\irref{dV} uses the $\varepsilon$-progress formula $\rgvar_\rfvar$, together with the semialgebraic progress formulas $\sigliedsai{\genDE{x}}{(\lnot{\rfvar})}, \sigliedsai{\genDE{x}}{(\rgvar_\rfvar)}$ and \dL's completeness result for ODE invariants~\cite[Theorem 6.8]{DBLP:journals/jacm/PlatzerT20} for this purpose.
These formulas give sufficient, although implicit, arithmetical conditions for proving liveness for $\rfvar$.
Rule~\irref{dVE} rephrases~\irref{dV} with an arithmetical premise, similar to how~\irref{dVcmpE} rephrases~\irref{dVcmp}, to give users the added flexibility of choosing between invoking an automated decision procedure or manually instantiating the existential quantifier for $\varepsilon$ and reasoning about the resulting progress formulas.
More explicit arithmetical premises for~\irref{dV+dVE} can be obtained by unfolding the definitions~\cite[Def. 6.4]{DBLP:journals/jacm/PlatzerT20} of $\sigliedsai{\genDE{x}}{(\lnot{\rfvar})}, \sigliedsai{\genDE{x}}{(\rgvar_\rfvar)}$ as exemplified below.

\begin{example}[Non-differentiable progress functions~\cite{DBLP:conf/fm/SogokonJ15}]
\label{ex:nondifferentiable}
Consider the following liveness formula with two inequalities in its postcondition:
\begin{equation}
\ddiamond{\D{u}=-u}{(-1 \leq u \leq 1)}
\label{eq:twoconjuncts}
\end{equation}

Using the $\min$ function, formula~\rref{eq:twoconjuncts} can be written equivalently with a single atomic inequality:
\begin{equation}
\ddiamond{\D{u}=-u}{\min(1-u, u+1) \geq 0}
\label{eq:oneatom}
\end{equation}

However, the postcondition of~\rref{eq:oneatom} is not a formula of real arithmetic (\rref{subsec:syntax}) and it does not have well-defined \dL semantics.
Indeed, rule~\irref{dVcmp} does not prove~\rref{eq:oneatom} because the Lie derivative of its postcondition is not well-defined.
One possible solution is to generalize~\irref{dVcmp} by considering directional derivatives of continuous (but non-differentiable) functions such as $\min,\max$~\cite[Section 5.2]{DBLP:conf/fm/SogokonJ15}.
However, justifying the correctness of this option would require delicate changes to \dL semantics~\cite{DBLP:conf/cade/BohrerFP19,DBLP:journals/jar/Platzer17}.
Rule~\irref{dV} instead proves~\rref{eq:twoconjuncts} directly without requiring rephrasing, nor complications associated with directional derivatives.
The proof is as follows, with $\constt{\varepsilon} \mnodefeq 1$ and $\rfvar \mnodefequiv u+1 \geq 0 \land 1-u \geq 0, \rgvar_\rfvar \mnodefequiv u+1-(b + t) \geq 0 \land 1-u-(b + t) \geq 0$:
{\footnotesizeoff\renewcommand{\arraystretch}{1.3}%
\begin{sequentdeduction}[array]
\linfer[dV]{
\linfer[]{
\linfer[]{
\linfer[qear]{
  \lclose
}
   {\lsequent{u + 1 < 0 \lor 1 - u < 0, u+1-(b + t) \geq 0 \land 1-u-(b + t) \geq 0}{(u+1-(b + t) = 0 \limply -u-1 > 0) \land \dots }}
}
  { \lsequent{\lnot{\rfvar}, \rgvar_\rfvar}{\sigliedsai{\genDE{x}}{(\rgvar_\rfvar)}}}
}
  { \lsequent{\lnot{\rfvar}, \sigliedsai{\genDE{x}}{(\lnot{\rfvar})}, \rgvar_\rfvar}{\sigliedsai{\genDE{x}}{(\rgvar_\rfvar)}}}
}
  {\lsequent{} {\ddiamond{\D{u}=-u}{(-1 \leq u \leq 1)}}}
\end{sequentdeduction}
}%

The proof starts by using rule~\irref{dV}, where the assumption $ \sigliedsai{\genDE{x}}{(\lnot{\rfvar})}$ in its premise is weakened as it is unnecessary for the proof.
Unfolding the definition of $\sigliedsai{\genDE{x}}{(\rgvar_\rfvar)}$ and simplifying leaves an arithmetical question in the succedent with two conjuncts; the right conjunct is omitted for brevity since the subsequent argument is symmetric.
The left conjunct in the succedent is proved by~\irref{qear} because the assumptions $u+1-(b + t) = 0$ and $u+1-(b + t) \geq 0 \land 1-u-(b + t) \geq 0$ imply $1-u \geq u+1$.
This, in turn, implies $-u -1 > 0$ using the assumption $u + 1 < 0 \lor 1 - u < 0$.

More generally, for a liveness postcondition comprising a conjunction of atomic inequalities $\ptermA \cmp 0 \land \ptermB \cmp 0$ (where $\cmp$ is either $\geq$ or $>$ in either conjunct), the premise resulting from applying~\irref{dV} can be simplified in real arithmetic to the following arithmetical premise:
\begin{equation}
\lsequent{\lnot{(\ptermA \cmp 0 \land \ptermB \cmp 0)}}{(\ptermA < \ptermB \limply \lied[]{\genDE{x}}{\ptermA} > \constt{\varepsilon}) \land (\ptermA > \ptermB \limply \lied[]{\genDE{x}}{\ptermB} > \constt{\varepsilon}) \land (\ptermA = \ptermB \limply \lied[]{\genDE{x}}{\ptermA} > \constt{\varepsilon} \land \lied[]{\genDE{x}}{\ptermB} > \constt{\varepsilon})}
\label{eq:premise}
\end{equation}

The arithmetical premise~\rref{eq:premise} is equivalent to the arithmetical progress conditions for $\min(p,q) \geq 0$~\cite[Example 14]{DBLP:conf/fm/SogokonJ15}, and both are decidable in real arithmetic.
The intuition behind~\rref{eq:premise} is that whenever $\ptermA$ is further from the goal than $\ptermB$, then $\ptermA$ is required to make $\varepsilon$ progress towards the goal (symmetrically when $\ptermB$ is further than $\ptermA$ from the goal).
A similar simplification of~\irref{dV} for a disjunctive postcondition $\ptermA \cmp 0 \lor \ptermB \cmp 0$ is shown in~\rref{eq:disjpremise}, which asks for the term closer to the goal to make $\varepsilon$ progress towards the goal instead.
Further simplifications for semialgebraic formulas $\rfvar$ are obtained as nested combinations of~\rref{eq:premise} and~\rref{eq:disjpremise}.
\begin{equation}
\lsequent{\lnot{(\ptermA \cmp 0 \lor \ptermB \cmp 0)}}{(\ptermA < \ptermB \limply \lied[]{\genDE{x}}{\ptermB} > \constt{\varepsilon}) \land (\ptermA > \ptermB \limply \lied[]{\genDE{x}}{\ptermA} > \constt{\varepsilon}) \land (\ptermA = \ptermB \limply \lied[]{\genDE{x}}{\ptermA} > \constt{\varepsilon} \lor \lied[]{\genDE{x}}{\ptermB} > \constt{\varepsilon})}
\label{eq:disjpremise}
\end{equation}

This example shows the intricate definition of semialgebraic progress formulas, even for the simple-looking conjunctive postcondition $ -1 \leq u \leq 1$, which highlights the need for a careful and trustworthy implementation of rules~\irref{dV+dVE}, as provided by \KeYmaeraX.
\end{example}

The variations of~\irref{dVcmp} shown in Corollaries~\ref{cor:atomicdvcmpexist} and~\ref{cor:semialgdv} (and their implementation) allow users to focus on high-level liveness arguments in \KeYmaeraX rather than low-level derivation steps.
Another key usability improvement afforded by an implementation is the sound and automatic enforcement of the appropriate side conditions for every proof rule.
The common side conditions for ODE liveness proof rules presented in this article can be broadly classified as follows:
\begin{enumerate}
\item Freshness side conditions on variables, e.g., in rules~\irref{dVcmp+dVcmpE+dV+dVE}. These are automatically enforced in the implementation because \KeYmaeraX's kernel insists on fresh names when required for soundness. Renaming with fresh variables is also automatically supported.
\item Global existence of ODE solutions. These are semi-automatically proved (\rref{subsec:support}).
\item Topological side conditions, e.g., in axiom~\irref{CORef} and rules~\irref{dVcmpQ+TTQ}.
These conditions are important to correctly enforce because they may otherwise lead to the subtle soundness errors (\rref{sec:withdomconstraint}).
The implementation uses syntactic criteria for checking these side conditions (\rref{app:topsidecalculus}).
\end{enumerate}

An example topological refinement axiom (\rref{lem:diatopaxioms}) and its corresponding proof rule implemented in \KeYmaeraX with syntactic topological side conditions is given next.

\begin{lemma}[Closed domain refinement axiom]
\label{lem:closeddomref}
The following topological $\didia{\cdot}$ ODE refinement axiom is sound, where formula $\ivr$ characterizes a topologically closed set over variables $x$, and formula $\interior{\ivr}$ characterizes the topological interior of the set characterized by $\ivr$.

\begin{calculuscollection}
\begin{calculus}
\cinferenceRule[CRef|CR]{}
{
\linferenceRule[impl]
  {\lnot{\rfvar} \land \dbox{\pevolvein{\D{x}=\genDE{x}}{\rrfvar \land \lnot{\rfvar}}}{\interior{\ivr}} }
  {\big(\ddiamond{\pevolvein{\D{x}=\genDE{x}}{\rrfvar}}{\rfvar} \limply \axkey{\ddiamond{\pevolvein{\D{x}=\genDE{x}}{\ivr}}{\rfvar}}\big)}
}{}
\end{calculus}
\end{calculuscollection}
\end{lemma}
\begin{proofsketchb}[app:refinementcalculus]{proof:proof26}
\end{proofsketchb}

\begin{corollary}[Closed domain refinement rule]
\label{cor:closeddomref}
The following proof rule is derivable in \dL, where formula $\ivr$ is formed from finite conjunctions and disjunctions of non-strict inequalities $\geq,\leq$, and formula $\strictineq{\ivr}$ is identical to $\ivr$ but with strict inequalities $>,<$ in place of $\geq,\leq$ respectively.

\begin{calculuscollection}
\begin{calculus}
\dinferenceRule[cRef|cR]{}
{\linferenceRule
  {\lsequent{\Gamma}{ \ivr } &
   \lsequent{\Gamma}{ \dbox{\pevolvein{\D{x}=\genDE{x}}{\rrfvar \land \lnot{\rfvar} \land \ivr}}{\strictineq{\ivr}}} &
   \lsequent{\Gamma}{\ddiamond{\pevolvein{\D{x}=\genDE{x}}{\rrfvar}}{\rfvar}}
   }
  {\lsequent{\Gamma}{\ddiamond{\pevolvein{\D{x}=\genDE{x}}{\ivr}}{\rfvar}} }
}{}
\end{calculus}
\end{calculuscollection}
\end{corollary}
\begin{proofsketchb}[app:implementationproofs]{proof:proof27}
\end{proofsketchb}

Axiom~\irref{CRef} is a variant of axiom~\irref{CORef} with different topological conditions.
It says that if the ODE solution can reach goal $\rfvar$ while staying in domain $\rrfvar$ then it can also reach that goal while staying in the new (closed) domain $\ivr$, provided that it stays within the \emph{interior} $\interior{\ivr}$ of the new domain while it has not yet reached $\rfvar$.
Solutions cannot sneak out of the topologically open interior $\interior{\ivr}$ as it enters the goal because, by definition of an open set, the solution must locally remain in  $\interior{\ivr}$ for a short time as it enters the goal (see the~\hyperlink{proof:proof26}{proof} for a detailed explanation).
In contrast to the semantical conditions of~\irref{CRef}, its corresponding derived rule~\irref{cRef} gives syntactic side conditions for the formulas $\ivr,\strictineq{\ivr}$ which are easily checked in an implementation.
In particular, formula $\strictineq{\ivr}$, which syntactically underapproximates the interior $\interior{\ivr}$, can be automatically generated from $\ivr$ through its syntactic structure.
Another advantage of the derived rule~\irref{cRef} is that the closed domain constraint $\ivr$ can be additionally assumed when proving that solutions stay within $\strictineq{\ivr}$ in its middle premise.
This addition is soundly justified using \dL's ODE invariance proof rules~\cite{DBLP:journals/jacm/PlatzerT20} (see~\hyperlink{proof:proof27}{proof}) and it makes rule~\irref{cRef} a powerful primitive for refining domain constraints amongst other options such as axiom~\irref{dDR}.

\subsection{Proof Support}
\label{subsec:support}
Beyond enabling the sound implementation of complex ODE liveness proof rules such as those in~\rref{subsec:complex}, tactics can also provide substantial proof support for users.

\subsubsection{Automatic Dependency Ordering}
Recall derived axiom~\irref{GBEx} from~\rref{cor:boundedexistgen}, which proves (global) existence of solutions for an ODE $\D{x}=\genDE{x}$.
Users of the axiom must still identify precisely which dependency order~\rref{eq:deporder} to use, and provide the sequence of bounded sets $B_i$ for each group of variables $y_i$ involving nonlinear ODEs.
The canonical choice of such a dependency order can be automatically produced by a tactic using a topological sort of the \emph{strongly connected components} (SCCs)\footnote{A strongly connected component of a directed graph is a maximal subset of vertices that are pairwise connected by paths.} of the dependency graph of the ODE.

More precisely, to prove global existence for an ODE $\D{x}=\genDE{x}$, consider the dependency graph $G$ where each variable $x_i$ is a vertex and with a directed edge $x_i \longrightarrow x_j$ if the RHS $f_j(x)$ for $\D{x_j}$ depends on free variable $x_i$.
First, compute the SCCs of $G$, and then topologically sort the SCCs.
The groups of variables $y_i$ in dependency order can be chosen according to the vertices in each SCC in topological order.
An illustrative dependency graph with four SCCs for the following $8$-dimensional ODE is shown in~\rref{fig:scc}.
\begin{align}
\D{x_1} = x_5,
\D{x_2} = x_3 + x_6^2,
\D{x_3} = x_3^2,
\D{x_4} = x_1 + x_3^2 + x_6^2,
\D{x_5} = x_4,
\D{x_6} = x_2^2,
\D{x_7} = x_8,
\D{x_8} = -x_7
\label{eq:bigode}
\end{align}

\begin{figure}
\centering
\begin{tikzpicture}
  [scale=.8,auto=left,every node/.style={circle,fill=blue!20}]
  \tikzset{vertex/.style = {shape=circle,draw,minimum size=1.5em}}
  \tikzset{edge/.style = {->,> = triangle 45}}
  \node (n1) at (2,6)  {$x_1$};
  \node (n4) at (4,6)  {$x_4$};
  \node (n5) at (3,4.5)  {$x_5$};
  \node[fill=white] (l4) at (1.5,4.0)  {$y_4$};
  \node (n2) at (8,4.5)  {$x_2$};
  \node (n6) at (8,6)  {$x_6$};
  \node[fill=white] (l3) at (6.5,4.0)  {$y_3$};
  \node (n3) at (13,5.25) {$x_3$};
  \node[fill=white] (l2) at (11.5,4.0)  {$y_2$};
  \node (n7) at (18,6) {$x_7$};
  \node (n8) at (18,4.5) {$x_8$};
  \node[fill=white] (l1) at (16.5,4.0)  {$y_1$};

  \foreach \from/\to in {n1/n4,n4/n5,n5/n1,n6/n2,n2/n6,n7/n8,n8/n7}
    \draw[edge] (\from) to[bend left] (\to);
  \draw[edge] (n3) to[bend left] (n2);
  \draw[edge] (n6) to[bend right] (n4);
  \draw[edge] (n3) to[loop right] (n3);
  \draw[edge] (n3) to[bend right=33] (n4);

  \draw[dashed] (1,3.7) rectangle   (5,6.9);
  \draw[dashed] (6,3.7) rectangle  (10,6.9);
  \draw[dashed] (11,3.7) rectangle (15,6.9);
  \draw[dashed] (16,3.7) rectangle (20,6.9);
\end{tikzpicture}
\caption{A dependency graph for the ODE~\rref{eq:bigode} over the variables $x_1,\dots,x_8$.
There is a directed edge $x_i \longrightarrow x_j$ if the RHS for $\D{x_j}$ depends on free variable $x_i$.
Each dashed rectangle is a strongly connected component.
Topologically sorting these components (according to the order induced by the edges) yields one possible grouping of the variables $y_1,\dots,y_4$ in dependency order.
The vertices in $y_1$ are not connected to those in $y_2,y_3,y_4$, so the order between these groups can be chosen arbitrarily.}
\label{fig:scc}
\end{figure}
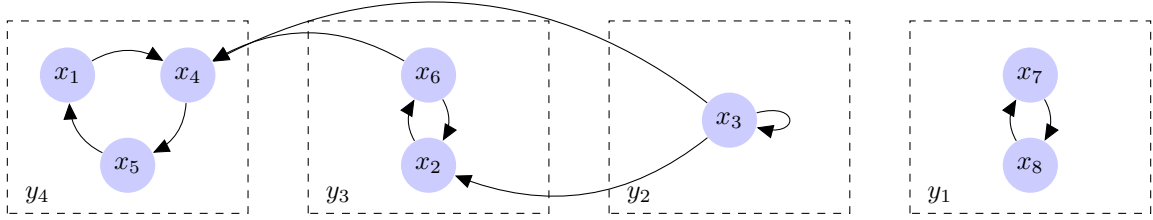

After finding the appropriate SCC-induced dependency order (as in~\rref{fig:scc}), the global existence tactic can prove global existence for the variable clusters $y_i$ that have affine dependencies within the cluster automatically.
For example, the SCC $y_4 \mnodefequiv \{x_1,x_4,x_5\}$ has affine dependencies because the RHS of the ODEs $\D{x_1}, \D{x_4}, \D{x_5}$ are affine in $x_1, x_4, x_5$, so the solution of ODE~\rref{eq:bigode} is automatically proved to be global in those variables following the proof of~\rref{cor:globalexistaffine}.
The generated dependency order enables such a proof even though the RHS of $\D{x_4}$ depends nonlinearly on variables $x_3, x_6$ from earlier clusters.
On the other hand, the SCC $y_3 \mnodefequiv \{x_2,x_6\}$ has nonlinear dependencies on $x_2, x_6$, so users are prompted to input a bounded set (or a bound on derivatives) over variables $x_2,x_6$ in order to prove global existence for those variables.
This continues similarly for the SCCs $y_2$ (nonlinear dependency) and $y_1$ (affine dependency) until global existence is proved for the full ODE.
This semi-automated proof support minimizes the manual effort required of the user in proving global existence by focusing their attention on the automatically identified nonlinear parts of the ODE that may cause finite-time blowup of solutions.

To drive global existence proof automation further, key special cases can be added to the method described above.
One such special case for univariate ODEs is shown below.

\begin{remark}[Global existence for univariate ODEs]
\label{rem:univariate}
Consider the case where a variable group has just one variable and no incoming dependencies, e.g., $y_2 \mnodefequiv \{x_3\}$ in~\rref{fig:scc} or $\exblowup$~\rref{eq:exblowup}.
Global existence for such univariate polynomials ODEs is decidable~\cite{DBLP:journals/entcs/GracaBC08}, even if the RHS is highly nonlinear, because all of its solutions either asymptotically approach a root of the polynomial RHS or diverge to infinity.

This result is best illustrated through the dynamical systems view of ODEs shown in~\rref{fig:onedim} for the ODE $\D{x}=x^4-5x^2+4$.
This example ODE has global solutions from all initial states satisfying $x \leq r_4$ because the solution from all such states are globally attracted to one of the fixed points.
Conversely, for all other initial conditions ($x > r_4$), the ODE blows up in finite time because the RHS is quartic in $x$.
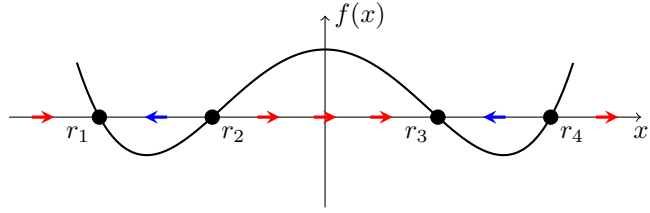
\begin{figure}
    \centering
    \begin{tikzpicture}[domain=-2.2:2.2, scale=1.5]
        \draw[->] (-2.8,0) -- (2.8,0) node[below] {$x$};
        \draw[->] (0,-0.8) -- (0,0.9) node[right] {$\genDE{x}$};
        \draw[>=stealth, very thick, color=red, ->] (-2.6,0) -- (-2.4,0);
        \draw[>=stealth, very thick, color=blue, <-] (-1.6,0) -- (-1.4,0);
        \draw[>=stealth, very thick, color=red, ->] (-0.6,0) -- (-0.4,0);
        \draw[>=stealth, very thick, color=red, ->] (-0.1,0) -- (0.1,0);
        \draw[>=stealth, very thick, color=red, ->] (0.4,0) -- (0.6,0);
        \draw[>=stealth, very thick, color=blue, <-] (1.4,0) -- (1.6,0);
        \draw[>=stealth, very thick, color=red, ->] (2.4,0) -- (2.6,0);

        \draw[color=black, thick] plot[id=sin,mark=none,samples=100] (\x, {0.15*((\x)^4 - 5*(\x)^2 + 4)})  node[right] {};
        \draw[black, fill=black] (-2,0) circle (.4ex) node[below left] {$r_1$};
        \draw[black, fill=black] (-1,0) circle (.4ex) node[below right] {$r_2$};
        \draw[black, fill=black] (1,0) circle (.4ex) node[below left] {$r_3$};
        \draw[black, fill=black] (2,0) circle (.4ex) node[below right] {$r_4$};
    \end{tikzpicture}
    \caption{The univariate ODE $\D{x}=x^4-5x^2+4$ is can be illustrated by plotting its RHS $\genDE{x} \mnodefeq x^4-5x^2+4$ (vertical axis) against $x$ (horizontal axis).
    Points on the horizontal axis evolve towards the right (\redc{red} arrow) when $\genDE{x} \geq 0$ and towards the left (\bluec{blue} arrow) when $\genDE{x} \leq 0$.
    The fixed points $r_1,r_2,r_3, r_4$ are roots of the polynomial RHS where $\genDE{x}=0$.
    These fixed points either attract trajectories (like $r_1, r_3$), or repel them (like $r_2,r_4$).
    All points on the horizontal axis evolve asymptotically towards exactly one fixed point or approach $\infty$.}
    \label{fig:onedim}
\end{figure}

More generally, for a nonlinear univariate polynomial ODE $\D{x}=\genDE{x}$ and initial assumptions $\Gamma$, it suffices to check validity of the following arithmetical sequent to decide global existence:
\[ \lsequent{\Gamma}{\lexists{r}{\big( f(r)=0 \land (\underbrace{f(x) \geq 0 \land r \geq x}_{\circled{a}}  \lor \underbrace{f(x) \leq 0 \land r\leq x}_{\circled{b}})\big)}} \]

The existentially quantified variable $r$ corresponds to a fixed point (a root with $f(r)=0$).
Disjunct~\circled{a} checks whether the solution approaches $r$ from the left, e.g., the points between $r_2$ and $r_3$ in~\rref{fig:onedim} approach $r_3$ from the left.
Alternatively, disjunct \circled{b} checks whether the solution approaches $r$ from the right.
The implementation checks validity of this sequent for univariate nonlinear ODEs and then proves global existence using~\irref{dBDG} because the solution is provably trapped between the initial value of $x$ and the fixed point $r$.
\end{remark}

\subsubsection{Differential Cuts for Liveness Proofs}
Differential cuts~\irref{dC} provide a convenient way to structure and stage safety proofs for ODEs in \dL.
An in-depth discussion is available elsewhere~\cite{Platzer18}, but the idea is illustrated by the following derivation:

\begin{sequentdeduction}[array]
  \linfer[dC]{
    \lsequent{\Gamma}{\dbox{\pevolvein{\D{x}=\genDE{x}}{\ivr}}{\rcfvar_1}}!
    \linfer[dC]{
      \dots !
      \linfer[dC]{
        \linfer[dW]{
        \lsequent{\ivr \land \rcfvar_1 \land \rcfvar_2 \land \dots \land \rcfvar_n}{\rfvar}
        }
        {\vdots}
      }
      {\lsequent{\Gamma}{\dbox{\pevolvein{\D{x}=\genDE{x}}{\ivr \land \rcfvar_1 \land \rcfvar_2}}{\rfvar}}}
    }
    {\lsequent{\Gamma}{\dbox{\pevolvein{\D{x}=\genDE{x}}{\ivr \land \rcfvar_1}}{\rfvar}}}
  }
  {\lsequent{\Gamma}{\dbox{\pevolvein{\D{x}=\genDE{x}}{\ivr}}{\rfvar}} }
\end{sequentdeduction}

The derivation uses a sequence of differential cut steps to progressively add the cuts $\rcfvar_1, \rcfvar_2, \dots, \rcfvar_n$ to the domain constraint.
A final~\irref{dW} step completes the proof when the postcondition $\rfvar$ is already implied by the (now strengthened) domain constraint.
Intuitively, the differential cuts are akin to dynamical lemmas in this derivation.
For example, by proving the premise $\lsequent{\Gamma}{\dbox{\pevolvein{\D{x}=\genDE{x}}{\ivr}}{\rcfvar_1}}$, the cut $\rcfvar_1$ can now be assumed in the domain constraints of subsequent steps.
Just like the~\irref{cut} rule from sequent calculus, differential cuts~\irref{dC} allow safety proofs for ODEs to be staged through a sequence of lemmas about those ODEs.

For proof modularity and maintainability, it is desirable to enable a similar staging for ODE liveness proofs.
Concretely, suppose that the formula $\dbox{\pevolvein{\D{x}=\genDE{x}}{\ivr}}{\rcfvar}$ has been proved as a cut:
\begin{sequentdeduction}[array]
  \linfer[cut]{
    \lsequent{\Gamma}{\dbox{\pevolvein{\D{x}=\genDE{x}}{\ivr}}{\rcfvar}}!
    \cdots
  }
  {\lsequent{\Gamma}{\ddiamond{\pevolvein{\D{x}=\genDE{x}}{\ivr}}{\rfvar}} }
\end{sequentdeduction}

The challenge is how to (soundly) use this lemma in subsequent derivation steps (shown as $\cdots$).
Na\"ively replacing $\ivr$ with $\ivr \land \rcfvar$ in the domain constraint of the succedent does not work.
This may even do more harm than good because the resulting ODE liveness question becomes more difficult (\rref{sec:withdomconstraint}).

The refinement-based approach to ODE liveness provides a natural answer: recall that each refinement step in the chain~\rref{eq:refinementchain} requires the user to prove an additional box modality formula.
The insight is that, for these box modality formulas, any relevant lemmas that have been proved can be soundly added to the domain constraint.
For example, suppose that rule~\irref{Prog} is used to continue the proof after the cut.
The left premise of~\irref{Prog} can now be strengthened to include $\rcfvar$ in its domain constraint:
\begin{sequentdeduction}[array]
    \linfer[Prog]{
    \linfer[dC]{
      \lsequent{\Gamma,\dbox{\pevolvein{\D{x}=\genDE{x}}{\ivr}}{\rcfvar}}{\dbox{\pevolvein{\D{x}=\genDE{x}}{\ivr \land \lnot{\rfvar} \land \rcfvar}}{\lnot{\rgvar}}}
    }
    {\lsequent{\Gamma,\dbox{\pevolvein{\D{x}=\genDE{x}}{\ivr}}{\rcfvar}}{\dbox{\pevolvein{\D{x}=\genDE{x}}{\ivr \land \lnot{\rfvar}}}{\lnot{\rgvar}}}}
      !
      \lsequent{\Gamma,\dbox{\pevolvein{\D{x}=\genDE{x}}{\ivr}}{\rcfvar}}{\ddiamond{\pevolvein{\D{x}=\genDE{x}}{\ivr}}{\rgvar}}
    }
    {\lsequent{\Gamma,\dbox{\pevolvein{\D{x}=\genDE{x}}{\ivr}}{\rcfvar}}{\ddiamond{\pevolvein{\D{x}=\genDE{x}}{\ivr}}{\rfvar}}}
\end{sequentdeduction}

Users could manually track and apply lemmas using~\irref{dC} as shown above, but this becomes tedious in larger liveness proofs.
The implementation instead provides users with tactics that automatically search the antecedents $\Gamma$ for compatible assumptions that can be used to strengthen the domain constraints.
These tactics also use a form of \emph{ODE unification} when determining compatibility.
More precisely, consider the sequent $\lsequent{\Gamma}{\dbox{\pevolvein{\D{x}=\genDE{x}}{\ivr}}{\rfvar}}$, which may arise as a box refinement during a liveness proof.
An antecedent formula $\dbox{\pevolvein{\D{y}=g(y)}{\rrfvar}}{\rcfvar}$ in $\Gamma$ is called a \emph{compatible assumption} for the succedent $\dbox{\pevolvein{\D{x}=\genDE{x}}{\ivr}}{\rfvar}$ if:
\begin{enumerate}
\item The set of ODEs $\D{y}=g(y)$ is a subset of the set of ODEs $\D{x}=\genDE{x}$. This is order-agnostic, e.g., the ODE $\D{u}=v,\D{v}=u$ is a subset of the ODE $\D{v}=u,\D{u}=v,\D{w}=u+v+w$.
\item The domain constraint $\ivr$ implies domain constraint $\rrfvar$, i.e., $\ivr \limply \rrfvar$ is valid.
\end{enumerate}

Under these conditions, the ODE $\pevolvein{\D{y}=g(y)}{\rrfvar}$ permits more trajectories than the ODE $\pevolvein{\D{x}=\genDE{x}}{\ivr}$.
Thus, if formula $\rcfvar$ is always true along solutions of the former ODE, then it also stays true along solutions of the latter.
Combining compatible assumptions with implementations of liveness proof rules yields turbo-charged versions of those rules.
For example, in rule~\irref{dVcmpE}, instead of simply assuming the negation of the postcondition ($\lnot{(\ptermA \cmp 0)} \limply \cdots$) when determining the existence of suitable $\varepsilon$, all postconditions of compatible assumptions can be assumed, e.g., with $\lnot{(\ptermA \cmp 0)} \land \rcfvar \limply \cdots$ for postcondition $\rcfvar$ of a compatible assumption.

\subsubsection{Microbenchmarks}
\label{subsec:autoexamples}

The \KeYmaeraX implementation is used to formally prove all of the ODE liveness examples from this article and elsewhere~\cite{DBLP:conf/fm/SogokonJ15,DBLP:journals/ral/BohrerTMSP19}.
\rref{tab:micro} provides a summary of statistics from these proofs:
``Tactic Steps'' counts the number of (manual) user proof steps;
``Kernel Steps'' counts the number of internal steps taken by the soundness-critical \KeYmaeraX kernel; and
``Proof Time'' measures the time taken (in seconds, averaged over 5 runs, rounded to 3 decimal places) for the proof to execute in \KeYmaeraX.
All experiments were run on an Ubuntu 18.04 laptop with a 2.70 GHz Intel Core i7-6820HQ CPU and 16GB memory.
None of the proofs have been optimized to favor any specific metric.
The specific timings and proof steps are naturally subject to change on different hardware and as various aspects of the \KeYmaeraX theorem prover are improved.
Nevertheless, the key takeaways from these microbenchmarks remain broadly applicable.

\begin{table}
\centering
\caption{Proof statistics for ODE existence and liveness properties proved using the implementation.
For this article's Examples~\ref{ex:velocity}--\ref{ex:nondifferentiable}, two proofs are presented: (M)anual proofs closely follow the pen-and-paper derivations shown in this article, while (A)utomatic proofs make extensive use of the implemented proof support.
The cells in \textbf{bold} font indicate lower (more desirable) values.
The stronger ODE liveness property proved in Example~\ref{ex:strengthen} implies those from Examples~\ref{ex:nonlinproof} and~\ref{ex:nonlindomproof}.
Examples 11, 12 and 15 refer to the correspondingly numbered examples from Sogokon and Jackson~\cite{DBLP:conf/fm/SogokonJ15}.
}
\label{tab:micro}

\begin{tabular}{@{}lrrrrrr@{}}
\toprule
\multirow{2}{*}{Liveness Property} & \multicolumn{2}{c}{Tactic Steps}            & \multicolumn{2}{c}{Kernel Steps}                  & \multicolumn{2}{c}{Proof Time (s)}                \\
                                   & \multicolumn{1}{c}{(M)} & \multicolumn{1}{c}{(A)} & \multicolumn{1}{c}{(M)} & \multicolumn{1}{c}{(A)} & \multicolumn{1}{c}{(M)} & \multicolumn{1}{c}{(A)} \\ \midrule
\rref{ex:velocity}                                     & 7                       & \textbf{3}              & \textbf{2040}           & 12156                   & \textbf{1.778}          & 3.716                   \\
\rref{ex:multiaffine}                                  & 8                       & \textbf{2}              & 962                     & \textbf{898}            & 0.220                   & \textbf{0.203}          \\
\rref{ex:trappedsol}                                   & 7                       & \textbf{3}              & \textbf{1562}           & 1580                    & \textbf{0.551}          & 0.759                  \\
\rref{ex:linproof}                                     & 29                      & \textbf{5}              & 3958                    & \textbf{3501}           & 3.286                   & \textbf{3.034}          \\
\rref{ex:strengthen} (Examples~\ref{ex:nonlinproof} and~\ref{ex:nonlindomproof}) & 50                      & \textbf{20}             & \textbf{5549}           & 6141                    & \textbf{1.714}          & 1.952                   \\
\rref{ex:nondifferentiable}                            & 28                      & \textbf{1}              & \textbf{1747}           & 2571                    & \textbf{0.575}          & 0.990                  \\
\cite{DBLP:conf/fm/SogokonJ15} Example 11              & -                       & 50                      & -                       & 11272                   & -                       & 9.090                   \\
\cite{DBLP:conf/fm/SogokonJ15} Example 12              & -                       & 19                      & -                       & 4818                    & -                       & 1.388                   \\
\cite{DBLP:conf/fm/SogokonJ15} Example 15              & -                       & 1                       & -                       & 4781                    & -                       & 1.730                   \\
\cite{DBLP:journals/ral/BohrerTMSP19}  Goal Position Reachability     & -                       & 34                      & -                       & 8159                    & -                       & 2.182                   \\
\cite{DBLP:journals/ral/BohrerTMSP19} Velocity Bounds Reachability    & -                       & 37                      & -                       & 10521                   & -                       & 3.042                   \\ \bottomrule
\end{tabular}

\end{table}

\paragraph{(M)anual and (A)utomatic Proofs.} The implementation provides users with powerful proof support but also exposes low-level primitives for users who prefer more fine-grained control over (parts of) their proofs.
Both types of proofs are shown for this article's examples in~\rref{tab:micro}.
Proofs that heavily exploit the proof support and automation are more convenient for users and require fewer manual tactic invocations.
This gap is most pronounced for~\rref{ex:nondifferentiable}, where the 28 step manual proof requires just one automated~\irref{dV} step.

On the other hand, the automated proofs are slower than their manual counterparts on four out of six examples.
Most of this overhead arises when there is significant proof search in the automation.
In particular, the automated proof of~\rref{ex:velocity} is significantly slower and requires almost six times more kernel steps compared to its manual counterpart.
This gap arises because the automated proof uses the decision procedure for univariate global existence outlined in~\rref{rem:univariate} while the manual proof uses the direct argument in~\rref{ex:velocity}.
However, the latter proof required user insight about the physical system.
This illustrates the need for a flexible implementation that lets users navigate the convenience and efficiency tradeoff according to their needs and proof insights.

Finally, the automated proofs are in fact \emph{faster} for Examples~\ref{ex:multiaffine} and~\ref{ex:linproof}, which both involve linear ODEs.
The speedups here can be attributed to the well-tuned implementation of global existence proofs for affine systems and to the use of rule~\irref{dVcmpE} for the latter example.
Thus, the aforementioned tradeoff can be further skewed towards favoring user convenience by tuning the implemented automation.

\paragraph{Trusted Kernel with Untrusted Tactics.} All of the proofs in~\rref{tab:micro} make extensive use of \KeYmaeraX's existing tactics framework~\cite{DBLP:conf/itp/FultonMBP17} to handle low-level interactions with \KeYmaeraX soundness-critical kernel, as shown by the large number of kernel steps that each proof requires.
The soundness guarantee provided by the \KeYmaeraX kernel makes this implementation effort a worthy tradeoff because it ensures that the proved results in~\rref{tab:micro} are trustworthy \emph{without} needing to trust the implementation of the tactics.

\paragraph{Applicability.} The insights of this section are not limited to this article's examples and they scale to larger proofs from elsewhere (\rref{tab:micro}).
Notably, Sogokon and Jackson~\cite{DBLP:conf/fm/SogokonJ15} Example 11 proves two separate liveness properties for the same ODE, which makes it the largest (and slowest) microbenchmark.
The examples from Bohrer et al.~\cite{DBLP:journals/ral/BohrerTMSP19} are liveness properties drawn from a larger case study with a hybrid system model of a robot driving along circular arcs in the plane~\cite{DBLP:journals/ral/BohrerTMSP19}.
The use of proof automation (in both senses of the preceding paragraphs) is indispensable for handling the scale of these proofs.

\section{Related Work}
\label{sec:relatedwork}

\paragraph{Existence and Liveness Proof Rules.}
The ODE liveness arguments surveyed in this article were originally presented in various notations, ranging from proof rules~\cite{DBLP:journals/logcom/Platzer10,DBLP:conf/fm/SogokonJ15,DBLP:conf/emsoft/TalyT10} to other mathematical notation~\cite{DBLP:conf/hybrid/PrajnaR05,DBLP:journals/siamco/PrajnaR07,DBLP:journals/siamco/RatschanS10,DBLP:conf/fm/SogokonJ15}.
All of them were justified directly through semantical or mathematical means.
This article unifies and corrects all of these arguments, and presents them as \dL proof rules which are syntactically derived by refinement from \dL axioms.

To the best of the authors' knowledge, this article is also the first to present a deductive approach for syntactic proofs of existence properties for ODEs.
In the surveyed liveness arguments~\cite{DBLP:journals/logcom/Platzer10,DBLP:conf/hybrid/PrajnaR05,DBLP:journals/siamco/PrajnaR07,DBLP:journals/siamco/RatschanS10,DBLP:conf/fm/SogokonJ15,DBLP:conf/emsoft/TalyT10}, sufficient existence duration is either assumed explicitly or is implicitly used in the correctness proofs.
Such a hypothesis is unsatisfactory, since the global existence of solutions for (nonlinear) ODEs is a non-trivial question; in fact, it is undecidable even for polynomial ODEs~\cite{DBLP:journals/entcs/GracaBC08}.
Formal proofs of any underlying existence assumptions thus yield stronger (unconditional) ODE liveness proofs.
Of course, such existence properties are an additional proof burden, but~\rref{sec:impl} also shows that proof support can help by automating easy existence questions, e.g., for affine systems where global existence is well-known.
A related problem arising in the study of hybrid systems is \emph{Zeno phenomena}~\cite{DBLP:conf/lics/Henzinger96,doi:10.1002/rnc.592}, where a trajectory of a hybrid model makes infinitely many (discrete) transitions in finite (continuous) time.
Like finite-time blow up, Zeno phenomena typically occur as abstraction artifacts of hybrid systems models, and they do not occur in real systems.
Thus, analogous to the question of global existence, absence of Zeno phenomena must either be assumed (or Zeno trajectories explicitly excluded)~\cite{DBLP:conf/lics/Henzinger96,DBLP:journals/logcom/Platzer10}, or proved when specifying and verifying properties of such systems~\cite{doi:10.1002/rnc.592}.

The refinement-based approach to ODE existence and liveness proofs underlies this article's implementation described in~\rref{sec:impl}.
Compared to an earlier implementation~\cite{DBLP:conf/cade/PlatzerQ08}, where rules like~\irref{dVcmpQ} are implemented monolithically, this article's approach and implementation build those rules from smaller building blocks which yields a flexible implementation together with powerful (untrusted) proof support.
The high-level lessons discussed in~\rref{sec:impl} are also broadly applicable to other deductive tools for ODEs and hybrid systems~\cite{DBLP:conf/icfem/WangZZ15,DBLP:conf/RelMiCS/FosterMS20} that currently lack support for ODE liveness proofs.

\paragraph{Other Liveness Properties.}
The liveness property studied in this article is the continuous analog of \emph{eventually}~\cite{DBLP:books/daglib/0077033} or \emph{eventuality}~\cite{DBLP:journals/siamco/PrajnaR07,DBLP:conf/fm/SogokonJ15} from temporal logics.
In discrete settings, temporal logic specifications give rise to a zoo of other liveness properties~\cite{DBLP:books/daglib/0077033}.
In continuous settings, \emph{weak eventuality} (requiring \emph{almost all} initial states to reach the goal region) and \emph{eventuality-safety} have been studied~\cite{DBLP:conf/hybrid/PrajnaR05,DBLP:journals/siamco/PrajnaR07}. %
In adversarial settings, \emph{differential game variants}~\cite{DBLP:journals/tocl/Platzer17} enable proofs of winning strategies for differential games.
In dynamical systems and controls, the study of \emph{asymptotic stability} requires both stability (an invariance property) with asymptotic attraction towards a fixed point or periodic orbit (an eventuality-like property)~\cite{Chicone2006,DBLP:journals/siamco/RatschanS10}.
For hybrid systems, various authors have proposed generalizations of classical asymptotic stability, such as \emph{persistence}~\cite{DBLP:journals/jar/SogokonJJ19}, \emph{stability}~\cite{DBLP:conf/hybrid/PodelskiW06}, and \emph{inevitability}~\cite{DBLP:conf/hybrid/DuggiralaM12}.
\emph{Controlled} versions of these properties are also of interest, e.g., \emph{(controlled) reachability and attractivity}~\cite{ABATE2009150,DBLP:conf/emsoft/TalyT10}.
Eventuality(-like) properties are fundamental to all of these advanced liveness properties.
The formal understanding of eventuality in this article is therefore a key step towards enabling formal analysis of more advanced liveness properties.

\paragraph{Automated Liveness Proofs.}
Automated reachability analysis tools~\cite{DBLP:conf/cav/ChenAS13,DBLP:conf/cav/FrehseGDCRLRGDM11} can also be used to answer certain liveness verification questions.
For an ODE and initial set $\bigchi_0$, computing an over-approximation $\mathcal{O}$ of the reachable set $\bigchi_{t} \subseteq \mathcal{O}$ at time $t$ shows that \emph{all} states in $\bigchi_0$ reach $\mathcal{O}$ at time $t$~\cite{DBLP:journals/jar/SogokonJJ19} (if solutions do not blow up).
Similarly, an under-approximation $\mathcal{U} \subseteq \bigchi_{t}$ shows that \emph{some} state in $\bigchi_0$ eventually reaches $\mathcal{U}$~\cite{DBLP:conf/hybrid/GoubaultP17} (if $\mathcal{U}$ is non-empty).
Neither approach handles domain constraints directly~\cite{DBLP:conf/hybrid/GoubaultP17,DBLP:journals/jar/SogokonJJ19} and, unlike deductive approaches, the use of reachability tools limits them to concrete time bounds $t$ and bounded initial sets $\bigchi_0$.
Deductive liveness approaches can also be (partially) automated, as shown in~\rref{sec:impl}.
Lyapunov functions guaranteeing (asymptotic) stability can be found by sum-of-squares (SOS) optimization~\cite{1184414}.
Liveness arguments can be similarly combined with SOS optimization to find suitable differential variants~\cite{DBLP:conf/hybrid/PrajnaR05,DBLP:journals/siamco/PrajnaR07}.
Other approaches are possible, e.g., a constraint solving-based approach can be used for finding the so-called \emph{set Lyapunov functions}~\cite{DBLP:journals/siamco/RatschanS10} (e.g., the term $\ptermA$ used in~\irref{RS+RSQ}).
Crucially, automated approaches must ultimately be based on sound underlying liveness arguments.
The correct justification of these arguments is precisely what this article enables.

\paragraph{Refinement Calculi.}
This article's view of ODE liveness arguments as step-by-step refinements is closely related to \emph{refinement proof calculi}~\cite{DBLP:books/daglib/0096285,DBLP:journals/toplas/Kozen97}.
The shared idea is that the proof of a complex property, like ODE liveness or program correctness, should be broken down into (simpler) step-by-step refinements.
The key difference is that, for refinement calculi, refinement typically takes place between programs (or implementations) and their specification.
For example, a concrete implementation $\beta$ is said to \emph{refine} its abstract specification $\alpha$ if the set of transitions of $\beta$ is a subset of those of $\alpha$~\cite{DBLP:books/daglib/0096285}.
Proving such a refinement for hybrid programs $\alpha,\beta$ would, for example, prove the implication:
\begin{equation}
\ddiamond{\beta}{\rfvar} \limply \ddiamond{\alpha}{\rfvar}
\label{eq:refinementsys}
\end{equation}

Program refinement is not directly applicable to this article's focus on proving liveness for specific ODEs.
Instead, as hinted by~\rref{eq:refinementsys}, program refinement plays an important role for generalizing this article's results beyond ODEs to hybrid systems, where, e.g., one may use implications like~\rref{eq:refinementsys} as part of a refinement chain~\rref{eq:refinementchain}.
There are a number of refinement calculi for hybrid systems~\cite{DBLP:conf/lics/LoosP16,DBLP:conf/RelMiCS/FosterMS20,DBLP:books/crc/p/ButlerAB16,DBLP:journals/tcs/RonkkoRS03}.
Notably, \emph{differential refinement logic}~\cite{DBLP:conf/lics/LoosP16} formally extends \dL with a refinement operator $\beta \leq \alpha$, and can be used together with this article's results.
Another direction for generalizing this article's results is to consider larger classes of continuous dynamics, such as differential inclusions, differential-algebraic constraints~\cite{DBLP:journals/logcom/Platzer10}, and differential games~\cite{DBLP:journals/tocl/Platzer17}.
These open up the possibility of proving refinements between concrete (ODE) descriptions and their more abstract continuous counterparts~\cite{DBLP:journals/logcom/Platzer10,DBLP:conf/RelMiCS/FosterMS20,DBLP:conf/tase/DupontAPS19,DBLP:journals/tocl/Platzer17}.

\section{Conclusion}
This article presents a refinement-based approach for proving liveness and, as a special case, global existence properties for ODEs in \dL.
The associated \KeYmaeraX implementation demonstrates the utility of this approach for formally proving concrete ODE liveness questions.
Beyond the particular proof rules derived in the article, the exploration of new and more general ODE liveness proof rules is enabled by simply piecing together more refinement steps in \dL, or in the \KeYmaeraX implementation of those steps.
Given its wide applicability and correctness guarantees, this approach is a suitable framework for justifying ODE liveness arguments, even for readers less interested in the logical aspects.

\appendix

\section{Proof Calculus}
\label{app:coreproofs}

This appendix presents the \dL proof calculus that underlies the refinement approach of this article.
For ease of reference, all of the core axioms and proof rules presented in the main article are summarized here, along with their proofs (where necessary).

\subsection{Base Calculus}
\label{app:basecalculus}

The following lemma summarizes the base \dL axioms and proof rules used in this article.
\rref{lem:appdlaxioms} subsumes~\rref{lem:dlaxioms} with the addition of the differential ghost axiom (\irref{DG} from~\rref{subsec:proofcalculus}) and three axioms (\irref{band+DMP+DX}).
The latter three additions are used in derivations in~\rref{app:derivationproofs}.

\begin{lemma}[Axioms and proof rules of \dL~\cite{DBLP:journals/jar/Platzer17,Platzer18,DBLP:journals/jacm/PlatzerT20}]
\label{lem:appdlaxioms}
The following are sound axioms and proof rules of \dL.
In axiom~\irref{DG}, the $\exists$ quantifier can be replaced with a $\forall$ quantifier. \\
\begin{calculuscollection}
\begin{calculus}
\cinferenceRuleQuote{diamond}
\end{calculus}
\qquad\qquad
\begin{calculus}
\cinferenceRuleQuote{K}
\end{calculus}\\
\begin{calculus}
\dinferenceRuleQuote{dIcmp}
\end{calculus}\\
\begin{calculus}
\dinferenceRuleQuote{dC}
\dinferenceRuleQuote{MbW}
\end{calculus}
\qquad
\begin{calculus}
\dinferenceRuleQuote{dW}
\dinferenceRuleQuote{MdW}
\end{calculus}\\
\begin{calculus}
\cinferenceRuleQuote{DG}

\dinferenceRule[band|${\dibox{\cdot}\land}$]{}
{\linferenceRule[equiv]
  {\dbox{\alpha}{\rfvar} \land \dbox{\alpha}{\rrfvar}}
  {\axkey{\dbox{\alpha}{(\rfvar \land \rrfvar)}}}
}{}

\cinferenceRule[DMP|DMP]{differential modus ponens}
{\linferenceRule[impl]
  {\dbox{\pevolvein{\D{x}=\genDE{x}}{\ivr}}{(\ivr \limply \rrfvar)}}
  {(\dbox{\pevolvein{\D{x}=\genDE{x}}{\rrfvar}}{\rfvar} \limply \axkey{\dbox{\pevolvein{\D{x}=\genDE{x}}{\ivr}}{\rfvar}})}
}{}
\cinferenceRule[DX|DX]{}
{\linferenceRule[equiv]
  {(\ivr \limply \rfvar \land \dbox{\pevolvein{\D{x}=\genDE{x}}{\ivr}}{\rfvar})}
  {\axkey{\dbox{\pevolvein{\D{x}=\genDE{x}}{\ivr}}{\rfvar}}}
}{\text{$\D{x} \not\in \rfvar,\ivr$}}
\end{calculus}
\end{calculuscollection}
\end{lemma}
\begin{proof}[\textbf{Proof of \rref{lem:appdlaxioms} (implies \rref{lem:dlaxioms})\hypertarget{proof:proof1}}]
The soundness of all axioms and proof rules in~\rref{lem:appdlaxioms} are proved elsewhere~\cite{DBLP:journals/jar/Platzer17,Platzer18,DBLP:journals/jacm/PlatzerT20}.
\end{proof}

Axiom~\irref{band} is derived from axiom~\irref{K}~\cite{DBLP:journals/jar/Platzer17,Platzer18}.
It commutes box modalities and their conjunctive postconditions because the conjunction $\rfvar \land \rrfvar$ is true after all runs of hybrid program $\alpha$ iff the individual conjuncts $\rfvar,\rrfvar$ are themselves true after all runs of $\alpha$.
Axiom~\irref{DMP} is the modus ponens principle for domain constraints.
The \emph{differential skip} axiom~\irref{DX} is a reflexivity property of differential equation solutions.
The ``$\lylpmi$'' direction says if domain constraint $\ivr$ is initially false, then the formula $\dbox{\pevolvein{\D{x}=\genDE{x}}{\ivr}}{\rfvar}$ is trivially true in that initial state because no solution of the ODE stays in the domain constraint.
Thus, this direction of~\irref{DX} allows domain constraint $\ivr$ to be assumed true initially when proving $\dbox{\pevolvein{\D{x}=\genDE{x}}{\ivr}}{\rfvar}$ (shown below, on the left).
The ``$\limply$'' direction has the following equivalent contrapositive reading using~\irref{diamond} and propositional simplification: $\ivr \land \rfvar \limply \ddiamond{\pevolvein{\D{x}=\genDE{x}}{\ivr}}{\rfvar}$, i.e., if the domain constraint $\ivr$ and postcondition $\rfvar$ were both true initially, then $\ddiamond{\pevolvein{\D{x}=\genDE{x}}{\ivr}}{\rfvar}$ is true because of the trivial solution of duration zero.
When proving the liveness property $\ddiamond{\pevolvein{\D{x}=\genDE{x}}{\ivr}}{\rfvar}$, one can therefore always additionally assume $\lnot{(\ivr \land \rfvar)}$ because, by~\irref{DX}, there is nothing to prove otherwise (shown below, on the right).
\[
\begin{minipage}[b]{0.5\textwidth}
{\footnotesizeoff%
\begin{sequentdeduction}[array]
  \linfer[DX]{
    \lsequent{\Gamma,\ivr}{\dbox{\pevolvein{\D{x}=\genDE{x}}{\ivr}}{\rfvar}}
  }
  {\lsequent{\Gamma}{\dbox{\pevolvein{\D{x}=\genDE{x}}{\ivr}}{\rfvar}} }
\end{sequentdeduction}
}%
\end{minipage}
\begin{minipage}[b]{0.5\textwidth}
{\footnotesizeoff%
\begin{sequentdeduction}[array]
  \linfer[DX]{
    \lsequent{\Gamma,\lnot{(\ivr \land \rfvar)}}{\ddiamond{\pevolvein{\D{x}=\genDE{x}}{\ivr}}{\rfvar}}
  }
  {\lsequent{\Gamma}{\ddiamond{\pevolvein{\D{x}=\genDE{x}}{\ivr}}{\rfvar}} }
\end{sequentdeduction}
}%
\end{minipage}
\]

Rule~\irref{dGt} from~\rref{sec:globexist} is useful for adding a fresh time variable $\timevar$ in ODE existence and liveness proofs.
Its derivation is shown below, using axiom~\irref{diamond} to switch between the box and diamond modalities and axiom~\irref{DG} to introduce a universally quantified time variable $\timevar$  which is then instantiated by~\irref{alll} to $\timevar=0$.
\\~\\
\begin{minipage}[b]{0.35\textwidth}
\[
\dinferenceRuleQuote{dGt}
\]
~\\
\end{minipage}\quad\;%
\begin{minipage}[b]{0.6\textwidth}
{\footnotesizeoff%
\begin{sequentdeduction}[array]
  \linfer[diamond+notr]{
  \linfer[DG]{
  \linfer[alll]{
  \linfer[diamond+notl]{
    \lsequent{\Gamma,\timevar{=}0}{\ddiamond{\pevolvein{\D{x}=\genDE{x},\D{\timevar}=1}{\ivr}}{\rfvar}}
  }
    {\lsequent{\Gamma,\timevar{=}0,\dbox{\pevolvein{\D{x}=\genDE{x},\D{\timevar}=1}{\ivr}}{\lnot{\rfvar}}}{\lfalse}}
  }
    {\lsequent{\Gamma,\lforall{\timevar}{\dbox{\pevolvein{\D{x}=\genDE{x},\D{\timevar}=1}{\ivr}}{\lnot{\rfvar}}}}{\lfalse}}
  }
    {\lsequent{\Gamma,\dbox{\pevolvein{\D{x}=\genDE{x}}{\ivr}}{\lnot{\rfvar}}}{\lfalse}}
  }
  {\lsequent{\Gamma}{\ddiamond{\pevolvein{\D{x}=\genDE{x}}{\ivr}}{\rfvar}}}
\end{sequentdeduction}
}%
\end{minipage}\\

The bounded differential ghost axiom~\irref{BDG} from~\rref{lem:boundeddg} (quoted and proved below) is a new vectorial generalization of~\irref{DG} which allows differential ghosts with provably bounded ODEs to be added.

\cinferenceRuleQuote{BDG}

\begin{proof}[\textbf{Proof of \rref{lem:boundeddg}\hypertarget{proof:proof2}}]
The proof of~\irref{BDG} is similar to that for the differential ghosts axiom~\cite{DBLP:journals/jar/Platzer17}, but generalizes it to support vectorial, nonlinear ODEs by adding a precondition on boundedness of solutions.
Let $y$ be a vector of $m$ fresh variables and $\D{y}=g(x,y)$ be its corresponding vector of ghost ODEs.
Both directions of the (inner) equivalence of axiom~\irref{BDG} are proved separately.
\begin{enumerate}
\item[``$\limply$'']
The (easier) ``$\limply$'' direction does not require the outer bounding assumption of~\irref{BDG}, i.e., the implication $\dbox{\pevolvein{\D{x}=\genDE{x}}{\ivr(x)}}{\rfvar(x)} \limply \axkey{\dbox{\pevolvein{\D{x}=\genDE{x}, \D{y}=g(x,y)}{\ivr(x)}}{\rfvar(x)}}$ is valid for any ODE $\D{y}=g(x,y)$ meeting the freshness condition on $y$.
The proof for this direction is identical to the proof of soundness for differential ghosts~\cite[Theorem 38]{DBLP:journals/jar/Platzer17}.
\item[``$\lylpmi$'']
In the ``$\lylpmi$'' direction, consider an initial state $\iget[state]{\I} \in \States$ and let $\solvar : [0, T) \to \States, 0<T\leq\infty$ be the unique, right-maximal solution~\cite{Chicone2006,Walter1998} to the ODE $\D{x}=\genDE{x}$ with initial value $\solvar(0)=\iget[state]{\I}$.
Similarly, let $\solvar_y : [0, T_y) \to \States, 0<T_y\leq\infty$ be the unique, right-maximal solution to the ODE $\D{x}=\genDE{x},\D{y}=g(x,y)$ with initial value $\solvar_y(0)=\iget[state]{\I}$.
Assume that $\iget[state]{\I}$ satisfies both of the following assumptions in~\irref{BDG}:
\begin{align}
&\iget[state]{\I} \in \imodel{\I}{\dbox{\pevolvein{\D{x}=\genDE{x},\D{y}=g(x,y)}{\ivr(x)}}{\,\norm{y}^2 \leq \ptermA(x)}}{}
\label{eq:bdgassm1}\\
&\iget[state]{\I} \in \imodel{\I}{\dbox{\pevolvein{\D{x}=\genDE{x}, \D{y}=g(x,y)}{\ivr(x)}}{\rfvar(x)}}
\label{eq:bdgassm2}
\end{align}

To show $\iget[state]{\I} \in \imodel{\I}{\dbox{\pevolvein{\D{x}=\genDE{x}}{\ivr(x)}}{\rfvar(x)}}$, unfold the semantics of the box modality and consider any finite time $\tau$ with $0 \leq \tau < T$ where $\solvar(\zeta) \in \imodel{\I}{\ivr(x)}$ for all $0 \leq \zeta \leq \tau$.
It is proved further below that $\tau$ is also in the existence interval for solution $\solvar_y$, i.e., \circled{$\star$}: $\tau < T_y$.
By uniqueness, $\solvar,\solvar_y$ agree on the values of $x$ on their common existence interval, which includes the time interval $[0,\tau]$ by \circled{$\star$}.
Therefore, by coincidence for terms and formulas~\cite{DBLP:journals/jar/Platzer17}, $\solvar_y(\zeta) \in \imodel{\I}{\ivr(x)}$ for all $0 \leq \zeta \leq \tau$.
Thus, by~\rref{eq:bdgassm2}, $\solvar_y(\tau) \in \imodel{\I}{\rfvar(x)}$ and by coincidence for formulas~\cite{DBLP:journals/jar/Platzer17}, $\solvar(\tau) \in \imodel{\I}{\rfvar(x)}$.

In order to prove \circled{$\star$}, suppose for contradiction that $T_y \leq \tau$.
Let $x(\cdot) : [0,T) \to \reals^n$ denote the projection of solution $\solvar$ onto its $x$ coordinates, and let $\ptermA(x(\cdot)) : [0,T) \to \reals$ denote the evaluation of term $\ptermA$ along $x(\cdot)$.
Since the projection $x(\cdot)$ and its composition with a polynomial evaluation function $\ptermA(x(\cdot))$ are continuous in $t$~\cite{DBLP:journals/jar/Platzer17}, $\ptermA(x(\cdot))$ is bounded above by (and attains) its maximum value $\ptermA_{\max} \in \reals$ on the compact interval $[0,\tau]$.

Let $y(\cdot) : [0,T_y) \to \reals^m$ similarly denote the projection of $\solvar_y$ onto its $y$ coordinates and $\norm{y(\cdot)}^2$ denote the squared norm evaluated along $y(\cdot)$.
Since $T_y \leq \tau < T$, note that $y(\cdot)$ must be the unique right-maximal solution of the time-dependent differential equation $\D{y}= g(x(t),y)$.
Otherwise, if there is a longer solution $\psi : [0,\zeta) \to \reals^m$ for $\D{y}= g(x(t),y)$ which exists for time $\zeta$ with $T_y < \zeta \leq T$, then the combined solution given by $(x(t),\psi(t)) : [0,\zeta) \to \reals^n \times \reals^m$ extends $\solvar_y$ beyond $T_y$ (by keeping all variables other than $x,y$ constant at their initial values in state $\iget[state]{\I}$).
This contradicts right-maximality of $\solvar_y$.
Moreover, for all times $0 \leq \zeta < T_y$, by assumption $\zeta \leq \tau$ and $\solvar(\zeta) \in \imodel{\I}{\ivr(x)}$, so the solution $\solvar_y$ satisfies $\solvar_y(\zeta) \in \imodel{\I}{\ivr(x)}$ by coincidence for formulas~\cite{DBLP:journals/jar/Platzer17}.
Thus, from~\rref{eq:bdgassm1}, for all times $0 \leq \zeta < T_y$, the squared norm is bounded by $\ptermA_{\max}$, with:
\[ \norm{y(\zeta)}^2 \leq \ptermA(x(\zeta)) \leq \ptermA_{\max} \]

Hence, $y(\cdot)$ remains trapped within the compact $\reals^m$ ball of radius $\sqrt{\ptermA_{\max}}$ on its domain of definition $[0,T_y)$.
By~\cite[Theorem 1.4]{Chicone2006}, and right-maximality of $y(\cdot)$ for the time-dependent ODE $\D{y}= g(x(t),y)$, the domain of definition of solution $y(\cdot)$ is equal to the domain of definition of $\D{y}= g(x(t),y)$, i.e., $T_y = T$, which contradicts $T_y \leq \tau < T$.
\qedhere
\end{enumerate}
\end{proof}

The following lemma presents additional \dL ODE invariance proof rules that are used in the derivations in~\rref{app:derivationproofs}.
These invariance proof rules are not the main focus of this article but they are nevertheless useful for simplifying or deriving the premises of this article's ODE existence and liveness proof rules.

\begin{lemma}[ODE invariance proof rules of \dL~\cite{DBLP:journals/jacm/PlatzerT20}]
\label{lem:appodeinvariance}
The following are derived ODE invariance proof rules of \dL.
In rule~\irref{dbxineq}, $\cofterm$ is any polynomial cofactor term.
In rule~\irref{sAIQ}, $\sigliedsai{\genDE{x}}{\ivr}$, $\sigliedsai{\genDE{x}}{\rfvar}$, $\sigliedsai[-]{-\genDE{x}}{\ivr}$, $\sigliedsai[-]{-\genDE{x}}{(\lnot{\rfvar})}$ are semialgebraic progress formulas~\cite[Def. 6.4]{DBLP:journals/jacm/PlatzerT20} with respect to $\D{x}=\genDE{x}$.
In rule~\irref{Enc}, formula $\rfvar$ is formed from finite conjunctions and disjunctions of strict inequalities $>,<$, and formula $\relaxineq{\rfvar}$ is identical to $\rfvar$ but with non-strict inequalities $\geq,\leq$ in place of $>,<$ respectively.
\\
\begin{calculuscollection}
\begin{calculus}
\dinferenceRule[dbxineq|dbx${_\cmp}$]{Darboux inequality}
{\linferenceRule
  {\lsequent{\ivr} {\lied[]{\genDE{x}}{\ptermA}\geq \cofterm\ptermA}}
  {\lsequent{\ptermA\cmp0} {\dbox{\pevolvein{\D{x}=\genDE{x}}{\ivr}}{\ptermA\cmp0}}}
}{\text{where $\cmp$ is either $\geq$ or $>$}}

\dinferenceRule[sAIQ|sAI{$\&$}]{Semianalytic invariant with domains}
{\linferenceRule
  {
  \lsequent{\rfvar,\ivr,\sigliedsai{\genDE{x}}{\ivr}}{\sigliedsai{\genDE{x}}{\rfvar}} &
  \lsequent{\lnot{\rfvar},\ivr,\sigliedsai[-]{-\genDE{x}}{\ivr}}{\sigliedsai[-]{-\genDE{x}}{(\lnot{\rfvar})}}
  }
  {\lsequent{\rfvar}{\dbox{\pevolvein{\D{x}=\genDE{x}}{\ivr}}{\rfvar}}}
}{}

\dinferenceRule[BC|Barr]{}
{\linferenceRule
  { \lsequent{\ivr, \ptermA=0}{\lied[]{\genDE{x}}{p} > 0}
  }
  {\lsequent{\Gamma, \ptermA \cmp 0}{\dbox{\pevolvein{\D{x}=\genDE{x}}{\ivr}}{\ptermA \cmp 0}} }  \quad
}{where $\cmp$ is either $\geq$ or $>$}

\dinferenceRule[Enc|Enc]{Enclosure}
{\linferenceRule
  {
  \lsequent{\Gamma}{\relaxineq{\rfvar}} &
  \lsequent{\Gamma}{\dbox{\pevolvein{\D{x}=\genDE{x}}{\ivr \land \relaxineq{\rfvar}}}{\rfvar}}
  }
  {\lsequent{\Gamma}{\dbox{\pevolvein{\D{x}=\genDE{x}}{\ivr}}{\rfvar}}}
}{}

\end{calculus}
\end{calculuscollection}
\end{lemma}
\begin{proof}[\textbf{Proof}]
These ODE invariance proof rules are all derived from the complete \dL axiomatization for ODE invariants~\cite{DBLP:journals/jacm/PlatzerT20}.
\end{proof}

Rule~\irref{dbxineq} is the Darboux inequality proof rule for the invariance of $\ptermA \cmp 0$ which is derived using rules~\irref{dIcmp+dC+DG}.
An extensive explanation of the rule is available elsewhere~\cite[Section 3.2]{DBLP:journals/jacm/PlatzerT20}.
Rule~\irref{sAIQ} is \dL's complete proof rule for ODE invariants, i.e., the formula $\rfvar$ is invariant for ODE $\pevolvein{\D{x}=\genDE{x}}{\ivr}$ iff it can be proved invariant by rule~\irref{sAIQ}.
For closed (resp. open) semialgebraic formulas $\rfvar$, the right (resp. left) premise of rule~\irref{sAIQ} closes trivially~\cite{DBLP:journals/jacm/PlatzerT20}.
This simplification is useful for obtaining more succinct proof rules, e.g., rule~\irref{dV} makes use of~\irref{sAIQ} with a closed semialgebraic formula.
Rule~\irref{BC} is a \dL rendition of the strict barrier certificates proof rule~\cite{DoyenFPP18,DBLP:journals/tac/PrajnaJP07} for invariance of $\ptermA \cmp 0$, which is derived as a special case of rule~\irref{sAIQ}.
Intuitively, the premise says that $\ptermA = 0$ is a \emph{barrier} along which the value of $\ptermA$ is increasing along solutions (succedent $\lied[]{\genDE{x}}{p} > 0$), so it is impossible for solutions starting from $\ptermA \cmp 0$ to cross this barrier into $\ptermA \pmc 0$.
Finally, rule~\irref{Enc} says that, in order to prove that solutions stay in postcondition $\rfvar$ which characterizes an open set, it suffices to prove it assuming $\relaxineq{\rfvar}$ in the domain constraint, where $\relaxineq{\rfvar}$ relaxes all strict inequalities in $\rfvar$ and thus provides an over-approximation of the topological closure of the set characterized by $\rfvar$.
The rule can also be understood in the contrapositive: if a continuous solution leaves $\rfvar$, then it either already started outside the closure (ruled out by left premise), or it starts in the closure and leaves $\rfvar$ on its topological boundary (included in the closure).
The latter case is ruled out by the right premise of~\irref{Enc} because solutions that remain in the closure must stay in $\rfvar$.

\subsection{Refinement Calculus}
\label{app:refinementcalculus}

The following ODE liveness refinement axioms are quoted from~\rref{lem:diarefaxioms}, and their syntactic derivations in the \dL proof calculus are given below.

\begin{calculuscollection}
\begin{calculus}
\dinferenceRuleQuote{Prog}

\dinferenceRuleQuote{dDR}

\dinferenceRuleQuote{dBDG}

\dinferenceRuleQuote{dDDG}
\end{calculus}
\end{calculuscollection}

\begin{proof}[\textbf{Proof of \rref{lem:diarefaxioms}\hypertarget{proof:proof3}}]
The four axioms are derived in order.

\begin{itemize}
\item[\irref{Prog}] Axiom~\irref{Prog} is derived as follows, starting with \irref{diamond+notl+notr} to dualize the diamond modalities in the antecedent and succedent to box modalities. A \irref{dC} step using the right antecedent completes the proof.
{\footnotesizeoff%
\begin{sequentdeduction}[array]
\linfer[diamond+notl+notr]{
\linfer[dC]{
  \lclose
}
  {\lsequent{\dbox{\pevolvein{\D{x}=\genDE{x}}{\ivr \land \lnot{\rfvar}}}{\lnot{\rgvar}}, \dbox{\pevolvein{\D{x}=\genDE{x}}{\ivr}}{\lnot{\rfvar}}} {\dbox{\pevolvein{\D{x}=\genDE{x}}{\ivr}}{\lnot{\rgvar}}}}
}
  {\lsequent{\dbox{\pevolvein{\D{x}=\genDE{x}}{\ivr \land \lnot{\rfvar}}}{\lnot{\rgvar}}, \ddiamond{\pevolvein{\D{x}=\genDE{x}}{\ivr}}{\rgvar}} {\ddiamond{\pevolvein{\D{x}=\genDE{x}}{\ivr}}{\rfvar}}}
\end{sequentdeduction}
}%

\item[\irref{dDR}] Axiom~\irref{dDR} is similarly derived from axiom~\irref{DMP} with~\irref{diamond}~\cite{DBLP:journals/jacm/PlatzerT20}.

\item[\irref{dBDG}] Axiom~\irref{dBDG} is derived from axiom~\irref{BDG} after using axiom~\irref{diamond} to dualize diamond modalities to box modalities.
The leftmost antecedent is abbreviated $\rrfvar \mnodefequiv \dbox{\pevolvein{\D{x}=\genDE{x},\D{y}=g(x,y)}{\ivr(x)}}{\,\norm{y}^2 \leq \ptermA(x)}$.
{\footnotesizeoff%
\begin{sequentdeduction}[array]
\linfer[diamond+notl+notr]{
\linfer[BDG]{
  \lclose
}
  {\lsequent{\rrfvar,\dbox{\pevolvein{\D{x}=\genDE{x}, \D{y}=g(x,y)}{\ivr(x)}}{\lnot{\rfvar(x)}} } {\dbox{\pevolvein{\D{x}=\genDE{x}}{\ivr(x)}}{\lnot{\rfvar(x)}}}}
}
  {\lsequent{\rrfvar, \ddiamond{\pevolvein{\D{x}=\genDE{x}}{\ivr(x)}}{\rfvar(x)}} {\ddiamond{\pevolvein{\D{x}=\genDE{x}, \D{y}=g(x,y)}{\ivr(x)}}{\rfvar(x)}}}
\end{sequentdeduction}
}%

\item[\irref{dDDG}] Axiom~\irref{dDDG} is derived as a differential version of axiom~\irref{dBDG} with the aid of~\irref{DG}.
The derivation starts with~\irref{diamond+notl+notr} to turn diamond modalities in the sequent to box modalities.
Axiom~\irref{DG} then introduces a fresh ghost ODE $\D{z}= L(x) z + M(x)$, where the antecedents are universally quantified over ghost variable $z$ by~\irref{DG}, while the succedent is existentially quantified.
All quantifiers are then instantiated using~\irref{alll+existsr}, with $z = \norm{y}^2$ so that $z$ stores the initial value of the squared norm of $y$.
Axiom~\irref{BDG} is used with $\D{y}=g(x,y)$ as the ghost ODEs and with $\ptermA(x,z) \mnodefeq z$.
The antecedents are abbreviated:
\begin{align*}
\rrfvar &\mnodefequiv \dbox{\pevolvein{\D{x}=\genDE{x},\D{y}=g(x,y)}{\ivr(x)}}{\,2\dotp{y}{g(x,y)} \leq L(x) \norm{y}^2 + M(x)} \\
\rrfvar_z &\mnodefequiv \dbox{\pevolvein{\D{x}=\genDE{x},\D{y}=g(x,y),\D{z}=L(x) z + M(x)}{\ivr(x)}}{\,2\dotp{y}{g(x,y)} \leq L(x) \norm{y}^2 + M(x)} \\
\rsfvar &\mnodefequiv \dbox{\pevolvein{\D{x}=\genDE{x},\D{y}=g(x,y)}{\ivr(x)}}{\lnot{\rfvar(x)}} \\
\rsfvar_z &\mnodefequiv \dbox{\pevolvein{\D{x}=\genDE{x},\D{y}=g(x,y),\D{z}=L(x) z + M(x)}{\ivr(x)}}{\lnot{\rfvar(x)}}
\end{align*}

{\footnotesizeoff%
\begin{sequentdeduction}[array]
\linfer[diamond+notl+notr]{
\linfer[DG]{
\linfer[alll+existsr]{
\linfer[BDG]{
  \lsequent{z=\norm{y}^2, \rrfvar_z} {\dbox{\pevolvein{\D{x}=\genDE{x},\D{y}=g(x,y),\D{z}=L(x) z + M(x)}{\ivr(x)}}{\norm{y}^2 \leq z}}
}
  {\lsequent{z=\norm{y}^2, \rrfvar_z, \rsfvar_z} {\dbox{\pevolvein{\D{x}=\genDE{x}, \D{z}=L(x) z + M(x)}{\ivr(x)}}{\lnot{\rfvar(x)}}}}
}
  {\lsequent{\lforall{z}{\rrfvar_z}, \lforall{z}{\rsfvar_z}} {\lexists{z}{\dbox{\pevolvein{\D{x}=\genDE{x}, \D{z}=L(x) z + M(x)}{\ivr(x)}}{\lnot{\rfvar(x)}}}}}
}
  {\lsequent{\rrfvar,\rsfvar} {\dbox{\pevolvein{\D{x}=\genDE{x}}{\ivr(x)}}{\lnot{\rfvar(x)}}}}
}
  {\lsequent{\rrfvar, \ddiamond{\pevolvein{\D{x}=\genDE{x}}{\ivr(x)}}{\rfvar(x)}} {\ddiamond{\pevolvein{\D{x}=\genDE{x}, \D{y}=g(x,y)}{\ivr(x)}}{\rfvar(x)}}}
\end{sequentdeduction}
}%

From the resulting open premise, a~\irref{dC} step adds the postcondition of $\rrfvar_z$ to the domain constraint of the succedent, while~\irref{MbW} rearranges the postcondition into the form expected by rule~\irref{dbxineq}.
The proof is completed using~\irref{dbxineq} with cofactor $\cofterm \mnodefeq L(x)$.
Its resulting arithmetical premise is proved by~\irref{qear} because the Lie derivative of $z - \norm{y}^2$ is bounded above by the following calculation, where the inequality from the domain constraint is used in the second step.
\begin{align*}
\lie[]{\D{x}=\genDE{x},\D{y}=g(x,y),\D{z}=L(x) z + M(x)}{z - \norm{y}^2} & = L(x) z + M(x) - 2\dotp{y}{g(x,y)} \\
& \geq L(x) z + M(x) - (L(x) \norm{y}^2 + M(x) )\\
& = L(x) ( z  - \norm{y}^2)
\end{align*}

The ODEs $\D{x}=\genDE{x},\D{y}=g(x,y),\D{z}=L(x) z + M(x)$ are abbreviated $\ldots$ in the derivation below.
{\footnotesizeoff%
\begin{sequentdeduction}[array]
\linfer[dC]{
\linfer[MbW]{
\linfer[dbxineq]{
\linfer[qear]{
  \lclose
}
  {\lsequent{2\dotp{y}{g(x,y)} \leq L(x) \norm{y}^2 + M(x)} {L(x) z + M(x) - 2\dotp{y}{g(x,y)} \geq L(x) (z - \norm{y}^2)}}
}
  {\lsequent{z=\norm{y}^2} {\dbox{\pevolvein{\ldots}{\ivr(x) \land 2\dotp{y}{g(x,y)} \leq L(x) \norm{y}^2 + M(x)}}{\,z - \norm{y}^2 \geq 0}}}
}
  {\lsequent{z=\norm{y}^2} {\dbox{\pevolvein{\ldots}{\ivr(x) \land 2\dotp{y}{g(x,y)} \leq L(x) \norm{y}^2 + M(x)}}{\,\norm{y}^2 \leq z}}}
}
  {\lsequent{z=\norm{y}^2, \rrfvar_z} {\dbox{\pevolvein{\ldots}{\ivr(x)}}{\,\norm{y}^2 \leq z}}}
  \\[-\normalbaselineskip]\tag*{\qedhere}
\end{sequentdeduction}
}%
\end{itemize}
\end{proof}

The following topological $\didia{\cdot}$ ODE refinement axioms are quoted from Lemmas~\ref{lem:diatopaxioms} and~\ref{lem:closeddomref}.
The topological side conditions for these axioms are listed in Lemmas~\ref{lem:diatopaxioms} and~\ref{lem:closeddomref} respectively.
For semialgebraic postcondition $\rfvar$ and domain constraints $\ivr,\rrfvar$, these refinement axioms are derived syntactically from \dL's real induction axiom~\cite[Lemma A.2]{DBLP:journals/jacm/PlatzerT20}.
For the sake of generality, the proofs below directly use the topological conditions.

\begin{calculuscollection}
\begin{calculus}
\cinferenceRuleQuote{CORef}

\cinferenceRuleQuote{CRef}

\cinferenceRuleQuote{SARef}
\end{calculus}
\end{calculuscollection}

\begin{proof}[\textbf{Proof of Lemmas~\ref{lem:diatopaxioms} and~\ref{lem:closeddomref}\hypertarget{proof:proof4}}]\hypertarget{proof:proof26}
Let $\iget[state]{\I} \in \States$ and $\solvar : [0, T) \to \States, 0<T\leq\infty$ be the unique, right-maximal solution~\cite{Chicone2006,Walter1998} to the ODE $\D{x}=\genDE{x}$ with initial value $\solvar(0)=\iget[state]{\I}$.
By definition, $\solvar$ is differentiable, and therefore continuous.
This proof uses the fact that preimages under continuous functions of open sets are open~\cite[Theorem 4.8]{MR0385023}.
In particular, for an open set $\openset$, if $\solvar(t) \in \openset$ at some time $0 < t < T$ then the preimage of a sufficiently small open ball $\openset_\varepsilon \subseteq \openset$ centered at $\solvar(t)$ is open.
Thus, if $t>0$ and $\solvar(t) \in \openset$, then $\solvar$ stays in the open set $\openset$ for some open time interval\footnote{In case $t=0$, the time interval in~\rref{eq:openball} is truncated to the left with $\solvar(\zeta) \in \openset~\text{for all}~0 \leq \zeta < t+\varepsilon$.}around $t$, i.e., for some $\varepsilon > 0$:
\begin{align}
\solvar(\zeta) \in \openset~\text{for all}~t-\varepsilon \leq \zeta \leq t+\varepsilon
\label{eq:openball}
\end{align}

For the soundness proof of axioms~\irref{CORef},~\irref{CRef}, and~\irref{SARef}, assume that $\iget[state]{\I} \in \imodel{\I}{\ddiamond{\pevolvein{\D{x}=\genDE{x}}{\rrfvar}}{\rfvar}}$, i.e., there is a time $\tau \in [0,T)$ such that $\solvar(\tau) \in \imodel{\I}{\rfvar}$ and $\solvar(\zeta) \in \imodel{\I}{\rrfvar}$ for all $0 \leq \zeta \leq \tau$.
The proofs make use of the following set $\timeset$ containing all times $t$ such that the solution $\solvar$ never enters $\rfvar$ on the time interval $[0,t]$.
\begin{equation}
\timeset \mnodefequiv \{t~|~\solvar(\zeta) \notin \imodel{\I}{\rfvar}~\text{for all}~0 \leq \zeta \leq t \}
\label{eq:timeset}
\end{equation}

\begin{itemize}
\item[\irref{CORef}] For axiom~\irref{CORef}, assume that~$\iget[state]{\I} \in \imodel{\I}{\lnot{\rfvar} \land \dbox{\pevolvein{\D{x}=\genDE{x}}{\rrfvar \land \lnot{\rfvar}}}{\ivr}}$.
The set of times $\timeset$~\rref{eq:timeset} is non-empty since $\iget[state]{\I} = \solvar(0) \notin \imodel{\I}{\rfvar}$ so it has a supremum $t$ with $0 \leq t \leq \tau$ and $\solvar(\zeta) \notin \imodel{\I}{\rfvar}$ for all $0 \leq \zeta < t$.

\begin{itemize}
\item Suppose $\rfvar,\ivr$ both characterize topologically closed sets.
Since $\rfvar$ characterizes a topologically closed set, its complement formula $\lnot{\rfvar}$ characterizes a topologically open set. If $\solvar(t) \notin \imodel{\I}{\rfvar}$, i.e., $\solvar(t) \in \imodel{\I}{\lnot{\rfvar}}$, then $t < \tau$ and by~\rref{eq:openball}, the solution stays in $\lnot{\rfvar}$ until time $t + \varepsilon$ for some $\varepsilon > 0$, so $t$ is not the supremum of $\timeset$, which is a contradiction.
Thus, $\solvar(t) \in \imodel{\I}{\rfvar}$ and $0 < t$ because $\solvar(0)\notin \imodel{\I}{\rfvar}$.
Hence, $\solvar(\zeta) \in \imodel{\I}{\rrfvar \land \lnot{\rfvar}}$ for all $0 \leq \zeta < t$, which, together with the assumption $\iget[state]{\I} \in \imodel{\I}{\dbox{\pevolvein{\D{x}=\genDE{x}}{\rrfvar \land \lnot{\rfvar}}}{\ivr}}$ implies $\solvar(\zeta) \in \imodel{\I}{\ivr}$ for all $0 \leq \zeta < t$.
Since $\ivr$ characterizes a topologically closed set, this implies $\solvar(t) \in \imodel{\I}{\ivr}$; otherwise, $\solvar(t) \in \imodel{\I}{\lnot{\ivr}}$ and $\lnot{\ivr}$ characterizes an open set, so~\rref{eq:openball} implies $\solvar(\zeta) \in \imodel{\I}{\lnot{\ivr}}$ for some $0 \leq \zeta < t$, which contradicts the earlier observation that $\solvar(\zeta) \in \imodel{\I}{\ivr}$ for all $0 \leq \zeta < t$.
Thus, $\iget[state]{\I} \in \imodel{\I}{\ddiamond{\pevolvein{\D{x}=\genDE{x}}{\ivr}}{\rfvar}}$ because $\solvar(t) \in \imodel{\I}{\rfvar}$ and $\solvar(\zeta) \in \imodel{\I}{\ivr}$ for all $0 \leq \zeta \leq t$.

\item Suppose $\rfvar,\ivr$ both characterize topologically open sets.
Then, $\solvar(t) \notin \imodel{\I}{\rfvar}$; otherwise, $\solvar(t) \in \imodel{\I}{\rfvar}$ and since $\rfvar$ characterizes an open set, by~\rref{eq:openball}, there is a time $0 \leq \zeta < t$ where $\solvar(\zeta) \in \imodel{\I}{\rfvar}$, which contradicts $t$ being the supremum of $\timeset$.
Note that $t < \tau$ and $\solvar(\zeta) \in \imodel{\I}{\rrfvar \land \lnot{\rfvar}}$ for all $0 \leq \zeta \leq t$, which, together with the assumption $\iget[state]{\I} \in \imodel{\I}{\dbox{\pevolvein{\D{x}=\genDE{x}}{\rrfvar \land \lnot{\rfvar}}}{\ivr}}$ implies $\solvar(\zeta) \in \imodel{\I}{\ivr}$ for all $0 \leq \zeta \leq t$.
Since $\ivr$ characterizes a topologically open set, by~\rref{eq:openball}, there exists $\varepsilon > 0$ where $t+\varepsilon < \tau$ such that $\solvar(t+\zeta) \in \imodel{\I}{\ivr}$ for all $0 \leq \zeta \leq \varepsilon$.
By definition of the supremum, for every such $\varepsilon > 0$, there exists $\zeta$ where $0 < \zeta \leq \varepsilon$ and $\solvar(t+\zeta) \in \imodel{\I}{\rfvar}$, which yields the desired conclusion.
\end{itemize}

\item[\irref{CRef}] For axiom~\irref{CRef}, assume that $\iget[state]{\I} \in \imodel{\I}{\lnot{\rfvar}}$ and
\begin{align}
\iget[state]{\I} \in \imodel{\I}{\dbox{\pevolvein{\D{x}=\genDE{x}}{\rrfvar \land \lnot{\rfvar}}}{\interior{\ivr}}}
\label{eq:crefassum}
\end{align}

The set of times $\timeset$~\rref{eq:timeset} is non-empty since $\iget[state]{\I} = \solvar(0) \notin \imodel{\I}{\rfvar}$ so it has a supremum $t$ with $0 \leq t \leq \tau$ and $\solvar(\zeta) \notin \imodel{\I}{\rfvar}$ for all $0 \leq \zeta < t$.
Furthermore, $\solvar(\zeta) \in \imodel{\I}{\rrfvar \land \lnot{\rfvar}}$ for all $0 \leq \zeta < t$, so by~\rref{eq:crefassum}, $\solvar(\zeta) \in \imodel{\I}{\interior{\ivr}}$ for all $0 \leq \zeta < t$.
By assumption, formula $\interior{\ivr}$ characterizes the open topological interior of the closed formula $\ivr$ so by continuity of $\solvar$, $\solvar(t) \in \imodel{\I}{\ivr}$.
Furthermore, the interior of a set is contained in the set itself, i.e., $\imodel{\I}{\interior{\ivr}} \subseteq \imodel{\I}{\ivr}$, so $\solvar(\zeta) \in \imodel{\I}{\ivr}$ for all $0 \leq \zeta \leq t$ .
Classically, either $\solvar(t) \in \imodel{\I}{\rfvar}$ or $\solvar(t) \notin \imodel{\I}{\rfvar}$.

\begin{itemize}
\item If $\solvar(t) \in \imodel{\I}{\rfvar}$, then since $\solvar(\zeta) \in \imodel{\I}{\ivr}$ for all $0 \leq \zeta \leq t$, by definition, $\iget[state]{\I} \in \imodel{\I}{\ddiamond{\pevolvein{\D{x}=\genDE{x}}{\ivr}}{\rfvar}}$.

\item If $\solvar(t) \notin \imodel{\I}{\rfvar}$, then $t < \tau$ and furthermore, by~\rref{eq:crefassum}, $\solvar(t) \in \imodel{\I}{\interior{\ivr}}$.
Since the interior is topologically open, by~\rref{eq:openball}, there exists $\varepsilon > 0$ where $t+\varepsilon < \tau$ such that $\solvar(t+\zeta) \in \imodel{\I}{\interior{\ivr}} \subseteq \imodel{\I}{\ivr}$ for all $0 \leq \zeta \leq \varepsilon$.
By definition of the supremum, for every such $\varepsilon > 0$, there exists $\zeta$ where $0 < \zeta \leq \varepsilon$ and $\solvar(t+\zeta) \in \imodel{\I}{\rfvar}$, which yields the desired conclusion.

\end{itemize}

\item[\irref{SARef}] For axiom~\irref{SARef}, assume that
\begin{equation}
\iget[state]{\I} \in \imodel{\I}{\dbox{\pevolvein{\D{x}=\genDE{x}}{\rrfvar \land \lnot{(\rfvar \land \ivr)}}}{\ivr}}
\label{eq:sarefassum}
\end{equation}
If $\iget[state]{\I} \in \imodel{\I}{\rfvar \land \ivr}$, then $\iget[state]{\I} \in \ddiamond{\pevolvein{\D{x}=\genDE{x}}{\ivr}}{\rfvar}$ trivially by following the solution $\solvar$ for duration $0$.
Thus, assume $\iget[state]{\I} \notin \imodel{\I}{\rfvar \land \ivr}$.
From~\rref{eq:sarefassum}, $\iget[state]{\I} \in \imodel{\I}{\ivr}$ which further implies $\iget[state]{\I} \notin \imodel{\I}{\rfvar}$.
The set of times $\timeset$~\rref{eq:timeset} is non-empty since $\iget[state]{\I} = \solvar(0) \notin \imodel{\I}{\rfvar}$ and has a supremum $t$ with $0 \leq t \leq \tau$ and $\solvar(\zeta) \notin \imodel{\I}{\rfvar}$ for all $0 \leq \zeta < t$.
Thus, $\solvar(\zeta) \in \imodel{\I}{\rrfvar \land \lnot{(\rfvar \land \ivr)}}$ for all $0 \leq \zeta < t$.
By~\rref{eq:sarefassum}, $\solvar(\zeta) \in \imodel{\I}{\ivr}$ for all $0 \leq \zeta < t$.
Classically, either $\solvar(t) \in \imodel{\I}{\rfvar}$ or $\solvar(t) \notin \imodel{\I}{\rfvar}$.

\begin{itemize}
\item Suppose $\solvar(t) \in \imodel{\I}{\rfvar}$, if $\solvar(t) \in \imodel{\I}{\ivr}$, then $\solvar(\zeta) \in \imodel{\I}{\ivr}$ for all $0 \leq \zeta \leq t$ and so, by definition, $\iget[state]{\I} \in \imodel{\I}{\ddiamond{\pevolvein{\D{x}=\genDE{x}}{\ivr}}{\rfvar}}$.
On the other hand, if $\solvar(t) \notin \imodel{\I}{\ivr}$, then $\solvar(\zeta) \in \imodel{\I}{\rrfvar \land \lnot{(\rfvar \land \ivr)}}$ for all $0 \leq \zeta \leq t$, so from~\rref{eq:sarefassum}, $\solvar(t) \in \imodel{\I}{\ivr}$, which yields a contradiction.

If the formula $\rfvar$ is further assumed to characterize a closed set, this sub-case (with $\solvar(t) \in \imodel{\I}{\rfvar}$) is the only possibility.
Otherwise, $\solvar(t) \in \imodel{\I}{\lnot{\rfvar}}$ and $\lnot{\rfvar}$ characterizes an open set, so by~\rref{eq:openball}, for some $\varepsilon > 0 $, $\solvar(t+\zeta) \in \imodel{\I}{\lnot{\rfvar}}$ for all $0 \leq \zeta < \varepsilon$ which contradicts $t$ being the supremum of $\timeset$.

\item Suppose $\solvar(t) \notin \imodel{\I}{\rfvar}$, then $t < \tau$ and $\solvar(\zeta) \in \imodel{\I}{\rrfvar \land \lnot{(\rfvar \land \ivr)}}$ for all $0 \leq \zeta \leq t$, so from~\rref{eq:sarefassum}, $\solvar(t) \in \imodel{\I}{\ivr}$.
Since $\ivr$ is a formula of first-order real arithmetic, solutions of polynomial ODEs either locally progress into the set characterized by $\ivr$ or $\lnot{\ivr}$~\cite{DBLP:journals/jacm/PlatzerT20,DBLP:conf/fm/SogokonJ15}\footnote{This property is specific to sets characterized by first-order formulas of real arithmetic and polynomial ODEs (and certain topologically well-behaved extensions~\cite{DBLP:journals/jacm/PlatzerT20}) and is not true for arbitrary sets and ODEs.}, i.e., there exists $\varepsilon > 0$, where $t + \varepsilon < \tau$, such that either \circled{1} $\solvar(t+\zeta) \in \imodel{\I}{\ivr}$ for all $0 < \zeta \leq \varepsilon$ or \circled{2} $\solvar(t+\zeta) \notin \imodel{\I}{\ivr}$ for all $0 < \zeta \leq \varepsilon$.
Since $t$ is the supremum of $\timeset$, by definition, for every such $\varepsilon$ there exists $\zeta$ where $0 < \zeta \leq \varepsilon$ and $\solvar(t+\zeta) \in \imodel{\I}{\rfvar}$.
In case \circled{1}, since $\solvar(t+\zeta) \in \imodel{\I}{\rfvar}$ and $\solvar(\nu) \in \imodel{\I}{\ivr}$ for all $0 \leq \nu \leq t+\zeta$, then $\iget[state]{\I} \in \imodel{\I}{\ddiamond{\pevolvein{\D{x}=\genDE{x}}{\ivr}}{\rfvar}}$.
If the formula $\ivr$ is further assumed to characterize an open set, this sub-case (\circled{1}) is the only possibility, even if $\ivr$ is not a formula of first-order real arithmetic, because $\solvar(t) \in \imodel{\I}{\ivr}$ implies $\solvar$ continues to satisfy $\ivr$ for some time interval to the right of $t$ by~\rref{eq:openball}.
In case \circled{2}, observe that $\solvar(\nu) \in \imodel{\I}{\rrfvar \land \lnot{(\rfvar \land \ivr)}}$ for all $0 \leq \nu\leq t+\zeta$, from~\rref{eq:sarefassum}, $\solvar(t+\zeta) \in \imodel{\I}{\ivr}$, which yields a contradiction.
\qedhere
\end{itemize}

\end{itemize}
\end{proof}

The refinement axioms are pieced together in refinement chains~\rref{eq:refinementchain} to build ODE existence and liveness proof rules in a step-by-step manner.
However, all such refinement chains~\rref{eq:refinementchain} start from an initial hypothesis $\ddiamond{\pevolvein{\D{x}=\genDE{x}}{\ivr_0}}{\rfvar_0}$ from which the subsequent implications are proved.
The time existence axiom~\irref{TEx} from~\rref{sec:proveglobexist} provides the sole initial hypothesis that is needed for the refinement approach of this article.

\dinferenceRuleQuote{TEx}

\begin{proof}[\textbf{Proof of \rref{lem:timeexist}\hypertarget{proof:proof5}}]
Axiom~\irref{TEx} is derived directly from \dL's solution axiom~\cite{DBLP:journals/jar/Platzer17}.
It also has an easy semantic soundness proof which is given here.
Consider an initial state $\iget[state]{\I}$ and the corresponding modified state $\iget[state]{\I}_{\tau}^d$ where the value of variable $\tau$ is replaced by an arbitrary $d \in \reals$.
The (right-maximal) solution of ODE $\D{\timevar}=1$ from state $\iget[state]{\I}_{\tau}^d$ is given by the function $\solvar : [0,\infty) \to \States$, where $\solvar(\zeta)(t) = \iget[state]{\I}_{\tau}^d(t) + \zeta = \iget[state]{\I}(t) + \zeta$, and $\solvar(\zeta)(y) = \iget[state]{\I}_{\tau}^d(y)$ for all other variables $y$.
In particular, $\solvar(\zeta)(\tau) = d$.
Thus, at any time $\zeta > d - \iget[state]{\I}(t)$, $\solvar(\zeta)(t) = \iget[state]{\I}(t) + \zeta > d = \solvar(\zeta)(\tau)$.
This time $\zeta$ witnesses $\ddiamond{\pevolve{\D{\timevar}=1}}{\timevar > \tau}$.
\end{proof}

\subsection{Topological Side Conditions}
\label{app:topsidecalculus}

In~\rref{subsec:semantics}, topological conditions are defined for formulas $\fvarA$ that only mention free variables $x$ occurring in an ODE $\D{x}=\genDE{x}$.
For example, $\fvarA$ is said to characterize an open set with respect to $x$ iff the set $\imodel{\I}{\fvarA}$ is open when considered as a subset of $\reals^n$ (over variables $x=(x_1,\dots,x_n)$).
This section defines a more general notion, where $\fvarA$ is allowed to mention additional free parameters $y$ that do not occur in the ODE.
Adopting these (parametric) side conditions makes the topological refinement axioms that use them, like~\irref{CORef+CRef}, more general.
Let $(y_1,\dots,y_r) = \allvars \setminus \{x\}$ be parameters, and $\iget[state]{\I} \in \States$ be a state.
For brevity, write $y=(y_1,\dots,y_r)$ for the parameters and $\iget[state]{\I}(y) = (\iget[state]{\I}(y_1), \dots, \iget[state]{\I}(y_r)) \in \reals^r$ for the component-wise projection, and similarly for $\iget[state]{\I}(x) \in \reals^n$.
Given the set $\imodel{\I}{\fvarA} \subseteq \States$ and $\gamma \in \reals^r$, define:
\[
  \imodel{\I}{\fvarA}_\gamma \mdefeq \{ \iget[state]{\I}(x) \in \reals^n~|~\iget[state]{\I} \in \imodel{\I}{\fvarA}, \iget[state]{\I}(y) = \gamma \}
\]

The set $\imodel{\I}{\fvarA}_\gamma \subseteq \reals^n$ is the projection onto variables $x$ of all states $\iget[state]{\I}$ that satisfy $\fvarA$ and having values $\gamma$ for the parameters $y$.
Formula $\fvarA$ \emph{characterizes} a (topologically) open (resp. closed, bounded, compact) set with respect to variables $x$ iff for all $\gamma \in \reals^r$, the set $\imodel{\I}{\fvarA}_\gamma \subseteq \reals^n$ is topologically open (resp. closed, bounded, compact) with respect to the Euclidean topology.

These topological side conditions are decidable~\cite{Bochnak1998} for first-order formulas of real arithmetic $\rfvar,\ivr$ because in Euclidean spaces they can be phrased as conditions using first-order real arithmetic.
The following conditions are standard~\cite{Bochnak1998}, although special care is taken to universally quantify over the parameters $y$.
Let $\rfvar(x,y)$ be a formula mentioning variables $x$ and parameters $y$, then it is (with respect to variables $x$):
\begin{itemize}
\item \emph{open} if the formula $\lforall{y}{ \lforall{x}{\Big(\rfvar(x,y) \limply \lexists{\varepsilon {>} 0}{ \lforall{z}{\big( \norm{x-z}^2 < \varepsilon^2 \limply \rfvar(z,y) \big) } }\Big)} }$ is valid, where the variables $z = (z_1,\dots,z_n)$ are fresh for $\rfvar(x,y)$,
\item \emph{closed} if its complement formula $\lnot{\rfvar(x,y)}$ is open,
\item \emph{bounded} if the formula $\lforall{y}{ \lexists{r {>} 0}{ \lforall{x}{\big( \rfvar(x,y) \limply \norm{x}^2 {<} r^2 \big)} }}$ is valid, where variable $r$ is fresh for $\rfvar(x,y)$,
\item \emph{compact} if it is closed and bounded, by the Heine-Borel theorem~\cite[Theorem 2.4.1]{MR0385023}.
\end{itemize}

There are syntactic criteria that are sufficient (but not necessary\footnote{If there are no parameters $y$, these syntactic checks are ``necessary'' conditions in the sense that every open (resp. closed) formula $\rfvar$ is provably equivalent in real arithmetic to a (computable) formula formed from finite conjunctions and disjunctions of strict (resp. non-strict) inequalities~\cite[Theorem 2.7.2]{Bochnak1998}.}) for checking whether a formula satisfies the semantic conditions.
For example, the formula $\rfvar(x,y)$ is (with respect to variables $x$):
\begin{itemize}
\item \emph{open} if it is formed from finite conjunctions and disjunctions of strict inequalities $(\neq,>,<)$,
\item \emph{closed} if it is formed from finite conjunctions and disjunctions of non-strict inequalities $(=,\geq,\leq)$,
\item \emph{bounded} if it is of the form $\norm{x}^2 \pmc \ptermA(y) \land \rrfvar(x,y)$, where $\ptermA(y)$ is a term depending only on parameters $y$ and $\rrfvar(x,y)$ is a formula. This syntactic criterion uses the fact that the intersection of a bounded set (characterized by $\norm{x}^2 \pmc \ptermA(y)$) with any set (characterized by $\rrfvar(x,y)$) is bounded. The formula $\rfvar(x,y)$ is also \emph{compact} if $\pmc$ is $\leq$ and $\rrfvar(x,y)$ is closed.
\end{itemize}

These syntactic criteria are easily checkable by an implementation that inspects the syntactic shape of input formulas $\rfvar$.
In contrast, checking the semantic topological conditions for $\rfvar$ requires invoking expensive real arithmetic decision procedures.
For example, such a syntactic side condition enables the effective implementation of rule~\irref{cRef} from~\rref{cor:closeddomref} compared to its underlying axiom~\irref{CRef} from~\rref{lem:closeddomref} which is more general but uses requires checking semantic side conditions.

\section{Derived Existence and Liveness Proof Rules}
\label{app:derivationproofs}

This appendix syntactically derives all of the existence and liveness proof rules of the main article.
These derivations only use the sound \dL axioms and proof rules presented in~\rref{app:coreproofs}.
For ease of reference, this appendix is organized into four sections, corresponding to Sections~\ref{sec:globexist}--\ref{sec:impl} of the main article.
The high-level intuition behind these proofs is available as proof sketches in the main article while motivation for important proof steps is given directly in the subsequent proofs.
Further motivation for the surveyed liveness arguments can also be found in their original presentations~\cite{DBLP:journals/logcom/Platzer10,DBLP:conf/hybrid/PrajnaR05,DBLP:journals/siamco/PrajnaR07,DBLP:journals/siamco/RatschanS10,DBLP:conf/fm/SogokonJ15,DBLP:conf/emsoft/TalyT10}.

\subsection{Proofs for Finite-Time Blow Up and Global Existence}
\label{app:globexistproofs}

\begin{proof}[\textbf{Proof of \rref{cor:deporderexist}\hypertarget{proof:proof6}}]
Assume that the ODE $\D{x}=\genDE{x}$ is in dependency order~\rref{eq:deporder}.
The derivation successively removes the ODEs $y_k, y_{k-1}, \dots, y_1$ in reverse dependency order using either axiom~\irref{dBDG} or \irref{dDDG}, as shown below.
This continues until all of the ODEs are removed and the rightmost premise closes by axiom~\irref{TEx}.
The left premises arising from refinement with axioms~\irref{dBDG+dDDG} are the premises of rule~\irref{DEx}.
They are collectively labeled \circled{$\star$} and explained below.
{\footnotesizeoff%
\begin{sequentdeduction}[array]
\linfer[allr]{
\linfer[dBDG+dDDG]{
  \circled{$\star$}
  !
\linfer[dBDG+dDDG]{
  \circled{$\star$}
  !
\linfer[dBDG+dDDG]{
  \circled{$\star$}
  !
  \linfer[TEx]{
  \lclose
}
  {\lsequent{\Gamma}{\ddiamond{\pevolve{\D{\timevar}=1}}{\timevar > \tau}}}
}
    {\vdots}
}
{\lsequent{\Gamma}{\ddiamond{\pevolve{\D{y_1}=g_1(y_1),\dots,\D{y_{k-1}}=g_{k-1}(y_1,\dots,y_{k-1}),\D{\timevar}=1}}{\,\timevar > \tau}}}
}
  {\lsequent{\Gamma}{\ddiamond{\pevolve{\D{y_1}=g_1(y_1),\dots,\D{y_{k-1}}=g_{k-1}(y_1,\dots,y_{k-1}),\D{y_k}=g_k(y_1,\dots,y_k),\D{\timevar}=1}}{\,\timevar > \tau}}}
}
  {\lsequent{\Gamma}{\lforall{\tau}{\ddiamond{\pevolve{\underbrace{\D{y_1}=g_1(y_1),\dots,\D{y_{k-1}}=g_{k-1}(y_1,\dots,y_{k-1}),\D{y_k}=g_k(y_1,\dots,y_k)}_{\D{x}=\genDE{x}~\text{written in dependency order}},\D{\timevar}=1}}{\,\timevar > \tau}}}}
\end{sequentdeduction}
}%

At each step $i$, for $i=k,\dots,1$, the ODE $y_i$ in the succeedent is removed using either axiom~\irref{dBDG} or~\irref{dDDG}, depending on the user-chosen form (\rref{cor:deporderexist}) of postcondition $\rfvar_i$.

\begin{itemize}
\item[\bform]
In case formula $\rfvar_i \mnodefequiv \norm{y_i}^2 \leq \ptermA_i(t,y_1,\dots,y_{i-1})$ is of form \bform (as defined in~\rref{cor:deporderexist}), axiom~\irref{dBDG} is used.
This yields the two stacked premises shown below, where the top premise corresponds to premise \circled{$\star$} above.
The dependency order~\rref{eq:deporder} enables the sound use of axiom~\irref{dBDG} for this refinement step because the ODEs for $y_1,\dots,y_{i-1}$ are not allowed to depend on variables $y_i$.
The term $\ptermA(t,y_1,\dots,y_{i-1})$ also meets the dependency requirements of~\irref{dBDG} because it does not depend on $y_i$.
{\footnotesizeoff%
\begin{sequentdeduction}
\linfer[dBDG]{
  \begin{array}{l}
  \lsequent{\Gamma}{\dbox{\pevolve{\D{y_1}=g_1(y_1),\dots,\D{y_{i-1}}=g_{i-1}(y_1,\dots,y_{i-1}),\D{y_i}=g_i(y_1,\dots,y_i),\D{\timevar}=1}}{\rfvar_i}} \\
  \lsequent{\Gamma}{\ddiamond{\pevolve{\D{y_1}=g_1(y_1),\dots,\D{y_{i-1}}=g_{i-1}(y_1,\dots,y_{i-1}),\D{\timevar}=1}}{\,\timevar > \tau}}
  \end{array}
}
  {\lsequent{\Gamma}{\ddiamond{\pevolve{\D{y_1}=g_1(y_1),\dots,\D{y_{i-1}}=g_{i-1}(y_1,\dots,y_{i-1}),\D{y_i}=g_i(y_1,\dots,y_i),\D{\timevar}=1}}{\,\timevar > \tau}}}
\end{sequentdeduction}
}%

\item[\dform]
In case formula $\rfvar_i \mnodefequiv 2\dotp{y_i}{g_i(y_1,\dots,y_i)} \leq L_i(t,y_1,\dots,y_{i-1}) \norm{y_i}^2+M_i(t,y_1,\dots,y_{i-1})$ is of form \dform (as defined in~\rref{cor:deporderexist}), axiom~\irref{dDDG} is used instead.
Again, terms $L_i(t,y_1,\dots,y_{i-1}), M_i(t,y_1,\dots,y_{i-1})$ meet the dependency requirements of~\irref{dDDG} because they do not depend on $y_i$.
The top premise corresponds to premise \circled{$\star$} above, while the ODE for $y_i$ is removed in the bottom premise.
{\footnotesizeoff%
\begin{sequentdeduction}
\linfer[dDDG]{
  \begin{array}{l}
  \lsequent{\Gamma}{\dbox{\pevolve{\D{y_1}=g_1(y_1),\dots,\D{y_{i-1}}=g_{i-1}(y_1,\dots,y_{i-1}),\D{y_i}=g_i(y_1,\dots,y_i),\D{\timevar}=1}}{\rfvar_i}} \\
  \lsequent{\Gamma}{\ddiamond{\pevolve{\D{y_1}=g_1(y_1),\dots,\D{y_{i-1}}=g_{i-1}(y_1,\dots,y_{i-1}),\D{\timevar}=1}}{\,\timevar > \tau}}
  \end{array}
}
  {\lsequent{\Gamma}{\ddiamond{\pevolve{\D{y_1}=g_1(y_1),\dots,\D{y_{i-1}}=g_{i-1}(y_1,\dots,y_{i-1}),\D{y_i}=g_i(y_1,\dots,y_i),\D{\timevar}=1}}{\,\timevar > \tau}}}
  \\[-\normalbaselineskip]\tag*{\qedhere}
\end{sequentdeduction}
}%
\end{itemize}
\end{proof}

\begin{proof}[\textbf{Proof of \rref{cor:globalexistbase}\hypertarget{proof:proof7}}]
The proof closely follows the proof sketch for~\rref{cor:globalexistbase} but with an extra step to ensure that the chosen terms $L,M$ are within the term language of \dL.
Let the ODE $\D{x}=\genDE{x}$ be globally Lipschitz and $C$ be the (positive) Lipschitz constant for $f$, i.e., $\norm{\genDE{x}-\genDE{y}} \leq C \norm{x-y}$ .
Then $f$ satisfies the following inequality, where the first step~\rref{eq:globlipschitzboundtwo} is proved in the sketch but its RHS contains norms $\norm{\cdot}$ which are not in the term syntax (\rref{subsec:syntax}).
The inequality~\rref{eq:globlipschitzboundtwo} is prolonged by using inequality~\rref{eq:norminequality} to remove these non-squared norm terms, which yields corresponding choices of bounding \dL terms $L,M$.
\begin{align}
2\dotp{x}{f(x)} \overset{\rref{eq:globlipschitzboundtwo}}{\leq} \big(2C + \norm{f(0)}\big) \norm{x}^2  + \norm{f(0)} \overset{\rref{eq:norminequality}}{\leq} \underbrace{\big(2C + \frac{1}{2} (1+\norm{f(0)}^2)\big)}_{L} \norm{x}^2  + \underbrace{\frac{1}{2}(1+\norm{f(0)}^2) }_{M}
\label{eq:globalexistbaseproof}
\end{align}

The inequality~\rref{eq:globalexistbaseproof} is a valid real arithmetic formula and is thus provable by rule~\irref{qear}.
This enables the derivation below using axiom~\irref{dDDG} because $L,M$ satisfy the respective variable constraints of the axiom.
The resulting left premise is proved, after a~\irref{dW} step, by~\irref{qear}.
The resulting right premise, after the ODEs $\D{x}=\genDE{x}$ have been removed, is proved by axiom~\irref{TEx}.
{\footnotesizeoff%
\begin{sequentdeduction}[array]
\linfer[allr]{
\linfer[dDDG]{
  \linfer[dW]{
    \linfer[qear]{
        \lclose
    }
    {\lsequent{}{2\dotp{x}{f(x)} \leq L \norm{x}^2 + M}}
  }
  {\lsequent{}{\dbox{\pevolve{\D{x}=\genDE{x},\D{\timevar}=1}}{\,2\dotp{x}{f(x)} \leq L \norm{x}^2 + M}}} !
  \linfer[TEx]{
      \lclose
  }
  {\lsequent{}{\ddiamond{\pevolve{\D{\timevar}=1}}{\,\timevar > \tau}}}
}
  {\lsequent{}{\ddiamond{\pevolve{\D{x}=\genDE{x},\D{\timevar}=1}}{\,\timevar > \tau}}}
}
  {\lsequent{}{\lforall{\tau}{\ddiamond{\pevolve{\D{x}=\genDE{x},\D{\timevar}=1}}{\,\timevar > \tau}}}}
  \\[-\normalbaselineskip]\tag*{\qedhere}
\end{sequentdeduction}
}%
\end{proof}

\begin{proof}[\textbf{Proof of \rref{cor:globalexistaffine}\hypertarget{proof:proof8}}]
Assume that the ODE $\D{x}=\genDE{x}$ has affine dependency order~\rref{eq:deporder}, i.e., where each ODE $\D{y_i} = g_i(y_1,\dots,y_i)$ is of the affine form $\D{y_i} = A_i(y_1,\dots,y_{i-1}) y_i + b_i(y_1,\dots,y_{i-1})$ for some matrix and vector terms $A_i,b_i$ respectively with the indicated variable dependencies.
From the proof sketch for \rref{cor:globalexistaffine}, $A_i,b_i$ satisfy inequality~\rref{eq:globexistaffinebound} for each $i = 1,\dots,k$.
Like the proof of inequality~\rref{eq:globalexistbaseproof}, inequality~\rref{eq:globexistaffinebound} is prolonged by inequality~\rref{eq:norminequality} to remove non-squared norm terms in its RHS, which yields corresponding choices of bounding \dL terms $L_i, M_i$.
\begin{align}
2\dotp{y_i}{(A_iy_i + b_i)} \overset{\rref{eq:globexistaffinebound}}{\leq} (2\norm{A_i}+\norm{b_i})\norm{y_i}^2 + \norm{b_i}
\overset{\rref{eq:norminequality}}{\leq} \underbrace{\big(1{+}\norm{A_i}^2+\frac{1}{2}(1{+}\norm{b_i}^2)\big)}_{L_i}\norm{y_i}^2 + \underbrace{\frac{1}{2}(1{+}\norm{b_i}^2)}_{M_i}
\label{eq:globalexistaffineproof}
\end{align}

The inequality from~\rref{eq:globalexistaffineproof} is a valid real arithmetic formula, and thus provable by~\irref{qear} for each $i=1,\dots,k$.
The derivation uses rule~\irref{DEx}, where the postcondition of each premise is chosen to be of form \dform.
The resulting premises are all proved, after a~\irref{dW} step, by~\irref{qear} with the above choice of $L_i,M_i$ for each $i = 1,\dots,k$.
{\footnotesizeoff%
\begin{sequentdeduction}[array]
\linfer[DEx]{
  \linfer[dW]{
    \linfer[qear]{\lclose}{\lsequent{}{2\dotp{y_1}{(A_1y_1 + b_1)} \leq L_1 \norm{y_1}^2 + M_1}}
  }
  {\lsequent{}{\dbox{\pevolve{\D{y_1}=g_1(y_1),\D{\timevar}=1}}{\rfvar_1}}} !
  \cdots !
  \linfer[dW]{
    \linfer[qear]{\lclose}{\lsequent{}{2\dotp{y_k}{(A_ky_k + b_k)} \leq L_k \norm{y_k}^2 + M_k}}
  }
  {\lsequent{}{\dbox{\pevolve{\D{y_1}=g_1(y_1),\dots,\D{y_k}=g_k(y_1,\dots,y_k),\D{\timevar}=1}}{\rfvar_k}}}
}
  {\lsequent{}{\lforall{\tau}{\ddiamond{\pevolve{\D{x}=\genDE{x},\D{\timevar}=1}}{\,\timevar > \tau \qedhere}}}}
\end{sequentdeduction}
}%
\end{proof}

\begin{proof}[\textbf{Proof of \rref{cor:boundedexistbase}\hypertarget{proof:proof9}}]
The derivation starts by Skolemizing with~\irref{allr}, then switching the diamond modality in the succedent to a box modality in the antecedent using~\irref{diamond+notr}.
The postcondition of the box modality is simplified using the propositional tautologies $\lnot{(\fvarA \lor \fvarB)} \lbisubjunct \lnot{\fvarA} \land \lnot{\fvarB}$ and $\lnot{\lnot{\fvarA}} \lbisubjunct \fvarA$.
Axiom~\irref{band+andl} splits the conjunction in the antecedent, before~\irref{diamond} is used again to flip the left antecedent to a diamond modality in the succedent.
These (mostly) propositional steps recover the more verbose phrasing of~\irref{BEx} from~\rref{eq:existence3}.
{\footnotesizeoff%
\begin{sequentdeduction}[array]
\linfer[allr]{
\linfer[diamond+notr]{
\linfer[band+andl]{
\linfer[diamond+notl]{
   \lsequent{\dbox{\pevolve{\D{x}=\genDE{x},\D{\timevar}=1}}{\boundedf(x)}}{\ddiamond{\pevolve{\D{x}=\genDE{x},\D{\timevar}=1}}{\,\timevar > \tau}}
}
  {\lsequent{\dbox{\pevolve{\D{x}=\genDE{x},\D{\timevar}=1}}{\lnot{(\timevar > \tau)}}, \dbox{\pevolve{\D{x}=\genDE{x},\D{\timevar}=1}}{\boundedf(x)}}{\lfalse}}
}
  {\lsequent{\dbox{\pevolve{\D{x}=\genDE{x},\D{\timevar}=1}}{(\lnot{(\timevar > \tau)} \land \boundedf(x))}}{\lfalse}}
}
  {\lsequent{}{\ddiamond{\pevolve{\D{x}=\genDE{x},\D{\timevar}=1}}{(\timevar > \tau \lor \lnot{\boundedf(x)})}}}
}
  {\lsequent{}{\lforall{\tau}{\ddiamond{\pevolve{\D{x}=\genDE{x},\D{\timevar}=1}}{(\timevar > \tau \lor \lnot{\boundedf(x)})}}}}
\end{sequentdeduction}
}%

The formula $\boundedf(x)$ is assumed to characterize a bounded set with respect to the variables $x$.
The closure of this set (with respect to $x$) is compact so the continuous norm function $\norm{\cdot}^2$ attains its maximum value on that set.
Hence, the formula $\lexists{D}{\lforall{x}{(\boundedf(x) \limply \norm{x}^2 \leq D)}}$ is valid in first-order real arithmetic, and is thus provable by~\irref{qear}.
The derivation continues with a~\irref{cut} of this formula and Skolemizing with~\irref{existsl}.
Axiom~\irref{dBDG} is then used to remove the ODE $\D{x}=\genDE{x}$ with $\ptermA(x) \mnodefeq D$.
The resulting right premise is proved by axiom~\irref{TEx}, while the resulting left premise is labeled \circled{1} and continued below.
{\footnotesizeoff%
\begin{sequentdeduction}[array]
\linfer[cut+qear+existsl]{
\linfer[dBDG]{
    \circled{1} !
    \linfer[TEx]{\lclose}
    {\lsequent{}{\ddiamond{\pevolve{\D{\timevar}=1}}{\,\timevar > \tau}}}
}
  {\lsequent{\dbox{\pevolve{\D{x}=\genDE{x},\D{\timevar}=1}}{\boundedf(x)},\lforall{x}{(\boundedf(x) \limply \norm{x}^2 \leq D)}}{\ddiamond{\pevolve{\D{x}=\genDE{x},\D{\timevar}=1}}{\,\timevar > \tau}}}
}
  {\lsequent{\dbox{\pevolve{\D{x}=\genDE{x},\D{\timevar}=1}}{\boundedf(x)}}{\ddiamond{\pevolve{\D{x}=\genDE{x},\D{\timevar}=1}}{\,\timevar > \tau}}}
\end{sequentdeduction}
}%

From premise \circled{1}, a~\irref{dC} step adds the postcondition of the leftmost antecedent, $\boundedf(x)$, to the domain constraint.
Since the remaining antecedent is universally quantified over variables $x$, it is soundly kept across an application of a subsequent~\irref{dW} step, and the proof is completed with~\irref{alll+implyl}.
{\footnotesizeoff%
\begin{sequentdeduction}[array]
 \linfer[dC]{
 \linfer[dW]{
 \linfer[alll+implyl]{
    \lclose
  }
    {\lsequent{\lforall{x}{(\boundedf(x) \limply \norm{x}^2 \leq D)}, \boundedf(x)}{\norm{x}^2 \leq D}}
  }
    {\lsequent{\lforall{x}{(\boundedf(x) \limply \norm{x}^2 \leq D)}}{\dbox{\pevolvein{\D{x}=\genDE{x},\D{\timevar}=1}{\boundedf(x)}}{\,\norm{x}^2 \leq D}}}
    }
    {\lsequent{\dbox{\pevolve{\D{x}=\genDE{x},\D{\timevar}=1}}{\boundedf(x)},\lforall{x}{(\boundedf(x) \limply \norm{x}^2 \leq D)}}{\dbox{\pevolve{\D{x}=\genDE{x},\D{\timevar}=1}}{\,\norm{x}^2 \leq D}}}
  \\[-\normalbaselineskip]\tag*{\qedhere}
\end{sequentdeduction}
}%
\end{proof}

\begin{proof}[\textbf{Proof of \rref{cor:boundedexistgen}\hypertarget{proof:proof10}}]
Assume the ODE $\D{x}=\genDE{x}$ is in dependency order~\rref{eq:deporder}, and the indices $i=1,\dots,k$ are partitioned into disjoint sets $L,N$ as in~\rref{cor:boundedexistgen}.
The first step Skolemizes with~\irref{allr}.
{\footnotesizeoff%
\begin{sequentdeduction}[array]
\linfer[allr]{
  \lsequent{}{\ddiamond{\pevolve{\D{x}=\genDE{x},\D{\timevar}=1}}{\big(\timevar > \tau \lor \lorfold_{j \in N}\lnot{\boundedf_j(y_j)}\big)}}
}
  {\lsequent{}{\lforall{\tau}{\ddiamond{\pevolve{\D{x}=\genDE{x},\D{\timevar}=1}}{\big(\timevar > \tau \lor \lorfold_{j \in N}\lnot{\boundedf_j(y_j)}\big)}}}}
\end{sequentdeduction}
}%

The derivation uses ideas from Corollaries~\ref{cor:deporderexist},\ref{cor:globalexistaffine}, and~\ref{cor:boundedexistbase} to remove the ODE $\D{y_i} = g_i(y_1,\dots,y_i)$ at each step.
The corresponding disjunct $\lnot{\boundedf_i(y_i)}$ (if present) is also removed from the succedent when $i \in N$.
More precisely, at each step $i$, the derivation turns a succedent of the form~\rref{eq:preremove} to the form~\rref{eq:postremove} below which removes the variables $y_i$ from the formula.

\begin{align}
\ddiamond{\pevolve{\D{y_1}=g_1(y_1),\dots,\D{y_{i-1}}=g_{i-1}(y_1,\dots,y_{i-1}),\D{y_i}=g_i(y_1,\dots,y_i),\D{\timevar}=1}}{\big(\timevar > \tau \lor \lorfold_{j \in N \cap \{ 1, \dots, i\}}\lnot{\boundedf_j(y_j)}\big)}
\label{eq:preremove}
\end{align}

\begin{align}
\ddiamond{\pevolve{\D{y_1}=g_1(y_1),\dots,\D{y_{i-1}}=g_{i-1}(y_1,\dots,y_{i-1}),\D{\timevar}=1}}{\big(\timevar > \tau \lor \lorfold_{j \in N \cap \{ 1, \dots, i-1\}}\lnot{\boundedf_j(y_j)}\big)}
\label{eq:postremove}
\end{align}
The derivation proceeds with two cases depending on whether $i \in L$ or $i \in N$.
\begin{itemize}
\item For each $i \in L$ (similarly to~\rref{cor:globalexistaffine}), the ODE $\D{y_i} = A_i(y_1,\dots,y_{i-1}) y_i + b_i(y_1,\dots,y_{i-1})$ is affine for some matrix and vector terms $A_i,b_i$ respectively with the indicated variable dependencies.
The RHS of this affine ODE satisfies the inequality~\rref{eq:globalexistaffineproof} with terms $L_i,M_i$ as given in~\rref{eq:globalexistaffineproof}.
Axiom~\irref{dDDG} is used with those choices of $L_i,M_i$, which removes the ODEs for $y_i$ in the resulting right premise.
The resulting left premise is labeled \circled{1} and explained below.
Note that the freshness conditions of axiom~\irref{dDDG} are met because the postcondition of the succedent does not mention variables $y_i$ for $i \in L$.
Similarly, the indices from $j\in N  \cap \{ 1, \dots, i\}$ are equal to those from $j\in N  \cap \{ 1, \dots, i-1\}$ because $i \notin N$.
{\footnotesizeoff%
\begin{sequentdeduction}
\linfer[dDDG]{
  \circled{1} &
  \lsequent{}{\ddiamond{\pevolve{\D{y_1}=g_1(y_1),\dots,\D{y_{i-1}}=g_{i-1}(y_1,\dots,y_{i-1}),\D{\timevar}=1}}{\big(\timevar > \tau \lor \lorfold_{i\in N  \cap \{ 1, \dots, i-1\}}\lnot{\boundedf_i(y_i)}\big)}}
}
  {\lsequent{}{\ddiamond{\pevolve{\D{y_1}=g_1(y_1),\dots,\D{y_{i-1}}=g_{i-1}(y_1,\dots,y_{i-1}),\D{y_i}=g_i(y_1,\dots,y_i),\D{\timevar}=1}}{\big(\timevar > \tau \lor \lorfold_{i\in N  \cap \{ 1, \dots, i\}}\lnot{\boundedf_i(y_i)}\big)}}}
\end{sequentdeduction}
}%

From premise \circled{1}, the proof is completed with a~\irref{dW} and~\irref{qear} step using inequality~\rref{eq:globalexistaffineproof}.
{\footnotesizeoff%
\begin{sequentdeduction}[array]
\linfer[dW]{
    \linfer[qear]{\lclose}{\lsequent{}{2\dotp{y_i}{(A_iy_i + b_i)} \leq L_i \norm{y_i}^2 + M_i}}
  }
  {\lsequent{}{\dbox{\pevolve{\D{y_1}=g_1(y_1),\dots,\D{y_i}=g_i(y_1,\dots,y_i),\D{\timevar}=1}}{\,2\dotp{y_i}{(A_iy_i + b_i)} \leq L_i \norm{y_i}^2 + M_i}}}
\end{sequentdeduction}
}%

\item For each $i \in N$ (similarly to~\rref{cor:boundedexistbase}), the boundedness assumption on $y_i$ is first extracted from the succedent, with the abbreviation $\rrfvar \mnodefequiv (\timevar > \tau \lor \lorfold_{j\in N \cap \{1,\dots,i-1\}}\lnot{\boundedf_j(y_j)})$.
The bottommost succedent is similarly abbreviated using the propositional tautology $\big(\timevar > \tau \lor \lorfold_{j\in N \cap \{1,\dots,i\}}\lnot{\boundedf_j(y_j)}\big) \lbisubjunct \rrfvar \lor \lnot{\boundedf_i(y_i)}$.
{\footnotesizeoff%
\begin{sequentdeduction}
\linfer[diamond+notr]{
  \linfer[band+andl]{
\linfer[diamond+notl]{
   \lsequent{\dbox{\pevolve{\D{y_1}=g_1(y_1),\dots,\D{y_i}=g_i(y_1,\dots,y_i),\D{\timevar}=1}}{\boundedf_i(y_i)}}{\ddiamond{\pevolve{\D{y_1}=g_1(y_1),\dots,\D{y_i}=g_i(y_1,\dots,y_i),\D{\timevar}=1}}{\rrfvar}}
}
  {\lsequent{\dbox{\pevolve{\D{y_1}=g_1(y_1),\dots,\D{y_i}=g_i(y_1,\dots,y_i),\D{\timevar}=1}}{\lnot{\rrfvar}}, \dbox{\pevolve{\D{y_1}=g_1(y_1),\dots,\D{y_i}=g_i(y_1,\dots,y_i),\D{\timevar}=1}}{\boundedf_i(y_i)}}{\lfalse}}
}
  {\lsequent{\dbox{\pevolve{\D{y_1}=g_1(y_1),\dots,\D{y_i}=g_i(y_1,\dots,y_i),\D{\timevar}=1}}{\big( \lnot{\rrfvar} \land \boundedf_i(y_i)\big)}}{\lfalse}}
}
  {\lsequent{}{\ddiamond{\pevolve{\D{y_1}=g_1(y_1),\dots,\D{y_i}=g_i(y_1,\dots,y_i),\D{\timevar}=1}}{(\rrfvar \lor \lnot{\boundedf_i(y_i)})}}}
\end{sequentdeduction}
}%

The formula $\boundedf_i(y_i)$ is assumed to characterize a bounded set with respect to the variables $y_i$.
Thus, like~\rref{cor:boundedexistbase}, the~\irref{cut} of the formula $\lexists{D_i}{\lforall{y_i}{(\boundedf_i(y_i) \limply \norm{y_i}^2 \leq D_i)}}$ is proved by~\irref{qear}.
The derivation continues by Skolemizing, abbreviating $\rsfvar \mnodefequiv \dbox{\pevolve{\D{y_1}=g_1(y_1),\dots,\D{y_i}=g_i(y_1,\dots,y_i),\D{\timevar}=1}}{\boundedf_i(y_i)}$.
Axiom~\irref{dBDG} is then used with $\ptermA(y_i) \mnodefeq D_i$, which removes the ODEs for $y_i$ in the resulting right premise.
The resulting left premise is labeled \circled{2} and explained below.
{\footnotesizeoff%
\begin{sequentdeduction}[array]
\linfer[cut+qear+existsl]{
\linfer[dBDG]{
  \circled{2} ! \lsequent{}{\ddiamond{\pevolve{\D{y_1}=g_1(y_1),\dots,\D{y_{i-1}}=g_{i-1}(y_1,\dots,y_{i-1}),\D{\timevar}=1}}{\rrfvar}}
}
  {\lsequent{\rsfvar,\lforall{y_i}{(\boundedf_i(y_i) \limply \norm{y_i}^2 \leq D_i)}}{\ddiamond{\pevolve{\D{y_1}=g_1(y_1),\dots,\D{y_i}=g_i(y_1,\dots,y_i),\D{\timevar}=1}}{\rrfvar}}}
}
  {\lsequent{\rsfvar}{\ddiamond{\pevolve{\D{y_1}=g_1(y_1),\dots,\D{y_i}=g_i(y_1,\dots,y_i),\D{\timevar}=1}}{\rrfvar}}
}
\end{sequentdeduction}
}%

The derivation continues from premise \circled{2} identically to~\rref{cor:boundedexistbase}, with a~\irref{dC} step to add the postcondition of the antecedent $\rsfvar$ to the domain constraint.
The proof is completed with~\irref{dW} and~\irref{alll+implyl}.
The universally quantified antecedent $\lforall{y_i}{\dots}$ is soundly kept across the use of~\irref{dW} since it does not mention any of the bound variables $y_1,\dots,y_i,t$ of the ODE free.
{\footnotesizeoff%
\begin{sequentdeduction}[array]
 \linfer[dC]{
 \linfer[dW]{
 \linfer[alll+implyl]{
    \lclose
  }
    {\lsequent{\lforall{y_i}{(\boundedf(y_i) \limply \norm{y_i}^2 \leq D_i)}, \boundedf(y_i)}{\norm{y_i}^2 \leq D_i}}
  }
    {\lsequent{\lforall{y_i}{(\boundedf(y_i) \limply \norm{y_i}^2 \leq D_i)}}{\dbox{\pevolvein{\D{y_1}=g_1(y_1),\dots,\D{y_i}=g_i(y_1,\dots,y_i),\D{\timevar}=1}{\boundedf(y_i)}}{\,\norm{y_i}^2 \leq D_i}}}
    }
    {\lsequent{\rsfvar,\lforall{y_i}{(\boundedf(y_i) \limply \norm{y_i}^2 \leq D_i)}}{\dbox{\pevolve{\D{y_1}=g_1(y_1),\dots,\D{y_i}=g_i(y_1,\dots,y_i),\D{\timevar}=1}}{\,\norm{y_i}^2 \leq D_i}}}
\end{sequentdeduction}
}%
\end{itemize}

Using the steps for $i=k,\dots,1$ (where either $i\in L$ or $i \in N$) successively removes the ODEs for $y_k,\dots,y_i$ from the succedent.
This is shown in the derivation below and the proof is completed using~\irref{TEx}.
{\footnotesizeoff%
\begin{sequentdeduction}[array]
\linfer[]{
\linfer[]{
\linfer[]{
\linfer[]{
\linfer[TEx]{
    \lclose
}
   {\lsequent{}{\ddiamond{\pevolve{\D{\timevar}=1}}{\timevar > \tau }}}
}
  {\lsequent{}{\ddiamond{\pevolve{\D{\timevar}=1}}{\big(\timevar > \tau \lor \lorfold_{j \in N \cap \emptyset}\lnot{\boundedf_j(y_j)}\big)}}}
}
  {\vdots}
}
  {\lsequent{}{\ddiamond{\pevolve{\D{y_1}=g_1(y_1),\dots,\D{y_{k-1}}=g_{k-1}(y_1,\dots,y_{k-1}),\D{\timevar}=1}}{\big(\timevar > \tau \lor \lorfold_{j \in N \cap \{1, \dots, k-1\}}\lnot{\boundedf_j(y_j)}\big)}}}
}
  {\lsequent{}{\ddiamond{\pevolve{\D{y_1}=g_1(y_1),\dots,\D{y_{k-1}}=g_{k-1}(y_1,\dots,y_{k-1}),\D{y_k}=g_k(y_1,\dots,y_k),\D{\timevar}=1}}{\big(\timevar > \tau \lor \lorfold_{j \in N}\lnot{\boundedf_j(y_j)}\big)}}}
  \\[-\normalbaselineskip]\tag*{\qedhere}
\end{sequentdeduction}
}%
\end{proof}

\begin{proof}[\textbf{Proof of \rref{prop:globalexistcomplete}\hypertarget{proof:proof11}}]
The ODE $\D{x}=\genDE{x}$ is assumed to have a global solution that is syntactically representable by polynomial term $X(t)$ in the term language (\rref{subsec:syntax}).
Formally, this representability condition means that for any initial state $\omega$, the mathematical solution $\solvar : [0, \infty) \to \States$ exists globally and in addition, for each time $\tau \in [0,\infty)$, the solution satisfies $\solvar(\tau) = \ivaluation{\dLint[state=\omega_t^\tau]}{X(t)}$, where $\ivaluation{\dLint[state=\omega_t^\tau]}{X(t)}$ is the value of term $X(t)$ in state $\omega$ with the value of time variable $t$ set to $\tau$.
This implies that the following formula is valid because terms $x,t-t_0$ have value $\solvar(\tau)$ and $\tau$ respectively at time $\tau \in [0,\infty)$ along the ODE $\D{x}=\genDE{x},\D{\timevar}=1$.
The variables $x_0,t_0$ store the initial values of $x,t$ respectively, which may be needed for the syntactic representation $X(t)$ of the solution.
Additionally, the syntactic representation $X(t)$ may mention parameters $y \notin x$ that remain constant for the ODE $\D{x}=\genDE{x}$.
\begin{align}
t=t_0 \land x=x_0 \limply \dbox{\D{x}=\genDE{x},\D{t}=1}{\,x = X(t-t_0)}
\label{eq:globalexistcompleterep}
\end{align}

Validity of formula~\rref{eq:globalexistcompleterep} further implies that~\rref{eq:globalexistcompleterep} is provable because of the \dL completeness theorem for equational invariants~\cite[Theorem 4.5]{DBLP:journals/jar/Platzer17,DBLP:journals/jacm/PlatzerT20}.
The derivation of global existence for $\D{x}=\genDE{x}$ first Skolemizes with~\irref{allr}, then introduces fresh variables $x_0,t_0$ storing the initial values of $x,t$ with~\irref{cut+qear+existsl}.
Axiom~\irref{dBDG} is used with $\ptermA(t) \mnodefeq \norm{X(t-t_0)}^2$ to remove the ODEs $\D{x}=\genDE{x}$.
The resulting right premise is proved by axiom~\irref{TEx}, while the resulting left premise is abbreviated \circled{1} and proved below.
{\footnotesizeoff%
\begin{sequentdeduction}[array]
\linfer[allr]{
\linfer[cut+qear+existsl]{
\linfer[dBDG]{
\circled{1}
!
\linfer[TEx]{
    \lclose
}
   {\lsequent{}{\ddiamond{\pevolve{\D{\timevar}=1}}{\,\timevar > \tau }}}
}
  {\lsequent{t=t_0 \land x=x_0 }{\ddiamond{\pevolve{\D{x}=\genDE{x},\D{\timevar}=1}}{\,\timevar > \tau}}}
}
  {\lsequent{}{\ddiamond{\pevolve{\D{x}=\genDE{x},\D{\timevar}=1}}\,{\timevar > \tau}}}
}
  {\lsequent{}{\lforall{\tau}{\ddiamond{\pevolve{\D{x}=\genDE{x},\D{\timevar}=1}}{\,\timevar > \tau}}}}
\end{sequentdeduction}
}%

From \circled{1}, the derivation continues with a~\irref{dC} using the provable formula~\rref{eq:globalexistcompleterep}.
The premise after~\irref{dW} is proved by~\irref{qear} after rewriting the succedent with the equality $x=X(t-t_0)$ and by reflexivity of $\leq$.
{\footnotesizeoff%
\begin{sequentdeduction}[array]
\linfer[dC]{
\linfer[dW]{
\linfer[qear]{
    \lclose
}
  {\lsequent{x = X(t-t_0)}{\norm{x}^2 \leq \norm{X(t-t_0)}^2}}
}
  {\lsequent{}{\dbox{\pevolvein{\D{x}=\genDE{x},\D{\timevar}=1}{x = X(t-t_0)}}{\,\norm{x}^2 \leq \norm{X(t-t_0)}^2}}}
}
{\lsequent{t=t_0 \land x=x_0}{\dbox{\pevolve{\D{x}=\genDE{x},\D{\timevar}=1}}{\,\norm{x}^2 \leq \norm{X(t-t_0)}^2}}}
\end{sequentdeduction}
}%

Note that, instead of assuming that $X(t)$ is a syntactically representable (global) solution for the ODE $\D{x}=\genDE{x}$, it also suffices for this derivation to assume that premise \circled{1} is provable, i.e., that the term $\norm{X(t-t_0)}^2$ (with free variables $t,x_0,t_0$ and parameters $y$) is a provable upper bound on the squared norm of $x$ along solutions of the ODE.
\end{proof}

\subsection{Proofs for Liveness Without Domain Constraints}
\label{app:livenodomproofs}

\begin{proof}[\textbf{Proof of \rref{cor:atomicdvcmp}\hypertarget{proof:proof12}}]
The complete derivation of rule~\irref{dVcmpA} using refinement axiom~\irref{Prog} and rule~\irref{dIcmp} is already given in the proof sketch for~\rref{cor:atomicdvcmp} so it is not repeated here.

The derivation of~\irref{dVcmp} (as a corollary of~\irref{dVcmpA}) starts by introducing fresh variables $\ptermA_0, i$ representing the initial values of $\ptermA$ and the multiplicative inverse of $\varepsilon()$ respectively using arithmetic cuts (\irref{cut+qear}) and Skolemizing (\irref{existsl}).
It then uses \irref{dGt} to introduce a fresh time variable to the system of differential equations:
{\footnotesizeoff%
\begin{sequentdeduction}[array]
  \linfer[cut+qear]{
  \linfer[existsl]{
  \linfer[dGt]{
    \lsequent{\Gamma, \constt{\varepsilon} > 0, \ptermA=\ptermA_0, i\varepsilon() = 1, \timevar=0}{\ddiamond{\pevolve{\D{x}=\genDE{x},\D{\timevar}=1}}{\ptermA \cmp 0}}
  }
    {\lsequent{\Gamma, \constt{\varepsilon} > 0, \ptermA=\ptermA_0, i\varepsilon() = 1}{\ddiamond{\pevolve{\D{x}=\genDE{x}}}{\ptermA \cmp 0}}}
  }
  {\lsequent{\Gamma, \constt{\varepsilon} > 0, \lexists{\ptermA_0}{(\ptermA=\ptermA_0)}, \lexists{i}{(i\varepsilon() = 1)}}{\ddiamond{\pevolve{\D{x}=\genDE{x}}}{\ptermA \cmp 0}}}
  }
  {\lsequent{\Gamma, \constt{\varepsilon} > 0}{\ddiamond{\pevolve{\D{x}=\genDE{x}}}{\ptermA \cmp 0}} }
\end{sequentdeduction}
}%

Next, an initial liveness assumption $\ddiamond{\pevolve{\D{x}=\genDE{x},\D{\timevar}=1}}{\ptermA_0 + \constt{\varepsilon}\timevar > 0}$ is cut into the antecedents after which rule~\irref{dVcmpA} is used to obtain the premise of~\irref{dVcmp}.
Intuitively, this initial liveness assumption says that the solution exists for sufficiently long, so that the term $\ptermA_0 + \constt{\varepsilon}\timevar$ (which is proved to lower bound $\ptermA$) becomes positive for sufficiently large $\timevar$.
This cut is abbreviated \circled{1} and proved below.
{\footnotesizeoff%
\begin{sequentdeduction}[array]
  \linfer[cut]{
  \linfer[dVcmpA]{
    \lsequent{\lnot{(\ptermA \cmp 0)}}{\lied[]{\genDE{x}}{\ptermA}\geq \constt{\varepsilon}}
  }
    {\lsequent{\Gamma, \ptermA=\ptermA_0, \timevar=0, \ddiamond{\pevolve{\D{x}=\genDE{x},\D{\timevar}=1}}{\,\ptermA_0 + \constt{\varepsilon}\timevar > 0} }{\ddiamond{\pevolve{\D{x}=\genDE{x},\D{\timevar}=1}}{\ptermA \cmp 0}}} \quad
    \circled{1}
  }
  {\lsequent{\Gamma, \constt{\varepsilon} > 0, \ptermA=\ptermA_0, i\varepsilon() = 1, \timevar=0}{\ddiamond{\pevolve{\D{x}=\genDE{x},\D{\timevar}=1}}{\ptermA \cmp 0}}}
\end{sequentdeduction}
}%

From premise~\circled{1}, a monotonicity step~\irref{MdW} equivalently rephrases the postcondition of the cut in real arithmetic.
The arithmetic rephrasing works using the constant assumption $\constt{\varepsilon} > 0$ and the choice of $i$ as the multiplicative inverse of $\constt{\varepsilon}$.
Since the ODE $\D{x}=\genDE{x}$ is assumed to have provable global solutions, axiom~\irref{GEx} finishes the derivation by instantiating $\tau \mnodefeq - i \ptermA_0$, which is constant for the ODE.
{\footnotesizeoff%
\begin{sequentdeduction}[array]
  \linfer[qear+MdW]{
  \linfer[GEx]{
    \lclose
  }
    {\lsequent{\Gamma}{\ddiamond{\pevolve{\D{x}=\genDE{x},\D{\timevar}=1}}{\,\timevar > -i \ptermA_0}}}
  }
  {\lsequent{\Gamma, \constt{\varepsilon} > 0, i\varepsilon() = 1}{\ddiamond{\pevolve{\D{x}=\genDE{x},\D{\timevar}=1}}{\,\ptermA_0 + \constt{\varepsilon}\timevar > 0}}}
\\[-\normalbaselineskip]\tag*{\qedhere}
\end{sequentdeduction}
}%
\end{proof}

\begin{proof}[\textbf{Proof of \rref{cor:tt}\hypertarget{proof:proof13}}]
Rule \irref{TT} is derived directly from \irref{dVeq} with a~\irref{MdW} monotonicity step:
{\footnotesizeoff%
\begin{sequentdeduction}[array]
  \linfer[MdW]{
  \lsequent{\ptermA=0}{\rfvar}
  !
  \linfer[dVeq]{
    \lsequent{\ptermA < 0}{\lied[]{\genDE{x}}{\ptermA}\geq \constt{\varepsilon}}
  }
  {\lsequent{\Gamma,\constt{\varepsilon} > 0, \ptermA \leq 0}{\ddiamond{\pevolve{\D{x}=\genDE{x}}}{\ptermA = 0}}}
  }
  {\lsequent{\Gamma,\constt{\varepsilon} > 0, \ptermA \leq 0}{\ddiamond{\pevolve{\D{x}=\genDE{x}}}{\rfvar}}}
\end{sequentdeduction}
}%

The derivation of rule~\irref{dVeq} starts by using axiom~\irref{Prog} with $\rgvar \mnodefequiv \ptermA \geq 0$ and rule \irref{dVcmp} (with $\cmp$ being $\geq$) on the resulting right premise, which yields the sole premise of~\irref{dVeq} (on the right, after~\irref{dVcmp}):
{\footnotesizeoff%
\begin{sequentdeduction}[array]
  \linfer[Prog]{
  \lsequent{\ptermA \leq 0}{\dbox{\pevolvein{\D{x}=\genDE{x}}{\ptermA \neq 0}}{\ptermA< 0}}
   !
  \linfer[dVcmp]{
     \lsequent{\ptermA < 0}{\lied[]{\genDE{x}}{\ptermA}\geq \constt{\varepsilon}}
  }
    {\lsequent{\Gamma,\constt{\varepsilon} > 0}{\ddiamond{\pevolve{\D{x}=\genDE{x}}}{\ptermA \geq 0}}}
  }
  {\lsequent{\Gamma,\constt{\varepsilon} > 0, \ptermA \leq 0}{\ddiamond{\pevolve{\D{x}=\genDE{x}}}{\ptermA = 0}}}
\end{sequentdeduction}
}%

From the left premise after using~\irref{Prog}, axiom~\irref{DX} allows the domain constraint $p \neq 0$ to be assumed true initially, which strengthens the antecedent $\ptermA \leq 0$ to $\ptermA < 0$.
Rule~\irref{BC} proves the invariance of formula $\ptermA < 0$ for the ODE $\pevolvein{\D{x}=\genDE{x}}{\ptermA \neq 0}$ because the antecedents $\ptermA {\neq} 0, \ptermA {=} 0$ in its resulting premise are contradictory.
{\footnotesizeoff%
\begin{sequentdeduction}[array]
  \linfer[DX]{
  \linfer[BC]{
  \linfer[qear]{
    \lclose
  }
    {\lsequent{\ptermA \neq 0, \ptermA = 0}{\lied[]{\genDE{x}}{\ptermA} < 0}}
  }
    {\lsequent{\ptermA < 0}{\dbox{\pevolvein{\D{x}=\genDE{x}}{\ptermA \neq 0}}{\ptermA< 0}}}
  }
  {\lsequent{\ptermA \leq 0}{\dbox{\pevolvein{\D{x}=\genDE{x}}{\ptermA \neq 0}}{\ptermA< 0}}}
\\[-\normalbaselineskip]\tag*{\qedhere}
\end{sequentdeduction}
}%
\end{proof}

\begin{proof}[\textbf{Proof of \rref{cor:higherdv}\hypertarget{proof:proof14}}]
Rule~\irref{dVcmpK} can be derived in several ways.
For example, because $\lied[k]{\genDE{x}}{\ptermA}$ is strictly positive, one can prove that the solution successively reaches states where $\lied[k-1]{\genDE{x}}{\ptermA}$ is strictly positive and remains positive thereafter, followed by reaching states where $\lied[k-2]{\genDE{x}}{\ptermA}$ is strictly positive (and remains positive thereafter), and so on.
The following derivation shows how \irref{dC} can be elegantly used for this argument.
The idea is to extend the derivation of rule~\irref{dVcmp} to higher Lie derivatives by (symbolically) integrating with respect to the time variable $\timevar$ using the following sequence of inequalities, where $\lied[i]{\genDE{x}}{\ptermA}_0$ is a symbolic constant that represents the initial value of the $i$-th Lie derivative of $\ptermA$ along $\D{x}=\genDE{x}$ for $i = 0,1,\dots,k-1$:
\begin{align}
\label{eq:integration}
\lied[k]{\genDE{x}}{\ptermA} &\geq \constt{\varepsilon} \nonumber \\
\lied[k-1]{\genDE{x}}{\ptermA} &\geq \lied[k-1]{\genDE{x}}{\ptermA}_0 + \constt{\varepsilon}\timevar \nonumber \\
\lied[k-2]{\genDE{x}}{\ptermA} &\geq \lied[k-2]{\genDE{x}}{\ptermA}_0 + \lied[k-1]{\genDE{x}}{\ptermA}_0\timevar + \constt{\varepsilon}\frac{\timevar^2}{2} \nonumber \\
&~\vdots \\
\lied[1]{\genDE{x}}{\ptermA} &\geq \lied[1]{\genDE{x}}{\ptermA}_0 + \dots + \lied[k-1]{\genDE{x}}{\ptermA}_0 \frac{\timevar^{k-2}}{(k-2)!} + \constt{\varepsilon} \frac{\timevar^{k-1}}{(k-1)!}  \nonumber \\
\ptermA &\geq \underbrace{\ptermA_0 + \lied[1]{\genDE{x}}{\ptermA}_0\timevar + \dots + \lied[k-1]{\genDE{x}}{\ptermA}_0 \frac{\timevar^{k-1}}{(k-1)!} + \constt{\varepsilon} \frac{\timevar^k}{k!}}_{\ptermB(\timevar)} \nonumber
\end{align}

The RHS of the final inequality in~\rref{eq:integration} is a polynomial in the time variable $\timevar$, denoted $\ptermB(\timevar)$, which is positive for sufficiently large values of $\timevar$ because its leading coefficient $\constt{\varepsilon}$ is strictly positive.
That is, with antecedent $\constt{\varepsilon} > 0$, the formula $\lexists{\timevar_1}{\lforall{\timevar > \timevar_1} {\ptermB(\timevar) > 0}}$ is provable in real arithmetic.

The derivation of~\irref{dVcmpK} starts by introducing fresh ghost variables that remember the initial values of $\ptermA$ and the (higher) Lie derivatives $\lied[1]{\genDE{x}}{\ptermA}, \dots, \lied[k-1]{\genDE{x}}{\ptermA}$ using~\irref{cut+qear+existsl}.
The resulting antecedents are abbreviated with $\Gamma_0 \mnodefequiv \big(\Gamma, \ptermA = \ptermA_0, \dots, \lied[k-1]{\genDE{x}}{\ptermA} = \lied[k-1]{\genDE{x}}{\ptermA}_0\big)$.
It also uses \irref{dGt} to introduce a fresh time variable $\timevar$ into the system.
The arithmetic fact that $\ptermB(\timevar)$ is eventually positive for all times $\timevar > \timevar_1$ is introduced with~\irref{cut+qear+existsl}.
{\footnotesizeoff%
\begin{sequentdeduction}[array]
  \linfer[cut+qear+existsl]{
  \linfer[dGt]{
  \linfer[cut+qear+existsl]{
    \lsequent{\Gamma_0, \timevar=0, \lforall{\timevar > \timevar_1}{\ptermB(\timevar) > 0}}{\ddiamond{\pevolve{\D{x}=\genDE{x},\D{\timevar}=1}}{\ptermA \cmp 0}}
  }
    {\lsequent{\Gamma_0, \constt{\varepsilon} > 0, \timevar=0}{\ddiamond{\pevolve{\D{x}=\genDE{x},\D{\timevar}=1}}{\ptermA \cmp 0}}}
  }
    {\lsequent{\Gamma, \constt{\varepsilon} > 0, \ptermA=\ptermA_0, \dots, \lied[k-1]{\genDE{x}}{\ptermA} = \lied[k-1]{\genDE{x}}{\ptermA}_0}{\ddiamond{\pevolve{\D{x}=\genDE{x}}}{\ptermA \cmp 0}}}
  }
  {\lsequent{\Gamma, \constt{\varepsilon} > 0}{\ddiamond{\pevolve{\D{x}=\genDE{x}}}{\ptermA \cmp 0}} }
\end{sequentdeduction}
}%

Next, an initial liveness assumption, $\ddiamond{\pevolve{\D{x}=\genDE{x},\D{\timevar}=1}}{\ptermB(t) > 0}$, is cut into the assumptions.
Like the derivation of rule~\irref{dVcmp}, this initial liveness assumption says that the solution exists for sufficiently long so that the term $\ptermB(\timevar)$ from~\rref{eq:integration}, which is proved to lower bound $\ptermA$, becomes positive for sufficiently large $\timevar$.
The cut premise is abbreviated \circled{1} and further proved below.
The derivation continues from the remaining (unabbreviated) premise by refinement axiom~\irref{Prog}, with $\rgvar \mnodefequiv \ptermB(\timevar) > 0$:
{\footnotesizeoff%
\begin{sequentdeduction}[array]
  \linfer[cut]{
  \linfer[Prog]{
    \lsequent{\Gamma_0, \timevar=0}{\dbox{\pevolvein{\D{x}=\genDE{x},\D{\timevar}=1}{\lnot{(\ptermA \cmp 0)}}}{\ptermB(t) \leq 0}}
  }
    {\lsequent{\Gamma_0, \timevar=0,\ddiamond{\pevolve{\D{x}=\genDE{x},\D{\timevar}=1}}{\ptermB(t) > 0}}{\ddiamond{\pevolve{\D{x}=\genDE{x},\D{\timevar}=1}}{\ptermA \cmp 0}}} \quad
    \circled{1}
  }
  {\lsequent{\Gamma_0, \timevar=0, \lforall{\timevar > \timevar_1}{\ptermB(\timevar) > 0}}{\ddiamond{\pevolve{\D{x}=\genDE{x},\D{\timevar}=1}}{\ptermA \cmp 0}}}
\end{sequentdeduction}
}%

From the resulting open premise after~\irref{Prog}, monotonicity \irref{MbW} strengthens the postcondition to $\ptermA \geq \ptermB(t)$ using the domain constraint $\lnot{(\ptermA \cmp 0)}$ and the provable real arithmetic fact $\lnot{(\ptermA \cmp 0)} \land \ptermA \geq \ptermB(t) \limply \ptermB(\timevar) \leq 0$.
Notice that the resulting postcondition $\ptermA \geq \ptermB(\timevar)$ is the final inequality from the sequence of inequalities~\rref{eq:integration}:
{\footnotesizeoff%
\begin{sequentdeduction}[array]
  \linfer[MbW]{
    \lsequent{\Gamma_0, \timevar=0}{\dbox{\pevolvein{\D{x}=\genDE{x},\D{\timevar}=1}{\lnot{(\ptermA \cmp 0)}}}{\ptermA \geq \ptermB(\timevar)}}
  }
  {\lsequent{\Gamma_0, \timevar=0}{\dbox{\pevolvein{\D{x}=\genDE{x},\D{\timevar}=1}{\lnot{(\ptermA \cmp 0)}}}{\ptermB(\timevar) \leq 0}}}
\end{sequentdeduction}
}%

The derivation continues by using~\irref{dC} to sequentially cut in the inequality bounds outlined in~\rref{eq:integration}.
The first differential cut~\irref{dC} step adds $\lied[k-1]{\genDE{x}}{\ptermA} \geq \lied[k-1]{\genDE{x}}{\ptermA}_0 + \constt{\varepsilon}\timevar $ to the domain constraint.
The proof of this cut yields the premise of~\irref{dVcmpK} after a~\irref{dIcmp} step, see the derivation labeled \circled{$\star$} immediately below.
{\footnotesizeoff%
\begin{sequentdeduction}[array]
  \linfer[dC]{
    \lsequent{\Gamma_0, \timevar=0}{\dbox{\pevolvein{\D{x}=\genDE{x},\D{\timevar}=1}{\lnot{(\ptermA \cmp 0)} \land \lied[k-1]{\genDE{x}}{\ptermA} \geq \lied[k-1]{\genDE{x}}{\ptermA}_0 + \constt{\varepsilon}\timevar }}{\ptermA \geq \ptermB(\timevar)}}
    \quad \circled{$\star$}
  }
  {\lsequent{\Gamma_0, \timevar=0}{\dbox{\pevolvein{\D{x}=\genDE{x},\D{\timevar}=1}{\lnot{(\ptermA \cmp 0)}}}{\ptermA \geq \ptermB(\timevar)}}}
\end{sequentdeduction}
}%
From~\circled{$\star$}:
{\footnotesizeoff\renewcommand{\arraystretch}{1.4}%
\begin{sequentdeduction}[array]
  \linfer[dIcmp]{
    \lsequent{\lnot{(\ptermA \cmp 0)}}{\lied[k]{\genDE{x}}{\ptermA} \geq \constt{\varepsilon} }
  }
  {\lsequent{\Gamma_0, \timevar=0}{\dbox{\pevolvein{\D{x}=\genDE{x},\D{\timevar}=1}{\lnot{(\ptermA \cmp 0)}}}{\lied[k-1]{\genDE{x}}{\ptermA} \geq \lied[k-1]{\genDE{x}}{\ptermA}_0 + \constt{\varepsilon}\timevar}}}
\end{sequentdeduction}
}%

Subsequent~\irref{dC+dIcmp} steps progressively add the inequality bounds from~\rref{eq:integration} to the domain constraint until the last step where the postcondition is proved invariant with~\irref{dIcmp}:
{\footnotesizeoff\renewcommand{\arraystretch}{1.5}%
\begin{sequentdeduction}[array]
    \linfer[dC+dIcmp]{
    \linfer[dC+dIcmp]{
    \linfer[dC+dIcmp]{
    \linfer[dIcmp]{
      \lclose
    }
      {\lsequent{\Gamma_0, \timevar=0}{\dbox{\pevolvein{\D{x}=\genDE{x},\D{\timevar}=1}{\dots \land \lied[1]{\genDE{x}}{\ptermA} \geq \lied[1]{\genDE{x}}{\ptermA}_0 + \dots + \constt{\varepsilon} \frac{\timevar^{k-1}}{(k-1)!} }}{\ptermA \geq \ptermB(\timevar)}}}
    }
      {\vdots}
    }
    {\lsequent{\Gamma_0, \timevar=0}{\dbox{\pevolvein{\D{x}=\genDE{x},\D{\timevar}=1}{\dots \land \lied[k-2]{\genDE{x}}{\ptermA} \geq \lied[k-2]{\genDE{x}}{\ptermA}_0 + \lied[k-1]{\genDE{x}}{\ptermA}_0\timevar + \constt{\varepsilon}\frac{\timevar^2}{2}}}{\ptermA \geq \ptermB(\timevar)}}}
    }
    {\lsequent{\Gamma_0, \timevar=0}{\dbox{\pevolvein{\D{x}=\genDE{x},\D{\timevar}=1}{\lnot{(\ptermA \cmp 0)} \land \lied[k-1]{\genDE{x}}{\ptermA} \geq \lied[k-1]{\genDE{x}}{\ptermA}_0 + \constt{\varepsilon}\timevar }}{\ptermA \geq \ptermB(\timevar)}}}
\end{sequentdeduction}
}%

From premise~\circled{1}, a monotonicity step~\irref{MdW} rephrases the postcondition of the cut using the (constant) assumption $\lforall{\timevar > \timevar_1}{\ptermB(\timevar) > 0}$.
Axiom~\irref{GEx}, with instance $\tau \mnodefeq \timevar_1$, finishes the derivation because the ODE $\D{x}=\genDE{x}$ is assumed to have provable global solutions.
{\footnotesizeoff%
\begin{sequentdeduction}[array]
  \linfer[MdW]{
  \linfer[GEx]{
    \lclose
  }
    {\lsequent{\Gamma}{\ddiamond{\pevolve{\D{x}=\genDE{x},\D{\timevar}=1}}{ \timevar > \timevar_1}}}
  }
  {\lsequent{\Gamma, \lforall{\timevar > \timevar_1}{\ptermB(\timevar) > 0}}{\ddiamond{\pevolve{\D{x}=\genDE{x},\D{\timevar}=1}}{\ptermB(\timevar) > 0}}}
\\[-\normalbaselineskip]\tag*{\qedhere}
\end{sequentdeduction}
}%
\end{proof}

\begin{proof}[\textbf{Proof of \rref{cor:SP}\hypertarget{proof:proof15}}]
The derivation of rule~\irref{SP} begins by using axiom~\irref{Prog} with $\rgvar \mnodefequiv \lnot{\rsfvar}$.
The resulting left premise is the left premise of~\irref{SP}, which is the staging property of the formula $\rsfvar$ expressing that solutions of the ODE $\D{x}=\genDE{x}$ can only leave $\rsfvar$ by entering $\rfvar$:
{\footnotesizeoff%
\begin{sequentdeduction}[array]
\linfer[Prog]{
  \lsequent{\Gamma}{\dbox{\pevolvein{\D{x}=\genDE{x}}{\lnot{\rfvar}}}{\rsfvar}} !
  \lsequent{\Gamma,\constt{\varepsilon}>0}{\ddiamond{\pevolve{\D{x}=\genDE{x}}}{\lnot{\rsfvar}}}
}
{\lsequent{\Gamma,\constt{\varepsilon}>0}{\ddiamond{\pevolve{\D{x}=\genDE{x}}}{\rfvar}}}
\end{sequentdeduction}
}%

The derivation continues on the right premise, similarly to~\irref{dVcmp}, by introducing fresh variables $\ptermA_0, i$ representing the initial value of $\ptermA$ and the multiplicative inverse of $\varepsilon()$ respectively using arithmetic cuts (\irref{cut+qear}).
It then uses \irref{dGt} to introduce a fresh time variable:
{\footnotesizeoff%
\begin{sequentdeduction}[array]
  \linfer[cut+qear]{
  \linfer[existsl]{
  \linfer[dGt]{
    \lsequent{\Gamma, \constt{\varepsilon} > 0, \ptermA=\ptermA_0, i\varepsilon() = 1, \timevar=0}{\ddiamond{\pevolve{\D{x}=\genDE{x},\D{\timevar}=1}}{\lnot{\rsfvar}}}
  }
    {\lsequent{\Gamma, \constt{\varepsilon} > 0, \ptermA=\ptermA_0, i\varepsilon() = 1}{\ddiamond{\pevolve{\D{x}=\genDE{x}}}{\lnot{\rsfvar}}}}
  }
  {\lsequent{\Gamma, \constt{\varepsilon} > 0, \lexists{\ptermA_0}{(\ptermA=\ptermA_0)}, \lexists{i}{(i\varepsilon() = 1)}}{\ddiamond{\pevolve{\D{x}=\genDE{x}}}{\lnot{\rsfvar}}}}
  }
  {\lsequent{\Gamma, \constt{\varepsilon} > 0}{\ddiamond{\pevolve{\D{x}=\genDE{x}}}{\lnot{\rsfvar}}} }
\end{sequentdeduction}
}%

The next cut introduces an initial liveness assumption, where the cut premise is abbreviated \circled{1}.
The premise~\circled{1} is proved identically to the correspondingly abbreviated premise from the derivation of~\irref{dVcmp} using axiom~\irref{GEx} because the ODE $\D{x}=\genDE{x}$ is assumed have provable global solutions.
{\footnotesizeoff%
\begin{sequentdeduction}[array]
  \linfer[cut]{
    \lsequent{\Gamma, \ptermA = \ptermA_0, \timevar = 0, \ddiamond{\pevolve{\D{x}=\genDE{x},\D{\timevar}=1}}{\,\ptermA_0 + \constt{\varepsilon}\timevar > 0} }{\ddiamond{\pevolve{\D{x}=\genDE{x}, \D{\timevar}=1}}{\lnot{\rsfvar}}} \quad\quad \circled{1}
  }
  {\lsequent{\Gamma, \constt{\varepsilon} > 0, \ptermA = \ptermA_0, i > 0, i\varepsilon() = 1, \timevar = 0}{\ddiamond{\pevolve{\D{x}=\genDE{x},\D{\timevar}=1}}{\lnot{\rsfvar}}}}
\end{sequentdeduction}
}%
From the remaining open premise, axiom \irref{Prog} is used with $\rgvar \mnodefequiv \ptermA_0 + \constt{\varepsilon} \timevar > 0$:
{\footnotesizeoff%
\begin{sequentdeduction}[array]
  \linfer[Prog]{
    \lsequent{\Gamma,\ptermA = \ptermA_0, \timevar = 0}{\dbox{\pevolvein{\D{x}=\genDE{x},\D{\timevar}=1}{\rsfvar}}{\,\ptermA_0 + \constt{\varepsilon}\timevar \leq 0}}
  }
  {\lsequent{\Gamma, \ptermA = \ptermA_0, \timevar = 0, \ddiamond{\pevolve{\D{x}=\genDE{x},\D{\timevar}=1}}{\,\ptermA_0 + \constt{\varepsilon}\timevar > 0} }{\ddiamond{\pevolve{\D{x}=\genDE{x},\D{\timevar}=1}}{\lnot{\rsfvar}}}}
\end{sequentdeduction}
}%

A monotonicity step~\irref{MbW} simplifies the postcondition using domain constraint $\rsfvar$, yielding the left conjunct of the right premise of rule~\irref{SP}.
The right premise after monotonicity is abbreviated~\circled{2} and continued below.
{\footnotesizeoff%
\begin{sequentdeduction}[array]
  \linfer[MbW]{
  \linfer[qear]{
    \lsequent{\rsfvar}{\ptermA \leq 0}
  }
    {\lsequent{\rsfvar,\ptermA \geq \ptermA_0 + \constt{\varepsilon}\timevar}{\ptermA_0 + \constt{\varepsilon}\timevar \leq 0}} !
    \circled{2}
    }
    {\lsequent{\Gamma,\ptermA = \ptermA_0,\timevar = 0}{\dbox{\pevolvein{\D{x}=\genDE{x},\D{\timevar}=1}{\rsfvar}}{\,\ptermA_0 + \constt{\varepsilon}\timevar \leq 0}}}
\end{sequentdeduction}
}%

From~\circled{2}, rule~\irref{dIcmp} yields the right conjunct of the right premise of rule~\irref{SP}.

{\footnotesizeoff%
\begin{sequentdeduction}[array]
\linfer[dIcmp]{
    \lsequent{\rsfvar}{\lied[]{\genDE{x}}{\ptermA}\geq \constt{\varepsilon}}
    }
    {\lsequent{\Gamma,\ptermA = \ptermA_0,\timevar = 0}{\dbox{\pevolvein{\D{x}=\genDE{x},\D{\timevar}=1}{\rsfvar}}{\,\ptermA \geq \ptermA_0 + \constt{\varepsilon}\timevar}}}
\\[-\normalbaselineskip]\tag*{\qedhere}
\end{sequentdeduction}
}%
\end{proof}

\begin{proof}[\textbf{Proof of \rref{cor:boundedandcompact}\hypertarget{proof:proof16}}]
Rule~\irref{SPb} is derived first since rule~\irref{SPc} follows from~\irref{SPb} as a corollary.
Both proof rules make use of the fact that continuous functions on compact domains attain their extrema~\cite[Theorem 4.16]{MR0385023}.
Polynomial functions are continuous, so the fact that a polynomial has bounded values on a compact (semialgebraic) domain can be stated and proved as a formula of first-order real arithmetic by~\irref{qear}~\cite{Bochnak1998}.
The derivation of~\irref{SPb} is essentially similar to~\irref{SP} except replacing the use of the global existence axiom~\irref{GEx} with the bounded existence axiom~\irref{BEx}.
It starts by using axiom~\irref{Prog} with $\rgvar \mnodefequiv \lnot{\rsfvar}$, yielding the left premise of~\irref{SPb}:
{\footnotesizeoff%
\begin{sequentdeduction}[array]
\linfer[Prog]{
  \lsequent{\Gamma}{\dbox{\pevolvein{\D{x}=\genDE{x}}{\lnot{\rfvar}}}{\rsfvar}} !
  \lsequent{\Gamma,\constt{\varepsilon}>0}{\ddiamond{\pevolve{\D{x}=\genDE{x}}}{\lnot{\rsfvar}}}
}
{\lsequent{\Gamma,\constt{\varepsilon}>0}{\ddiamond{\pevolve{\D{x}=\genDE{x}}}{\rfvar}}}
\end{sequentdeduction}
}%

Continuing on the resulting right from~\irref{Prog} (similarly to~\irref{SP}), the derivation introduces fresh variables $\ptermA_0, i$ representing the initial value of $\ptermA$ and the multiplicative inverse of $\varepsilon()$ respectively using arithmetic cuts and Skolemizing (\irref{cut+qear+existsl}). Rule \irref{dGt} is also used to introduce a fresh time variable $\timevar$ with $\timevar = 0$ initially.
{\footnotesizeoff%
\begin{sequentdeduction}[array]
  \linfer[cut+qear+existsl+dGt]{
    \lsequent{\Gamma, \constt{\varepsilon} > 0, \ptermA=\ptermA_0, i\varepsilon() = 1, \timevar=0}{\ddiamond{\pevolve{\D{x}=\genDE{x},\D{\timevar}=1}}{\lnot{\rsfvar}}}
  }
  {\lsequent{\Gamma, \constt{\varepsilon} > 0}{\ddiamond{\pevolve{\D{x}=\genDE{x}}}{\lnot{\rsfvar}}} }
\end{sequentdeduction}
}%

The set characterized by formula $\rsfvar$ is bounded so its closure is compact (with respect to variables $x$).
On this compact closure, the continuous polynomial function $\ptermA$ attains its maximum value, which implies that the value of $\ptermA$ is bounded above in $\rsfvar$ and cannot increase unboundedly while staying in $\rsfvar$.
That is, the formula $\lexists{\ptermA_1}{\rrfvar(\ptermA_1)}$ where $\rrfvar(\ptermA_1) \mnodefequiv \lforall{x}{(\rsfvar(x) \limply \ptermA \leq \ptermA_1)}$ is valid in first-order real arithmetic and thus provable by \irref{qear}.
This formula is added to the assumptions with a~\irref{cut}, and the existential quantifier is Skolemized with \irref{existsl}.
The resulting symbolic constant $\ptermA_1$ represents the upper bound of $\ptermA$ on $\rsfvar$.
Note that $\rrfvar(\ptermA_1)$ is constant for the ODE $\D{x}=\genDE{x},\D{\timevar}=1$ because it does not mention any of the variables $x$ (nor $\timevar$) free:
{\footnotesizeoff%
\begin{sequentdeduction}[array]
  \linfer[cut+qear]{
  \linfer[existsl]{
    \lsequent{\Gamma, \constt{\varepsilon} > 0, \ptermA=\ptermA_0, i\varepsilon() = 1, \timevar=0, \rrfvar(\ptermA_1)}{\ddiamond{\pevolve{\D{x}=\genDE{x},\D{\timevar}=1}}{\lnot{\rsfvar}}}
  }
    {\lsequent{\Gamma, \constt{\varepsilon} > 0, \ptermA=\ptermA_0, i\varepsilon() = 1, \timevar=0, \lexists{\ptermA_1}{\rrfvar(\ptermA_1)} }{\ddiamond{\pevolve{\D{x}=\genDE{x},\D{\timevar}=1}}{\lnot{\rsfvar}}}}
  }
  {\lsequent{\Gamma, \constt{\varepsilon} > 0, \ptermA=\ptermA_0, i\varepsilon() = 1, \timevar=0}{\ddiamond{\pevolve{\D{x}=\genDE{x},\D{\timevar}=1}}{\lnot{\rsfvar}}}}
\end{sequentdeduction}
}%

Next, a~\irref{cut} introduces an initial liveness assumption saying that \emph{either} the solution exists for sufficient time for the bound $\ptermA_0 + \constt{\varepsilon}\timevar > \ptermA_1$ to be satisfied (at sufficiently large $\timevar$) \emph{or} the solution leaves $\rsfvar$.
This assumption is abbreviated $\rtfvar \mnodefequiv \ddiamond{\pevolve{\D{x}=\genDE{x},\D{\timevar}=1}}{(\ptermA_0 + \constt{\varepsilon}\timevar > \ptermA_1 \lor \lnot{\rsfvar})}$.
The main difference from~\irref{SP} is that the postcondition of assumption $\rtfvar$ adds a disjunction for the possibility of leaving $\rsfvar$ (which characterizes a bounded set).
This cut premise is abbreviated \circled{1} and proved further below.
{\footnotesizeoff%
\begin{sequentdeduction}[array]
  \linfer[cut]{
    \lsequent{\Gamma, \ptermA=\ptermA_0, \timevar=0, \rrfvar(\ptermA_1), \rtfvar}{\ddiamond{\pevolve{\D{x}=\genDE{x},\D{\timevar}=1}}{\lnot{\rsfvar}}} \quad
    \circled{1}
  }
  {\lsequent{\Gamma, \constt{\varepsilon} > 0, \ptermA=\ptermA_0, i\varepsilon() = 1, \timevar=0, \rrfvar(\ptermA_1)}{\ddiamond{\pevolve{\D{x}=\genDE{x},\D{\timevar}=1}}{\lnot{\rsfvar}}}}
\end{sequentdeduction}
}%
Continuing from the open premise on the left, axiom \irref{Prog} is used with $\rgvar \mnodefequiv  \ptermA_0 + \constt{\varepsilon}\timevar > \ptermA_1 \lor \lnot{\rsfvar}$:
{\footnotesizeoff%
\begin{sequentdeduction}[array]
  \linfer[Prog]{
    \lsequent{\Gamma,\ptermA = \ptermA_0, \timevar=0, \rrfvar(\ptermA_1)}{\dbox{\pevolvein{\D{x}=\genDE{x},\D{\timevar}=1}{\rsfvar}}{(\ptermA_0 + \constt{\varepsilon}\timevar \leq \ptermA_1 \land \rsfvar)}}
  }
  {\lsequent{\Gamma, \ptermA=\ptermA_0, \timevar=0, \rrfvar(\ptermA_1), \rtfvar}{\ddiamond{\pevolve{\D{x}=\genDE{x},\D{\timevar}=1}}{\lnot{\rsfvar}}}}
\end{sequentdeduction}
}%

The postcondition of the resulting box modality is simplified to $\ptermA \geq \ptermA_0 + \constt{\varepsilon}\timevar$ with a \irref{MbW} monotonicity step.
This crucially uses the assumption $\rrfvar(\ptermA_1)$ which is constant for the ODE.
A \irref{dIcmp} step yields the remaining premise of~\irref{SPb} on the right, see the derivation labeled \circled{$\star$} immediately below:
{\footnotesizeoff%
\begin{sequentdeduction}[array]
\linfer[MbW]{
    \linfer[qear]{
      \linfer[qear]{
        \lclose
      }
      {\lsequent{\rsfvar,\rrfvar(\ptermA_1)}{\ptermA \leq \ptermA_1}}
    }
      {\lsequent{\rsfvar,\rrfvar(\ptermA_1),\ptermA \geq \ptermA_0 + \constt{\varepsilon}\timevar}{\ptermA_0 + \constt{\varepsilon}\timevar \leq \ptermA_1 \land \rsfvar}}
    !
    \circled{$\star$}
  }
  {\lsequent{\Gamma,\ptermA = \ptermA_0, \timevar=0, \rrfvar(\ptermA_1)}{\dbox{\pevolvein{\D{x}=\genDE{x},\D{\timevar}=1}{\rsfvar}}{(\ptermA_0 + \constt{\varepsilon}\timevar \leq \ptermA_1 \land \rsfvar)}}
}
\end{sequentdeduction}
}%
From \circled{$\star$}:
{\footnotesizeoff%
\begin{sequentdeduction}[array]
\linfer[dIcmp]{
    \lsequent{\rsfvar}{\lied[]{\genDE{x}}{\ptermA}\geq \constt{\varepsilon}}
    }
    {\lsequent{\Gamma,\ptermA = \ptermA_0,\timevar=0}{\dbox{\pevolvein{\D{x}=\genDE{x},\D{\timevar}=1}{\rsfvar}}{\ptermA \geq \ptermA_0 + \constt{\varepsilon}\timevar}}}
\end{sequentdeduction}
}%

From premise~\circled{1}, a monotonicity step~\irref{MdW} equivalently rephrases the postcondition of the cut.
Axiom~\irref{BEx} finishes the proof because formula $\rsfvar(x)$ is assumed to be bounded over variables $x$.
{\footnotesizeoff%
\begin{sequentdeduction}[array]
  \linfer[qear+MdW]{
  \linfer[BEx]{
    \lclose
  }
    {\lsequent{}{\ddiamond{\pevolve{\D{x}=\genDE{x},\D{\timevar}=1}}{ (\timevar > i(\ptermA_1-\ptermA_0) \lor \lnot{\rsfvar})}}}
  }
  {\lsequent{\constt{\varepsilon} > 0, i\varepsilon() = 1}{\rtfvar}}
\end{sequentdeduction}
}%

To derive rule~\irref{SPc} from~\irref{SPb}, the compactness of the set characterized by $\rsfvar(x)$ implies that the formula $\lexists{\varepsilon {>} 0}{A(\varepsilon)}$ where $A(\varepsilon) \mnodefequiv \lforall{x}{(\rsfvar(x) {\limply} \lied[]{\genDE{x}}{p} \geq \varepsilon)}$ and the formula $B \mnodefequiv \lforall{x}{(\rsfvar(x) {\limply} \lied[]{\genDE{x}}{p} > 0)}$ are provably equivalent in first-order real arithmetic.
This provable real arithmetic equivalence follows from the fact that the continuous polynomial function $\lied[]{\genDE{x}}{p}$ is bounded below by its minima on the compact set characterized by $\rsfvar(x)$ and this minima is strictly positive.
The following derivation of~\irref{SPc} threads these two formulas through the use of rule~\irref{SPb}.
After Skolemizing $\lexists{\varepsilon {>} 0}{A(\varepsilon)}$ with \irref{existsl}, the resulting formula $A(\varepsilon)$ is constant for the ODE $\D{x}=\genDE{x}$ so it is kept as a constant assumption across the use of~\irref{SPb}, leaving only the two premises of rule~\irref{SPc}:
{\footnotesizeoff%
\renewcommand{\linferPremissSeparation}{\hspace{5pt}}%
\begin{sequentdeduction}[array]
  \linfer[cut]{
  \linfer[existsl]{
  \linfer[SPb]{
    \lsequent{\Gamma}{\dbox{\pevolvein{\D{x}=\genDE{x}}{\lnot{\rfvar}}}{\rsfvar}} !
    \linfer[qear]{\lclose}
    {\lsequent{\rsfvar,A(\varepsilon)}{ \lied[]{\genDE{x}}{p} \geq \varepsilon}}
  }
    {\lsequent{\Gamma,\varepsilon > 0, A(\varepsilon)}{\ddiamond{\pevolve{\D{x}=\genDE{x}}}{\rfvar}}}
  }
  {\lsequent{\Gamma,\lexists{\varepsilon {>} 0}{A(\varepsilon)}}{\ddiamond{\pevolve{\D{x}=\genDE{x}}}{\rfvar}}}
  !
    \linfer[qear]{
    \linfer[allr+implyr]{
      \lsequent{\rsfvar}{\lied[]{\genDE{x}}{p} > 0}
    }
      \lsequent{}{B}
    }
    {\lsequent{}{\lexists{\varepsilon {>} 0}{A(\varepsilon)}}}
  }
  {\lsequent{\Gamma}{\ddiamond{\pevolve{\D{x}=\genDE{x}}}{\rfvar}} }
\\[-\normalbaselineskip]\tag*{\qedhere}
\end{sequentdeduction}
}%
\end{proof}

\begin{proof}[\textbf{Proof of \rref{cor:rs}\hypertarget{proof:proof17}}]
Rule \irref{RS} is derived from rule~\irref{SPc} with $\rsfvar \mnodefequiv \lnot{\rfvar} \land K$, since the intersection of a closed set (characterized by $\lnot{\rfvar}$) with a compact set (characterized by $K$) is compact.
The resulting right premise from using~\irref{SPc} is the right premise of~\irref{RS}:
{\footnotesizeoff%
\begin{sequentdeduction}[array]
\linfer[SPc]{
  \lsequent{\Gamma, \ptermA \cmp 0}{\dbox{\pevolvein{\D{x}=\genDE{x}}{\lnot{\rfvar}}}{(\lnot{\rfvar} \land K)}} !
  \lsequent{\lnot{\rfvar},K}{\lied[]{\genDE{x}}{p} > 0}
}
{\lsequent{\Gamma, \ptermA \cmp 0}{\ddiamond{\pevolve{\D{x}=\genDE{x}}}{\rfvar}}}
\end{sequentdeduction}
}%

Continuing from the left premise, a monotonicity step with the premise $\lsequent{\ptermA \geq 0}{K}$ turns the postcondition to $\ptermA \cmp 0$.
Rule~\irref{BC} is used, which, along with the premise $\lsequent{\ptermA \geq 0}{K}$ results in the premises of rule~\irref{RS}:
{\footnotesizeoff%
\begin{sequentdeduction}[array]
\linfer[MbW]{
    \linfer[qear]{
    \lsequent{\ptermA \geq 0}{K}}
    {\lsequent{\lnot{\rfvar}, \ptermA \cmp 0}{\lnot{\rfvar} \land K}}
    !
    \linfer[BC]{
    \linfer[cut]{
      \lsequent{\lnot{\rfvar},K}{\lied[]{\genDE{x}}{\ptermA} > 0} !
      \linfer[qear]{
      \lsequent{\ptermA\geq 0}{K}
      }
      {\lsequent{\lnot{\rfvar},\ptermA=0}{K}}
    }
      {\lsequent{\lnot{\rfvar}, \ptermA = 0}{\lied[]{\genDE{x}}{\ptermA} > 0}}
    }
    {\lsequent{\ptermA \cmp 0}{\dbox{\pevolvein{\D{x}=\genDE{x}}{\lnot{\rfvar}}}{\ptermA \cmp 0}}}
  }
  {\lsequent{\Gamma, \ptermA \cmp 0}{\dbox{\pevolvein{\D{x}=\genDE{x}}{\lnot{\rfvar}}}{(\lnot{\rfvar} \land K)}}}
\\[-\normalbaselineskip]\tag*{\qedhere}
\end{sequentdeduction}
}%
\end{proof}

\subsection{Proofs for Liveness With Domain Constraints}
\label{app:livewithdomproofs}

\begin{proof}[\textbf{Proof of \rref{cor:atomicdvcmpQ}\hypertarget{proof:proof18}}]
The derivation uses axiom~\irref{CORef} choosing $\rrfvar \mnodefequiv \ltrue$ and noting that $\ptermA \geq 0$ (resp. $\ptermA > 0$) characterizes a topologically closed (resp. open) set so the appropriate topological requirements of~\irref{CORef} are satisfied. The resulting left premise is the left premise of~\irref{dVcmpQ}:
{\footnotesizeoff%
\begin{sequentdeduction}[array]
\linfer[CORef]{
  \lsequent{\Gamma}{\dbox{\pevolvein{\D{x}=\genDE{x}}{\lnot{(\ptermA \cmp 0)}}}{\ivr}} !
  \lsequent{\Gamma,\constt{\varepsilon} > 0}{\ddiamond{\pevolve{\D{x}=\genDE{x}}}{p\cmp 0}}
}
{\lsequent{\Gamma,\constt{\varepsilon} > 0,\lnot{(\ptermA \cmp 0)}}{\ddiamond{\pevolvein{\D{x}=\genDE{x}}{\ivr}}{\ptermA \cmp 0}}}
\end{sequentdeduction}
}%

The proof continues from the resulting right premise (after~\irref{CORef}) identically to the derivation of~\irref{dVcmp} until the step where \irref{dVcmpA} is used.
The steps are repeated briefly here.

{\footnotesizeoff%
\begin{sequentdeduction}[array]
  \linfer[cut+qear+existsl]{
  \linfer[dGt]{
  \linfer[cut+GEx]{
    \lsequent{\Gamma, \ptermA = \ptermA_0, \timevar = 0, \ddiamond{\pevolve{\D{x}=\genDE{x},\D{\timevar}=1}}{\,\ptermA_0 + \constt{\varepsilon}\timevar > 0} }{\ddiamond{\pevolve{\D{x}=\genDE{x},\D{\timevar}=1}}{\ptermA \cmp 0}}
  }
  {\lsequent{\Gamma, \constt{\varepsilon} > 0, \ptermA = \ptermA_0, i\constt{\varepsilon}=1, \timevar = 0}{\ddiamond{\pevolve{\D{x}=\genDE{x},\D{\timevar}=1}}{\ptermA \cmp 0}} }
  }
  {\lsequent{\Gamma, \constt{\varepsilon} > 0, \ptermA = \ptermA_0, i\constt{\varepsilon}=1}{\ddiamond{\pevolve{\D{x}=\genDE{x}}}{\ptermA \cmp 0}} }
  }
  {\lsequent{\Gamma, \constt{\varepsilon} > 0}{\ddiamond{\pevolve{\D{x}=\genDE{x}}}{\ptermA \cmp 0}} }
\end{sequentdeduction}
}%

Like the derivation of~\irref{dVcmpA}, axiom \irref{Prog} is used with $\rgvar \mnodefequiv \constt{\ptermA_0} + \constt{\varepsilon} \timevar > 0$.
The key difference is an additional~\irref{dC} step, which adds $\ivr$ to the domain constraint.\footnote{Notably, the differential cuts proof support from~\rref{subsec:support} can add such a cut automatically.}
The proof of this differential cut uses the left premise of~\irref{dVcmpQ}, it is labeled \circled{1} and shown below.
{\footnotesizeoff%
\begin{sequentdeduction}[array]
  \linfer[Prog]{
  \linfer[dC]{
    \lsequent{\qquad\Gamma,\ptermA = \constt{\ptermA_0},\timevar=0}{\dbox{\pevolvein{\D{x}=\genDE{x},\D{\timevar}=1}{\lnot{(\ptermA \cmp 0) \land \ivr}}}{\,\constt{\ptermA_0} + \constt{\varepsilon} \timevar \leq 0}}
    \quad \circled{1}
  }
    {\lsequent{\Gamma,\ptermA = \constt{\ptermA_0},\timevar=0}{\dbox{\pevolvein{\D{x}=\genDE{x},\D{\timevar}=1}{\lnot{(\ptermA \cmp 0)}}}{\,\constt{\ptermA_0} + \constt{\varepsilon} \timevar \leq 0}} \qquad\;\;\,\,}
  }
  {\lsequent{\Gamma, \ptermA=\ptermA_0, \timevar=0, \ddiamond{\pevolve{\D{x}=\genDE{x},\D{\timevar}=1}}{\ptermA_0 + \constt{\varepsilon}\timevar > 0} }{\ddiamond{\pevolve{\D{x}=\genDE{x},\D{\timevar}=1}}{\ptermA \cmp 0}}}
\end{sequentdeduction}
}%

The derivation from the resulting left premise (after the cut) continues similarly to~\irref{dVcmpA} using a monotonicity step~\irref{MbW} to rephrase the postcondition, followed by~\irref{dIcmp} which results in the right premise of~\irref{dVcmpQ}:
{\footnotesizeoff%
\begin{sequentdeduction}[array]
  \linfer[MbW]{
  \linfer[dIcmp]{
    \lsequent{\lnot{(\ptermA \cmp 0)}, \ivr}{\lied[]{\genDE{x}}{\ptermA}\geq \constt{\varepsilon}}
  }
    {\lsequent{\Gamma, \ptermA = \constt{\ptermA_0}, \timevar=0}{\dbox{\pevolvein{\D{x}=\genDE{x},\D{\timevar}=1}{\lnot{(\ptermA \cmp 0)}\land \ivr}}{\,\ptermA \geq \constt{\ptermA_0} + \constt{\varepsilon} \timevar}}}
  }
  {\lsequent{\Gamma,\ptermA = \constt{\ptermA_0}, \timevar=0}{\dbox{\pevolvein{\D{x}=\genDE{x},\D{\timevar}=1}{\lnot{(\ptermA \cmp 0)}\land \ivr}}{\,\constt{\ptermA_0} + \constt{\varepsilon} \timevar \leq 0}}}
\end{sequentdeduction}
}%

The derivation from~\circled{1} removes the time variable $t$ using the inverse direction of rule~\irref{dGt}~\cite{DBLP:journals/jar/Platzer17,Platzer18,DBLP:journals/jacm/PlatzerT20}.
Just as rule~\irref{dGt} allows introducing a \emph{fresh} time variable $t$ for the sake of proof, its inverse direction simply removes the variable $t$ since it is irrelevant for the proof of the differential cut.
{\footnotesizeoff%
\begin{sequentdeduction}[array]
  \linfer[dGt]{
  \lsequent{\Gamma}{\dbox{\pevolvein{\D{x}=\genDE{x}}{\lnot{(\ptermA \cmp 0)}}}{\ivr}}
  }
  {\lsequent{\Gamma,\ptermA = \constt{\ptermA_0},\timevar=0}{\dbox{\pevolvein{\D{x}=\genDE{x},\D{\timevar}=1}{\lnot{(\ptermA \cmp 0)}}}{\ivr}}}
\\[-\normalbaselineskip]\tag*{\qedhere}
\end{sequentdeduction}
}%
\end{proof}

\begin{proof}[\textbf{Proof of \rref{cor:ttq}\hypertarget{proof:proof19}}]
The derivations of rules~\irref{dVeqQ+TTQ} are similar to the derivations of rules~\irref{dVeq+TT} respectively.
Rule~\irref{TTQ} is derived from~\irref{dVeqQ} by monotonicity:
{\footnotesizeoff%
\renewcommand{\linferPremissSeparation}{\hspace{4pt}}%
\begin{sequentdeduction}[array]
  \linfer[MdW]{
  \lsequent{\ivr,\ptermA = 0}{\rfvar}
  !
  \linfer[dVeqQ]{
     \lsequent{\Gamma}{\dbox{\pevolvein{\D{x}=\genDE{x}}{\ptermA < 0}}{\ivr}} !
     \lsequent{\ptermA < 0,\ivr}{\lied[]{\genDE{x}}{\ptermA}\geq \constt{\varepsilon}} !
  }
    {\lsequent{\Gamma,\constt{\varepsilon} > 0, \ptermA \leq 0, \ivr}{\ddiamond{\pevolvein{\D{x}=\genDE{x}}{\ivr}}{\ptermA = 0}}}
  }
  {\lsequent{\Gamma,\constt{\varepsilon} > 0, \ptermA \leq 0, \ivr}{\ddiamond{\pevolvein{\D{x}=\genDE{x}}{\ivr}}{\rfvar}}}
\end{sequentdeduction}
}%

The derivation of rule~\irref{dVeqQ} starts by using axiom~\irref{Prog} with $\rgvar \mnodefequiv \ptermA \geq 0$.
The resulting box modality (right) premise is abbreviated~\circled{1} and proved below.
On the resulting left premise, a~\irref{DX} step adds the negated postcondition $\ptermA < 0$ as an assumption to the antecedents since the domain constraint $\ivr$ is true initially.
Following that, rule~\irref{dVcmpQ} is used (with $\cmp$ being $\geq$, since $\ivr$ characterizes a closed set). This yields the two premises of~\irref{dVeqQ}:
{\footnotesizeoff%
\begin{sequentdeduction}[array]
  \linfer[Prog]{
  \linfer[DX]{
  \linfer[dVcmpQ]{
     \lsequent{\Gamma}{\dbox{\pevolvein{\D{x}=\genDE{x}}{\ptermA < 0}}{\ivr}} !
     \lsequent{\ptermA < 0,\ivr}{\lied[]{\genDE{x}}{\ptermA}\geq \constt{\varepsilon}}
  }
    {\lsequent{\Gamma,\constt{\varepsilon} > 0, \ptermA < 0}{\ddiamond{\pevolvein{\D{x}=\genDE{x}}{\ivr}}{\ptermA \geq 0}}}}
  {\lsequent{\Gamma,\constt{\varepsilon} > 0, \ivr}{\ddiamond{\pevolvein{\D{x}=\genDE{x}}{\ivr}}{\ptermA \geq 0}}  \qquad \circled{1}}
  }
  {\lsequent{\Gamma,\constt{\varepsilon} > 0, \ptermA \leq 0, \ivr}{\ddiamond{\pevolvein{\D{x}=\genDE{x}}{\ivr}}{\ptermA = 0}}}
\end{sequentdeduction}
}%
From premise \circled{1}, the derivation is completed similarly to~\irref{dVeq} using~\irref{DX} and \irref{BC}:
{\footnotesizeoff%
\begin{sequentdeduction}[array]
  \linfer[DX]{
  \linfer[BC]{
  \linfer[qear]{
    \lclose
  }
    {\lsequent{\ptermA \neq 0, \ptermA = 0}{\lied[]{\genDE{x}}{\ptermA} < 0}}
  }
    {\lsequent{\ptermA < 0}{\dbox{\pevolvein{\D{x}=\genDE{x}}{\ivr \land \ptermA \neq 0}}{\ptermA< 0}}}
  }
  {\lsequent{\ptermA \leq 0}{\dbox{\pevolvein{\D{x}=\genDE{x}}{\ivr \land \ptermA \neq 0}}{\ptermA< 0}}}
\\[-\normalbaselineskip]\tag*{\qedhere}
\end{sequentdeduction}
}%
\end{proof}

\begin{proof}[\textbf{Proof of \rref{cor:rsq}\hypertarget{proof:proof20}}]
Rule \irref{RSQM} is derived from \irref{RSQ} by a~\irref{dDR} monotonicity step followed by~\irref{dW} on its resulting left premise and~\irref{RSQ} on its resulting right premise:
{\footnotesizeoff%
\begin{sequentdeduction}[array]
  \linfer[dDR]{
  \linfer[dW]{
    \lsequent{\ptermA > 0}{\ivr}
  }
  {\lsequent{\Gamma, \ptermA > 0}{\dbox{\pevolvein{\D{x}=\genDE{x}}{p > 0}}{\ivr}} }
  !
  \linfer[RSQ]{
     \lsequent{\ptermA \geq 0}{K} !
     \lsequent{\lnot{\rfvar},K}{\lied[]{\genDE{x}}{\ptermA} > 0}
  }
  {\lsequent{\Gamma, \ptermA > 0}{\ddiamond{\pevolvein{\D{x}=\genDE{x}}{\ptermA > 0}}{\rfvar}} }
  }
  {\lsequent{\Gamma, \ptermA > 0}{\ddiamond{\pevolvein{\D{x}=\genDE{x}}{\ivr}}{\rfvar}} }
\end{sequentdeduction}
}%

The derivation of rule~\irref{RSQ} starts by adding assumption $\lnot{\rfvar}$ to the antecedents, because if both $\ptermA > 0$ (which is already in the antecedents) and $\rfvar$ were true initially, then the liveness succedent is trivially true by~\irref{DX}.
Next, axiom~\irref{CORef} is used with  $\rrfvar \mnodefequiv \ltrue$, its topological restrictions are met since both formulas $\rfvar$ and $\ptermA > 0$ characterize open sets.
From the resulting right premise, rule~\irref{RS} yields the corresponding two premises of~\irref{RSQ} because formula $K$ (resp. $\rfvar$) characterizes a compact set (resp. open set):
{\footnotesizeoff%
\begin{sequentdeduction}[array]
  \linfer[DX]{
  \linfer[CORef]{
    \lsequent{\Gamma, \ptermA > 0}{\dbox{\pevolvein{\D{x}=\genDE{x}}{\lnot{\rfvar}}}{\ptermA > 0}}
    !
    \linfer[RS]{
     \lsequent{\ptermA \geq 0}{K} !
     \lsequent{\lnot{\rfvar},K}{\lied[]{\genDE{x}}{\ptermA} > 0}
    }
    {\lsequent{\Gamma, \ptermA > 0}{\ddiamond{\pevolve{\D{x}=\genDE{x}}}{\rfvar}}}
  }
    {\lsequent{\Gamma, \ptermA > 0, \lnot{\rfvar}}{\ddiamond{\pevolvein{\D{x}=\genDE{x}}{\ptermA > 0}}{\rfvar}}}
  }
  {\lsequent{\Gamma, \ptermA > 0}{\ddiamond{\pevolvein{\D{x}=\genDE{x}}{\ptermA > 0}}{\rfvar}} }
\end{sequentdeduction}
}%

From the leftmost open premise after~\irref{CORef}, rule~\irref{BC} is used and the resulting $\ptermA = 0$ assumption is turned into $K$ using the left premise of~\irref{RSQ}.
The resulting open premises are the premises of~\irref{RSQ}:
{\footnotesizeoff%
\begin{sequentdeduction}[array]
  \linfer[BC]{
  \linfer[cut]{
    \lsequent{\lnot{\rfvar}, K}{ \lied[]{\genDE{x}}{\ptermA} > 0}
    !
    \linfer[qear]{
      \lsequent{\ptermA \geq 0}{K}
    }
    {\lsequent{\ptermA = 0}{K}}
  }
    {\lsequent{\lnot{\rfvar}, \ptermA = 0}{ \lied[]{\genDE{x}}{\ptermA} > 0}}
  }
  {\lsequent{\Gamma, \ptermA > 0}{\dbox{\pevolvein{\D{x}=\genDE{x}}{\lnot{\rfvar}}}{\ptermA > 0}}}
\\[-\normalbaselineskip]\tag*{\qedhere}
\end{sequentdeduction}
}%
\end{proof}

\begin{proof}[\textbf{Proof of \rref{cor:SPQ}\hypertarget{proof:proof21}}]
The derivation starts by using axiom \irref{SARef} which results in two premises.
From the left premise after axiom~\irref{SARef}, a monotonicity step turns the postcondition into $\rsfvar$, yielding the left premise and first conjunct of the right premise of~\irref{SPQ}.
{\footnotesizeoff%
\begin{sequentdeduction}[array]
\linfer[SARef]{
  \linfer[MbW]{
  \lsequent{\rsfvar}{\ivr} !
  \lsequent{\Gamma}{\dbox{\pevolvein{\D{x}=\genDE{x}}{\lnot{(\rfvar \land \ivr)}}}{\rsfvar}}
  }
  {\lsequent{\Gamma}{\dbox{\pevolvein{\D{x}=\genDE{x}}{\lnot{(\rfvar \land \ivr)}}}{\ivr}}} !
  \lsequent{\Gamma}{\ddiamond{\pevolve{\D{x}=\genDE{x}}}{\rfvar}}
}
{\lsequent{\Gamma}{\ddiamond{\pevolvein{\D{x}=\genDE{x}}{\ivr}}{\rfvar}}}
\end{sequentdeduction}
}%

From the right premise after axiom~\irref{SARef}, rule~\irref{SP} yields the remaining two premises of~\irref{SPQ}:
{\footnotesizeoff%
\begin{sequentdeduction}[array]
\linfer[SP]{
  \linfer[dW+DMP]{
    \lsequent{\Gamma}{\dbox{\pevolvein{\D{x}=\genDE{x}}{\lnot{(\rfvar \land \ivr)}}}{\rsfvar}}
  }
  {\lsequent{\Gamma}{\dbox{\pevolvein{\D{x}=\genDE{x}}{\lnot{\rfvar}}}{\rsfvar}}} !
   \lsequent{\rsfvar}{\ptermA \leq 0 \land \lied[]{\genDE{x}}{p}\geq \constt{\varepsilon}}
}
  {\lsequent{\Gamma}{\ddiamond{\pevolve{\D{x}=\genDE{x}}}{\rfvar}}}
\end{sequentdeduction}
}%

The~\irref{dW+DMP} step uses the propositional tautology $\lnot{\rfvar} \limply \lnot{(\rfvar \land \ivr)}$ to weaken the domain constraint so that it matches the left premise of rule~\irref{SPQ}.
\end{proof}

\begin{proof}[\textbf{Proof of \rref{cor:combination}\hypertarget{proof:proof22}}]
The chimeric proof rule~\irref{SPcQ} amalgamates ideas behind the rules~\irref{SPQ+dVcmpK+SPc}.
It is therefore unsurprising that the derivation of~\irref{SPcQ} uses various steps from the derivations of those rules.
The derivation of~\irref{SPcQ} starts similarly to~\irref{SPQ} (following~\rref{cor:SPQ}) using axiom \irref{SARef}:
{\footnotesizeoff%
\begin{sequentdeduction}[array]
\linfer[SARef]{
  \lsequent{\Gamma}{\dbox{\pevolvein{\D{x}=\genDE{x}}{\lnot{(\rfvar \land \ivr)}}}{\ivr}} !
  \lsequent{\Gamma}{\ddiamond{\pevolve{\D{x}=\genDE{x}}}{\rfvar}}
}
{\lsequent{\Gamma}{\ddiamond{\pevolvein{\D{x}=\genDE{x}}{\ivr}}{\rfvar}}}
\end{sequentdeduction}
}%

From the left premise after~\irref{SARef}, a monotonicity step turns the postcondition into $\rsfvar$, yielding the left premise and first conjunct of the right premise of~\irref{SPcQ}.
{\footnotesizeoff%
\begin{sequentdeduction}[array]
  \linfer[MbW]{
  \lsequent{\Gamma}{\dbox{\pevolvein{\D{x}=\genDE{x}}{\lnot{(\rfvar \land \ivr)}}}{\rsfvar}} !
  \lsequent{\rsfvar}{\ivr}
  }
  {\lsequent{\Gamma}{\dbox{\pevolvein{\D{x}=\genDE{x}}{\lnot{(\rfvar \land \ivr)}}}{\ivr}}}
\end{sequentdeduction}
}%

From the right premise after~\irref{SARef}, the derivation continues using~\irref{Prog} with $\rgvar \mnodefequiv \lnot{\rsfvar}$, followed by~\irref{dW+DMP}.
The resulting left premise is (again) the left premise of~\irref{SPcQ}, while the resulting right premise is abbreviated \circled{1} and continued below:
{\footnotesizeoff%
\begin{sequentdeduction}[array]
\linfer[Prog]{
  \linfer[dW+DMP]{
    \lsequent{\Gamma}{\dbox{\pevolvein{\D{x}=\genDE{x}}{\lnot{(\rfvar \land \ivr)}}}{\rsfvar}}
  }
  {\lsequent{\Gamma}{\dbox{\pevolvein{\D{x}=\genDE{x}}{\lnot{\rfvar}}}{\rsfvar}}} !
  \circled{1}
}
  {\lsequent{\Gamma}{\ddiamond{\pevolve{\D{x}=\genDE{x}}}{\rfvar}}}
\end{sequentdeduction}
}%

The derivation continues from \circled{1} by intertwining proof ideas from~\rref{cor:higherdv} and~\rref{cor:boundedandcompact}.
First, compactness of the set characterized by $\rsfvar(x)$ implies that the formula
$\lexists{\varepsilon {>} 0}{A(\varepsilon)}$ where $A(\varepsilon) \mnodefequiv \lforall{x}{(\rsfvar(x) \limply \lied[k]{\genDE{x}}{p} \geq \varepsilon)}$ and the formula
$B \mnodefequiv \lforall{x}{(\rsfvar(x) \limply \lied[k]{\genDE{x}}{p} > 0)}$ are provably equivalent in first-order real arithmetic.
These facts are added to the assumptions similarly to the derivation of~\irref{SPc}.
The resulting right open premise is the right conjunct of the right premise of~\irref{SPcQ}:

{\footnotesizeoff%
\begin{sequentdeduction}[array]
  \linfer[cut]{
  \linfer[existsl]{
    \lsequent{\Gamma,\varepsilon > 0, A(\varepsilon)}{\ddiamond{\pevolve{\D{x}=\genDE{x}}}{\lnot{\rsfvar}}}
  }
  {\lsequent{\Gamma,\lexists{\varepsilon {>} 0}{A(\varepsilon)}}{\ddiamond{\pevolve{\D{x}=\genDE{x}}}{\lnot{\rsfvar}}}}
  !
    \linfer[qear]{
    \linfer[allr+implyr]{
      \lsequent{\rsfvar}{\lied[k]{\genDE{x}}{p} > 0}
    }
      \lsequent{}{B}
    }
    {\lsequent{}{\lexists{\varepsilon {>} 0}{A(\varepsilon)}}}
  }
  {\lsequent{\Gamma}{\ddiamond{\pevolve{\D{x}=\genDE{x}}}{\lnot{\rsfvar}}} }
\end{sequentdeduction}
}%

From the left premise, recall the derivation from~\rref{cor:higherdv} which introduces fresh variables for the initial values of the Lie derivatives with~\irref{cut+qear+existsl}.
The derivation continues similarly here, with the resulting antecedents abbreviated $\Gamma_0 \mnodefequiv \big(\Gamma,\ptermA=\ptermA_0, \dots, \lied[k-1]{\genDE{x}}{\ptermA} = \lied[k-1]{\genDE{x}}{\ptermA}_0\big)$.
Rule~\irref{dGt} is also used to add time variable $\timevar$ to the system of equations with initial value $\timevar=0$.

{\footnotesizeoff%
\begin{sequentdeduction}[array]
  \linfer[cut+qear+existsl]{
  \linfer[dGt]{
    \lsequent{\Gamma_0, \varepsilon > 0, A(\varepsilon),\timevar=0}{\ddiamond{\pevolve{\D{x}=\genDE{x},\D{\timevar}=1}}{\lnot{\rsfvar}}}
  }
    {\lsequent{\Gamma_0, \varepsilon > 0, A(\varepsilon)}{\ddiamond{\pevolve{\D{x}=\genDE{x}}}{\lnot{\rsfvar}}}}
  }
  {\lsequent{\Gamma,\varepsilon > 0, A(\varepsilon)}{\ddiamond{\pevolve{\D{x}=\genDE{x}}}{\lnot{\rsfvar}}}}
\end{sequentdeduction}
}%

Recall from~\rref{cor:boundedandcompact} that the formula $\rrfvar(\ptermA_1) \mnodefequiv \lforall{x}{(\rsfvar(x) \limply \ptermA \leq \ptermA_1)}$ can be added to the assumptions using~\irref{cut+qear+existsl}, for some fresh variable $\ptermA_1$ symbolically representing the maximum value of $\ptermA$ on the compact set characterized by $\rsfvar$:
{\footnotesizeoff%
\begin{sequentdeduction}[array]
  \linfer[cut+qear+existsl]{
    \lsequent{\Gamma_0, \varepsilon > 0, A(\varepsilon),\timevar=0,\rrfvar(\ptermA_1)}{\ddiamond{\pevolve{\D{x}=\genDE{x},\D{\timevar}=1}}{\lnot{\rsfvar}}}
  }
  {\lsequent{\Gamma_0, \varepsilon > 0, A(\varepsilon),\timevar=0}{\ddiamond{\pevolve{\D{x}=\genDE{x},\D{\timevar}=1}}{\lnot{\rsfvar}}}}
\end{sequentdeduction}
}%

One last arithmetic cut is needed to set up the sequence of differential cuts~\rref{eq:integration}.
Recall the polynomial $\ptermB(\timevar)$ from~\rref{eq:integration} is eventually positive for sufficiently large values of $\timevar$ because its leading coefficient is strictly positive.
The same applies to the polynomial $\ptermB(\timevar)-\ptermA_1$ so~\irref{cut+qear} (and Skolemizing with~\irref{existsl}) adds the formula $\lforall{\timevar > \timevar_1} {(\ptermB(\timevar) - \ptermA_1 > 0)}$ to the assumptions:
{\footnotesizeoff%
\begin{sequentdeduction}[array]
  \linfer[cut+qear+existsl]{
    \lsequent{\Gamma_0, \varepsilon > 0, A(\varepsilon),\timevar=0,\rrfvar(\ptermA_1),\lforall{\timevar > \timevar_1} {\ptermB(\timevar) - \ptermA_1 > 0}}{\ddiamond{\pevolve{\D{x}=\genDE{x},\D{\timevar}=1}}{\lnot{\rsfvar}}}
  }
  {\lsequent{\Gamma_0, \varepsilon > 0, A(\varepsilon),\timevar=0,\rrfvar(\ptermA_1)}{\ddiamond{\pevolve{\D{x}=\genDE{x},\D{\timevar}=1}}{\lnot{\rsfvar}}}
}
\end{sequentdeduction}
}%

Once all the arithmetic cuts are in place, an additional cut introduces a (bounded) sufficient duration assumption $\ddiamond{\pevolve{\D{x}=\genDE{x},\D{\timevar}=1}}{(\ptermB(t) - \ptermA_1 > 0 \lor \lnot{\rsfvar})}$ (antecedents temporarily abbreviated with $\dots$ for brevity).
The cut premise, abbreviated~\circled{1}, is proved further below:
{\footnotesizeoff%
\begin{sequentdeduction}[array]
  \linfer[cut]{
    \lsequent{\Gamma_0, \dots, \ddiamond{\pevolve{\D{x}=\genDE{x},\D{\timevar}=1}}{(\ptermB(t) - \ptermA_1 > 0 \lor \lnot{\rsfvar})}}{\ddiamond{\pevolve{\D{x}=\genDE{x},\D{\timevar}=1}}{\lnot{\rsfvar}}} \quad
    \circled{1}
  }
  {\lsequent{\Gamma_0, \varepsilon > 0, A(\varepsilon),\timevar=0,\rrfvar(\ptermA_1),\lforall{\timevar > \timevar_1} {(\ptermB(\timevar) - \ptermA_1 > 0)}}{\ddiamond{\pevolve{\D{x}=\genDE{x},\D{\timevar}=1}}{\lnot{\rsfvar}}}}
\end{sequentdeduction}
}%

From the open premise on the left, axiom~\irref{Prog} is used with $\rgvar \mnodefequiv \ptermB(t) - \ptermA_1 > 0 \lor \lnot{\rsfvar}$:
{\footnotesizeoff%
\begin{sequentdeduction}[array]
  \linfer[Prog]{
    \lsequent{\Gamma_0, \varepsilon > 0, A(\varepsilon),\timevar=0,\rrfvar(\ptermA_1)}{\dbox{\pevolvein{\D{x}=\genDE{x},\D{\timevar}=1}{\rsfvar}}{(\ptermB(t) - \ptermA_1 \leq 0 \land \rsfvar)}}
  }
  {\lsequent{\Gamma_0,\dots, \ddiamond{\pevolve{\D{x}=\genDE{x},\D{\timevar}=1}}{(\ptermB(t) - \ptermA_1 > 0 \lor \lnot{\rsfvar})}}{\ddiamond{\pevolve{\D{x}=\genDE{x},\D{\timevar}=1}}{\lnot{\rsfvar}}}}
\end{sequentdeduction}
}%

Next, a monotonicity step~\irref{MbW} simplifies the postcondition using the (constant) assumption $\rrfvar(\ptermA_1)$ and the domain constraint $\rsfvar$:

{\footnotesizeoff%
\begin{sequentdeduction}[array]
  \linfer[MbW]{
    \lsequent{\Gamma_0, \timevar=0, A(\varepsilon)}{\dbox{\pevolvein{\D{x}=\genDE{x},\D{\timevar}=1}{\rsfvar}}{\ptermA \geq \ptermB(t)}}
  }
    {\lsequent{\Gamma_0, \varepsilon > 0, A(\varepsilon),\timevar=0,\rrfvar(\ptermA_1)}{\dbox{\pevolvein{\D{x}=\genDE{x},\D{\timevar}=1}{\rsfvar}}{(\ptermB(t) - \ptermA_1 \leq 0 \land \rsfvar)}}}
\end{sequentdeduction}
}%

The derivation closes using the chain of differential cuts from~\rref{eq:integration}.
In the first~\irref{dC} step, the (constant) assumption $A(\varepsilon)$ is used, see the derivation labeled \circled{$\star$} immediately below:
{\footnotesizeoff%
\begin{sequentdeduction}[array]
  \linfer[dC]{
    \lsequent{\Gamma_0, \timevar=0}{\dbox{\pevolvein{\D{x}=\genDE{x},\D{\timevar}=1}{\rsfvar \land \lied[k-1]{\genDE{x}}{\ptermA} \geq \lied[k-1]{\genDE{x}}{\ptermA}_0 + \constt{\varepsilon}\timevar}}{\ptermA \geq \ptermB(t)}} \quad \circled{$\star$}
  }
  {\lsequent{\Gamma_0, \timevar=0, A(\varepsilon)}{\dbox{\pevolvein{\D{x}=\genDE{x},\D{\timevar}=1}{\rsfvar}}{\ptermA \geq \ptermB(t)}}}
\end{sequentdeduction}
}%
From~\circled{$\star$}:
{\footnotesizeoff\renewcommand{\arraystretch}{1.4}%
\begin{sequentdeduction}[array]
  \linfer[dIcmp]{
  \linfer[qear]{
    \lclose
  }
    {\lsequent{A(\varepsilon), \rsfvar}{\lied[k]{\genDE{x}}{\ptermA} \geq \constt{\varepsilon} }}
  }
  {\lsequent{\Gamma_0, \timevar=0,A(\varepsilon)}{\dbox{\pevolvein{\D{x}=\genDE{x},\D{\timevar}=1}{\rsfvar}}{\lied[k-1]{\genDE{x}}{\ptermA} \geq \lied[k-1]{\genDE{x}}{\ptermA}_0 + \constt{\varepsilon}\timevar}}}
\end{sequentdeduction}
}%

Subsequent~\irref{dC+dIcmp} steps are similar to the derivation in~\rref{cor:higherdv}:
{\footnotesizeoff\renewcommand{\arraystretch}{1.5}%
\begin{sequentdeduction}[array]
    \linfer[dC+dIcmp]{
    \linfer[dC+dIcmp]{
    \linfer[dIcmp]{
      \lclose
    }
      {\lsequent{\Gamma_0, \timevar=0}{\dbox{\pevolvein{\D{x}=\genDE{x},\D{\timevar}=1}{\dots \land \lied[1]{\genDE{x}}{\ptermA} \geq \lied[1]{\genDE{x}}{\ptermA}_0 + \dots + \constt{\varepsilon} \frac{\timevar^{k-1}}{(k-1)!} }}{\ptermA \geq \ptermB(\timevar)}}}
    }
      {\vdots}
    }
    {\lsequent{\Gamma_0, \timevar=0}{\dbox{\pevolvein{\D{x}=\genDE{x},\D{\timevar}=1}{\rsfvar \land \lied[k-1]{\genDE{x}}{\ptermA} \geq \lied[k-1]{\genDE{x}}{\ptermA}_0 + \constt{\varepsilon}\timevar }}{\ptermA \geq \ptermB(\timevar)}}}
\end{sequentdeduction}
}%

From premise \circled{1}, a monotonicity step~\irref{MdW} rephrases the postcondition of the cut using the assumption $\lforall{\timevar > \timevar_1} {(\ptermB(\timevar) - \ptermA_1 > 0)}$.
Axiom~\irref{BEx} finishes the derivation since formula $\rsfvar(x)$ characterizes a compact (and hence bounded) set:
{\footnotesizeoff%
\begin{sequentdeduction}[array]
  \linfer[MdW]{
  \linfer[BEx]{
    \lclose
  }
    {\lsequent{}{\ddiamond{\pevolve{\D{x}=\genDE{x},\D{\timevar}=1}}{(\timevar > \timevar_1 \lor \lnot{\rsfvar})}}}
  }
  {\lsequent{\lforall{\timevar > \timevar_1} {(\ptermB(\timevar) - \ptermA_1 > 0)}}{\ddiamond{\pevolve{\D{x}=\genDE{x},\D{\timevar}=1}}{(\ptermB(t) - \ptermA_1 > 0 \lor \lnot{\rsfvar})}}}
\\[-\normalbaselineskip]\tag*{\qedhere}
\end{sequentdeduction}
}%
\end{proof}

\begin{proof}[\textbf{Proof of \rref{cor:prq}\hypertarget{proof:proof23}}]
Rule~\irref{PRQ} is derived from~\irref{SPcQ} with $\rsfvar \mnodefequiv \ivr \land \lnot{\rfvar}$ and $k\mnodefeq1$ because formula $\ivr \land \lnot{\rfvar}$ is assumed to characterize a compact set, as required by rule~\irref{SPcQ}:
{\footnotesizeoff%
\begin{sequentdeduction}[array]
\linfer[SPc]{
  \linfer[MbW]{
    \lsequent{\Gamma}{\dbox{\pevolvein{\D{x}=\genDE{x}}{\lnot{(\rfvar \land \ivr)}}}{\ivr}}
  }
  {\lsequent{\Gamma}{\dbox{\pevolvein{\D{x}=\genDE{x}}{\lnot{(\rfvar \land \ivr)}}}{(\ivr \land \lnot{\rfvar})}}} !
  \linfer[]{
    \lsequent{\ivr,\lnot{\rfvar}}{\lied[]{\genDE{x}}{p} > 0}
  }
  {\lsequent{\ivr,\lnot{\rfvar}}{\ivr \land \lied[]{\genDE{x}}{p} > 0}}
}
  {\lsequent{\Gamma}{\ddiamond{\pevolvein{\D{x}=\genDE{x}}{\ivr}}{\rfvar}} }
\end{sequentdeduction}
}%

The~\irref{MbW} step uses the propositional tautology $\lnot{(\rfvar \land \ivr)} \land \ivr \limply \ivr \land \lnot{\rfvar}$.
\end{proof}

\subsection{Proofs for ODE Liveness Proofs in Practice}
\label{app:implementationproofs}

\begin{proof}[\textbf{Proof of \rref{cor:atomicdvcmpexist}\hypertarget{proof:proof24}}]
The derivation starts with a \irref{cut} of the sole premise of~\irref{dVcmpE} (the left premise below).
The existentially bound variable is renamed to $\delta$ throughout the derivation for clarity.
After Skolemizing (with~\irref{existsl}), rule~\irref{dVcmp} is used with $\constt{\varepsilon} \mnodefeq \delta$.
The universally quantified antecedent is constant for the ODE $\D{x}=\genDE{x}$ so it is soundly kept across the application of~\irref{dVcmp}.
This proof is completed propositionally~\irref{alll+implyl}.
{\footnotesizeoff%
\renewcommand{\linferPremissSeparation}{\hspace{5pt}}%
\begin{sequentdeduction}[array]
\linfer[cut]{
   \lsequent{\Gamma}{\lexists{\delta>0}{\lforall{x}{\big(\lnot{(\ptermA \cmp 0)} \limply \lied[]{\genDE{x}}{\ptermA}\geq \delta} \big)}} !
   \linfer[existsl+andl]{
   \linfer[dVcmp]{
   \linfer[alll+implyl]{
      \lclose
    }
   {\lsequent{\lforall{x}{\big(\lnot{(\ptermA \cmp 0)} \limply \lied[]{\genDE{x}}{\ptermA}\geq \delta\big)}, \lnot{(\ptermA \cmp 0)}}{ \lied[]{\genDE{x}}{\ptermA}\geq \delta}}
   }
    {\lsequent{\delta>0,\lforall{x}{\big(\lnot{(\ptermA \cmp 0)} \limply \lied[]{\genDE{x}}{\ptermA}\geq \delta\big)}}{\ddiamond{\pevolvein{\D{x}=\genDE{x}}{\ivr}}{\ptermA \cmp 0}}}
  }
  {\lsequent{\lexists{\delta>0}{\lforall{x}{\big(\lnot{(\ptermA \cmp 0)} \limply \lied[]{\genDE{x}}{\ptermA}\geq \delta} \big)}}{\ddiamond{\pevolvein{\D{x}=\genDE{x}}{\ivr}}{\ptermA \cmp 0}}}
}
  {\lsequent{\Gamma}{\ddiamond{\pevolvein{\D{x}=\genDE{x}}{\ivr}}{\ptermA \cmp 0}} }
\\[-\normalbaselineskip]\tag*{\qedhere}
\end{sequentdeduction}
}%
\end{proof}

\begin{proof}[\textbf{Proof of \rref{cor:semialgdv}\hypertarget{proof:proof25}}]
Assume that formulas $\rfvar,\rgvar_\rfvar$ are in normal form as in~\rref{cor:semialgdv}.
Rule~\irref{dV} is derived first since rule~\irref{dVE} follows from~\irref{dV} as a corollary.
The derivation of rule~\irref{dV} uses variable $b$ as a symbolic lower bound on the initial values of all terms $\ptermA_{ij},\ptermB_{ij}$ appearing in formula $\rfvar$.
The formula $\lexists{b}{\landfold_{i=0}^{M} \Big(\landfold_{j=0}^{m(i)} \ptermA_{ij} \geq b \land \landfold_{j=0}^{n(i)} \ptermB_{ij} \geq b\Big)}$ is a valid formula of real arithmetic and is proved as a~\irref{cut} by~\irref{qear} because $\rfvar$ is a finite formula so there exists a lower bound $b$ smaller than the value all of the terms $\ptermA_{ij}, \ptermB_{ij}$.

The derivation starts similarly to~\irref{dVcmp} by introducing fresh variables $b$ (for the bound above), and $i$ representing the multiplicative inverse of $\constt{\varepsilon}$ using arithmetic cuts~\irref{cut+qear}.
It then Skolemizes (\irref{existsl}) and uses \irref{dGt} to introduce a fresh time variable to the system of differential equations:
{\footnotesizeoff%
\begin{sequentdeduction}[array]
  \linfer[cut+qear]{
  \linfer[existsl]{
  \linfer[dGt]{
    \lsequent{\Gamma, \constt{\varepsilon} > 0, \landfold_{i=0}^{M} \Big(\landfold_{j=0}^{m(i)} \ptermA_{ij} \geq b \land \landfold_{j=0}^{n(i)} \ptermB_{ij} \geq b\Big), i\varepsilon() = 1, \timevar=0}{\ddiamond{\pevolve{\D{x}=\genDE{x},\D{\timevar}=1}}{\rfvar}}
  }
    {\lsequent{\Gamma, \constt{\varepsilon} > 0, \landfold_{i=0}^{M} \Big(\landfold_{j=0}^{m(i)} \ptermA_{ij} \geq b \land \landfold_{j=0}^{n(i)} \ptermB_{ij} \geq b\Big), i\varepsilon() = 1}{\ddiamond{\pevolve{\D{x}=\genDE{x}}}{\rfvar}}}
  }
  {\lsequent{\Gamma, \constt{\varepsilon} > 0, \lexists{b}{\landfold_{i=0}^{M} \Big(\landfold_{j=0}^{m(i)} \ptermA_{ij} \geq b \land \landfold_{j=0}^{n(i)} \ptermB_{ij} \geq b\Big)}, \lexists{i}{(i\varepsilon() = 1)}}{\ddiamond{\pevolve{\D{x}=\genDE{x}}}{\rfvar}}}
  }
  {\lsequent{\Gamma, \constt{\varepsilon} > 0}{\ddiamond{\pevolve{\D{x}=\genDE{x}}}{\rfvar}} }
\end{sequentdeduction}
}%

Next, the refinement axiom~\irref{Prog} is used with $\rgvar \mnodefequiv (b + \constt{\varepsilon} \timevar > 0 )$.
This yields two premises, the right of which is proved by~\irref{GEx} (after monotonic rephrasing with~\irref{qear+MdW}) because the ODE $\D{x}=\genDE{x}$ is assumed to have provable global solutions.
The left premise from~\irref{Prog} is abbreviated \circled{1} and continued below.
{\footnotesizeoff%
\begin{sequentdeduction}[array]
  \linfer[Prog]{
  \linfer[qear+MdW]{
  \linfer[GEx]{
    \lclose
  }
    {\lsequent{\Gamma}{\ddiamond{\pevolve{\D{x}=\genDE{x},\D{\timevar}=1}}{ \timevar > -i b}}}
  }
  {\circled{1} \qquad\qquad
    \lsequent{\Gamma,\constt{\varepsilon} > 0, i\varepsilon() = 1}{\ddiamond{\pevolve{\D{x}=\genDE{x},\D{\timevar}=1}}{ (b + \constt{\varepsilon} \timevar > 0 )}}}
  }
  {\lsequent{\Gamma, \constt{\varepsilon} > 0, \landfold_{i=0}^{M} \Big(\landfold_{j=0}^{m(i)} \ptermA_{ij} \geq b \land \landfold_{j=0}^{n(i)} \ptermB_{ij} \geq b\Big), i\varepsilon() = 1, \timevar=0}{\ddiamond{\pevolve{\D{x}=\genDE{x},\D{\timevar}=1}}{\rfvar}}}
\end{sequentdeduction}
}%

Continuing from premise \circled{1}, monotonicity strengthens the postcondition from $b + \constt{\varepsilon} \timevar \leq 0$ to $\rgvar_\rfvar$ under the domain constraint assumption $\lnot{\rfvar}$.
This strengthening works because, assuming that $\lnot{\rfvar}$ and $\rgvar_\rfvar$ are true in a given state, then propositionally, at least one of the following pairs (each pair listed horizontally) of sub-formulas of $\lnot{\rfvar}$ and $\rgvar_\rfvar$ for some indices $i,j$ is true in that state:
\[ p_{ij} < 0 \qquad\qquad p_{ij} - (b + \constt{\varepsilon}t) \geq 0 \]
\[ q_{ij} \leq 0 \qquad\qquad q_{ij} - (b + \constt{\varepsilon}t) \geq 0 \]

Either pair of formulas imply that formula $b + \constt{\varepsilon}t \leq 0$ is also true in that state, so the strengthening is proved by~\irref{MbW+qear}.
Next, a~\irref{cut+qear} step adds the formula $\rgvar_\rfvar$ to the antecedents using the assumptions $ \landfold_{i=0}^{M} \Big(\landfold_{j=0}^{m(i)} \ptermA_{ij} \geq b \land \landfold_{j=0}^{n(i)} \ptermB_{ij} \geq b\Big)$ and $t=0$.
Rule~\irref{sAIQ} yields the sole premise of rule~\irref{dV} because $\rgvar_\rfvar$ characterizes a closed set~\cite{DBLP:journals/jacm/PlatzerT20}.
{\footnotesizeoff%
\begin{sequentdeduction}[array]
  \linfer[MbW+qear]{
  \linfer[cut+qear]{
  \linfer[sAIQ]{
    \lsequent{\lnot{\rfvar}, \sigliedsai{\genDE{x}}{(\lnot{\rfvar})}, \rgvar_\rfvar}{\sigliedsai{\genDE{x}}{(\rgvar_\rfvar)}}
  }
    {\lsequent{\rgvar_\rfvar}{\dbox{\pevolvein{\D{x}=\genDE{x},\D{\timevar}=1}{\lnot{\rfvar}}}{\rgvar_\rfvar}}}
  }
    {\lsequent{\Gamma, \landfold_{i=0}^{M} \Big(\landfold_{j=0}^{m(i)} \ptermA_{ij} \geq b \land \landfold_{j=0}^{n(i)} \ptermB_{ij} \geq b\Big), \timevar=0}{\dbox{\pevolvein{\D{x}=\genDE{x},\D{\timevar}=1}{\lnot{\rfvar}}}{\rgvar_\rfvar}}}
  }
  {\lsequent{\Gamma, \landfold_{i=0}^{M} \Big(\landfold_{j=0}^{m(i)} \ptermA_{ij} \geq b \land \landfold_{j=0}^{n(i)} \ptermB_{ij} \geq b\Big), \timevar=0}{\dbox{\pevolvein{\D{x}=\genDE{x},\D{\timevar}=1}{\lnot{\rfvar}}}{\big( b + \constt{\varepsilon} \timevar \leq 0 \big)}}}
\end{sequentdeduction}
}%

Rule~\irref{dVE} is derived from rule~\irref{dV} similarly to the derivation of rule~\irref{dVcmpE} from rule~\irref{dVcmp}.
The derivation starts with a \irref{cut} of the sole premise of~\irref{dVE} (the left premise below).
The existentially bound variable is renamed to $\delta$ throughout the derivation for clarity.
The right premise is abbreviated \circled{2} and shown below.
{\footnotesizeoff%
\begin{sequentdeduction}[array]
\linfer[cut]{
   \lsequent{\Gamma}{\lexists{\delta>0}{\lforall{b}{\lforall{t}{\lforall{x}{\big(\lnot{\rfvar} \land \sigliedsai{\genDE{x}}{(\lnot{\rfvar})} \land \rgvar_\rfvar \limply \sigliedsai{\genDE{x}}{(\rgvar_\rfvar)}\big)}}}}} !
   \circled{2}
   }
  {\lsequent{\Gamma}{\ddiamond{\pevolvein{\D{x}=\genDE{x}}{\ivr}}{\rfvar}} }
\end{sequentdeduction}
}%

From \circled{2}, after Skolemizing (with~\irref{existsl}), rule~\irref{dV} is used with $\constt{\varepsilon} \mnodefeq \delta$.
The universally quantified antecedent is constant for the ODE $\D{x}=\genDE{x}$ and the universal quantification over variables $b, \timevar$ ensure that those variables are fresh in the rest of the sequent so the antecedent is soundly kept across the application of rule~\irref{dV}.
This proof is completed propositionally~\irref{alll+implyl+andl}.
{\footnotesizeoff%
\begin{sequentdeduction}[array]
\linfer[existsl+andl]{
   \linfer[dV]{
   \linfer[alll+implyl+andl]{
      \lclose
    }
   {\lsequent{\lforall{b}{\lforall{t}{\lforall{x}{\big(\lnot{\rfvar} \land \sigliedsai{\genDE{x}}{(\lnot{\rfvar})} \land \rgvar_\rfvar \limply \sigliedsai{\genDE{x}}{(\rgvar_\rfvar)}\big)}}}, \lnot{\rfvar}, \sigliedsai{\genDE{x}}{(\lnot{\rfvar})}, \rgvar_\rfvar}{\sigliedsai{\genDE{x}}{(\rgvar_\rfvar)}}}
   }
    {\lsequent{\delta>0,\lforall{b}{\lforall{t}{\lforall{x}{\big(\lnot{\rfvar} \land \sigliedsai{\genDE{x}}{(\lnot{\rfvar})} \land \rgvar_\rfvar \limply \sigliedsai{\genDE{x}}{(\rgvar_\rfvar)}\big)}}}}{\ddiamond{\pevolvein{\D{x}=\genDE{x}}{\ivr}}{\rfvar}}}
  }
  {\lsequent{\lexists{\delta>0}{\lforall{b}{\lforall{t}{\lforall{x}{\big(\lnot{\rfvar} \land \sigliedsai{\genDE{x}}{(\lnot{\rfvar})} \land \rgvar_\rfvar \limply \sigliedsai{\genDE{x}}{(\rgvar_\rfvar)}\big)}}}}}{\ddiamond{\pevolvein{\D{x}=\genDE{x}}{\ivr}}{\rfvar}}}
  \\[-\normalbaselineskip]\tag*{\qedhere}
  \end{sequentdeduction}
}%
\end{proof}

\begin{proof}[\textbf{Proof of \rref{cor:closeddomref}\hypertarget{proof:proof27}}]
The derivation of rule~\irref{cRef} is seemingly straightforward using axiom~\irref{CRef} followed by rule~\irref{Enc} on the resulting middle premise.
There is a minor subtlety to address because the formula $\strictineq{\ivr}$ (with strict inequalities replacing non-strict ones in $\ivr$) is only a syntactic \emph{under-approximation} of the interior of the set characterized by $\ivr$, and so the axiom~\irref{CRef} does \emph{not} immediately apply as stated.
For example, formula $x < x$ characterizes the empty set, while the formula $x \leq x$ characterizes the set of all states, whose interior is also the set of all states.
However, since $\ivr$ is a semialgebraic formula, there is a computable quantifier-free formula $\interior{\ivr}$ that exactly characterizes its topological interior~\cite{Bochnak1998} which can be used with~\irref{CRef} in the syntactic derivation below.

The derivation starts with a~\irref{cut} of the formula $\ivr$ which yields the leftmost premise of rule~\irref{cRef}.
This is followed with~\irref{DX}, which adds formula $\lnot{\rfvar}$ to the antecedents because there is nothing to prove if both formulas $\ivr$ and $\rfvar$ are already true initially.
The derivation then uses~\irref{CRef} with the computable formula $\interior{\ivr}$ characterizing the topological interior of formula $\ivr$.
This yields two premises, the right of which corresponds to the rightmost premise of rule~\irref{cRef}.
From the resulting left premise (with postcondition $\interior{\ivr}$), an~\irref{MbW+qear} monotonicity step strengthens the postcondition because $\strictineq{\ivr} \limply \interior{\ivr}$ is a provable formula of real arithmetic.
Rule~\irref{Enc} completes the derivation because formula $\strictineq{\ivr}$ is formed from finite conjunctions and disjunctions of strict inequalities, and $\relaxineq{(\strictineq{\ivr})}$ is syntactically equal to $\ivr$ by definition.
{\footnotesizeoff\renewcommand{\arraystretch}{1.4}%
\begin{sequentdeduction}[array]
  \linfer[cut]{
  \lsequent{\Gamma}{\ivr} !
  \linfer[DX]{
  \linfer[CRef]{
    \linfer[MbW+qear]{
    \linfer[Enc]{
      \lsequent{\Gamma}{\dbox{\pevolvein{\D{x}=\genDE{x}}{\rrfvar \land \lnot{\rfvar} \land \ivr}}{\strictineq{\ivr}}}
    }
      {\lsequent{\Gamma,\ivr}{\dbox{\pevolvein{\D{x}=\genDE{x}}{\rrfvar \land \lnot{\rfvar}}}{\strictineq{\ivr}}}}
    }
    {\lsequent{\Gamma,\ivr}{\dbox{\pevolvein{\D{x}=\genDE{x}}{\rrfvar \land \lnot{\rfvar}}}{\interior{\ivr}}}}
   !
   \lsequent{\Gamma}{\ddiamond{\pevolvein{\D{x}=\genDE{x}}{\rrfvar}}{\rfvar}}
  }
    {\lsequent{\Gamma,\ivr, \lnot{\rfvar}}{\ddiamond{\pevolvein{\D{x}=\genDE{x}}{\ivr}}{\rfvar}}}
  }
    {\lsequent{\Gamma,\ivr}{\ddiamond{\pevolvein{\D{x}=\genDE{x}}{\ivr}}{\rfvar}}}
  }
  {\lsequent{\Gamma}{\ddiamond{\pevolvein{\D{x}=\genDE{x}}{\ivr}}{\rfvar}} }
\\[-\normalbaselineskip]\tag*{\qedhere}
\end{sequentdeduction}
}%
\end{proof}

\section{Counterexamples}

\newcommand{\uvar}{u}
\newcommand{\vvar}{v}
\label{app:counterexamples}
This appendix gives explicit counterexamples to illustrate the soundness errors identified in Sections~\ref{sec:nodomconstraint} and~\ref{sec:withdomconstraint}.

\subsection{Finite-Time Blow Up}
The soundness errors identified in~\rref{sec:nodomconstraint} all arise because of incorrect handling of the fact that solutions may blow up in finite time.
This phenomenon is studied in detail in~\rref{sec:globexist}, and it is illustrated by $\exnonlinear$~\rref{eq:exnonlinear}, see~\rref{fig:odeexamples}, or $\exblowup$~\rref{eq:exblowup}, see~\rref{ex:velocity}.
The following is a counterexample for the original presentation of~\irref{dVeq} (and~\irref{TT+dVeqQ+TTQ})~\cite{DBLP:conf/emsoft/TalyT10}.
Similar counterexamples can be constructed for~\cite[Remark 3.6]{DBLP:journals/siamco/PrajnaR07} and for the original presentation of~\irref{RS+RSQ}~\cite{DBLP:journals/siamco/RatschanS10}.

\begin{counterexample}
\irlabel{dVeqbad|dV$_=$\usebox{\Lightningval}}
Consider rule~\irref{dVeq} \emph{without} the restriction that the ODE has provable global solutions.
This unrestricted rule, denoted~\irref{dVeqbad}, is unsound as shown by the following derivation using it with $\constt{\varepsilon}{\mnodefeq}1$:
{\footnotesizeoff
\begin{sequentdeduction}[array]
  \linfer[dVeqbad]{
  \linfer[qear]{
    \lclose
  }
    {\lsequent{\vvar - 2 < 0}{ 1 \geq 1}}
  }
  {\lsequent{\vvar - 2 \leq 0}{\ddiamond{\pevolve{\D{\uvar}=\uvar^2,\D{\vvar}=1}}{\vvar - 2 = 0}}}
\end{sequentdeduction}
}%

The conclusion of this derivation is not valid.
Consider the initial state $\iget[state]{\I}$ with values $\iget[state]{\I}(\uvar)=1$ and $\iget[state]{\I}(\vvar)=0$.
The explicit solution of the ODE from $\iget[state]{\I}$ is given by $\uvar(t) = \frac{1}{1-t}, \vvar(t) = t$ for $t \in [0,1)$.
This solution \emph{does not exist} beyond the time interval $[0,1)$ because the $\uvar$-coordinate asymptotically approaches $\infty$, i.e., blows up, as time approaches $t=1$.
It is impossible to reach a state satisfying $\vvar-2=0$ from $\iget[state]{\I}$ along this solution since at least $2$ time units are required.

This counterexample further illustrates the difficulty in handling nonlinear ODEs.
Neither the precondition ($\vvar-2 \leq 0$) nor postcondition ($\vvar-2=0$) mention the variable $\uvar$, and the ODEs $\D{\uvar}=\uvar^2, \D{\vvar}=1$ do not depend on variables $\vvar,\uvar$ respectively, so it is tempting to disregard the variable $\uvar$ entirely.
Indeed, the liveness property $\vvar - 2 \leq 0 \limply \ddiamond{\pevolve{\D{\vvar}=1}}{\vvar - 2 = 0}$ is valid.
Yet, for liveness questions about the (original) ODE, $\D{\uvar}=\uvar^2, \D{\vvar}=1$, the two variables are inextricably linked through the time axis of solutions to the ODE.
\end{counterexample}

\subsection{Topological Considerations}
The soundness errors identified in~\rref{sec:withdomconstraint} arise because of incorrect topological reasoning in subtle cases where the topological boundaries of the sets characterized by the domain constraint and desired liveness postcondition intersect.
The original presentation of~\irref{dVcmpQ}~\cite{DBLP:journals/logcom/Platzer10} gives the following proof rule for atomic inequalities $\ptermA \cmp 0$.
For simplicity, assume that the ODE $\D{x}=\genDE{x}$ is globally Lipschitz continuous so that solutions exist for all time.
\[
\dinferenceRule[dVcmpQbad|dV$_\cmp\&$\usebox{\Lightningval}]{}
{\linferenceRule
  {
    \lsequent{\Gamma}{\dbox{\pevolvein{\D{x}=\genDE{x}}{\ptermA \leq 0}}{\ivr}} \quad
    \lsequent{\lnot{(\ptermA \cmp 0)}, \ivr}{\lied[]{\genDE{x}}{\ptermA}\geq \constt{\varepsilon}}
  }
  {\lsequent{\Gamma,\constt{\varepsilon} > 0}{\ddiamond{\pevolvein{\D{x}=\genDE{x}}{\ivr}}{\ptermA \cmp 0}} }
}{}
\]

Compared to~\irref{dVcmpQ}, this omits the assumption $\lnot{(\ptermA \cmp 0)}$, makes no topological assumptions on the domain constraint $\ivr$, and uses syntactic weak negation~\cite{DBLP:journals/logcom/Platzer10} for the domain constraint of its left premise.
The following two counterexamples show that the two assumptions are necessary.

\begin{counterexample}
Consider the following derivation using the unsound rule~\irref{dVcmpQbad} with $\constt{\varepsilon} \mnodefeq 1$:
{\footnotesizeoff
\begin{sequentdeduction}[array]
  \linfer[dVcmpQbad]{
    \linfer[dW+qear]{
      \lclose
    }
    {\lsequent{\uvar > 1}{\dbox{\pevolvein{\D{\uvar}=1}{\uvar \leq 0}}{\uvar \leq 1}}} !
    \linfer[qear]{
      \lclose
    }
    {\lsequent{\uvar < 0, \uvar \leq 1}{1 \geq 1}}
  }
  {\lsequent{\uvar > 1}{\ddiamond{\pevolvein{\D{\uvar}=1}{\uvar \leq 1}}{\uvar \geq 0}}}
\end{sequentdeduction}
}%

The conclusion of this derivation is not valid. In states where $\uvar > 1$ is true initially, the domain constraint is violated immediately so the diamond modality in the succedent is trivially false in those states.
\end{counterexample}

\begin{counterexample}[\cite{Sogokon16}]
This counterexample is adapted from~\cite[Example 142]{Sogokon16}, which has a minor typographical error (the sign of an inequality is flipped).
Consider the following derivation using the unsound rule~\irref{dVcmpQbad} with $\constt{\varepsilon} \mnodefeq 1$:
{\footnotesizeoff
\begin{sequentdeduction}[array]
  \linfer[dVcmpQbad]{
    \linfer[dW+qear]{
      \lclose
    }
    {\lsequent{}{\dbox{\pevolvein{\D{\uvar}=1}{\uvar \leq 1}}{\uvar \leq 1}}} !
    \linfer[qear]{
      \lclose
    }
    {\lsequent{\uvar \leq 1, \uvar \leq 1}{1 \geq 1}}
  }
  {\lsequent{}{\ddiamond{\pevolvein{\D{\uvar}=1}{\uvar \leq 1}}{\uvar > 1}}}
\end{sequentdeduction}
}%

The conclusion of this derivation is not valid and, in fact, unsatisfiable. The domain constraint $\uvar \leq 1$ and postcondition $\uvar > 1$ are contradictory so no solution can reach a state satisfying both simultaneously.
\end{counterexample}

The next two counterexamples are for the liveness arguments from~\cite[Corollary 1]{DBLP:conf/hybrid/PrajnaR05} and~\cite[Theorem 3.5]{DBLP:journals/siamco/PrajnaR07}.
For clarity, the original notation from~\cite[Theorem 3.5]{DBLP:journals/siamco/PrajnaR07} is used.
The following conjecture is quoted from~\cite[Theorem 3.5]{DBLP:journals/siamco/PrajnaR07}:

\begin{conjecture}
Consider the system $\D{x}=\genDE{x}$, with $f \in C(\reals^n,\reals^n)$. Let $\bigchi \subset \reals^n$, $\bigchi_0 \subseteq \bigchi$, and $\bigchi_r \subseteq \bigchi$ be bounded sets. If there exists a function $B \in C^1(\reals^n)$ satisfying:
\begin{align}
&B(x) \leq 0                              && \forall x \in \bigchi_0  \label{eq:init} \\
&B(x) > 0                                 && \forall x \in \closure{\bdr{\bigchi} \setminus \bdr{\bigchi_r}} \label{eq:unsoundbdr} \\
&\Dp[x]{B}f(x) < 0    && \forall x \in \closure{\bigchi \setminus \bigchi_r}
\end{align}

Then the eventuality property holds, i.e., for all initial conditions $x_0 \in \bigchi_0$, the trajectory $x(t)$ of the system starting at $x(0)=x_0$ satisfies $x(T) \in \bigchi_r$ and $x(t) \in \bigchi$ for all $t \in [0,T]$ for some $T \geq 0$.
The notation $\closure{\bigchi}$ (resp. $\bdr{\bigchi}$) denotes the topological closure (resp. boundary) of the set $\bigchi$.
\end{conjecture}

In~\cite[Corollary 1]{DBLP:conf/hybrid/PrajnaR05}, stronger conditions are required.
In particular, the sets $\bigchi_0,\bigchi_r,\bigchi$ are additionally required to be topologically open, and the inequality in~\rref{eq:init} is strict, i.e., $B(x) < 0$ instead of $B(x) \leq 0$.

The soundness errors in both of these liveness arguments stem from the condition~\rref{eq:unsoundbdr} being too permissive.
For example, notice that if the sets $\bdr{\bigchi}, \bdr{\bigchi_r}$ are equal then~\rref{eq:unsoundbdr} is vacuously true.
The first counterexample below applies for the requirements of~\cite[Theorem 3.5]{DBLP:journals/siamco/PrajnaR07}, while the second applies even for the more restrictive requirements of~\cite[Corollary 1]{DBLP:conf/hybrid/PrajnaR05}.

\begin{counterexample}
\label{cex:prajnarantzer1}
Let the system $\D{x}=\genDE{x}$ be $\D{\uvar}=0,\D{\vvar}=1$.
Let $\bigchi_r$ be the open unit disk characterized by $\uvar^2 + \vvar^2 < 1$, $\bigchi$ be the closed unit disk characterized by $\uvar^2+\vvar^2 \leq 1$, and $\bigchi_0$ be the single point characterized by $\uvar=0 \land \vvar=1$.
All of these sets are bounded.
Note that $\bdr{\bigchi} \setminus \bdr{\bigchi_r} = \emptyset$ since both topological boundaries are the unit circle $\uvar^2+\vvar^2=1$.
Let $B(\uvar,\vvar) \mnodefeq -\vvar$, so that $\Dp[x]{B}f(x) = \Dp[\uvar]{B}0 + \Dp[\vvar]{B}1 = -1 < 0$ and $B \leq 0$ on $\bigchi_0$.

All conditions of~\cite[Theorem 3.5]{DBLP:journals/siamco/PrajnaR07} are met but the eventuality property is false.
The trajectory from $\bigchi_0$ leaves $\bigchi$ immediately and never enters $\bigchi_r$.
This is visualized in~\rref{fig:cex} (Left).
\end{counterexample}

\begin{counterexample}
\label{cex:prajnarantzer2}
Let the system $\D{x}=\genDE{x}$ be $\D{\uvar}=0,\D{\vvar}=1$.
Let $\bigchi_r$ be the set characterized by the formula $\uvar^2 + \vvar^2 < 5 \land \vvar > 0$, $\bigchi$ be the set characterized by the formula $\uvar^2+\vvar^2 < 5 \land \vvar \neq 0$, and $\bigchi_0$ be the set characterized by the formula $\uvar^2+(\vvar+1)^2 < \frac{1}{2}$.
All of these sets are bounded and topologically open.
Let $B(\uvar,\vvar) \mnodefeq -\vvar +\uvar^2 - 2$, so that $\Dp[x]{B}f(x) = \Dp[\uvar]{B}0 + \Dp[\vvar]{B}1 = -1 < 0$, and $B < 0$ on $\bigchi_0$.
The set $\closure{\bdr{\bigchi} \setminus \bdr{\bigchi_r}}$ is characterized by formula $\uvar^2+\vvar^2=5 \land \vvar \leq 0$ and $B$ is strictly positive on this set.
These claims can be checked arithmetically, see~\rref{fig:cex} (Right) for a plot of the curve $B=0$.

All conditions of~\cite[Corollary 1]{DBLP:conf/hybrid/PrajnaR05} are met but the eventuality property is false.
Solutions starting in $\bigchi_0$ eventually enter $\bigchi_r$ but can only do so by leaving the domain constraint $\bigchi$ at $\vvar=0$, see~\rref{fig:cex} (Right).

\end{counterexample}

\begin{figure}
\centering
\includegraphics[width=0.27\textwidth]{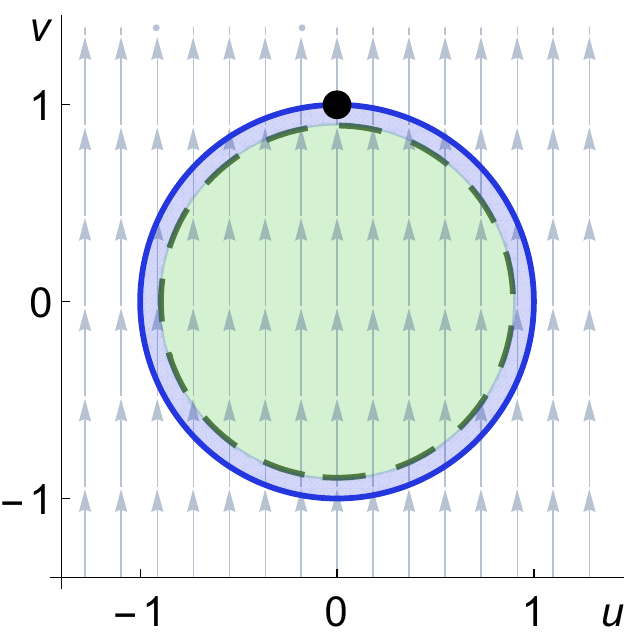}
\qquad\qquad
\includegraphics[width=0.27\textwidth]{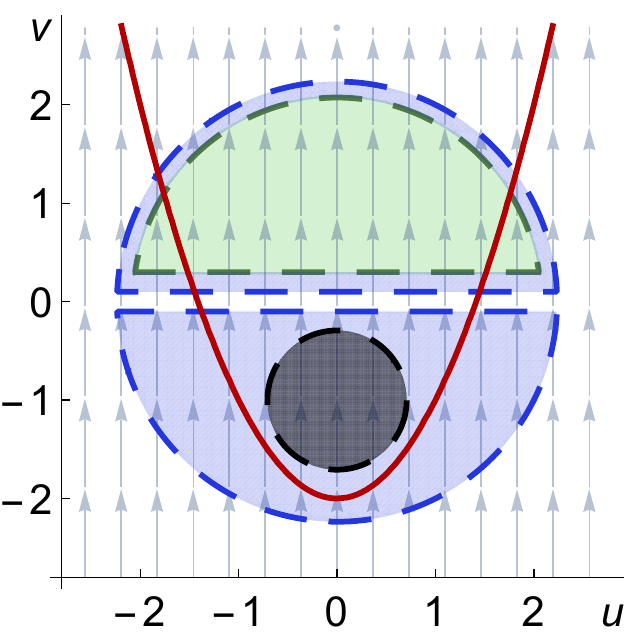}
\caption{\textbf{(Left)} Visualization of~\rref{cex:prajnarantzer1}. The solution from initial point $\uvar=0, \vvar=1$ ($\bigchi_0$, in black) leaves the domain unit disk ($\bigchi$, boundary in \bluec{blue}) immediately without ever reaching its interior ($\bigchi_r$, in \greenc{green} with dashed boundary). The interior is slightly shrunk for clarity in the visualization: the \bluec{blue} and \greenc{green} boundaries should actually overlap exactly.
\textbf{(Right)} Visualization of~\rref{cex:prajnarantzer2}.
Solutions from the initial set ($\bigchi_0$, in black with dashed boundary) eventually enter the goal region ($\bigchi_r$, in \greenc{green} with dashed boundary).
However, the domain ($\bigchi$, in \bluec{blue} with dashed boundary) shares an (open) boundary with $\bigchi_r$ at $\vvar=0$ which solutions are not allowed to cross.
As before, the sets are slightly shrunk for clarity in the visualization: the \bluec{blue} and \greenc{green} boundaries should actually overlap exactly.
The level curve $B = 0$ is plotted in \redc{red}.
All points above the curve satisfy $B < 0$, while all points below it satisfy $B > 0$.}
\label{fig:cex}
\end{figure}

\paragraph{Acknowledgments.} The authors thank the anonymous reviewers for their insightful comments and feedback.
The authors also thank members of the Logical Systems Lab at Carnegie Mellon University for their feedback on the \KeYmaeraX implementation, and Katherine Cordwell, Frank Pfenning, Andrew Sogokon, and the FM'19 anonymous reviewers for their feedback on the earlier conference version.
This material is based upon work supported by the Alexander von Humboldt Foundation and the AFOSR under grant number FA9550-16-1-0288.
The first author was also supported by A*STAR, Singapore.

The views and conclusions contained in this document are those of the authors and should not be interpreted as representing the official policies, either expressed or implied, of any sponsoring institution, the U.S. government or any other entity.

\bibliographystyle{halpha}
\bibliography{paper}

\end{document}